\RequirePackage{scrlfile}          %
\PreventPackageFromLoading{natbib} %
\PassOptionsToPackage{dvipsnames}{xcolor}
\PassOptionsToPackage{final}{microtype}
\PassOptionsToPackage{final}{hyperref}
\documentclass[nonatbib]{elsarticle}

\usepackage{geometry}
\usepackage{etoolbox}

\usepackage[utf8]{inputenc}

\newcommand{\ms}{\mathsf}
\newcommand{\mb}{\mathbf}
\newcommand{\mi}{\mathit}
\newcommand{\mt}{\mathtt}
\newcommand{\mc}{\mathcal}

\usepackage[dvipsnames]{xcolor}
\usepackage{color}
\newcommand{\red}[1]{{\color{Maroon}{#1}}}
\newcommand{\blue}[1]{{\color{MidnightBlue}{#1}}}

\usepackage{hyperref}

\usepackage{adjustbox}

\usepackage{appendix}

\usepackage[canadian]{babel}

\makeatletter
\let\c@author\relax
\makeatother
\usepackage[backend=biber,style=numeric-comp,mincrossrefs=10,sortcites=true]{biblatex}
\DeclareNameAlias{author}{family-given/given-family}
\DeclareFieldFormat[book,thesis]{title}{\mkbibquote{#1\isdot}}
\renewbibmacro*{volume+number+eid}{%
  \printfield{volume}%
  \setunit*{\addnbthinspace}%
  \printfield{eid}}
\DeclareFieldFormat[article,book]{volume}{\mkbibbold{#1}}
\AtEveryBibitem{%
  \ifentrytype{article}{%
    \clearfield{day}%
    \clearfield{month}%
    \clearfield{endday}%
    \clearfield{endmonth}%
    \clearfield{endyear}%
  }}
\renewbibmacro{in:}{%
  \ifentrytype{article}{}{\printtext{\bibstring{in}\intitlepunct}}}
\newbibmacro*{publisher+location+date}{%
  \printlist{publisher}%
  \setunit*{\addcomma\space}%
  \printlist{location}%
  \setunit*{\addcomma\space}%
  \usebibmacro{date}%
  \newunit}
\NewBibliographyString{extendedof}
\DefineBibliographyStrings{english}{extendedof={},}
\newbibmacro*{related:extendedof}[1]{}
\DeclareFieldFormat{relatedstring:extendedof}{}
\usepackage{bookmark}

\usepackage{booktabs}

\usepackage{amsmath}
\usepackage{amsthm}
\usepackage{amssymb}
\usepackage{MnSymbol}

\makeatletter
\renewcommand{\thm@space@setup}{%
  \thm@preskip=.5\baselineskip\@plus.2\baselineskip \@minus.2\baselineskip
  \thm@postskip=\thm@preskip
}
\newtheoremstyle{amsplain}
{\thm@preskip}
{\thm@postskip}
{\itshape}
{\parindent}
{\scshape}
{.}
{ }
{}
\makeatother
\usepackage{cleveref} %
\crefname{diagram}{diagram}{diagrams}
\crefname{property}{property}{properties}
\crefname{intn}{interpretation}{interpretations}
\crefname{msr}{rule}{rules}
\creflabelformat{msr}{(#2#1#3)}
\crefrangelabelformat{msr}{(#3#1#4) to~(#5#2#6)}
\creflabelformat{intn}{#2(#1)#3}
\crefrangelabelformat{intn}{(#3#1#4) to~(#5#2#6)}

\usepackage[autostyle=true]{csquotes}

\usepackage{enumitem}

\usepackage{forest}

\usepackage[abbreviations,symbols,record=alsoindex,stylemods={tree,topic},toc=false]{glossaries-extra}
\usepackage{glossary-mcols}%
\glsxtrRevertMarks %

\GlsXtrLoadResources[type={main},match={entrytype=entry}]
\GlsXtrLoadResources[type={abbreviations},match={entrytype=abbreviation}]
\GlsXtrLoadResources[type={symbols},
field-aliases={identifier=parent},
match={entrytype=symbol,entrytype=indexplural},
missing-sort-fallback={description}]

\usepackage{varindex}
\def\Index{\varindex(){\varindextwoScan}{\varindextwo}[][]}
\def\DIndex{\varindex(){\varindextwoScan}{\varindextwo}[][|defin]}

\usepackage{subcaption}
\usepackage{tikz}
\usetikzlibrary{arrows}
\usetikzlibrary{arrows.meta}
\usetikzlibrary{automata}
\usetikzlibrary{calc}
\usetikzlibrary{petri}
\usetikzlibrary{positioning}
\usetikzlibrary{shapes}
\usepackage{pgf-extrect}

\tikzset{
  place/.style={
    circle,
    thick,
    draw=blue!75,
    fill=blue!20,
    minimum size=6mm
  },
  htransition/.style={
    rectangle,
    thick,
    fill=black,
    minimum width=8mm,
    inner ysep=1pt
  },
  vtransition/.style={
    rectangle,
    thick,
    fill=black,
    minimum height=8mm,
    inner xsep=1pt
  }
}

\tikzset{
  >={Computer Modern Rightarrow[scale=0.8]}
}

\tikzset{triple/.style={-,preaction={draw,double distance=3pt}}}

\tikzset{
  process common/.style = {
    draw,
    thick,
    rounded corners,
    text centered,
  },
  node distance=0.5cm and 1cm,
  medium box/.style={
    rectangle,
    draw=black,
    fill=white,
    shape=rectangle,
    minimum height=0.75cm,
    minimum width=0.75cm
  },
  none/.style={},
  ctxt/.style = {
    process common,
    fill = gray!30,
  },
  obs/.style = {
    fill,
    orange,
    midway,
    circle,
    inner sep=0pt,
    minimum size=0.5em,
    anchor=center
  },
  node/.style = {
    process common,
    align=center,
  },
}

\pgfdeclarelayer{backlayer}
\pgfdeclarelayer{edgelayer}
\pgfdeclarelayer{nodelayer}

\pgfsetlayers{backlayer,edgelayer,nodelayer,main}

\definecolor{definedrblue}{RGB}{8,0,128}

\definecolor{definedrred}{RGB}{128,0,8}

\usepackage{tikz-cd}
\tikzcdset{
  arrow style=tikz,
}

\usepackage{xspace}

\usepackage{enumitem}
\newlist{thmlist}{enumerate}{1}
\setlist[thmlist]{label=(\arabic{thmlisti}), ref=\thetheorem(\arabic{thmlisti}),noitemsep}
\newlist{lemlist}{enumerate}{1}
\setlist[lemlist]{label=(\arabic{lemlisti}), ref=\thelemma(\arabic{lemlisti}),noitemsep}
\newlist{proplist}{enumerate}{1}
\setlist[proplist]{label=(\arabic{proplisti}), ref=\theproposition(\arabic{proplisti}),noitemsep}
\crefname{thmlisti}{theorem}{theorems}
\crefname{lemlisti}{lemma}{lemmas}
\crefname{proplisti}{proposition}{propositions}

\makeatletter
\def\Montserrat@scale{0.8}
\def\lato@scale{0.88}
\makeatother
\def\definfontfamily{lato-OsF}
\newcommand{\defin}[2][]{%
  \ifcsempty{#1}{%
    {{\usefont{T1}{\definfontfamily}{b}{n}{#2}}}%
  }{%
    {{\usefont{T1}{\definfontfamily}{b}{n}{#2}}}\index{#1}%
  }}
\newcommand{\mdefin}[2][]{%
  \ifcsempty{#1}{%
    \emph{#2}%
  }{%
    \emph{#2}\index{#1}%
  }}

\newcommand{\ie}{i.e.\@\xspace}
\newcommand{\eg}{e.g.\@\xspace}
\newcommand{\cf}{cf.\@\xspace}
\makeatletter
\newcommand*{\etc}{%
  \@ifnextchar{.}%
  {etc}%
  {etc.\@\xspace}%
}
\makeatother

\usepackage{mathtools}

\usepackage{proof}

\theoremstyle{amsplain}
\newtheorem{theorem}{Theorem}[section]
\newtheorem{conjecture}[theorem]{Conjecture}
\newtheorem{corollary}[theorem]{Corollary}
\newtheorem{lemma}[theorem]{Lemma}
\newtheorem{proposition}[theorem]{Proposition}
\newtheorem*{proposition*}{Proposition}
\newtheorem*{falsehood}{Falsehood}

\theoremstyle{definition} %
\newtheorem{definition}[theorem]{Definition}
\newtheorem{example}[theorem]{Example}

\newtheorem{problem}[theorem]{Problem}

\theoremstyle{remark} %
\newtheorem{remark}[theorem]{Remark}
\crefname{Remark}{remark}{remarks}
\newtheorem{assumption}[theorem]{Assumption}
\crefname{Assumption}{assumption}{assumptions}

\newlist{proofcases}{description}{3}
\setlist[proofcases]{font=\normalfont, labelindent=\listparindent,
  leftmargin=0pt, style=sameline}
\newcommand{\theproofcaselabel}[1]{{\normalfont \textsc{Case} #1.}}
\setlist[proofcases,1]{format=\theproofcaselabel}
\newcommand{\theproofsubcaselabel}[1]{{\normalfont \textsc{Subcase} #1.}}
\setlist[proofcases,2]{format=\theproofsubcaselabel}
\newcommand{\theproofsubsubcaselabel}[1]{{\normalfont \textsc{Subsubcase} #1.}}
\setlist[proofcases,3]{format=\theproofsubsubcaselabel}

\newcommand{\uscore}{\mbox{\tt\char`\_}}

\newcommand{\limplies}{\supset}

\newcommand{\Iota}{\mkern2mu\mathrm{I}}

\providecommand\given{}
\newcommand\SetSymbol[1][]{%
  \nonscript\:#1\vert
  \allowbreak
  \nonscript\:
  \mathopen{}}
\DeclarePairedDelimiterX\Set[1]\{\}{%
  \renewcommand\given{\,\SetSymbol[\delimsize]\,}
  #1
}

\newcommand{\gact}[2]{#1 \mathbin{\cdot} #2}

\newcommand{\N}{\mathbb{N}}

\newcommand{\ordinals}[1]{\mb{#1}}

\DeclarePairedDelimiterX{\sembr}[1]{\lsem}{\rsem}{
  \ifblank{#1}{\:\cdot\:}{#1}
}

\DeclarePairedDelimiterX{\obsbr}[1]{\langlebar}{\ranglebar}{
  \ifblank{#1}{\:\cdot\:}{#1}
}

\DeclarePairedDelimiterX{\gq}[1]{\ulcorner}{\urcorner}{
  \ifblank{#1}{\:\cdot\:}{#1}
}

\DeclarePairedDelimiterX{\sqgq}[1]{\ullcorner}{\ulrcorner}{
  \ifblank{#1}{\:\cdot\:}{#1}
}

\DeclarePairedDelimiterX{\ceil}[1]{\lceil}{\rceil}{
  \ifblank{#1}{\:\cdot\:}{#1}
}

\newcommand{\rn}[1]{(\textup{\textsc{#1}})}

\ExplSyntaxOn
\cs_new_protected:Npn \rakthesis_define_rule:nnnn #1#2#3#4
{
  \tl_const:cn { g_rule_#1_name } { #2 }
  \tl_const:cn { g_rule_#1_conclusion } { #3 }
  \seq_new:c { g_rule_#1_hypotheses }
  \seq_clear_new:N \l__temp_hypotheses
  \seq_set_split:Nnn \l__temp_hypotheses { & } { #4 }
  \seq_gset_eq:cN { g_rule_#1_hypotheses } \l__temp_hypotheses
}
\cs_new_protected:Npn \rakthesis_get_rule_name_rn:n #1
{
  \rn{ \tl_use:c { g_rule_#1_name } }
}
\cs_new_protected:Npn \rakthesis_get_rule:n #1
{
  \infer[\rakthesis_get_rule_name_rn:n { #1 }]{
    \tl_use:c { g_rule_#1_conclusion }
  }{
    \seq_use:cnnn { g_rule_#1_hypotheses } { & } { & } { & }
  }
}
\cs_new_protected:Npn \rakthesis_define_interp:nnnn #1#2#3#4
{
  \bool_const:cn   { g_interp_#2_split } { #1 }
  \tl_const:cn     { g_interp_#2_label } { eq:interp:#2 }
  \tl_const:cn     { g_interp_#2_lhs   } { #3 }
  \tl_const:cn     { g_interp_#2_rhs   } { #4 }
  \int_gzero_new:c { g_interp_#2_counter }
}
\cs_new_protected:Npn \rakthesis_ref_intlabel:n #1
{
  \seq_set_split:Nnn \l_tmpa_seq { , } { #1 }
  \seq_set_map:NNn \l_tmpb_seq \l_tmpa_seq { \tl_use:c { g_interp_##1_label } }
  \cref{ \seq_use:Nn \l_tmpb_seq { , } }
}
\cs_new_protected:Npn \rakthesis_Ref_intlabel:n #1
{
  \seq_set_split:Nnn \l_tmpa_seq { , } { #1 }
  \seq_set_map:NNn \l_tmpb_seq \l_tmpa_seq { \tl_use:c { g_interp_##1_label } }
  \Cref{ \seq_use:Nn \l_tmpb_seq { , } }
}
\msg_new:nnn { rakthesis } { splitnoalign } { Split\ was\ set\ but\ #1\ used\ outside\ of\ alignment\ environment:\ #2 }
\cs_new_protected:Npn \rakthesis_get_interp:nnn #1 #2 #3
{
  \tl_use:c { g_interp_#3_lhs  }
  \tl_if_eq:NnTF \c_novalue_tl { #2 } {
    \legacy_if:nTF { inalign@ } {
      \bool_if:cTF { g_interp_#3_split } {
        \nonumber\\&=
      } {
        &=
      }
    } {
      \legacy_if:nF { inalign@ } {
        \msg_warning:nnxx { rakthesis } { splitnoalign } { #3 } { \msg_line_context: }
      }
      =
    }
  } {
    #2
  }
  \legacy_if:nF { st@rred} {
    \int_compare:nNnTF { \int_use:c { g_interp_#3_counter } } { = } { 0 } {
      \label{\tl_use:c { g_interp_#3_label }}
      \legacy_if:nF { measuring@ } {
        \int_gincr:c { g_interp_#3_counter }
      }
    } {
      \tag{\ref{\tl_use:c { g_interp_#3_label }}}
    }
  }
  \tl_use:c { g_interp_#3_rhs  }
}
\NewDocumentCommand{\defrule}{m m m m}{
  \rakthesis_define_rule:nnnn { #1 } { #2 } { #3} { #4 }
}
\NewDocumentCommand{\getrule}{m}{
  \rakthesis_get_rule:n { #1 }
}
\NewDocumentCommand{\getrn}{m}{
  \rakthesis_get_rule_name_rn:n { #1 }
}
\NewDocumentCommand{\getrc}{m}{
  \tl_use:c { g_rule_#1_conclusion }
}
\NewDocumentCommand{\getrh}{m m}{
  \seq_item:cn { g_rule_#1_hypotheses } { #2 }
}

\NewDocumentCommand{\definterp}{s m m m}{
  \rakthesis_define_interp:nnnn { #1 } { #2 } { #3 } { #4 }
}
\NewDocumentCommand{\getinterp}{s o m}{
  \rakthesis_get_interp:nnn { #1 } { #2 } { #3 }
}
\NewDocumentCommand{\getlhs}{m}{
  \tl_use:c { g_interp_#1_lhs  }
}
\NewDocumentCommand{\getrhs}{m}{
  \tl_use:c { g_interp_#1_rhs  }
}
\NewDocumentCommand{\getintlbl}{m}{
  \tl_use:c { g_interp_#1_label }
}
\NewDocumentCommand{\refint}{m}{
  \rakthesis_ref_intlabel:n { #1 }
}
\NewDocumentCommand{\Refint}{m}{
  \rakthesis_Ref_intlabel:n { #1 }
}
\NewDocumentCommand{\getictr}{m}{
  \int_to_arabic:n { \int_use:c { g_interp_#1_counter } }
}
\ExplSyntaxOff

\newcommand{\outm}[1]{\red{#1}}
\newcommand{\inm}[1]{\blue{#1}}

\newcommand{\subst}[3]{[#1/#2]#3}

\newcommand{\derivable}[3]{\ensuremath{\blacktriangleright_{#1}^{#2}\,#3}}

\newcommand{\pderiv}[2]{#1 \mathrel{\Updownline} #2}

\newcommand{\rename}[3]{#1 \colon #2 \leftrightarrow #3}

\newcommand{\apprs}[2]{[#1]{#2}}

\newcommand{\relfnt}{\mathfrak}

\newcommand{\opr}[1]{{#1}^{\mathrm{op}}}

\DeclarePairedDelimiterX{\relcpl}[1]{\lceil}{\rceil}{
  \ifblank{#1}{\:\cdot\:}{#1}
}

\newcommand{\reltransc}[1]{\mathrel{#1^+}}

\newcommand{\relreflc}[1]{\mathrel{#1^r}}

\newcommand{\unirel}{\mathrel{\relfnt{U}}}

\usepackage{stackengine,scalerel}
\newcommand\operatorupX[1]{\,\ThisStyle{\ensurestackMath{%
      #1\stackengine{-0pt}{\,}{\SavedStyle\!^{\mathord{\uparrow}}}{O}{l}{F}{T}{S}}}}
\makeatletter
\newcommand\operatorup[1]{\mathop{\operatorupX{#1}}}
\makeatother
\newcommand{\dirsqcup}{\operatorup{\bigsqcup}}

\newcommand{\dirsup}{\dirsqcup}

\DeclarePairedDelimiterX{\upim}[1]{[}{]}{
  \ifblank{#1}{\:\cdot\:}{#1}
}

\DeclarePairedDelimiterX{\prel}[1]{|}{|}{
  \ifblank{#1}{\:\cdot\:}{#1}
}

\newcommand{\sfix}[1]{\ensuremath{#1^\dag}}

\newcommand{\opc}[1]{{#1}^{\mathrm{op}}}

\newcommand{\olap}[2]{\Omega_{#1}(#2)}

\newcommand{\emptymset}{{\emptyset}}

\newcommand{\supp}{\ms{supp}}

\newcommand{\persfnt}[1]{\boldsymbol{\mathbf{#1}}}
\newcommand{\ephemfnt}{\mathsf}

\newcommand{\steparrow}[1][{}]{\xrightarrow{#1}}
\newcommand{\steps}[1][{}]{\steparrow[#1]^*}

\newcommand{\msinc}[2]{#1 \mathrel{;} #2}

\newcommand{\jstep}[3][{}]{#2 \xrightarrow{#1} #3}
\newcommand{\msstep}[1][{}]{\steparrow[#1]}
\newcommand{\mssteps}[1][{}]{\steps[#1]}

\newcommand{\tSendC}[3]{\mathsf{send}\ #1\ #2;\ #3}
\newcommand{\tSendL}[3]{#1.#2;\ #3}
\newcommand{\tSendS}[2]{\mathsf{send}\ #1\ \mathsf{shift};\ #2}
\newcommand{\tSendV}[3]{\uscore \leftarrow \ms{output}\ #1\ #2;\ #3}

\newcommand{\tSendU}[2]{\ms{send}\ #1\ \mathsf{unfold};\ #2}
\newcommand{\tRecvC}[3]{#1 \leftarrow \mathsf{recv}\ #2;\ #3}
\newcommand{\tRecvS}[2]{\ms{shift} \leftarrow \mathsf{recv}\ #1;\ #2}
\newcommand{\tRecvV}[3]{#1 \leftarrow \ms{input}\ #2;\ #3}

\newcommand{\tRecvU}[2]{\ms{unfold} \leftarrow \ms{recv}\ #1;\ #2}
\newcommand{\tClose}[1]{\mathsf{close}\ #1}
\newcommand{\tWait}[2]{\mathsf{wait}\ #1;\ #2}
\newcommand{\tCase}[2]{\mathsf{case}\ #1\ #2}

\newcommand{\tFwdP}[2]{#1 \rightarrow #2}
\newcommand{\tFwdN}[2]{#1 \leftarrow #2}
\newcommand{\tCut}[3]{#1 \leftarrow #2;\ #3}
\newcommand{\tFix}[2]{\ms{fix}\ #1.#2}
\newcommand{\tpFix}[3]{\ms{fix}^{#1}\ #2.#3}
\newcommand{\tProc}[3]{#1 \leftarrow \{ #2 \} \leftarrow #3}
\newcommand{\tProcQ}[3]{#1 \leftarrow \{ #2 \} \leftarrow #3}
\newcommand{\tProcU}[3]{#1 \leftarrow \{ #2 \} \leftarrow #3}

\newcommand{\Tu}{\mathbf{1}}
\newcommand{\Tplus}{\oplus}
\newcommand{\Tot}{\otimes}
\newcommand{\Tamp}{\&}
\newcommand{\Tlolly}{\multimap}
\newcommand{\Tand}[2]{#1 \land #2}
\newcommand{\Timp}[2]{#1 \limplies #2}

\newcommand{\Trec}[2]{\rho #1.#2}
\newcommand{\Trecn}[3]{\rho^{#1} #2.#3}
\newcommand{\Tus}[1]{{{\uparrow} #1}}
\newcommand{\Tds}[1]{{{\downarrow} #1}}
\newcommand{\Tproc}[2]{\{#1 \leftarrow #2\}}
\newcommand{\Tnat}{\mathbf{nat}}

\newcommand{\mClose}[1]{\tClose{#1}}
\newcommand{\mSendLP}[3]{\tSendL{#1}{#2}{\tFwdP{#3}{#1}}}
\newcommand{\mSendLN}[3]{\tSendL{#1}{#2}{\tFwdN{#1}{#3}}}
\newcommand{\mSendCP}[3]{\tSendC{#1}{#2}{\tFwdP{#3}{#1}}}
\newcommand{\mSendCN}[3]{\tSendC{#1}{#2}{\tFwdN{#1}{#3}}}
\newcommand{\mSendVP}[3]{\tSendV{#1}{#2}{\tFwdP{#3}{#1}}}
\newcommand{\mSendVN}[3]{\tSendV{#1}{#2}{\tFwdN{#1}{#3}}}
\newcommand{\mSendSP}[2]{\tSendS{#1}{\tFwdN{#2}{#1}}}
\newcommand{\mSendSN}[2]{\tSendS{#1}{\tFwdP{#1}{#2}}}
\newcommand{\mSendUP}[2]{\tSendU{#1}{\tFwdP{#2}{#1}}}
\newcommand{\mSendUN}[2]{\tSendU{#1}{\tFwdN{#1}{#2}}}

\newcommand{\cclose}{\ms{close}}
\newcommand{\cunfold}{\ms{unfold}}
\newcommand{\cshift}{\ms{shift}}
\newcommand{\cval}[1]{\ms{val}\ #1}

\newcommand{\capxn}[2]{\lfloor #1 \rfloor_{#2}}
\DeclarePairedDelimiter{\commbrdelim}{\llangle}{\rrangle}
\newcommand{\commbr}[2][{}]{\commbrdelim{#2}_{#1}}

\newcommand{\fnval}[1]{#1\ \ms{val}}
\ExplSyntaxOn
\msg_new:nnn {thesis} {fnstep-obs} {fnstep(s)~is~obsolete;~migrate~to~fneval}
\NewDocumentCommand{\fnstep}{o}{%
  \msg_warning:nn {thesis} {fnstep-obs}
  \steparrow[#1]
}
\NewDocumentCommand{\fnsteps}{o}{%
  \msg_warning:nn {thesis} {fnstep-obs}
  \steps[#1]
}
\ExplSyntaxOff
\newcommand{\fneval}[2]{#1 \mathrel{\Downarrow} #2}

\newcommand{\jsynt}[2]{\inm{#1} \mathrel{\varepsilon} \inm{#2}}

\newcommand{\jisst}[2][{}]{#2\ \ms{type}_{\ms{s}}^{#1}}

\newcommand{\jisft}[1]{#1\ \ms{type}_{\ms{f}}}

\newcommand{\jtypef}[3]{#1 \Vdash #2 : #3}

\newcommand{\jtypem}[5]{{#1}\mathrel{;} {#2} \vdash #3 \mathrel{{:}{:}} {#4 : #5}}

\newcommand{\jstype}[3][{}]{{#2} \vdash \jisst[#1]{#3}}

\newcommand{\jftype}[2]{{#1} \vdash \jisft{#2}}

\newcommand{\jcfgt}[4][]{
  \ifstrempty{#1}{%
    #2 \vdash #3 :: #4%
  }{%
    \pderiv{#1}{#2 \vdash #3 :: #4}%
  }}
\newcommand{\jcfgti}[5][]{%
  \ifstrempty{#1}{%
    #2 \medvert #3 \vdash #4 :: #5%
  }{%
    \pderiv{#1}{#2 \medvert #3 \vdash #4 :: #5}%
  }}

\newcommand{\freecn}{\mathsf{fc}}
\newcommand{\inpcn}{\mathsf{ic}}
\newcommand{\outcn}{\mathsf{oc}}
\newcommand{\carrcn}{\mathsf{cc}}
\newcommand{\contcn}{\mathsf{kc}}

\newcommand{\jproc}[2]{\ephemfnt{proc}(#1, #2)}
\newcommand{\jmsg}[2]{\ephemfnt{msg}(#1, #2)}

\newcommand{\jeval}[2]{{\persfnt{eval}(#1, #2)}}

\newcommand{\jttp}[3]{\inm{#1} \vdash \inm{#2} : \outm{#3}}

\newcommand{\jtoc}[5][{}]{\inm{#2} \leadsto^{#1} \outm{#3} \mathrel{\varepsilon} \outm{#5} \mathrel{/} \inm{#4}}

\DeclarePairedDelimiterX{\ppi}[1]{\langlebar}{\ranglebar}{
  \ifblank{#1}{\:\cdot\:}{#1}
}

\DeclarePairedDelimiterX{\mg}[1]{\langlebar}{\ranglebar^p}{
  \ifblank{#1}{\:\cdot\:}{#1}
}

\newcommand{\jcmf}[3]{#1 :_{\ms{f}} #2 \leadsto #3}
\newcommand{\jtrelf}[5]{\jtypef{#2}{#3 \mathrel{#1} #4}{#5}}
\newcommand{\jtrelp}[7]{\jtypem{#2}{#3}{#4 \mathrel{#1} #5}{#6}{#7}}
\newcommand{\jtrelc}[5]{\jcfgt[]{#2}{#3 \mathrel{#1} #4}{#5}}

\newcommand{\mkcg}[1]{\mathrel{#1^c}}
\newcommand{\mkocg}[1]{\mathrel{#1^{\mathcal{O}}}}
\newcommand{\mkbcg}[1]{\mathrel{#1^b}}
\newcommand{\mkpcg}[1]{\mathrel{#1^p}}

\newcommand{\sobssim}[1]{\leqdot_{#1}}
\newcommand{\sobseq}[1]{\doteq_{#1}}
\newcommand{\sobsprec}[1]{\left(\leqdot_{#1}\right)^c}
\newcommand{\sobscong}[1]{\left(\doteq_{#1}\right)^c}
\newcommand{\strobsc}{\sobseq{T}}
\newcommand{\eocssim}{\sobssim{E}}
\newcommand{\iocssim}{\sobssim{I}}
\newcommand{\tocssim}{\sobssim{T}}
\newcommand{\eocsprec}{\sobsprec{E}}
\newcommand{\iocsprec}{\sobsprec{I}}
\newcommand{\tocsprec}{\sobsprec{T}}

\newcommand{\icommeq}{\doteq_I}
\newcommand{\tcommeq}{\doteq_T}

\NewDocumentCommand{\commsim}{o}{%
  \IfNoValueTF{#1}{%
    \leqdot
  }{%
    \ensuremath{\mathrel{\raisebox{1pt}{$\leqdot$}/\raisebox{-1pt}{$#1$}}}
  }
}

\NewDocumentCommand{\ncommsim}{o}{%
  \IfNoValueTF{#1}{%
    \not\leqdot
  }{%
    \ensuremath{\mathrel{\raisebox{1pt}{$\not\leqdot$}/\raisebox{-1pt}{$#1$}}}
  }
}

\NewDocumentCommand{\commeq}{o}{%
  \IfNoValueTF{#1}{%
    \doteq
  }{%
    \ensuremath{\mathrel{\raisebox{1pt}{$\doteq$}/\raisebox{-1pt}{$#1$}}}
  }
}

\newcommand{\barb}[2][{}]{#2 \mskip\medmuskip{\downarrow}_{#1}}
\newcommand{\wbarb}[2][{}]{#2 \mskip\medmuskip{\Downarrow}_{#1}}
\newcommand{\nwbarb}[2][{}]{#2 \mskip\medmuskip{\nDownarrow}_{#1}}

\newcommand{\wbsim}{\precapprox}
\newcommand{\wbbisim}{\approx}

\newcommand{\ctxh}[4][\cdot]{#2[#1]^{#3}_{#4}}

\title{Fairness and Communication-Based Semantics for Session-Typed Languages}
\author[1]{Ryan Kavanagh}
\affiliation[1]{organization={School of Computer Science},
  addressline={McGill University},
  city={Montreal, Quebec},
  postcode={H3A 2A7},
  country={Canada}}

\begin{document}

\begin{abstract}
  We give communication-based semantics and reasoning techniques for Polarized SILL~\cite{toninho_2013:_higher_order_proces_funct_session, pfenning_griffith_2015:_polar_subst_session_types}, a rich session-typed programming language with general recursion.
  Its features include channel and code transmission, synchronous and asynchronous communication, and functional programming.
  Our contributions are distinguished by their faithfulness to the \emph{process abstraction}, \ie, to the premise that communication is the only observable phenomenon of processes.
  We give the first observed communication semantics that supports general recursion and code transmission.
  Observed communication semantics~\cite{atkey_2017:_obser_commun_seman_class_proces} define the meaning of processes in terms of their observed communications.
  We use this observational semantics to define experiments on processes, and we give a communication-based testing equivalences framework~\cite{denicola_hennessy_1984:_testin_equiv_proces} for defining observational simulations and equivalences on processes.
  This framework captures several natural equivalences, and we show that one of these coincides with barbed congruence, the canonical notion of process equivalence.

  Polarized SILL is defined using a substructural operational semantics based on multiset rewriting.
  To ensure that our contributions are well-defined in the presence of non-termination, we introduce fairness for multiset rewriting systems.
  We construct a fair scheduler, we give sufficient conditions for traces to be fair, and we study the effects of permutation on fair traces.
\end{abstract}


\maketitle

\section{Introduction}

Communicating systems are ubiquitous, but their complexity makes them hard to get right and reason about.
To reason about communicating systems, we can use \emph{program equivalence}.
``Program equivalence is arguably one of the most interesting and at the same time important problems in formal verification''~\cite{lahiri_2018:_progr_equiv_dagst_semin}, and its practical applications include both program optimization and compiler verification.
But before we can do so, we must know what it means for communicating systems to be equivalent.

To make this question tractable, we abstract away inessentials and consider systems of \emph{communicating processes}.
Processes are computational agents that interact with their environment only through communication.
Importantly, communication is their only phenomenon.
This means that we can observe processes' communication patterns, but that we cannot observe their internal states or workings.
It also implies that we can think of systems of processes as processes themselves.

To respect this process abstraction, we are interested in an \textit{extensional} notions of equivalence, where processes are equivalent ``if we cannot tell them apart without pulling them apart''~\cite[2]{milner_1980:_calcul_commun_system}.
This suggests that two processes should be equivalent if, regardless of the experiments to which they are subjected, we cannot observe any differences between them.
But what are legitimate experiments on processes, and what does it mean to observe processes?

To illustrate these questions' many subtleties, we consider a session-typed process that flips bits in a bit stream.
Session types specify communication protocols, and communication on a session-typed communication channel respects that channel's type.
Let the session type ``$\ms{bits}$'' specify a bit stream, \ie, a potentially infinite sequence of bits $\mt{1}$ and $\mt{0}$.
The following recursive process\footnote{It is written in an idealized fragment of the session-typed language Polarized SILL presented in \cref{cha:sill-background}. We revisit this process in \cref{ex:sill-background-typing-mult-rewr:4}.} receives a bit stream on the channel $i$, flips its bits, and sends the result on the channel~$o$:
\begin{align*}
  i : \ms{bits} \vdash \tFix{F}{\;\tCase{i}{\{\; &\mt{0} \Rightarrow \tSendL{o}{\mt{1}}{\tProc{i}{F}{o}}\\
  {} \mid\; &\mt{1} \Rightarrow \tSendL{o}{\mt{0}}{\tProc{i}{F}{o}} \}}} :: o : \ms{bits}
\end{align*}
Operationally, the process waits until it receives a bit $\mt{0}$ or $\mt{1}$ on the channel \(i\).
If it receives the bit $\mt{0}$, then it takes the $\mt{0}$ branch of the $\ms{case}$ statement, sends the bit $\mt{1}$ on $o$ (syntax: $o.\mt{1}$), and then recurses.
Its behaviour if it receives the bit $\mt{1}$ is analogous.
Call this process $F_{i,o}$.

Now consider the process $C_{i,o}$ obtained by composing two copies of $F$ to form a chain, as in the following picture, where we draw processes as boxes and channels as lines or wires.
Is it ``equivalent'' to the identity process $\tFwdP{i}{o}$ that copies all communications from $i$ to $o$ unchanged?
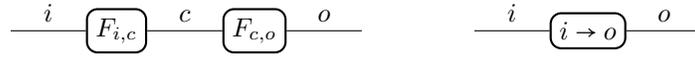
\begin{center}
  \begin{tikzpicture}
    \node [node] (F1) {$F_{i,c}$};
    \node [node, right=of F1] (F2) {$F_{c,o}$};
    \node [left=of F1] (i) {};
    \node [right=of F2] (o) {};
    \draw (i) -- (F1) node [pos=0.5, above] {$i$};
    \draw (F1) -- (F2) node [pos=0.5, above] {$c$};
    \draw (F2) -- (o) node [pos=0.5, above] {$o$};

    \node [right=of o] (i2) {};
    \node [node, right=of i2] (fwd) {$\tFwdP{i}{o}$};
    \node [right=of fwd] (o2) {};
    \draw (i2) -- (fwd) node [pos=0.5, above] {$i$};
    \draw (fwd) -- (o2) node [pos=0.5, above] {$o$};
  \end{tikzpicture}
\end{center}
If we are to give an extensional notion of equivalence, then we must determine if both processes are indistinguishable to external observers.
A long-standing idea in concurrency theory is that we can only observe or interact with processes through communication~\cites[2]{milner_1980:_calcul_commun_system}[12]{milner_1989:_commun_concur}.
Accordingly, a natural experiment to test their indistinguishability is to compose both processes with a process $S$ that sends a bit stream on $i$, and to compare the resulting communications on the channel $o$:
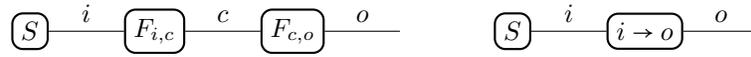
\begin{center}
  \begin{tikzpicture}
    \node [node] (S) {$S$};
    \node [node, right=of S] (F1) {$F_{i,c}$};
    \node [node, right=of F1] (F2) {$F_{c,o}$};
    \node [right=of F2] (o) {};
    \draw (S) -- (F1) node [pos=0.5, above] {$i$};
    \draw (F1) -- (F2) node [pos=0.5, above] {$c$};
    \draw (F2) -- (o) node [pos=0.5, above] {$o$};

    \node [node, right=of o] (S2) {$S$};
    \node [node, right=of S2] (fwd) {$\tFwdP{i}{o}$};
    \node [right=of fwd] (o2) {};
    \draw (S2) -- (fwd) node [pos=0.5, above] {$i$};
    \draw (fwd) -- (o2) node [pos=0.5, above] {$o$};
  \end{tikzpicture}
\end{center}
Intuitively, we should observe the same communications on $o$ in both cases.
In the first case, if $S$ sends a bit on \(i\), then $F_{i,c}$ flips it to its complement, and then $F_{c,o}$ flips it back, so we observe the original bit on $o$.
In the second case, $\tFwdP{i}{o}$ immediately copies the bit unchanged from $i$ to $o$.
So these two processes produce indistinguishable communications on $o$ when subjected to the above experiment.

Can we subject processes to any experiments other than composition with communication partners?
Not if we take seriously the premise that we can only interact with processes through communication.
Even then, only the contents of communications matter: timing differences are not meaningful because of non-deterministic process scheduling.
This example suggests that processes are observationally equivalent if we observe ``the same'' communications under all communication experiments.
Though this notion of equivalence is intuitively and semantically reasonable, it faces many~challenges.

First, assume that communication is higher-order, \ie, that we can send channels and programs.
Processes that send ``equivalent processes'' should be equivalent.
This means that higher-order communications do not need to be be equal on-the-nose to be considered ``the same''.
But then what does it mean for communications to be ``the same''?

Next, which channels should we deem to be observable in experiments?
We could treat experiments as black boxes that experiment on processes and then report their findings on their external channels.
In this ``external communication'' case, we observe the experiments' external channels, and we deem two processes to be equivalent if subjecting them to any experiment produces the same report on its external channels.
Alternatively, we could view experiments as questioners.
In this ``internal communication'' case, we observe the dialog between experiments and processes, and we deem two processes to be equivalent if subjecting them to any experiment produces the same dialogues.
Do these two styles of experimentation induce the same notion of process equivalence?
Is one preferable to the other?

Finally, assume that communication is asynchronous, \ie, that channels are buffered and that a process can send on a channel without first synchronizing with its recipient.
Though processes execute independently, an unfair scheduler could neglect to execute the process $F_{c,o}$.
In this case, the bits sent by $S$ and flipped by $F_{i,c}$ accumulate on the channel $c$, and they are never processed by $F_{c,o}$.
This means that, as external observers, we observe no communications on the channel $o$.
This contradicts our intuition that we should observe the same bit stream on $o$ as was sent by $S$.
It implies that fairness underlies our notion of process equivalence.

\subsection{Executive Summary of Contributions}
\label{sec:exec-summ-contr}

To make the above intuitions rigorous and answer the above questions, we make three contributions:
\begin{enumerate}
\item the first observed communication semantics for a session-typed language supporting general recursion and code transmission;
\item a communication-based testing equivalences framework; and
\item the first analysis of fairness for multiset rewriting systems, which are used to define the operational semantics of many session-typed programming languages.
\end{enumerate}
We give an overview of each of these contributions in turn.

\subsubsection{Observed Communication Semantics}

To make our notion of observation rigorous, we give an observed communication semantics~\cite{atkey_2017:_obser_commun_seman_class_proces} to Polarized SILL~\cite{toninho_2013:_higher_order_proces_funct_session, pfenning_griffith_2015:_polar_subst_session_types}.
Observed communication semantics define the meaning of a process to be the communications observed on its channels.
Polarized SILL is a rich session-typed language that cohesively integrates functional programming with session-typed message-passing concurrency.
Its functional layer is the simply-typed \(\lambda\)-calculus with a fixed-point operator, and it includes quoted processes as a base type.
Its process layer is based on a proofs-as-processes correspondence between intuitionistic linear logic and the session-typed \(\pi\)-calculus.
This layer's is specified by a substructural operational semantics in the form of a multiset rewriting system.
It supports recursive types and processes, synchronization (communication is asynchronous), choices (a form of branching), and value transmission (including quoted processes).
Polarized SILL provides an ideal setting for our analysis: it has many desirable real-world features while also remaining tractable.

We start by giving a semantic account of session-typed communications.
Communications are described using potentially infinite trees, and they are associated to session types by a typing judgment.
We then endow session-typed communications with an approximation order.
It serves two purposes.
First, we use it to specify what it means for communications to be ``the same''.
Second, we use it to characterize infinite communications in terms of their finite approximations.
Our account of session-typed communications is language independent: it describes in general what it means to be a session-typed communication.

Next, we explain how to observe the communications of processes in Polarized SILL.
Given a process execution (a sequence multiset rewriting steps), we use a coinductively defined judgment to observe the process's communications.
We show two important results.
First, if we consider only fair executions of a process, then all executions give the same observed communications.
This reflects the confluence property enjoyed by Polarized SILL.
Second, we show that whenever we observe a communication on a session-typed channel, that the type of the channel and the type of the communication agree.
This ensures that our observed communication semantics is semantically sound.

To help situate our observed communication semantics, we contrast it to the prior work.
\Textcite{atkey_2017:_obser_commun_seman_class_proces} introduced observed communication semantics to provide a notion of equivalence for Wadler's Classical Processes~\cite{wadler_2014:_propos_as_session}.
Classical Processes arises from a proofs-as-processes correspondence between classical linear logic and the session-typed $\pi$-calculus.
Our semantics differs on several key points.
First, we support a broader range of communication protocols, including general recursion and value transmission.
Second, our observed communication semantics requires no changes to the underlying operational semantics.
Third, we assume that communication is asynchronous rather than synchronous.
This assumption costs us nothing because synchronous communication can be encoded in asynchronous systems~\cite{pfenning_griffith_2015:_polar_subst_session_types}.
However, asynchronous communication simplifies observing communications (processes do not need to be provided with communication partners), and it lets us define a better-behaved ``external'' notion of equivalence.

\subsubsection{Communication-Based Testing Equivalences}

We make our notion of experimentation rigorous by giving a communication-based testing equivalences framework.
Testing equivalence frameworks~\cite{denicola_hennessy_1984:_testin_equiv_proces,hennessy_1983:_synch_async_exper_proces,denicola_1985:_testin_equiv_fully} use experiments to determine if processes are equivalent.
Intuitively, processes are deemed equivalent if no experiment can differentiate them.
Classical approaches rely on observing process states.
Instead of observing states, we observe communications using our observed communication semantics.
Advantageously, this gives a language-agnostic framework: equivalence depends only on observed communications and not on concrete operational details.

Concretely, we define an experiment to be a context with a hole, along with a choice of channels to observe.
Experimenting on a process then involves composing it with the context and observing the chosen channels.
Two processes are deemed equivalent according to an experiment if they both induce the same observed communications.
Because single experiments are often insufficient for discriminating between processes, we consider collections of experiments.
We then say that processes are equivalent if they are indistinguishable according to all experiments in a collection.

As motivated above, we have a certain latitude in choosing which channels to observe during experimentation.
One possibility is to observe the channels between an experiment and a process, leading to ``internal'' notions of equivalence à la \textcite{darondeau_1982:_enlar_defin_compl} and \textcite{atkey_2017:_obser_commun_seman_class_proces}.
A second possibility is to imagine that an experiment communicates with the process and reports its findings on its external channels.
Observing these external channels leads to an ``external'' notion of equivalence, and we show that this equivalence is a congruence.
These two equivalences are distinct, and they are both closed under process execution.
Our main result is that external equivalence coincides with barbed congruence, the canonical notion of progress equivalence.
This new characterization sheds light on the nature of barbed congruence, and it provides techniques for showing that processes are barbed congruent.

\subsubsection{Fairness for Multiset Rewriting Systems}

Polarized SILL's substructural operational semantics is specified by multiset rewriting system, and its rewrite rules can be applied non-deterministically.
In the case of Polarized SILL, this non-determinism means that processes that could make progress might not do so.
These ``unfair'' executions are undesirable in practice, and they make it harder to reason about processes.
As a result, we would like to restrict our attention to only ``fair'' executions.
To do so, we introduce and study fairness for multiset rewriting systems.

We start by introducing three varieties of fairness, each of which subdivides along the axis of weak and strong fairness.
We study the merits of each form of fairness, and we consider situations in which each might be desirable.
We show that these three varieties are independent, \ie, that no two varieties imply the third.

We then study properties of fairness.
We start by considering sufficient conditions for multiset rewriting systems to have fair executions, and we construct a fair scheduler.
An important sufficient condition is called ``interference-freedom''.
Assuming interference-freedom, we show that all varieties of fairness coincide, that fair executions are closed under permutation of steps, and that all fair executions are permutations of each other.
These results simplify reasoning about observed communications and processes: given a fair execution, we can often assume without loss of generality that steps of interest occur in sequence.

\subsection{Changes Since the Workshop Version}

This paper significantly expands on the work we presented at EXPRESS/SOS~2020~\cite{kavanagh_2020:_subst_obser_commun_seman}.
Several of the key changes include:
\begin{enumerate}
\item A more refined analysis of fairness.
  The previous single notion of fairness (called ``über fairness'' below) is extremely strong and it subsumes the various more refined notions of fairness considered below.
\item We support the entirety of Polarized SILL instead of a small idealized fragment.
\item Our observed communication semantics no longer requires changes to the language's substructural operational semantics.
\item We introduce a communication-based experimentation framework for defining and reasoning about program equivalences.
  We relate ``external'' observational equivalences to barbed congruence.
\item We explore the relationship between observational congruences and process congruences.
\end{enumerate}
This paper also builds on \cite[chapters~3, 5--7, and 10]{kavanagh_2021:_commun_based_seman}.
The key changes are added exposition and examples.

\subsection{Outline of Paper}

We give a survey of multiset rewriting systems in \cref{sec:ssos-fairness:mult-rewr-syst}.
In \cref{sec:three-vari-fairn}, we introduce fairness for multiset rewriting systems, and we study its properties in \cref{sec:ssos-fairness:prop-fair-trac}.
We present Polarized SILL in \cref{cha:sill-background}, our observed communication semantics in \cref{sec:sill-obs-equiv:observ-comm}, and our communication-based testing equivalences framework in \cref{cha:sill-obs-equiv}.
We discuss related work in \cref{sec:ssos-fairness:related-work}.
In \cref{sec:conclusion}, we take stock of our contributions and discuss future work.
For convenience, a listing of symbols follows our references.

Readers primarily interested in fairness are invited to focus on sections \cref{sec:ssos-fairness:mult-rewr-syst,sec:three-vari-fairn,sec:ssos-fairness:prop-fair-trac}, while skimming \cref{cha:sill-background,sec:sill-obs-equiv:observ-comm} for applications of fairness to session-typed processes.
Those primarily interested in session-typed processes can treat fairness as a ``black box'' concept and focus on \cref{cha:sill-background,sec:sill-obs-equiv:observ-comm,cha:sill-obs-equiv}.

\newcommand{\queue}{\ephemfnt{queue}}
\newcommand{\enq}{\ephemfnt{enq}}

\section{Multiset Rewriting Systems}
\label{sec:ssos-fairness:mult-rewr-syst}

In this section, we review (first-order) multiset rewriting systems.
For expository reasons, we start with the simpler formalism MSR${}_1$ of \textcite{cervesato_scedrov_2009:_relat_state_based} in \cref{sec:ssos-fairness:mult-rewr-syst:first-order-multiset}.
We extend it in \cref{sec:ssos-fairness:mult-rewr-syst:first-order-multiset-persistence} to handle the \emph{persistency} features of the multiset rewriting system MSR of \cite{cervesato_2005:_compar_between_stran}.
These first sections are expository, and they serve solely to give a uniform presentation to pre-existing work.
Our contribution comes in \cref{sec:ssos-fairness:mult-rewr-syst:parall-rule-appl}, where we extend parallel multiset rewriting~\cite[\S\S~5.3--5.4]{cervesato_2001:_typed_multis_rewrit} to support persistency.

\begin{definition}
  \label{def:ssos-fairness/mult-rewr-syst:1}
  A \defin{multiset}\DIndex{multiset} $M$ is a pair $(S, m)$ where $S$ is a set (the \mdefin{underlying set}) and $m : S \to \N$ is a function.
  It is finite if $\sum_{s \in S} m(s)$ is finite.
  We say $s$ is an \defin{element}\varindex{multiset element}{1!2}[|defin] of $M$, $s \in M$, if $m(s) > 0$.
  The \defin{support}\varindex{support of@ a@ multiset}{4!1 1!~234}[|defin] of $M$ is the set $\supp(M) = \{ s \in S \mid s \in M \}$\glsadd{supp}.
\end{definition}

\begin{remark}
  A multiset with a finite underlying set is always finite.
  The converse is false: the multiset \( (\N, \lambda x \in \N . 0) \) is finite, but the underlying set \(\N\) is infinite.
\end{remark}

When considering several multisets at once, we assume without loss of generality that they have equal underlying sets.

\begin{definition}
  \label{def:ssos-fairness/mult-rewr-syst:2}
  Multisets $M_1 = (S, m_1)$ and $M_2 = (S, m_2)$ are equipped with the following operations and relations:
  \begin{enumerate}
  \item the \defin{sum}\varindex{multiset sum}{1!2}[|defin] of $M_1$ and $M_2$ is the multiset $M_1,M_2 = (S, \lambda s \in S.m_1(s) + m_2(s))$;
  \item the \defin{union}\varindex{multiset union}{1!2}[|defin] of $M_1$ and $M_2$ is the multiset $M_1 \cup M_2 = (S, \lambda s \in S . \max(m_1(s), m_2(s)))$;
  \item the \defin{intersection}\varindex{multiset intersection}{1!2}[|defin] of $M_1$ and $M_2$ is the multiset $M_1 \cap M_2 = (S, \lambda s \in S . \min(m_1(s), m_2(s)))$;
  \item the \defin{difference}\varindex{multiset difference}{1!2}[|defin] of $M_1$ and $M_2$ is the multiset $M_1 \setminus M_2 = (S, \lambda s \in S . \max(0, m_1(s) - m_2(s)))$;
  \item $M_1$ is \defin{included}\varindex{multiset inclusion}{1!2}[|defin] in $M_2$, written $M_1 \subseteq M_2$, if $m_1(s) \leq m_2(s)$ for all $s \in S$.\qedhere
  \end{enumerate}
\end{definition}

We abuse terminology and call multisets $(S, m)$ \mdefin{sets} if $m(s) \leq 1$ for all $s \in S$.
We write $\emptymset$\glsadd{emptymset} for \defin{empty}\varindex{multiset empty}{1!2~}[|defin] multisets, \ie, for multisets $(S, m)$ such that $m(s) = 0$ for all $s \in S$.

\begin{example}
  \label{ex:ssos-fairness/mult-rewr-syst:1}
  A finite string $s$ over an alphabet $\Sigma$ describes a multiset $(\Sigma, m)$, where $m(\sigma)$ is the multiplicity of $\sigma$ in $s$.
\end{example}

\subsection{First-Order Multiset Rewriting}
\label{sec:ssos-fairness:mult-rewr-syst:first-order-multiset}

Consider finite multisets of first-order atomic formulas.
We call formulas \defin{facts}\DIndex{fact}.
We write $M(\vec x)$ to mean that the facts in the multiset $M$ draw their variables from $\vec x$, where $\vec x = x_1, \dotsc, x_m$ for some $m$.
Given $M(\vec x)$ and some choice of terms $\vec t$ for $\vec x$, we write $M(\vec t)$ for the simultaneous substitution $\subst{\vec t}{\vec x}{M}$.

Multiset rewrite rules describe localized changes to multisets of facts.
A \defin{multiset rewrite rule}\varindex{multiset rewrite rule}{123}[|defin] $r$ is a pair of multisets $F(\vec x)$ and $G(\vec x, \vec n)$, and it is schematically represented by:
\[
  r : \forall \vec x . F(\vec x) \to \exists \vec n . G(\vec x, \vec n).
\]
Informally, we interpret the variables $\vec x$ as being universally quantified in $F$ and $G$, and the variables $\vec n$ as being existentially quantified in $G$.
In particular, we treat $\vec x$ and $\vec n$ as bound variables, and assume that they can be freely $\alpha$-varied.
A \defin{multiset rewriting system}\varindex{multiset rewriting system}{123}[|defin] (MRS) is a set $\mc{R}$ of multiset rewrite rules.

Given a rule $r : \forall \vec x . F(\vec x) \to \exists \vec n . G(\vec x, \vec n)$ in $\mc{R}$ and some choice of constants $\vec t$ for $\vec x$, we say that the \defin{instantiation}\varindex{multiset rewrite rule instantiation}{123!~4}[|defin] $r(\vec t) : F(\vec t) \to \exists \vec n.G(\vec t, \vec n)$\glsadd{mrsinst} is \defin{applicable} to a multiset $M$ if there exists a multiset $M'$ such that $M = F(\vec t),M'$.
The rule $r$ is applicable\varindex{multiset rewrite rule applicable}{123!4~}[|defin] to $M$ if $r(\vec t)$ is applicable to $M$ for some $\vec t$.
In these cases, the \defin{result}\varindex{{multiset rewriting system} result as@ a@ multiset}{1!2!345}[|defin] of applying $r(\vec t)$ to $M$ is the multiset $G(\vec t, \vec d),M'$, where $\vec d$ is a choice of pairwise-distinct fresh constants.
In particular, we assume that the constants $\vec d$ do not appear in $M$ or in $\mc{R}$.
We call $\theta = \subst{\vec t}{\vec x}{}$ the \defin{matching substitution}\DIndex<{matching substitution}\varindex{multiset rewrite rule matching substitution}{123!5!4=}[|defin] and $\xi = \subst{\vec d}{\vec n}{}$ the \defin{fresh-constant substitution}\DIndex<{fresh-constant substitution}\varindex{multiset rewrite rule fresh-constant substitution}{123!5!4=}[|defin].
Intuitively, the matching substitution specifies to which portion of $M$ the rule $r$ is applied.
The \defin{instantiating substitution}\DIndex<{instantiating substitution}\varindex{multiset rewrite rule instantiating substitution}{123!5!4=}[|defin] for $r$ relative to $M$ is the composite substitution $\delta = (\theta, \xi)$.
We capture this relation using the syntax
\[
    F(\vec t),M' \xrightarrow{(r;\delta)} G(\vec t, \vec d),M'.
\]
We often abuse notation and write $r(\theta)$, $F(\theta)$, and $G(\theta,\xi)$ for $r(\vec t)$, $F(\vec t)$, and $G(\vec t, \vec d)$.
We call $F(\vec t)$ the \defin{active}\varindex{multiset rewriting system active}{1!4~ 123!41}[|defin] multiset and $M'$ the \defin{stationary}\varindex{multiset rewriting system stationary}{1!4~ 123!41}[|defin] multiset.

\begin{definition}
  Given an MRS $\mc{R}$ and a multiset $M_0$, a \defin{trace}\DIndex/>{{multiset rewriting system} trace}\index{trace!for multiset rewriting systems|see{multiset rewriting system, trace}} from $M_0$ is a countable sequence of steps
  \begin{equation}
    \label{eq:ssos-fairness/mult-rewr-syst:7}
    M_0 \xrightarrow{(r_1;\delta_1)} M_1 \xrightarrow{(r_2;\delta_2)} M_2 \xrightarrow{(r_3;\delta_3)} \cdots
  \end{equation}
  such that, where $\delta_i = (\theta_i, \xi_i)$, the constants in $M_i$ and $\xi_j$ are disjoint for all $i < j$.

  The notation $(M_0, (r_i;\delta_i)_{i \in I})$\glsadd{msrtrace} abbreviates the trace \eqref{eq:ssos-fairness/mult-rewr-syst:7}, where $I$ always ranges over $\N$ or $\ordinals{n} = \{1, \dotsc, n \}$\glsadd{cardinal} for some $n \in \N$.
  An \defin{execution}\DIndex/>{{multiset rewriting system} execution}\index{execution|see{multiset rewriting system, execution}} is a maximally long trace.\label{def:ssos-fairness/mult-rewr-syst:3}
\end{definition}

Let the \defin{support}\varindex{support of@ a@ trace}{1!~234 4!1}[|defin] $\supp(T)$\glsadd{supptrace} of a trace $T = (M_0;(r_i,\delta_i)_I)$ be the set $\supp(T) = \bigcup_{i \geq 0} \supp(M_i)$.
Write \(M \msstep M'\) if \(M \xrightarrow{(r,\theta)} M'\) for some \((r,\theta)\), and let $\mssteps$ be the reflexive, transitive closure of $\msstep$.

\begin{example}
  \label{ex:fair-mult-rewr:1}%
  We model computations with queues.
  Let the fact $\queue(q, \$)$ mean that $q$ is the empty queue, and let $\queue(q, v \leadsto q')$ mean that the queue $q$ has value $v$ at its head and that its tail is the queue $q'$.
  Then the multiset $Q = \queue(q, 0 \leadsto q'), \queue(q', \$)$ describes a one-element queue containing $0$.
  The following two rules capture enqueuing values on empty and non-empty queues, respectively, where the fact $\enq(q, v)$ is used to enqueue $v$ onto the queue $q$:
  \begin{gather*}
    e_1 : \forall x, y . \enq(x, y), \queue(x, \$) \to \exists z . \queue(x, y \leadsto z), \queue(z, \$),\\
    e_2 : \forall x, y, z, w . \enq(x, y), \queue(x, z \leadsto w) \to \queue(x, z \leadsto w), \enq(w, y).
  \end{gather*}

  The following execution from $Q,\enq(q,1)$ captures enqueuing 1 on the queue $q$:
  \begin{gather*}
    Q,\enq(q,1) \xrightarrow{(e_2;(\subst{q,1,0,q'}{x,y,z,w}{}, \emptyset))} Q,\enq(q',1)\hspace{12em}\\
    \hspace{8em}\xrightarrow{(e_1;(\subst{q',1}{x,y}{}, \subst{a}{z}{}))} \queue(q, 0 \leadsto q'), \queue(q', 1 \leadsto a), \queue(a, \$).\qedhere
  \end{gather*}
\end{example}

We will use a variant of \cref{ex:fair-mult-rewr:1} when we define the substructural operational semantics of Polarized SILL in \cref{cha:sill-background}.
Its substructural operational semantics is given by a multiset rewriting system, and it uses queues of messages to ensure that messages sent by processes arrive in order.

\begin{example}
  \label{ex:ssos-fairness/mult-rewr-syst:10}
  We define addition on unary natural numbers.
  The MRS uses the facts \(\ephemfnt{add}(m, n, l)\) and \(\ephemfnt{val}(l, v)\), where $\ephemfnt{val}(l,v)$ represents a memory cell $l$ with value $v$, and $\ephemfnt{add}(m, n, l)$ causes the sum of \(m\) and \(n\) to be stored in cell $l$.
  It is given by the following rules:
  \begin{gather}
    a_{\mt{z}} : \forall n, l . \ephemfnt{add}(\mt{z}, n, l) \to \ephemfnt{val}(l, n)\label{eq:ssos-fairness/mult-rewr-syst:1}\\
    a_{\mt{s}} : \forall m, n, l . \ephemfnt{add}(\mt{s}(m), n, l) \to \ephemfnt{add}(m, \mt{s}(n), l)
  \end{gather}
  Write \(3\) for the unary representation \(\mt{s}(\mt{s}(\mt{s}(\mt{z})))\) of three.
  The following execution stores the sum of two and three in \(l\):
  \begin{gather*}
    \ephemfnt{add}(\mt{s}(\mt{s}(\mt{z})), 3, l) \msstep \ephemfnt{add}(\mt{s}(\mt{z}), \mt{s}(3), l) \msstep \ephemfnt{add}(\mt{z}, \mt{s}(\mt{s}(3)), l) \msstep \ephemfnt{val}(l, \mt{s}(\mt{s}(3))).
  \end{gather*}
  The first two rules in this execution are instances of \(a_{\mt{s}}\), while the last rule is an instance of \(a_{\mt{z}}\).
\end{example}

\begin{example}
  \label{ex:ssos-fairness/mult-rewr-syst:8}
  We build on \cref{ex:ssos-fairness/mult-rewr-syst:10} to recursively compute the $n$-th Fibonacci number.
  The fact $\ephemfnt{fib}(n, l)$ causes the $n$-th Fibonacci number to be stored in cell $l$.
  The MRS is given by \(a_{\mt{z}}\), \(a_{\mt{s}}\), and the following new rules:
  \begin{gather}
    f_0 : \forall l . \ephemfnt{fib}(\mt{z}, l) \to \ephemfnt{val}(l, \mt{s}(\mt{z}))\label{eq:ssos-fairness/mult-rewr-syst:10}\\
    f_1 : \forall l . \ephemfnt{fib}(\mt{s}(\mt{z}), l) \to \ephemfnt{val}(l, \mt{s}(\mt{z}))\label{eq:ssos-fairness/mult-rewr-syst:111}\\
    f : \forall l, n . \ephemfnt{fib}(\mt{s}(\mt{s}(n)), l) \to \exists l', l'' . \ephemfnt{cont}(l, l', l''), \ephemfnt{fib}(\mt{s}(n), l'), \ephemfnt{fib}(n, l'')\label{eq:ssos-fairness/mult-rewr-syst:122}\\
    c : \forall l, l', l'', m, n . \ephemfnt{cont}(l, l', l''), \ephemfnt{val}(l', m), \ephemfnt{val}(l'', n) \to \ephemfnt{add}(m, n, l)\label{eq:ssos-fairness/mult-rewr-syst:13}
  \end{gather}
  Rules $f_0$ and $f_1$ directly calculate the zeroth and first Fibonacci numbers.
  The rule $f$ makes two ``recursive calls'' that will store their results in fresh locations $l'$ and $l''$.
  It uses the continuation fact $\ephemfnt{cont}(l, l', l'')$ to signal that the values in locations $l'$ and $l''$ need to be added and stored in location $l$.
  Once values \(m\) and \(n\) are available in locations \(l'\) and \(l''\), the rule \(c\) causes \(m\) and \(n\) to be added and stored in location \(l\).

  We remark that this implementation builds in garbage collection.
  Indeed, because we represent memory locations \(\ephemfnt{val}(l, n)\) using ephemeral facts, these locations are discarded as soon as they are no longer needed by future computation.
  However, this implementation is not very efficient: it repeatedly recomputes the same $n$-th Fibonacci number and requires exponential time.
  In \cref{ex:ssos-fairness/mult-rewr-syst:2}, we will use \emph{persistent} facts to implement a memoized version of this algorithm.
\end{example}

\begin{remark}
  The order in which rewrite rules are applied is non-deterministic and is outside of the control of a multiset rewriting system.\index{multiset rewriting system!non-determinism}
  For example, when computing \(\ephemfnt{fib}(2,l)\), a scheduler could non-deterministically choose to apply rule \(f_1\) or \(f_0\) after applying \(f_2\).
  Moreover, multiset rewriting systems need not satisfy any confluence properties.
  This means that, in general, finite executions from a given multiset need not result in the same final multiset.
  In \cref{ex:ssos-fairness/mult-rewr-syst:2}, we illustrate design considerations for multiset rewrite systems that force a scheduler to order certain rule applications.
  In \cref{sec:ssos-fairness:prop-fair-trac}, we present a condition on MRSs called ``interference-freedom''.
  Intuitively, it states that the order in which rules are applied does not matter, because rules do not interfere with or disable each other.
\end{remark}

Though the above presentation of multiset rewriting due to \textcite{cervesato_2005:_compar_between_stran} is relatively concise, its implicit treatment of eigenvariables and signatures introduces some ambiguity.
For greater clarity, we sometimes adopt an equivalent presentation due to \textcite{cervesato_scedrov_2009:_relat_state_based}.
In this presentation, we explicitly track the signatures used to form facts, and multiset rewriting systems rewrite multisets-in-context.\varindex{{multiset rewriting system} for@ multisets-in-context}{1!~23}\index{context!multiset-in-context|see{multiset, multiset-in-context}}
In particular, we assume that all facts appearing in rules are formed over some fixed initial signature \(\Sigma_i\).\footnote{The terminology ``initial signature'' follows \textcite[\S~3.2]{cervesato_scedrov_2009:_relat_state_based}.
  Here, we are not using ``initial'' in a category-theoretic sense.
  Instead, it should be interpreted as ``first'' and in opposition to ``subsequent''.
  Indeed, rewriting rules below extend signatures with fresh symbols, and the initial signature \(\Sigma_i\) is the first signature in this sequence of extensions.}
A \defin{multiset-in-context}\DIndex/>{multiset multiset-in-context} is a pair $\msinc{\Sigma}{M}$\glsadd{msinc} where the symbols used by $M$ appear in the signature $\Sigma$, and where $\Sigma$ contains the initial signature $\Sigma_i$.
We write $\Sigma \vdash t$ to mean that the term $t$ is valid over the signature $\Sigma$, and $\Sigma \vdash \vec{t}$ for the obvious extension to collections of terms $\vec{t}$.

Consider a rule $r : \forall \vec x . F(\vec x) \to \exists \vec n . G(\vec x, \vec n)$.
Given a signature $\Sigma$, an \defin{instantiation}\DIndex>{{multiset rewrite rule} instantiation} $r(\vec t)$ of $r$ is a rule of the form $r(\vec t) : F(\vec t) \to \exists \vec n . G(\vec t, \vec n)$ for some $\vec t$ with $\Sigma \vdash \vec t$.
This instantiation is \defin{applicable}\varindex{{multiset rewrite rule} applicable}{1!2~} to a multiset-in-context $\msinc{\Sigma}{M}$ if $M = F(\vec t), M'$ for some $M'$.
The \defin{result}\varindex{{multiset rewriting system} result as@ a@ multiset-in-context}{1!2!345}[|defin] of applying the instantiation $r(\vec t)$ to $\msinc{\Sigma}{M}$ is the multiset-in-context $\msinc{\Sigma,\vec n}{G(\vec t),M'}$, where we extend the signature $\Sigma$ with the globally fresh symbols\footnote{In particular, we expect the symbols to be disjoint from $M$ and from those appearing in any other rule.
  It is sufficient to choose \(\vec n\) disjoint from \(\Sigma\): all symbols in \(M\) appear in \(\Sigma\) by definition of multiset-in-context, and all symbols appearing in rules appear in \(\Sigma_i \subseteq \Sigma\) by assumption.} $\vec n$.
We can represent this transition schematically as:
\[
  \msinc{\Sigma}{F(\vec t), M'} \longrightarrow_{\mc{R},r : \forall \vec x . F(\vec x) \to \exists \vec n . G(\vec x, \vec n)} \msinc{\Sigma,\vec n}{G(\vec t),M'} \quad \text{if }\Sigma \vdash \vec t.
\]
Here, $\mc{R}$ lists the other rules in the multiset rewriting system, and the substitution $\subst{\vec t}{\vec x}{}$ applied to \(F\) and \(G\) corresponds to the matching substitution.
Extending the signature $\Sigma$ with $\vec n$ captures the fresh-constant substitution $\subst{\vec n}{\vec n}{}$.

Given an MRS $\mc{R}$ and a multiset-in-context $\msinc{\Sigma_0}{M_0}$, a \defin{trace} from $\msinc{\Sigma_0}{M_0}$ is a countable sequence of applications $\msinc{\Sigma_0}{M_0} \longrightarrow_{\mc{R}} \msinc{\Sigma_0,\Sigma_1}{M_1} \longrightarrow_{\mc{R}} \msinc{\Sigma_0,\Sigma_1,\Sigma_2}{M_2} \longrightarrow \dotsb$.
Again, an \defin{execution} is a maximally long trace.

\subsection{First-Order Multiset Rewriting with Persistence}
\label{sec:ssos-fairness:mult-rewr-syst:first-order-multiset-persistence}

Facts in multisets represent pieces of knowledge.
In \cref{sec:ssos-fairness:mult-rewr-syst:first-order-multiset}, these facts were ephemeral: they could be consumed or destroyed by applying multiset rewrite rules.
Often times, we would like some facts to be \textit{persistent}, \ie, for some facts to be reusable and preserved by all rules.
To this end, we partition facts as \defin{persistent}\DIndex<{persistent fact} (indicated by bold face, $\persfnt{p}$\glsadd{persfnt}) and \defin{ephemeral}\DIndex<{ephemeral fact} (indicated by sans serif face, $\ephemfnt{p}$\glsadd{ephemfnt}).
We then extend the multiset rewriting system of the previous section to support persistence.
In doing so, we diverge slightly from \textcite{cervesato_2005:_compar_between_stran} to allow for the set of persistent facts to grow across time.

Persistent facts are reusable, so we do not care about their multiplicities in multisets.
For simplicity, we assume throughout that they form a set.
We also separate persistent facts from ephemeral facts and write a generic multiset as $\persfnt{\Pi}, M$, where the set $\persfnt{\Pi}$ contains the persistent facts and the multiset $M$ contains the ephemeral facts.
A \defin{multiset-in-context}\DIndex/>{multiset multiset-in-context} is now a triple $\msinc{\Sigma}{\persfnt{\Pi}, M}$\glsadd{msinc}.
As before, all term-level symbols appearing in $\persfnt{\Pi}$ and $M$ are contained in the signature $\Sigma$.

A \defin{multiset rewrite rule}\DIndex{{multiset rewrite rule}} is now schematically represented by
\[
  r : \forall \vec x . \persfnt{\pi}(\vec x), F(\vec x) \to \exists \vec n . \persfnt{\pi'}(\vec x, \vec n), G(\vec x, \vec n),
\]
where $F$ and $G$ are as before, and $\persfnt{\pi}$ and $\persfnt{\pi'}$ are sets of persistent facts.
As before, multiset rewrite rules describe localized changes to multisets.
Fix a multiset rewriting system $\mc{R}$.
Given a rule $r : \forall \vec x . \persfnt{\pi}(\vec x), F(\vec x) \to \exists \vec n . \persfnt{\pi'}(\vec x, \vec n), G(\vec x, \vec n)$ in $\mc{R}$ and some choice of constants $\vec t$ for $\vec x$, we say that the \defin{instantiation}\varindex{multiset rewrite rule instantiation}{123!~4}[|defin] $r(\vec t) : \persfnt{\pi}(\vec t), F(\vec t) \to \exists \vec n . \persfnt{\pi'}(\vec t, \vec n), G(\vec t, \vec n)$ is \defin{applicable}\varindex{multiset rewrite rule applicable}{123!4~}[|defin] to a multiset $\persfnt{\Pi}, M$ if there exists a multiset $M'$ such that $M = F(\vec t),M'$ and if $\persfnt{\pi}(\vec t) \subseteq \persfnt{\Pi}$.
The rule $r$ is applicable to $\persfnt{\Pi}, M$ if $r(\vec t)$ is applicable to $\persfnt{\Pi}, M$ for some $\vec t$.
In these cases, the \defin{result}\varindex{{multiset rewriting system} result as@ a@ multiset}{1!2!345}[|defin] of applying $r(\vec t)$ to $\persfnt{\Pi}, M$ is the multiset $\left(\persfnt{\Pi} \cup \persfnt{\pi'}(\vec t,\vec d)\right), G(\vec t, \vec d), M'$, where $\vec d$ is a choice of fresh constants.
Again, we assume that the constants $\vec d$ do not appear in $M$ or in $\mc{R}$.
Because $\persfnt{\Pi}$ and $\persfnt{\pi'}(\vec t, \vec d)$ were both assumed to be sets, the multiset $\persfnt{\Pi} \cup \persfnt{\pi'}(\vec t, \vec d)$ in the result is again a set.

Multiset rewrite rules again equivalently describe localized changes to multisets-in-context.
Consider a rule $r : \forall \vec x . \persfnt{\pi}(\vec x), F(\vec x) \to \exists \vec n . \persfnt{\pi'}(\vec x, \vec n), G(\vec x, \vec n)$ in a multiset rewriting system $\mc{R}$.
Given a signature $\Sigma$, an \defin{instantiation}\varindex{multiset rewrite rule instantiation}{123!~4}[|defin] $r(\vec t)$ of $r$ is a rule of the form $r(\vec t) : \persfnt{\pi}(\vec t), F(\vec t) \to \exists \vec n . \persfnt{\pi'}(\vec t, \vec n), G(\vec t, \vec n)$ for some $\vec t$ with $\Sigma \vdash \vec t$.
This instantiation is \defin{applicable}\varindex{multiset rewrite rule applicable}{123!4~}[|defin] to a multiset-in-context $\msinc{\Sigma}{\persfnt{\Pi}, M}$ if there exists a multiset $M'$ such that $M = F(\vec t),M'$ and if $\persfnt{\pi}(\vec t) \subseteq \persfnt{\Pi}$.
The \defin{result}\varindex{{multiset rewriting system} result as@ a@ multiset-in-context}{1!2!345}[|defin] of applying the instantiation to $\msinc{\Sigma}{\persfnt{\Pi}, M}$ is the multiset-in-context $\msinc{\Sigma,\vec n}{\left(\persfnt{\Pi} \cup \persfnt{\pi'}(\vec t)\right),G(\vec t),M'}$, where we extend the signature $\Sigma$ with the globally fresh symbols $\vec n$ (modulo $\alpha$-renaming).
We can represent this transition schematically as:
\begin{gather*}
  \msinc{\Sigma}{\persfnt{\pi}(\vec t), \persfnt{\Pi'}, F(\vec t), M'} \longrightarrow_{\mc{R},(r : \forall \vec x . \persfnt{\pi}(\vec x), F(\vec x) \to \exists \vec n . \persfnt{\pi'}(\vec x, \vec n), G(\vec x, \vec n))} \qquad\qquad\\
  \qquad\qquad\qquad\msinc{\Sigma,\vec n}{\left(\persfnt{\pi}(\vec t) \cup \persfnt{\Pi'} \cup \persfnt{\pi'}(\vec t)\right), G(\vec t), M'} \quad \text{if }\Sigma \vdash \vec t,
\end{gather*}
where $\mc{R}$ lists the other rules in the multiset rewriting system.
The \defin{active multiset}\varindex{multiset rewriting system active}{1!4~ 123!41}[|defin] is $(\persfnt{\pi},F)(\vec t)$, while the \defin{stationary multiset}\varindex{multiset rewriting system stationary}{1!4~ 123!41}[|defin] is $\persfnt{\Pi'},M'$.

The definitions of trace and execution are as before.

\begin{example}
  \label{ex:ssos-fairness/mult-rewr-syst:2}
  We use memoization to improve the time complexity of \cref{ex:ssos-fairness/mult-rewr-syst:8}, which computed the $n$-th Fibonacci number.
  Memoization uses a persistent fact $\persfnt{memo}(n, m)$ that means ``$m$ is the $n$-th Fibonacci number''.
  Care is often needed when designing multiset rewriting systems to counter the effects of non-deterministic rule application.
  Indeed, in the case of memoization, the scheduler could choose to ignore available memoized values.
  To see how, consider a naïve implementation of memoization, where we extend \cref{ex:ssos-fairness/mult-rewr-syst:8} with the rule
  \[
    f_{\mi{memo}} : \forall l, n, m . \ephemfnt{fib}(n, l), \persfnt{memo}(n, m) \to \ephemfnt{val}(l, m).
  \]
  The scheduler could apply the rule $f$ to the multiset $\ephemfnt{fib}(\mt{s}(\mt{s}(\mt{z})), l), \persfnt{memo}(\mt{s}(\mt{s}(\mt{z})), 2)$ instead of \(f_{\mi{memo}}\), even though a memoized value is available.

  We can force the scheduler's hand by disabling rules after one use.
  We do so by making them depend on an ephemeral fact that is never replaced.
  Concretely, we use an ephemeral fact $\ephemfnt{notyet}(n)$ that is consumed on the first invocation of $\ephemfnt{fib}(n, l)$:
  \begin{gather}
    f_0 : \forall l . \ephemfnt{fib}(\mt{z}, l), \ephemfnt{notyet}(\mt{z}) \to \persfnt{memo}(\mt{z}, \mt{s}(\mt{z})), \ephemfnt{val}(l, \mt{s}(\mt{z}))\label{ex:ssos-fairness/mult-rewr-syst:3}\\
    f_1 : \forall l . \ephemfnt{fib}(\mt{s}(\mt{z}), l), \ephemfnt{notyet}(\mt{s}(\mt{z})) \to \persfnt{memo}(\mt{s}(\mt{z}), \mt{s}(\mt{z})), \ephemfnt{val}(l, \mt{s}(\mt{z}))\label{ex:ssos-fairness/mult-rewr-syst:4}\\
    \begin{split}
      f : \forall l, n . \ephemfnt{fib}(\mt{s}(\mt{s}(n)), l), \ephemfnt{notyet}(\mt{s}(\mt{s}(n))) \to  {} \hspace{8em} \\
      \hspace{8em} {} \to \exists l', l'' . \ephemfnt{fibcont}(\mt{s}(\mt{s}(n)), l, l', l''), \ephemfnt{fib}(\mt{s}(n), l'), \ephemfnt{fib}(n, l'')
    \end{split}\label{ex:ssos-fairness/mult-rewr-syst:5}
  \end{gather}
  The rule \(f\) uses the modified continuation fact \(\ephemfnt{fibcont}(k, l, l', l'')\).
  As with the continuation fact \(\ephemfnt{cont}(l, l', l'')\) of \cref{ex:ssos-fairness/mult-rewr-syst:8}, it means that the values in locations \(l'\) and \(l''\) should be added and stored in \(l\).
  We also use it to mean that this sum should be memoized as the value of the \(k\)-th Fibonacci number.
  Because the fact \(\ephemfnt{notyet}(n)\) gets consumed, subsequent attempts to compute $\ephemfnt{fib}(n, l')$ are forced to used the memoized value.
  Using the memoized value is captured by:
  \begin{gather}
    r : \forall l, n, m . \persfnt{memo}(n, m), \ephemfnt{fib}(n, l) \to \ephemfnt{val}(l, m)\label{ex:ssos-fairness/mult-rewr-syst:6}
  \end{gather}
  We split the rule \(c\) of \cref{ex:ssos-fairness/mult-rewr-syst:8} in two.
  The first rule \(c_f\) waits until values \(n\) and \(m\) are available in the locations \(l'\) and \(l''\).
  It then causes them to be added and stored in location \(l\).
  It also creates a continuation fact \(\ephemfnt{addcont}(k, l, l')\).
  This fact is used by the rule \(c_a\) to memoize the value in \(l'\) as the value of the \(k\)-th Fibonacci number before storing it in location \(l\):
  \begin{gather}
    \begin{split}
      c_f : \forall l, l', l'', k, m, n . \ephemfnt{fibcont}(k, l, l', l''), \ephemfnt{val}(l', n), \ephemfnt{val}(l'', m) \to {}\\
      {} \to \exists l''' . \ephemfnt{addcont}(k, l, l'''), \ephemfnt{add}(l''', n, m)\\
    \end{split}\label{ex:ssos-fairness/mult-rewr-syst:7}\\
    c_a : \forall l, l', k, m . \ephemfnt{addcont}(k, l, l'), \ephemfnt{val}(l', m) \to \persfnt{memo}(k, m), \ephemfnt{val}(l, m)
  \end{gather}

  To compute the $n$-th Fibonacci number, take an arbitrary execution from the multiset
  \[
    \ephemfnt{notyet}(\mt{z}), \ephemfnt{notyet}(\mt{s}(\mt{z})), \dotsc, \ephemfnt{notyet}(n), \ephemfnt{fib}(n, l).
  \]
  We can show that this execution is finite.
  Its final multiset will contain a fact $\ephemfnt{val}(l, m)$, where $m$ is the desired value.
\end{example}

Sometimes matching substitutions can make a pair of rules indistinguishable:

\begin{definition}
  \label{def:ssos-fairness/mult-rewr-syst:4}
  Two rule instantiations \(r_1(\theta_1)\) and \(r_2(\theta_2)\) are \defin{equivalent}\varindex{multiset rewrite rule instantiation equivalent}{123!~4!5=}[|defin], $r_1(\theta_1) \equiv r_2(\theta_2)$\glsadd{mrsequiv}, if they are applicable to the same multisets, and if whenever they are applicable to some multiset \(M\), then the result of applying either to \(M\) is the same (up to choice of fresh constants).
  Otherwise, they are \defin{distinct}\varindex{multiset rewrite rule instantiation distinct}{123!~4!5=}[|defin].
\end{definition}

We will use instantiation equivalence in \cref{sec:ssos-fairness:prop-fair-trac} to study the relationship between various forms of fairness and properties of fair traces.
We can characterize it as follows:

\begin{proposition}
  \label{prop:ssos-fairness/mult-rewr-syst:3}
  Consider rules $r_i : \forall \vec x_i . \persfnt{\pi_i}(\vec x_i), F_i(\vec x_i) \to \exists \vec n_i . \persfnt{\pi_i'}(\vec x_i, \vec n_i), G_i(\vec x_i, \vec n_i)$ and matching substitutions $\theta_i$ for $i = 1, 2$.
  The instantiations $r_1(\theta_1)$ and $r_2(\theta_2)$ are equivalent if and only if
  \begin{enumerate}
  \item $\persfnt{\pi_1}(\theta_1), F_1(\theta_1) = \persfnt{\pi_2}(\theta_2), F_2(\theta_2)$;
  \item $\exists \vec n_1 . G_1(\theta_1, \vec n_1) = \exists \vec n_2 . G_2(\theta_2, \vec n_2)$ (up to renaming of bound variables); and
  \item \( \exists \vec n_1 . \persfnt{\pi_1}(\theta_1) \cup \persfnt{\pi_1'}(\theta_1, \vec n_1) = \exists \vec n_2 . \persfnt{\pi_2}(\theta_2) \cup \persfnt{\pi_2'}(\theta_2, \vec n_2) \) (up to renaming of bound variables).
  \end{enumerate}
\end{proposition}

\begin{proof}
  We start with sufficiency.
  Assume that \(r_1(\theta_1)\) and \(r_2(\theta_2)\) are equivalent.
  Then they are both applicable to the multiset \(\persfnt{\pi_1}(\theta_1), F_1(\theta_1)\), so we deduce \(\left(\persfnt{\pi_2}(\theta_2), F_2(\theta_2)\right) \subseteq \left(\persfnt{\pi_1}(\theta_1), F_1(\theta_1)\right)\).
  A symmetric observation gives the opposite inclusion, so we deduce that the multisets are equal.
  Next, the result of applying either to \(M = \persfnt{\pi_1}(\theta_1), F_1(\theta_1)\) is the same, so
  \[
    \exists \vec n_1 . \left(\persfnt{\pi_1}(\theta_1) \cup \persfnt{\pi_1'}(\theta_1, \vec n_1) \right), G_1(\theta_1, \vec n_1) =
    \exists \vec n_2 . \left(\persfnt{\pi_2}(\theta_2) \cup \persfnt{\pi_2'}(\theta_2, \vec n_2) \right), G_2(\theta_2, \vec n_2)
  \]
  up to renaming of bound variables.
  Recall that the collections of ephemeral and persistent facts are disjoint, so this implies the remaining two conditions.

  Next, we show necessity.
  The first condition implies that \(r_1(\theta_1)\) and \(r_2(\theta_2)\) are applicable to the same multisets.
  The last two conditions imply that the result of applying either rule to a given multiset is the same.
  We deduce that the two rule instantiations are equivalent.
\end{proof}

\begin{example}
  \label{ex:ssos-fairness/mult-rewr-syst:9}
  Consider the rules
  \begin{gather*}
    r_1 : \forall x, y . A(x, y) \to \exists n . B(x, n),\\
    r_2 : \forall x . A(x, x) \to \exists n . B(x, n)
  \end{gather*}
  and matching substitutions $\theta_1 = [a,a/x,y]$ and $\theta_2 = [a/x]$.
  Then $r_1(\theta_1) \equiv r_2(\theta_2)$:
  \begin{gather*}
    r_1(\theta_1) : A(a, a) \to \exists n . B(a, n),\\
    r_2(\theta_2) : A(a, a) \to \exists n . B(a, n).
  \end{gather*}
  Moreover, applying either $r_1(\theta_1)$ or $r_2(\theta_2)$ to the multiset $A(a,a), C(b,b)$ gives $B(a,a'), C(b,b)$ for some fresh constant $a'$.
\end{example}

\subsection{Semantic Irrelevance of Fresh Constants}
\label{sec:ssos-fairness:mult-rewr-syst:semant-irrel-fresh}

The constants in fresh-constant substitutions are, by construction, not semantically meaningful.
Indeed, they are arbitrarily chosen globally fresh constants.
This fact that can be made precise by appealing to interpretations of multiset rewriting as fragments of linear logic~\cite{cervesato_scedrov_2009:_relat_state_based}.
In these interpretations, fresh constants come from symbols bound by an existential quantifier, and these symbols can be freely $\alpha$-varied.

As a result of these observations, we identify traces up to refreshing substitutions.
A \defin{refreshing substitution} for a trace $T = (M_0,(r_i;(\theta_i,\xi_i))_i)$ is a collection of fresh-constant substitutions $\eta = (\eta_i)_i$ such that, where \(\theta_1' = \theta_1\) and \(\theta_{i + 1}' = \subst{\eta_1, \dotsc, \eta_i}{\xi_1, \dotsc, \xi_i}{\theta_{i + 1}}\), $[\eta]T = (M_0, (r_i; (\theta_i', \eta_i))_i)$\glsadd{tracerefresh} is also a trace.
Intuitively, $[\eta]T$ is the trace obtained by redoing the steps in \(T\), but with the fresh constants \(\eta_i\) instead of \(\xi_i\).
Explicitly, we identify traces $T$ and $T'$ if there exists a refreshing substitution $\eta$ such that $T' = [\eta]T$.

\subsection{Parallel Rule Applications}
\label{sec:ssos-fairness:mult-rewr-syst:parall-rule-appl}

We have so far only considered sequential rule application.
However, we are interested in modelling parallel computation.
To do so, we would expect parallel (or simultaneous) rule application to be required.
To this end, we briefly discuss parallel multiset rewriting systems.
We show that we can emulate parallel rule application using sequential rule application and vice-versa.
As a result, it will be sufficient to consider only sequential rule applications.

We define a binary operator $\ast$ on rules that combines them for parallel application.
Given rules
\[
  r_i : \forall \vec x_i . \persfnt{\pi_i}(\vec x_i), F_i (\vec x_i) \to \exists \vec n_i . \persfnt{\pi_i'}(\vec x_i, \vec n_i), G_i (\vec x_i, \vec n_i)
\]
for $i = 1, 2$, let the rule $r_1 \ast r_2$ be given by:
\begin{multline*}
  r_1 \ast r_2 : \forall \vec x_1,\vec x_2 . (\persfnt{\pi_1}(\vec x_1) \cup \persfnt{\pi_2}(\vec x_2)), F_1(\vec x_1), F_2(\vec x_2) \to {} \\
  {} \to \exists \vec n_1, \vec n_2 . (\persfnt{\pi_1'}(\vec x_1, \vec n_1) \cup \persfnt{\pi_2'}(\vec x_2, \vec n_2)), G_1(\vec x_1, \vec n_1), G_2(\vec x_2, \vec n_2).
\end{multline*}
The multiset rewrite rule $r_1 \ast r_2$ captures applying $r_1$ and $r_2$ in parallel.
Intuitively, every application of this rule splits the ephemeral portion of a multiset in two, applies each of the rules separately, and then recombines the results at the end.
It is clear that $\ast$ is an associative and commutative operator with identity $1_\ast : \emptymset \to \emptymset$.

Given a multiset rewriting system $\mc{R}$, let the \mdefin{parallel multiset rewriting system}\DIndex<{parallel {multiset rewriting system}} $\mc{R}^\ast$\glsadd{parallelmrs} be the multiset rewriting system given by:
\[
  \mc{R}^\ast = \Set{ r_1 \ast \dotsb \ast r_n \given n \in \N, r_i \in \mc{R} }.
\]
The following proposition shows that parallel rewriting can be emulated by sequential rewriting.
At a high level, its proof replaces each rule \( r_1 \ast \dotsb \ast r_n \) appearing in a trace $M \mssteps_{\mc{R}^\ast} M'$ by the sequence of rules \(r_1, \dotsc, r_n\).
We can make this replacement because, by definition of \( \ast \), the rule \(r_1 \ast \dotsb \ast r_n\) describes making the localized changes described by each \(r_i\) to a disjoint portion (modulo persistence) of the multiset.
Because each \(r_i\) rewrites a disjoint portion of the multiset, the rules \(r_1, \dotsc, r_n\) do not disable or otherwise interfere with each other, so applying each \(r_i\) sequentially gives the same result.

\begin{proposition}
  \label{prop:ssos-fairness/mult-rewr-syst:2}
  For all multiset rewriting systems $\mc{R}$ and multisets $M$ and $M'$, ${M \mssteps_{\mc{R}} M'}$ if and only if $M \mssteps_{\mc{R}^\ast} M'$.
\end{proposition}

\begin{proof}
  Sufficiency is clear: every trace over $\mc{R}$ is a trace over $\mc{R}^\ast$.
  Conversely, assume that $M \mssteps_{\mc{R}^\ast} M'$.
  We proceed by induction on the number of steps taken.
  If no steps were taken, \ie, if $M \mssteps_{\mc{R}^\ast} M'$ by reflexivity, then we are done.
  Now assume that $m + 1$ steps were taken, \ie, that $M \msstep_{\mc{R}^\ast} M'' \mssteps_{\mc{R}^\ast} M'$ for some $M''$.
  Assume that the first step is given by the rule
  \begin{multline*}
    r_1 \ast \dotsb \ast r_k : \forall \vec x_1, \dotsc, \vec x_k . \left(\bigcup_{i = 1}^k \persfnt{\pi}_i(\vec x_i)\right),F_1(\vec x_1),\dotsc,F_k(\vec x_k) \to {} \\
    {} \to \exists \vec n_1, \dotsc, \vec n_k . \left(\bigcup_{i = 1}^k \persfnt{\pi}_i'(\vec x_i, \vec n_i)\right),G_1(\vec x_1, \vec n_1), \dotsc, G_k(\vec x_k, \vec n_k)
  \end{multline*}
  for some rules $r_i \in \mc{R}$ and $k \geq 1$, with matching substitution $\theta$ and fresh-constant substitution $\xi$.
  Let $\theta_i$ and $\xi_i$ be the obvious restrictions of $\theta$ and $\xi$ to $\vec x_i$ and $\vec n_i$, respectively.
  By definition of rule application, it follows that for some \(\persfnt{\Pi}'\) and \(N\),
  \begin{align}
    M &= \left(\bigcup_{i = 1}^k \persfnt{\pi}_i(\theta)\right),\persfnt{\Pi'},F_1(\theta),\dotsc,F_k(\theta), N\nonumber\\
      &= \left(\bigcup_{i = 1}^k \persfnt{\pi}_i(\theta_i)\right),\persfnt{\Pi'},F_1(\theta_1),\dotsc,F_k(\theta_k), N,\label{eq:ssos-fairness/mult-rewr-syst:11}\\
    M'' &= \left(\left(\bigcup_{i = 1}^k \persfnt{\pi}_i(\theta)\right) \cup \persfnt{\Pi'} \cup \left(\bigcup_{i = 1}^k \persfnt{\pi'}_i(\theta,\xi)\right)\right) ,G_1(\theta,\xi), \dotsc, G_k(\theta, \xi), N\nonumber\\
      &= \left(\left(\bigcup_{i = 1}^k \persfnt{\pi}_i(\theta_i)\right) \cup \persfnt{\Pi'} \cup \left(\bigcup_{i = 1}^k \persfnt{\pi'}_i(\theta_i,\xi_i)\right)\right) ,G_1(\theta_1,\xi_1), \dotsc, G_k(\theta_k, \xi_k), N.\label{eq:ssos-fairness/mult-rewr-syst:12}
  \end{align}

  We claim that, where $M_0 = M$, the following is a trace for some multisets $M_1$, \ldots, $M_{k}$, and that $M_{k} = M''$:
  \begin{equation}
    \label{eq:ssos-fairness/mult-rewr-syst:100}
    M_0 \msstep[(r_1;(\theta_1,\xi_1))] M_1 \msstep \dotsb \msstep M_{k - 1} \msstep[(r_k;(\theta_k,\xi_k))] M_k.
  \end{equation}
  This claim implies that $M \mssteps_{\mc{R}} M''$.
  By the induction hypothesis, $M'' \mssteps_{\mc{R}} M'$.
  Because $\mssteps_{\mc{R}}$ is transitive, we can then conclude $M \mssteps_{\mc{R}} M'$.

  We show that \eqref{eq:ssos-fairness/mult-rewr-syst:100} is a trace.
  To do so, we proceed by induction on $j$, $0 \leq j \leq k$, to show that
  \[
    M_j = \left(\bigcup_{i = 1}^k \persfnt{\pi}_i(\theta_i)\right) \cup \persfnt{\Pi'} \cup \left(\bigcup_{i = 1}^j \persfnt{\pi'}_i(\theta_i,\xi_i)\right), G_1(\theta_1,\xi_1), \dotsc, G_j(\theta_j, \xi_j), F_{j + 1}(\theta_{j + 1}), \dotsc, F_k(\theta_k), N.
  \]
  It will then follow that $r_j$ is applicable to $M_j$ for all $0 \leq j \leq k$, \ie, that \eqref{eq:ssos-fairness/mult-rewr-syst:100} is a trace.
  It will also be immediate by \eqref{eq:ssos-fairness/mult-rewr-syst:12} that $M_{k} = M''$.

  The case $k = 0$ is immediate by \cref{eq:ssos-fairness/mult-rewr-syst:11}.
  Assume the result for some $j$, then by definition of rule application, we have
  \[
    M_{j + 1} = \left(\bigcup_{i = 1}^k \persfnt{\pi}_i(\theta_i)\right) \cup \persfnt{\Pi'} \cup \left(\bigcup_{i = 1}^{j + 1} \persfnt{\pi'}_i(\theta_i,\xi_i)\right), G_1(\theta_1,\xi_1), \dotsc, G_{j + 1}(\theta_{j + 1}, \xi_{j + 1}), F_{j + 2}(\theta_{j + 2}), \dotsc, F_k(\theta_k), N.
  \]
  This is exactly what we wanted to show.
  We conclude the result.
\end{proof}

Parallel rule applications for multiset rewriting systems are discussed by \textcite[\S\S~5.3--5.4]{cervesato_2001:_typed_multis_rewrit}.
Our approach is inspired by Cervesato's: we both decompose a multiset into disjoint pieces, apply rules in parallel, and recombining multisets.
However, we diverge in the details from Cervesato's approach by defining a new MRS $\mc{R}^\ast$ that captures parallel execution, instead of using an inductively defined parallel rewriting judgment.
Our approach also allows for persistent facts.

\section{Three Varieties of Fairness}
\label{sec:three-vari-fairn}
\label{sec:ssos-fairness:three-vari-fairn}

\varindex{{multiset rewriting system} fairness}{1!2}[|see{fairness}]
Intuitively, fairness properties provide progress guarantees for components in computational systems.
Countless varieties of fairness were introduced in the 1980s, and they were classified according to various taxonomies by \textcite{francez_1986:_fairn,kwiatkowska_1989:_survey_fairn_notion}.
A common classification is along the axis of \emph{strength}.
\defin{Weak fairness}\DIndex<{weak fairness} ensures that components that are almost always able to make progress do make progress infinitely often.
In contrast, \defin{strong fairness}\DIndex<{strong fairness} ensures that components that are able to make progress infinitely often do make progress infinitely often.

We introduce three varieties of fairness for multiset rewriting systems---\textit{rule fairness}, \textit{fact fairness}, and \textit{instantiation fairness}---and we give a weak and a strong formulation for each.
We motivate each variety of fairness by an example.
In \cref{sec:ssos-fairness:comp-types-fairn}, we show that these three varieties are independent.

\subsection{Rule Fairness}
\label{sec:ssos-fairness:rule-fairness}

We motivate rule fairness using an encoding of Petri nets as multiset rewriting systems.
Petri nets~\cite{petri_1980:_introd_to_gener_net_theor, peterson_1977:_petri_nets} are structures used to model concurrency.
A \textit{Petri net} is given by two sets---a set $P$ of \textit{places} and a set $T$ of \textit{transitions}---and a pair of functions $I, O : T \to \powerset(P)$ specifying the \textit{inputs} and \textit{outputs} of the transitions in $T$.

Informally, one executes a Petri net by placing tokens in places and observing how the tokens move through the net.
An assignment $\mu : P \to \N$ of tokens to places is called a \textit{marking}.
A \textit{marked Petri net} $(P, T, I, O, \mu)$ is conveniently depicted as a directed graph, where circles represent places, solid dots represent tokens, thick lines represent transitions, and arrows represent input/output.
For example, \cref{fig:ssos-fairness:three-vari-fairn:petri-two-markings:a} depicts the marked Petri net given by $P = \{p_1, p_2, p_3\}$, $T = \{t_1\}$, $I(t_1) = \{p_1,p_2\}$, $O(t_1) = \{t_3\}$, and $\mu_0 = (p_1 \mapsto 2, p_2 \mapsto 1, p_3 \mapsto 0)$.

\begin{figure}[htb]
  \centering
  \begin{subfigure}{0.4\textwidth}
    \centering
    \begin{tikzpicture}[node distance=3em]
      \node [place,tokens=2] (p1) [label=left:$p_1$] {};
      \node [place,tokens=1] (p2) [below of=p1, label=left:$p_2$] {};
      \node [vtransition] (t1) at ($(p1)!0.5!(p2)+(3em,0)$) [label=above:$t_1$] {};
      \node [place] (p3) [right of=t1] [label=right:$p_3$] {};
      \draw (t1) edge [pre] (p1)
      edge [pre] (p2)
      edge [post] (p3);
    \end{tikzpicture}
    \caption{Marking $\mu_0$}
    \label{fig:ssos-fairness:three-vari-fairn:petri-two-markings:a}
  \end{subfigure}
  \begin{subfigure}{0.4\textwidth}
    \centering
    \begin{tikzpicture}[node distance=3em]
      \node [place,tokens=1] (p1) [label=left:$p_1$] {};
      \node [place,tokens=0] (p2) [below of=p1, label=left:$p_2$] {};
      \node [vtransition] (t1) at ($(p1)!0.5!(p2)+(3em,0)$) [label=above:$t_1$] {};
      \node [place,tokens=1] (p3) [right of=t1] [label=right:$p_3$] {};
      \draw (t1) edge [pre] (p1)
      edge [pre] (p2)
      edge [post] (p3);
    \end{tikzpicture}
    \caption{Marking $\mu_1$}
    \label{fig:ssos-fairness:three-vari-fairn:petri-two-markings:b}
  \end{subfigure}
  \caption{Two markings of the same Petri net}
  \label{fig:ssos-fairness:three-vari-fairn:petri-two-markings}
\end{figure}
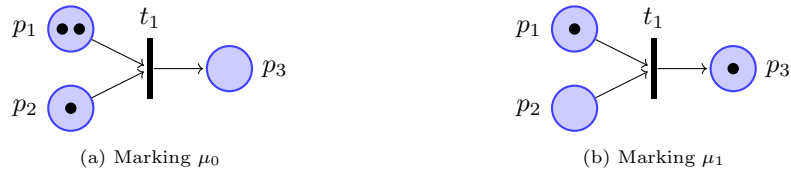

A transition is \textit{enabled} if all of its inputs have at least one token.
In this case, it \textit{fires} by taking one token from each of its input places and adding one token to each of its output places; this changes the marking of the Petri net.
A Petri net executes by repeatedly firing enabled transitions.
An \textit{execution} is then a sequence $\mu_0,\mu_1,\dotsc$ of markings, each obtained from its predecessor by firing an enabled transition.
For example, in the Petri net of \cref{fig:ssos-fairness:three-vari-fairn:petri-two-markings:a}, the transition $t_1$ is enabled, and the result of firing $t_1$ is the Petri net of \cref{fig:ssos-fairness:three-vari-fairn:petri-two-markings:b}.
When considering executions from a given marked Petri net, it is often more convenient to notate executions by their \textit{firing sequence}, \ie, by the sequence of transitions that fired, than by the sequence of markings.
The firing sequence for the execution $\mu_0, \mu_1$ is $t_1$.

Transitions fire non-deterministically.
Consider, for example, the marked Petri net in \cref{fig:ssos-fairness:three-vari-fairn:non-det:a}.
It could fire $t_2$ first to obtain the marking $\mu_1$ of \cref{fig:ssos-fairness:three-vari-fairn:non-det:b}.
From here, it can fire $t_1$ to return to the marking $\mu_0$.
Alternatively, it could have fired the transition $t_3$ to get the marking $\mu_2$ of \cref{fig:ssos-fairness:three-vari-fairn:non-det:c}.
No transitions are enabled in this marking, so an execution ends as soon as it reaches this marking.
It follows that all executions from $\mu_0$ are given by the firing sequences $t_2$, $(t_2t_1)^\infty$, or $(t_2t_1)^*t_3$.\footnote{We adopt notation from $\omega$-regular languages, where $\Sigma^*$ and $\Sigma^\omega$ respectively denote finite and infinite words over the alphabet $\Sigma$, and $\Sigma^\infty = \Sigma^* \cup \Sigma^\omega$.}

\begin{figure}[htb]
  \begin{subfigure}{0.3\textwidth}
    \centering
    \begin{tikzpicture}[node distance=3em]
      \node [vtransition] (t1) [label=above left:$t_1$] {};
      \node [vtransition] (t2) [below of=t1, label=below left:$t_2$] {};
      \node [vtransition] (t3) [below of=t2, label=below left:$t_3$] {};
      \node [place,tokens=1] (p1) [left of=t2, label=left:$p_1$] {};
      \node [place] (p2) [right of=t2, label=right:$p_2$] {};
      \node [place] (p3) [right of=t3, label=right:$p_3$] {};

      \draw (t1) edge [pre, bend left] (p2) edge [post, bend right] (p1);
      \draw (t2) edge [pre] (p1) edge [post] (p2);
      \draw (t3) edge [pre,bend left] (p1) edge [post] (p3);
    \end{tikzpicture}
    \caption{Marking $\mu_0$}
    \label{fig:ssos-fairness:three-vari-fairn:non-det:a}
  \end{subfigure}
  \begin{subfigure}{0.3\textwidth}
    \centering
    \begin{tikzpicture}[node distance=3em]
      \node [vtransition] (t1) [label=above left:$t_1$] {};
      \node [vtransition] (t2) [below of=t1, label=below left:$t_2$] {};
      \node [vtransition] (t3) [below of=t2, label=below left:$t_3$] {};
      \node [place] (p1) [left of=t2, label=left:$p_1$] {};
      \node [place, tokens=1] (p2) [right of=t2, label=right:$p_2$] {};
      \node [place] (p3) [right of=t3, label=right:$p_3$] {};

      \draw (t1) edge [pre, bend left] (p2) edge [post, bend right] (p1);
      \draw (t2) edge [pre] (p1) edge [post] (p2);
      \draw (t3) edge [pre,bend left] (p1) edge [post] (p3);
    \end{tikzpicture}
    \caption{Marking $\mu_1$}
    \label{fig:ssos-fairness:three-vari-fairn:non-det:b}
  \end{subfigure}
  \begin{subfigure}{0.3\textwidth}
    \centering
    \begin{tikzpicture}[node distance=3em]
      \node [vtransition] (t1) [label=above left:$t_1$] {};
      \node [vtransition] (t2) [below of=t1, label=below left:$t_2$] {};
      \node [vtransition] (t3) [below of=t2, label=below left:$t_3$] {};
      \node [place] (p1) [left of=t2, label=left:$p_1$] {};
      \node [place] (p2) [right of=t2, label=right:$p_2$] {};
      \node [place, tokens=1] (p3) [right of=t3, label=right:$p_3$] {};

      \draw (t1) edge [pre, bend left] (p2) edge [post, bend right] (p1);
      \draw (t2) edge [pre] (p1) edge [post] (p2);
      \draw (t3) edge [pre,bend left] (p1) edge [post] (p3);
    \end{tikzpicture}
    \caption{Marking $\mu_2$}
    \label{fig:ssos-fairness:three-vari-fairn:non-det:c}
  \end{subfigure}
  \caption{Markings reachable from \cref{fig:ssos-fairness:three-vari-fairn:non-det:a}, illustrating non-deterministic firings.}
  \label{fig:ssos-fairness:three-vari-fairn:non-det}
\end{figure}
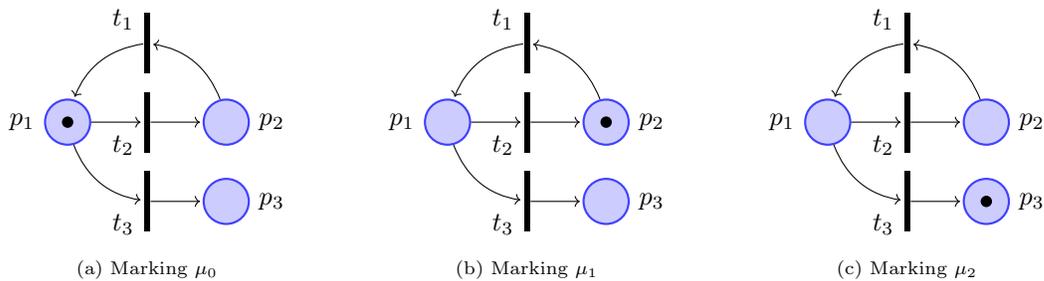

In applications, it is often desirable to rule out so-called ``unfair'' executions.
For example, we could deem the execution $(t_2t_1)^\omega$ to be unfair because, though the transition $t_3$ is enabled infinitely often, it never fires.
To this end, we recall the definitions of weak and strong fairness for Petri nets~\cite{leu_1988:_inter_among_various}.
We say that an execution $(\mu_i)_{i \in I}$ of a Petri net is
\begin{itemize}
\item \textit{weakly fair} if it is finite, or if it is infinite and for all transitions $t \in T$, if $t$ is enabled on all but finitely many markings $\mu_i$, then there exist infinitely many $i \in I$ such that $\mu_{i + 1}$ was obtained from $\mu_i$ by firing~$t$;
\item \textit{strongly fair} if it is finite, or if it is infinite and for all transitions $t \in T$, if $t$ is enabled on infinitely many markings $\mu_i$, then there exist infinitely many $i \in I$ such that $\mu_{i + 1}$ was obtained from $\mu_i$ by firing~$t$.
\end{itemize}

Weak and strong fairness rule out different executions.
For example, the firing sequence $(t_2t_1)^\omega$ describes a weakly fair execution from $\mu_0$: the execution alternates between the markings $\mu_0$ and $\mu_1$, so no transitions are enabled on all but finitely many markings.
The execution is not strongly fair: the transition $t_3$ is enabled infinitely often, but it never fires.
The only strongly fair executions from $\mu_0$ are given by the finite firing sequences $t_2$ and $(t_2t_1)^*t_3$.

We can encode a Petri net $(P, T, I, O)$ as a multiset rewriting system.\footnote{This encoding is very similar to the one used to encode Petri nets in Concurrent LF~\cite[\S~5.3.1]{cervesato_2003:_concur_logic_framew_ii}.}
Each transition $t \in T$ induces a rule
\[
  t : \ephemfnt{i}_1, \dotsc, \ephemfnt{i}_m \to \ephemfnt{o}_1, \dotsc, \ephemfnt{o}_n
\]
where $I(t) = \{ i_1,\dotsc,i_m \}$ and $O(t) = \{ o_1, \dotsc, o_n \}$.
A marking \(\mu\) specifies a multiset containing $\mu(p)$ many facts $\ephemfnt{p}$ for each place $p \in P$.
Firing a transition $t$ corresponds to applying the rule $t$.
Observe that a transition is enabled if and only if its rule is applicable.

\begin{example}
  \label{ex:ssos-fairness/three-vari-fairn:1}
  The Petri net of \cref{fig:ssos-fairness:three-vari-fairn:petri-two-markings} induces the single rule $t_1 : \ephemfnt{p}_1, \ephemfnt{p}_2 \to \ephemfnt{p}_3$.
  The marking $\mu_0$ corresponds to the multiset $\ephemfnt{p}_1, \ephemfnt{p}_1, \ephemfnt{p}_2$.
  Firing the transition $t_1$, \ie, applying the rule $t_1$, results in the multiset $\ephemfnt{p}_1, \ephemfnt{p}_3$.
\end{example}

Weak and strong fairness for Petri nets exactly correspond to the concepts of weak and strong rule-fairness for multiset rewriting systems.
Consider an MRS $\mc{R}$.
A trace $T = (M_0,(r_i;\delta_i)_{i \in I})$ is:
\begin{itemize}
\item \defin{weakly rule-fair}\DIndex<{rule fairness}\varindex{{multiset rewrite rule} fairness}{1!~2}[|defin] if it is finite, or if it is infinite and for all rules $r \in \mc{R}$, if $r$ is applicable to all but finitely many $M_i$, then there exist infinitely many $i \in I$ such that $r_i = r$;
\item \defin{strongly rule-fair} if it is finite, or if it is infinite and for all $r \in \mc{R}$, if $r$ is applicable to infinitely many $M_i$, then there exist infinitely many $i \in I$ such that $r_i = r$.
\end{itemize}

\subsection{Fact Fairness}
\label{sec:ssos-fairness:fact-fairness}

Rule fairness alone can be insufficient for an intuitively reasonable notion of fairness.
We illustrate this fact by means of an MRS that grows and shrinks trees of finite height.
Consider a collection of formulas $\ephemfnt{B}_n(a, s)$ capturing branches in a tree.
Here, $n \in \N$ denotes some amount of ``growth potential'', $a$ is the branch's ancestor, and $s$ is a globally unique symbol identifying the branch.
The root of a tree is given by a formula $\ephemfnt{B}_n(a,a)$.
We can depict trees-as-multisets graphically, using dots to represent growth potential.
For example, \cref{fig:ssos-fairness:three-vari-fairn:tree} depicts multiset $\ephemfnt{B}_2(a,a), \ephemfnt{B}_1(a,b), \ephemfnt{B}_0(a,c), \ephemfnt{B}_3(c,d)$.

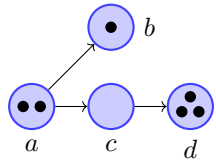
\begin{figure}[htb]
  \centering
  \begin{tikzpicture}[node distance=3em]
    \node [place, tokens=2] (a) [label=below:$a$] at (0,0) {};
    \node [place, tokens=0] (c) [right of=a, label=below:$c$] {};
    \node [place, tokens=1] (b) [above of=c, label=right:$b$] {};
    \node [place, tokens=3] (d) [right of=c, label=below:$d$] {};
    \draw [->] (a) edge (c);
    \draw [->] (c) edge (d);
    \draw [->] (a) edge (b);
  \end{tikzpicture}
  \caption{Visualization of the multiset $\ephemfnt{B}_2(a,a), \ephemfnt{B}_1(a,b), \ephemfnt{B}_0(a,c), \ephemfnt{B}_3(c,d)$ as a tree}
  \label{fig:ssos-fairness:three-vari-fairn:tree}
\end{figure}

A branch can sprout a new branch if it has positive growth potential.
Branching is given by a family of rules, where we have a ``branching rule'' $b_{j, k}$ for all $n > 0$ and $j, k \geq 0$ such that $j + k = n$:
\[
  b_{j,k} : \forall x y. \ephemfnt{B}_n(x,y) \to \exists z. \ephemfnt{B}_j(x,y), \ephemfnt{B}_k(y,z).
\]
It takes a branch $y$ with potential $n$ and creates a new descendent $z$ with potential $k$, leaving $y$ with potential~$j$.

Consider the following execution starting from a root $a_0$ with two units of growth potential:
\begin{equation}
  \begin{split}
    &\ephemfnt{B}_2(a_0,a_0)\\
    &\xrightarrow{b_{1,1}} \ephemfnt{B}_1(a_0,a_0), \ephemfnt{B}_1(a_0,a_1)\\
    &\xrightarrow{b_{0,1}} \ephemfnt{B}_1(a_0,a_0), \ephemfnt{B}_0(a_0,a_1), \ephemfnt{B}_1(a_1,a_2)\\
    &\vdotswithin{\xrightarrow{b_{0,1}}}\\
    &\xrightarrow{b_{0,1}} \ephemfnt{B}_1(a_0,a_0), \ephemfnt{B}_0(a_0,a_1), \ephemfnt{B}_0(a_1,a_2), \dotsc, \ephemfnt{B}_1(a_n,a_{n + 1})\\
    &\vdotswithin{\xrightarrow{b_{0,1}}}\\
  \end{split}
  \label{eq:ssos-fairness/three-vari-fairn:2}
\end{equation}
It is graphically depicted by the sequence of trees in \cref{fig:ssos-fairness/three-vari-fairn:1}.
This execution grows the tree forever by applying $b_{0,1}$ to the right-most branch in the tree.
It is both weakly and strongly rule-fair.
Indeed, the only rule that is applicable infinitely often is $b_{0,1}$, and it is applied infinitely often.
However, the execution could be deemed ``unfair'' to the branch $a_0$.
Indeed, though $a_0$ has a unit of potential left, it never gets to use it to grow a second branch.
This motivates the notion of fact fairness, which ensures that facts that could be used to take a step are not ignored.

\begin{figure}[tb]
  \centering
  \begin{tikzpicture}
    \node [place, tokens=2] (p1) [label=below:$a_0$] at (0,0) {};

    \node [place, tokens=1] (p2) [label=below:$a_0$] at (1,0) {};
    \node [place, tokens=1] (p3) [label=below:$a_1$] at (2,0) {};
    \draw (p2) edge [->] (p3);

    \node [place, tokens=1] (p4) [label=below:$a_0$] at (3,0) {};
    \node [place, tokens=0] (p5) [label=below:$a_1$] at (4,0) {};
    \node [place, tokens=1] (p6) [label=below:$a_2$] at (5,0) {};
    \draw (p4) edge [->] (p5);
    \draw (p5) edge [->] (p6);

    \node at (6,0) {$\cdots$};

    \node [place, tokens=1] (p7) [label=below:$a_0$] at (7,0) {};
    \node [place, tokens=0] (p8) [label=below:$a_1$] at (8,0) {};
    \node (p9) at (9,0) {$\cdots$};
    \node [place, tokens=0] (p10) [label=below:$a_n$] at (10,0) {};
    \node [place, tokens=1] (p11) [label=below:$a_{n + 1}$] at (11,0) {};
    \draw (p7) edge [->] (p8);
    \draw (p8) edge [->] (p9);
    \draw (p9) edge [->] (p10);
    \draw (p10) edge [->] (p11);

    \node at (12,0) {$\cdots$};
  \end{tikzpicture}
  \caption{Graphical depiction of the multisets in execution \eqref{eq:ssos-fairness/three-vari-fairn:2}}
  \label{fig:ssos-fairness/three-vari-fairn:1}
\end{figure}
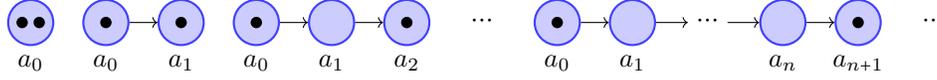

Consider an MRS $\mc{R}$.
We say that a fact $J \in M$ is \defin{enabled}\DIndex/>{fact enabled} in $M$ if there exists an instantiation $r(\theta) : F(\theta) \to \exists \vec n .G(\theta, \vec n)$ of a rule $r \in \mc{R}$ such that $J \in F(\theta)$.
A trace $T = (M_0,(r_i;\delta_i)_{i \in I})$~is:
\begin{itemize}
\item \defin{weakly fact-fair}\DIndex{fact fairness} if it is finite, or if it is infinite and for all facts $J \in \supp(T)$, if $J$ is enabled in all but finitely many $M_i$, then there exist infinitely many $i \in I$ such that $J$ is in the active multiset of $r_i(\theta_i)$;
\item \defin{strongly fact-fair} if it is finite, or if it is infinite and for all facts $J \in \supp(T)$, if $J$ is enabled in infinitely many $M_i$, then there exist infinitely many $i \in I$ such that $J$ is in the active multiset of $r_i(\theta_i)$.
\end{itemize}

\subsection{Instantiation Fairness}
\label{sec:ssos-fairness:inst-fairn}

To illustrate a final variety of fairness, we suspend disbelief and assume that we can water individual branches.
Watering a branch gives it a unit of potential:
\[
  w_n : \forall x y. B_n(x, y) \to B_{n + 1}(x, y).
\]
Intuitively, fairness should ensure that no branch in the tree is left unwatered forever.
Rule-fairness and fact-fairness are insufficient to ensure this in general.
To illustrate this, we make the simplifying assumption that we only have the rule $w_0$ plus a family of rules that redistribute potential:
\[
  r_{m,n} : \forall x y z . B_{m + 1}(x,y), B_{n}(y,z) \to B_m(x,y), B_{n + 1}(y, z).
\]
Consider the execution starting from the multiset $B_0(a,a), B_0(a,b), B_0(b,c)$ given by the sequence of rule instantiations
\begin{multline*}
  w_0(a), r_{0,0}(a,a,b), r_{0,0}(a,b,c),\\
  w_0(a), r_{0,0}(a,a,b), r_{0,1}(a,b,c),\\
  \dotsc, w_0(a), r_{0,0}(a,a,b), r_{0,n}(a,b,c), \dotsc.
\end{multline*}
After the $3n$-th rule, the multiset is $B_0(a,a), B_0(a,b), B_n(b,c)$.
This execution is strongly rule-fair: the only rules applicable infinitely often are $w_0$ and $r_{0,0}$, and they are both applied infinitely often.
It is strongly fact-fair: the only facts enabled infinitely often are $B_0(a,a)$, $B_1(a,a)$, $B_0(a,b)$, and $B_1(a,b)$.
Each of these appears in the active multisets of infinitely many rule applications.
However, the execution is still intuitively unfair because though the branch $b$ could be watered infinitely often, it never gets watered.
Explicitly, the rule instantiation $w_0(a,b)$ is applicable infinitely often, but it is never applied.
To address this, we introduce \textit{instantiation fairness}.

Recall the definition of equivalent instantiations $r_1(\theta_1) \equiv r_2(\theta_2)$ from \cref{def:ssos-fairness/mult-rewr-syst:4}.
Consider an MRS $\mc{R}$.
A trace $T = (M_0,(r_i;\delta_i)_{i \in I})$ is:
\begin{itemize}
\item \defin{weakly instantiation-fair}\DIndex<{instantiation fairness}\varindex{{multiset rewrite rule} instantiation fairness}{1!~2!=3}[|defin] if it is finite, or if it is infinite and for all rules $r \in \mc{R}$ and $\theta$, if $r(\theta)$ is applicable to all but finitely many $M_i$, then there exist infinitely many $i \in I$ such that $r_i(\theta_i) \equiv r(\theta)$;
\item \defin{strongly instantiation-fair} if it is finite, or if it is infinite and for all $r \in \mc{R}$ and $\theta$, if $r(\theta)$ is applicable to infinitely many $M_i$, then there exist infinitely many $i \in I$ such that $r_i(\theta) = r(\theta)$.
\end{itemize}

\subsection{Comparing Varieties of Fairness}
\label{sec:ssos-fairness:comp-types-fairn}

We have a proliferation of varieties of fairness, and it raises the question: are they all useful and independent from each other?
The first question is normative and implicitly presupposes an answer to the question: \textit{useful to what end}?
We motivated each kind of fairness by an application, but it is up to the practitioner to decide whether or not a particular kind of fairness is useful or desirable in a particular setting.
We can, however, answer the second question.
We show that rule, fact, and instantiation fairness are independent notions, no two of which imply the other.
To do so, we construct traces that satisfy two of the forms of fairness, but not the third.

Consider the MRS given by the following rules:
\begin{gather*}
  a : \forall x . A(x) \to A(x)\\
  b : \forall x . A(x), B \to A(x), B.
\end{gather*}
Consider the multiset $M_0 = A(c),A(d),B$.
The trace given by alternating applications of $a(c)$ and $b(d)$ is strongly fact-fair and strongly-rule fair, but it is not weakly or strongly instantiation fair.
Indeed, though the instantiation $b(c)$ is always applicable, it never gets applied.

Now consider the MRS given by the following rule:
\begin{gather*}
  r : \forall xy . A(x),B(y) \to \exists z . A(x),B(z).
\end{gather*}
Consider the multiset $M_0 = A(a),A(a'),B(b_0)$.
Assume that we generate the fresh constants $b_1, b_2, \dotsc$.
Then the trace given by $r(a,b_0),r(a,b_1),r(a,b_2),\dotsc$ is strongly instantiation-fair and strongly rule-fair, but it is not weakly or strongly fact-fair.
Indeed, the only rule $r$ is applied infinitely often, so the trace is strongly rule-fair, and all instantiations of $r$ are globally unique, so the trace is instantiation-fair.
It is not fact fair because, though the fact $A(a')$ is always enabled (it is in the active multiset of $r(a', b_n)$ for each $n$), it is never in the active multiset of a rule in the trace.

Finally, consider the MRS given by the following rules:
\begin{gather*}
  a : \forall x . A(x) \to \exists y . A(y)\\
  b : \forall x . B(x) \to \exists y . B(y)\\
  i : \forall x,y . A(x),B(y) \to A(x),B(y).
\end{gather*}
Consider the multiset $A(a_0),B(b_0)$.
Consider the trace given by alternating applications of $a$ and $b$:
\[
  A(a_0),B(b_0) \to A(a_1)B(b_0) \to A(a_1),B(b_1) \to A(a_2),B(b_1) \to A(a_2),B(b_2) \to \cdots
\]
It is strongly fact-fair: each fact appears at most twice in the trace.
It is strongly instantiation-fair: each rule instantiation is applicable at most twice in the trace.
However, it is not weakly or strongly rule-fair: the rule $i$ never gets applied.

\subsection{Weak, Strong, and Über Fairness}
\label{sec:ssos-fairness:weak-strong-fairness}

We say that a trace is \defin{weakly fair} if it is weakly rule-, fact-, and instantiation-fair, and \defin{strongly fair} if it is strongly rule-, fact-, and instantiation-fair.
Surprisingly, a much stronger notion of fairness arises naturally in applications of multiset rewriting systems.
We say that a trace $(M_0,(r_i;\delta_i)_{i \in I})$ is \defin{über fair}\DIndex<{über fairness} if it is finite, or if for all $i \in I$, $r \in \mc{R}$, and $\theta$, whenever $r(\theta)$ is applicable to $M_i$, there exists a $j > i$ such that $r_j(\theta_j) \equiv r(\theta)$.
Given an über fair trace $T$, we write $\upsilon_T(i,r,\theta)$ for the least such $j$ if it exists.
Every über fair trace is also strongly fair, and every strongly fair is also weakly fair.

\section{Properties of Fair Traces}
\label{sec:ssos-fairness:prop-fair-trac}

We study sufficient conditions for multiset rewriting systems to have fair traces.
One of these conditions, ``interference-freedom'', implies that all fair traces are permutations of each other.
We also study the effects of permuting steps in traces.
We start with various general fairness properties.

The fair tail property is immediate from the definitions of fairness:

\begin{proposition}[Fair Tail Property]\varindex{fairness {fair tail property}}{1!2}
  \label{prop:main:4}
  Let ``fair'' range over the nine notions of fairness considered in \cref{sec:ssos-fairness:weak-strong-fairness}.
  If $(M_0, (t_i;\delta_i)_{i \in I})$ is fair, then so is $(M_n, (t_i;\delta_i)_{\substack{n<i, i \in I}})$ for all $n \in I$.
\end{proposition}

Most notions of fairness are closed under concatenation with finite prefixes:

\begin{proposition}[Fair Concatenation Property]\varindex{fairness {fair concatenation property}}{1!2}
  \label{prop:ssos-fairness/prop-fair-trac:6}
  Let ``fair'' range over weak and strong rule, fact, and instantiation fairness.
  If $M_0 \mssteps M_n$, and $T$ is a fair trace from $M_n$, then $M_0 \mssteps M_n$ followed by $T$ is a fair trace from $M_0$.
  If $T$ is finite, then the result also holds for über fairness.
\end{proposition}

Interference-freedom roughly means that at any given point, the order in which we apply applicable rules does not matter.
Write $S_I$ for the group of bijections on $I$; its elements are called \mdefin{permutations}\DIndex{permutation}.
A permutation $\sigma \in S_I$ acts on a trace $T = (M_0, (t_i;\delta_i)_{i \in I})$ to produce a sequence $\gact{\sigma}{T} = (M_0, (t_{\sigma(i)};\delta_{\sigma(i)})_{i \in I})$.
This sequence $\gact{\sigma}{T}$\glsadd{gact} is a \defin{permutation of}\varindex{{multiset rewriting system} trace permutation}{1!2!3}[|defin] $T$ whenever it is also a trace.
We adopt group-theoretic notation for cyclic permutations and write $(x, \sigma(x), \sigma(\sigma(x)), \dotsc)$ for a cyclic permutation $\sigma : I \to I$; implicit is that all elements not in the orbit of $x$ are fixed by $\sigma$.
For example, let $T$ be given by the sequence of rule instances $t_1, t_2, t_3, t_4, t_5$.
Then $\gact{(4, 3, 2)}{T}$ is the sequence $t_1, t_4, t_2, t_3, t_5$.
Cyclic permutations of length two are called \mdefin{transpositions}\index{transposition|see{permutation}}.

Consider an MRS $\mc{R}$ and let $r_1(\theta_1), \dotsc, r_n(\theta_n)$ enumerate all distinct instantiations of rules in $\mc{R}$ applicable to $M_0$.
We say that $\mc{R}$ \defin{commutes}\DIndex<{commuting {multiset rewriting system}} on $M_0$ or is \defin{interference-free}\DIndex<{interference-free {multiset rewriting system}} on $M_0$ if for all corresponding pairwise-disjoint fresh-constant substitutions $\xi_i$, the following diagram commutes for all permutations $\sigma \in S_{\ordinals{n}}$, and both paths around it are traces:
\[
  \begin{adjustbox}{max width=\linewidth}
    \begin{tikzcd}[column sep=6em, row sep=1em]
      &
      M_1
      \ar[r, "{(r_2;(\theta_2,\xi_2))}"]
      &[1ex]
      \cdots
      \ar[r, "{(r_{n - 1};(\theta_{n - 1},\xi_{n - 1}))}"]
      &[2.5em]
      M_{n - 1}
      \ar[dr, "{(r_n;(\theta_n,\xi_n))}" pos=0.3]
      &
      \\
      M_0
      \ar[ur, "{(r_1;(\theta_1,\xi_1))}"  pos=0.7]
      \ar[dr, swap, "{(r_{\sigma(1)};(\theta_{\sigma(1)},\xi_{\sigma(1)}))}" pos=0.7]
      &
      &
      &
      &
      M_n
      \\
      &
      M_1'
      \ar[r, "{(r_{\sigma(2)};(\theta_{\sigma(2)},\xi_{\sigma(2)}))}"]
      &
      \cdots
      \ar[r, "{(r_{\sigma(n - 1)};(\theta_{\sigma(n - 1)},\xi_{\sigma(n - 1)}))}"]
      &
      M_{n - 1}'
      \ar[ur, swap, "{(r_{\sigma(n)};(\theta_{\sigma(n)},\xi_{\sigma(n)}))}" pos=0.3]
      &
    \end{tikzcd}
  \end{adjustbox}
\]
We note that interference-freedom is only defined if the enumeration of distinct applicable instantiations is finite.
The following proposition is an immediate consequence of the definition of commuting rules:

\begin{proposition}
  \label{prop:fair-mult-rewr:3}
  Let $\mc{R}$ commute on $M_0$, and let $r_i(\theta_i)$ with $1 \leq i \leq n$ be the distinct instantiations applicable on $M_0$.
  If $M_0 \xrightarrow{(r_1;(\theta_1,\xi_1))} M_1$, then $r_2(\theta_2), \dotsc, r_n(\theta_n)$ are applicable to and commute on $M_1$.
\end{proposition}

Interference-freedom implies the existence of über fair executions.
To prove this, we construct a scheduler that enqueues all applicable rules and applies them one by one.
Interference freedom ensures that a rule is still applicable once it reaches the front of the queue.

\begin{proposition}[Fair Scheduler]
  \label{prop:ssos-fairness/prop-fair-trac:9}\varindex{fairness {fair scheduling}}{1!2}
  If $\mc{R}$ is interference-free from $M_0$, then there exists an über fair execution from $M_0$.
\end{proposition}

\begin{proof}
  We define an über fair execution $(r_i;(\theta_i,\xi_i)_i)$ from $M_0$ by induction on the number $n$ of steps.

  Let $Q_0 = r_{01}(\theta_{01}),\dotsc,r_{0m}(\theta_{0m})$ be an enumeration of the distinct instantiations of rules from $\mc{R}$ applicable to $M_0$.
  The enumeration is finite because $\mc{R}$ commutes on $M_0$.
  We maintain for all $n$ the following invariants:
  \begin{itemize}
  \item all distinct rules instantiations applicable to $M_n$ appear in $Q_n$;
  \item all rules in $Q_n$ are applicable to $M_n$ and are distinct instantiations of rules in $\mc{R}$.
  \end{itemize}

  Let $Q_n$ be given; we construct $r_{n + 1}$, $\theta_{n + 1}$, $\xi_{n + 1}$, $M_{n + 1}$ and $Q_{n + 1}$.
  If $Q_n$ is empty, then no rules are applicable to $M_n$ and $(M_0, (r_i;(\theta_i,\xi_i)_i))$ is maximal.
  Otherwise, let $r_{n + 1}(\theta_{n + 1})$ be first element in $Q_n$ so that $Q_n = r_{n + 1}(\theta_{n + 1}),Q_n'$ for some $Q_n'$.
  Pick a suitably fresh $\xi_{n + 1}$.
  Take
  \[
    M_n \xrightarrow{(r_{n + 1};(\theta_{n + 1},\xi_{n + 1}))} M_{n + 1}.
  \]
  By assumption, all rules in $Q_n$ are applicable to $M_n$, and these are all instantiations of rules in $\mc{R}$.
  By interference-freedom, $\mc{R}$ commutes on $M_n$, so all rules in $Q_n'$ are applicable to $M_{n + 1}$ by \cref{prop:fair-mult-rewr:3}.
  Let $Q_{n + 1} = Q_n',N_{n + 1}$, where $N_{n + 1}$ is an enumeration of the distinct rule instantiations applicable to $M_{n + 1}$ not already in $Q_n'$.
  It is finite because $\mc{R}$ commutes on $M_{n + 1}$.
  The invariants hold by construction.

  The resulting trace is clearly über fair: for all $j$, if $r$ is applicable to $M_j$, then it appears at some finite depth $d$ in $Q_{j + 1}$ and will appear in the trace after $d$ steps.
\end{proof}

Though interference-freedom simplifies fair scheduling, it is primarily of interest for reasoning about executions.
For example, it is useful for showing confluence properties.
It also lets us safely permute certain steps in a trace without affecting observations for session-typed processes (see~\cref{sec:sill-obs-equiv:observ-comm}).
This can simplify process equivalence proofs, for it lets us assume that related steps in an execution happen one after another.

Interference-freedom is a strong property, but it holds for many multiset rewriting systems of interest.
This is because many systems can be captured using rules whose active multisets do not intersect, and rules with disjoint active multisets commute.
In fact, even if their active multisets intersect, rules do not disable each other so long as they preserve these intersections.
Because persistent facts are always preserved, it is sufficient to consider only ephemeral facts.

To make this intuition explicit, consider multisets $M_i \subseteq M$ for $1 \leq i \leq n$.
Their \defin{overlap}\DIndex/>{multiset overlap} in $M$ is $\olap{M}{M_1, \dotsc, M_n} = M_1,\dotsc,M_n \setminus M$\glsadd{olap}.
Fix an MRS \(\mc{R}\) and let $r_i(\theta_i) : \persfnt{\pi_i}(\theta_i), F_i(\theta_i) \to \exists \vec n_i . \persfnt{\pi_i'}(\theta_i), G_i(\theta_i, \vec n_i)$, $1 \leq i \leq n$, enumerate all distinct instantiations of rules in $\mc{R}$ applicable to $M$.
We say that $\mc{R}$ is \defin{non-overlapping}\DIndex<{non-overlapping {multiset rewriting system}} on $M$ if for all $1 \leq i \leq n$ and fresh-constant substitutions $\xi_i$, $F_i(\theta_i) \cap \olap{M}{F_1(\theta_1), \dotsc, F_n(\theta_n)} \subseteq G_i(\theta_i, \xi_i)$, \ie, if each rule instantiation preserves its portion of the overlap.

\begin{example}
  \label{ex:ssos-fairness/prop-fair-trac:1}
  The overlap of $A, B$ and $B, C$ in $A, B, C$ is $\olap{A, B, C}{(A, B), (B, C)} = B$.
  The overlap of $A, B$ and $B, C$ in $A, B, B, C$ is the empty multiset $\olap{A, B, B, C}{(A, B), (B, C)} = \emptymset$.
  The overlap of $A, B$ and $B, C$, and $C, A$ in $A, B, C$ is $\olap{A, B, C}{(A, B), (B, C), (C, A)} = A, B, C$.
\end{example}

\begin{example}
  \label{ex:fair-mult-rewr:2}
  The MRS given by \cref{ex:fair-mult-rewr:1} is non-overlapping from any multiset of the form $Q,E$ where $Q$ is a queue rooted at $q$, and $E$ contains at most one fact of the form $\enq(q,v)$.
\end{example}

\Cref{prop:ssos-fairness/prop-fair-trac:3} characterizes the application of non-overlapping rules, while \cref{prop:ssos-fairness/prop-fair-trac:1} characterizes the relationship between commuting and non-overlapping rules.
Because persistent facts pose no difficulty (multiset union is a commutative operation), we elide them from these results for clarity.

\begin{proposition}
  \label{prop:ssos-fairness/prop-fair-trac:3}
  Let $\mc{R}$ be non-overlapping on $M_0$ and let $r_i(\theta_i) : F_i(\theta_i) \to \exists \vec n_i.G(\theta_i, \vec n_i)$ with $1 \leq i \leq n$ be the distinct instantiations applicable to $M_0$.
  If $M_0 \xrightarrow{(r_1;(\theta_1,\xi_1))} M_1$ and $r_1, \dotsc, r_n$ are non-overlapping on $M_0$, then $r_2(\theta_2), \dotsc, r_n(\theta_n)$ are applicable to and non-overlapping on $M_1$.

  In particular, abbreviate $F_i(\theta_i)$ and $G_i(\theta_i,\xi_i)$ by $F_i$ and $G_i$, respectively, and let \(O\) be the overlap $O = \olap{M_0}{F_1, \dotsc, F_n} \cap F_1$.
  There exist $F_1'$ and $G_1'$ be such that $F_1 = O,F_1'$ and $G_1 = O,G_1'$, and there exists an $M$ such that $M_0 = O,F_1',M$ and $M_1 = O,G_1',M$.
  The instantiations $r_2(\theta_2), \dotsc, r_n(\theta_n)$ are all applicable to $O,M \subseteq M_1$.
\end{proposition}

\begin{proof}
  Where $M = (S, m)$ is a multiset and $s \in S$, we abuse notation and write $M(s)$ for the multiplicity $m(s)$ of \(s\) in \(M\).

  Set $O = \olap{M_0}{F_1, \dotsc, F_n} \cap F_1$.
  By assumption, $O \subseteq G_1$, so $F_1 = O,F_1'$ and $G_1 = O,G_1'$ for some $F_1'$ and $G_1'$.
  It follows that $M_0 = O,F_1',M$ and $M_1 = O,G_1',M$ for some $M$.
  We show that $r_2(\theta_2), \dotsc, r_n(\theta_n)$ are all applicable to $O,M$.
  Without loss of generality, we show that $r_2(\theta_2)$ is applicable to $O,M$.
  This requires that we show that $F_2 \subseteq O,M$.
  To do so, we use the function-based definitions of multiset operations and relations of \cref{def:ssos-fairness/mult-rewr-syst:2}.
  To show $F_2 \subseteq O,M$ means to show that $F_2(s) \leq O(s) + M(s)$ for all $s$.
  Recall that \((M_1 \setminus M_2)(s) = \max(0, M_1(s) - M_2(s))\), that \((M_1,M_2)(s) = M_1(s) + M_2(s)\), and that \((M_1 \cap M_2)(s) = \min(M_1(s), M_2(s))\).
  We compute:
  \begin{align}
    (\olap{M_0}{F_1,\dotsc,F_n})(s) &= \left(\left(F_1, \dotsc, F_n\right)\setminus M_0\right)(s)\nonumber\\
                                    &= \max\left(0, \left(\sum_{i = 1}^n F_i(s)\right) - M_0(s)\right)\nonumber\\
                                    &= \max\left(0, \left(\sum_{i = 2}^n F_i(s)\right) - M(s)\right),\label{eq:ssos-fairness/prop-fair-trac:1}\\
    O(s) &= \olap{M_0}{F_1, \dotsc, F_n} \cap F_1\nonumber\\
                                    &= \min\left(F_1(s), \max\left(0, \left(\sum_{i = 2}^n F_i(s)\right) - M(s)\right)\right).\label{eq:ssos-fairness/prop-fair-trac:2}
  \end{align}
  Because $r_2(\theta_2)$ is applicable to $M_0$, we have $F_2 \subseteq M_0$, \ie, $F_2(s) \leq M_0(s) = O(s) + F_1'(s) + M(s)$ for all $s$.
  If $F_1'(s) = 0$, then we are done.
  Assume now that $F_1'(s) > 0$.
  From this it follows that $F_1(s) = F_1'(s) + O(s) > O(s)$.
  We consider three cases for the value of $O(s)$, based on the three possibilities in \cref{eq:ssos-fairness/prop-fair-trac:2}:
  \begin{proofcases}
  \item[$O(s) = F_1(s)$]
    Impossible, for it contradicts $F_1(s) > O(s)$.
  \item[$O(s) = 0$]
    Then $F_1(s) = 0$ or $\max\left(0, \left(\sum_{i = 2}^n F_i(s)\right) - M(s)\right) = 0$.
    The case $F_1(s) = 0$ is impossible because $F_1(s) > O(s) = 0$.
    So $0 \geq \left(\sum_{i = 2}^n F_i(s)\right) - M(s)$, so $M(s) \geq \sum_{i = 2}^n F_i(s)$.
    Because the $F_i$ are all non-negative, $M(s) \geq F_2(s)$, so we are done.
  \item[$O(s) = \left(\sum_{i = 2}^n F_i(s)\right) - M(s)$]
    Then $O(s) + M(s) = \sum_{i = 2}^n F_i(s) \geq F_2(s)$, so we are done.
  \end{proofcases}
  We conclude that $r_2(\theta_2)$ is applicable to $M_1$.

  Next, we show that $r_2(\theta_2), \dotsc, r_n(\theta_n)$ is non-overlapping on $M_1$.
  In particular, we must show that for all $i$, $\olap{M_1}{F_2, \dotsc, F_n} \cap F_i \subseteq G_i$.
  By assumption,
  \[
    \olap{M_0}{F_1, \dotsc, F_n} \cap F_i \subseteq G_i
  \]
  for all $i$, so it is sufficient to show that
  \[
    \olap{M_1}{F_2, \dotsc, F_n} \subseteq \olap{M_0}{F_1, \dotsc, F_n}.
  \]
  To do so, we show that for all $s$,
  \[
    \left(\olap{M_1}{F_2, \dotsc, F_n}\right)(s) \leq \left(\olap{M_0}{F_1, \dotsc, F_n}\right)(s).
  \]
  Let $M_1 = O,G_1',M$ as before.
  We compute:
  \begin{align*}
    (\olap{M_1}{F_2,\dotsc,F_n})(s) &= \max\left(0, \left(\sum_{i = 2}^n F_i(s)\right) - M_1(s)\right)\\
                                    &= \max\left(0, \left(\sum_{i = 2}^n F_i(s)\right) - O(s) - G_1'(s) - M(s)\right).
  \end{align*}
  Because all values involved are non-negative,
  \[
    \left(\sum_{i = 2}^n F_i(s)\right) - O(s) - G_1'(s) - M(s) \leq \left(\sum_{i = 2}^n F_i(s)\right) - M(s).
  \]
  Recall that $\max$ is monotone and recall \cref{eq:ssos-fairness/prop-fair-trac:1}.
  It follows that
  \[
    \left(\olap{M_1}{F_2, \dotsc, F_n}\right)(s) \leq \left(\olap{M_0}{F_1, \dotsc, F_n}\right)(s)
  \]
  as desired.
  We conclude that $r_2(\theta_2),\dotsc,r_n(\theta_n)$ are non-overlapping on $M_1$.
\end{proof}

An inductive argument in the proof of \cref{prop:ssos-fairness/prop-fair-trac:1} relies on the following \namecref{lemma:fair-mult-rewr:2}:

\begin{lemma}
  \label{lemma:fair-mult-rewr:2}
  Consider distinct instantiations $r_i(\theta_i) : F_i(\theta_i) \to \exists \vec n_i.G_i(\theta_i, \vec n_i)$ that are applicable to $M_0$ for $i = 1, 2$.
  If they are non-overlapping on $M_0$, then they commute on $M_0$.
\end{lemma}

\begin{proof}
  We must show that for all disjoint fresh-constant substitutions $\xi_1$ and $\xi_2$, the following diagram commutes and both paths around it are traces:
  \[
    {\begin{tikzcd}[column sep=5em,ampersand replacement=\&]
        M_0
        \ar[r, "{(r_1;(\theta_1,\xi_1))}"]
        \ar[d, swap, "{(r_2;(\theta_2,\xi_2))}"]
        \&
        M_1
        \ar[d, "{(r_2;(\theta_2,\xi_2))}"]
        \\
        M_1'
        \ar[r, "{(r_1;(\theta_1,\xi_1))}"]
        \&
        M_2.
      \end{tikzcd}}
  \]
  Fix some such substitutions and abbreviate $F_i(\theta_i)$ and $G_i(\theta_i,\xi_i)$ by $F_i$ and $G_i$, respectively.
  Both paths around the square are traces by \cref{prop:ssos-fairness/prop-fair-trac:3}; we show that it commutes.

  Let $O_i' = F_i \cap \olap{M_0}{F_1, F_2}$, $O_{12} = O_1' \cap O_2'$, and $O_i = O_i' \setminus O_{12}$.
  By assumption, $O_i' \subseteq G_i$, so for some $F_i'$ and $G_i'$ we have $F_i(\theta_i) = O_i',F_i' = O_{12},O_i',F_i'$ and $G_i = O_i',G_i' = O_{12},O_i',G_i'$.
  Because $F_i \subseteq M_0$ for $i = 1,2$, it follows that $M_0 = O_{12},O_1,O_2,F_1',F_2',M$.
  Then the two paths are
  \begin{gather*}
    M_0 \xrightarrow{(r_1;(\theta_1,\xi_1))} O_{12},O_1,O_2,G_1',F_2',M \xrightarrow{(r_2;(\theta_2,\xi_2))} O_{12},O_1,O_2,G_1',G_2',M\\
    M_0 \xrightarrow{(r_2;(\theta_2,\xi_2))} O_{12},O_1,O_2,F_1',G_2',M \xrightarrow{(r_1;(\theta_1,\xi_1))} O_{12},O_1,O_2,G_1',G_2',M.
  \end{gather*}
  We conclude that the diagram commutes.
\end{proof}

In \cite[Proposition~5]{kavanagh_2020:_subst_obser_commun_seman}, we claimed that the converse of \cref{prop:ssos-fairness/prop-fair-trac:1} was false.
The counter-example used does not support this claim.

\begin{proposition}
  \label{prop:ssos-fairness/prop-fair-trac:1}
  An MRS commutes on $M_0$ if it is non-overlapping on $M_0$.
\end{proposition}

\begin{proof}
  Assume that the rules are non-overlapping on $M_0$.
  Let $r_i(\theta_i) : F_i(\theta_i) \to \exists \vec n_i.G(\theta_i, \vec n_i)$ enumerate the distinct instantiations that are applicable to $M_0$ with $1 \leq i \leq n$.
  Consider pairwise-disjoint fresh-constant substitutions $\xi_i$ for $1 \leq i \leq n$.
  We must show that $(M_0, (r_i;(\theta_i,\xi_i))_{1 \leq i \leq n})$ is a trace and that permuting its steps does not change the last multiset.

  We proceed by induction on $n$ to show that it is a trace.
  Informally, \cref{prop:ssos-fairness/prop-fair-trac:3} will ensure that no matter which instantiation we apply to get to the next multiset, the remaining instantiations will be applicable to and non-overlapping on that multiset.
  The result is immediate when $n = 0$.
  Assume that the result holds for some $k$, and assume that $n = k + 1$.
  By the induction hypothesis, the following sequence of $k$ steps is a trace:
  \[
    M_0 \xrightarrow{(r_1;(\theta_1,\xi_1))} M_1 \xrightarrow{(r_2;(\theta_2,\xi_2))} \cdots \xrightarrow{(r_{k - 1};(\theta_{k - 1},\xi_{k - 1}))} M_{k - 1} \xrightarrow{(r_k;(\theta_k,\xi_k))} M_k.
  \]
  We proceed by induction on $0 \leq j \leq k$ to show that $r_{j + 1}(\theta_{j + 1}), \dotsc, r_{k + 1}(\theta_{k + 1})$ are all applicable and non-overlapping on $M_j$.
  \begin{proofcases}
  \item[$j = 0$] Immediate from the assumption that the rules are non-overlapping on $M_0$.
  \item[$j = j' + 1$ with $j' < k$] The instantiations $r_j(\theta_j), \dotsc, r_{k + 1}(\theta_{k + 1})$ are all applicable to $M_j$ by the induction hypothesis on $j'$.
    By \cref{prop:ssos-fairness/prop-fair-trac:3}, the instantiations $r_{j + 1}(\theta_{j + 1}), \dotsc, r_{k + 1}(\theta_{k + 1})$ are all applicable and non-overlapping on $M_{j + 1}$.
  \end{proofcases}
  This completes the nested induction on $j$, and we conclude that $r_{k + 1}(\theta_{k + 1})$ is applicable to $M_k$.
  Pairwise-disjointness of the $\xi_i$ guarantees that the resulting sequence of $k + 1$ steps is a trace.
  This completes the induction on $n$.

  Next, we observe that any permutation of the above trace is a trace.
  Indeed, the numbering of the $r_i(\theta_i) : F_i(\theta_i) \to \exists \vec n_i.G(\theta_i, \vec n_i)$ was arbitrary, and the proof that $(M_0, (r_i;(\theta_i,\xi_i))_{1 \leq i \leq n})$ was a trace did not depend in any way on the numbering of the rules or on the individual rules themselves.

  Finally, we show that all permutations of the trace have the same final multiset.
  Concretely, fix some $\sigma$ and let $(N_0, (r_{\sigma(i)}; (\theta_{\sigma(i)},\xi_{\sigma(i)}))_{i \in \ordinals{n}})) = \gact{\sigma}{(M_0, (r_i;(\theta_i,\xi_i))_{i \in \ordinals{n})}}$.
  We show that $M_n = N_n$.
  The permutation $\sigma$ factors as a product of transpositions of adjacent steps by the proof of \cite[Corollary~I.6.5]{hungerford_1974:_algeb}.
  By the previous paragraph, each of these transpositions preserves the property of being a trace.
  By \cref{lemma:fair-mult-rewr:2}, each transposition preserves the multisets at its endpoints.
  In particular, each transposition preserves $M_n$.
  An induction on the number of transpositions in the factorization of $\sigma$ gives that $M_n = N_n$.
\end{proof}

Given an MRS $\mc{R}$ and a property $P$ of multisets, we say $P$ holds \defin{from} $M_0$ if for all traces $(M_0, (r_i;\delta_i)_{i \in I})$, $P$ holds for $M_0$ and for $M_i$ for all $i \in I$.
Weak, strong, and über fairness coincide in the presence of interference-freedom:

\begin{proposition}\varindex{fairness equivalence under interference-freedom}{1!234}
  \label{prop:ssos-fairness/prop-fair-trac:2}
  If $\mc{R}$ is interference-free from $M_0$ and $T$ is a trace from $M_0$, then the following are equivalent: $T$ is weakly instantiation-fair, $T$ is strongly instantiation-fair, $T$ is über fair.
\end{proposition}

\begin{proof}
  It is clear that über fairness implies both forms of instantiation-fairness, and that strong instantiation-fairness implies weak instantiation-fairness.
  It is sufficient to show that if $T$ is weakly instantiation-fair, then it is über fair.

  Assume that $T$ is weakly instantiation-fair.
  If $T$ is finite, then we are done, so assume that $T = (M_0;(r_i,\delta_i)_i)$ is infinite.
  Consider some arbitrary $M_i$, and assume that $r(\theta)$ is applicable to $M_i$.
  We must show that there exists some $k \geq 1$ such that $r_{i + k}(\theta_{i + k}) \equiv r(\theta)$.
  By induction on $k \geq 1$ with \cref{prop:ssos-fairness/prop-fair-trac:3}, we know that if $r_{i + k}(\theta_{i + k}) \not \equiv r(\theta)$, then $r(\theta)$ is applicable to $M_{i + k + 1}$.
  This implies that $r_{i + k}(\theta_{i + k}) \equiv r(\theta)$ for some $k \geq 1$ or that $r(\theta)$ is applicable to all but finitely many $M_j$ (\ie, at most the first $i$ multisets).
  In the first case, we are done.
  In the latter case, we know that $r(\theta)$ is applied infinitely often by weak instantiation-fairness.
  We conclude that $T$ is über fair.
\end{proof}

\begin{corollary}[Fair Concatenation for Über Fairness]\varindex{fairness {fair concatenation property}}{1!2}
  \label{cor:ssos-fairness/prop-fair-trac:2}
  If $\mc{R}$ is interference-free from $M_0$, $M_0 \mssteps M_n$, and $T$ is an über fair trace from $M_n$, then $M_0 \mssteps M_n$ followed by $T$ is an über fair trace from $M_0$.
\end{corollary}

\begin{proof}
  Immediate by \cref{prop:ssos-fairness/prop-fair-trac:6,prop:ssos-fairness/prop-fair-trac:2}.
\end{proof}

In light of \cref{prop:ssos-fairness/prop-fair-trac:2}, we hereinafter use the words ``fair trace'' to equivalently mean ``weakly fair trace'', ``strongly fair trace'', or ``über fair trace'' when assuming interference-freedom.

\begin{assumption}
  For the remainder of this section, we assume that if $(M_0,(r_i;\delta_i)_i)$ is a fair trace, then its MRS is interference-free from $M_0$.
\end{assumption}

In the remainder of this section, we study permutations of fair traces.
We first show that fairness is invariant under permutation.
Then, we show that all fair executions are permutations of each other.

Interference-freedom implies the ability to safely permute finitely many steps that do not depend on each other.
However, it is not obvious that finite permutations, let alone infinite permutations, preserve fairness.
We begin by showing that finite permutations preserve fairness.
Our proof relies on the fact that finite permutations factor as products of cycles, which themselves factor as products of transpositions.
We begin by showing that transpositions of adjacent steps preserve fairness.

\begin{lemma}
  \label{lemma:fair-mult-rewr:1}
  Let $\mc{R}$ be interference-free from $M_0$ and let $T = (M_0,(r_i;(\theta_i,\xi_i))_{i \in I})$ be a trace, an execution, or a fair trace.
  For all transpositions $(m + 1, m) \in S_I$, if $r_{m + 1}(\theta_{m + 1})$ is applicable to $M_{m - 1}$, then $\gact{(m + 1, m)}{T}$ is respectively a trace, an execution, or a fair trace.
\end{lemma}

\begin{proof}
  Consider some transpositions $(m + 1, m) \in S_I$ such that $r_{m + 1}(\theta_{m + 1})$ is applicable to $M_{m - 1}$.
  By non-interference, it follows that $\gact{(m + 1, m)}{T} = (M_0, (r_i';(\theta_i',\xi_i'))_i)$ is also a trace.
  Observe that for all $j$, if $j \neq m$, then $M_j' = M_j$, and if $j < m$ or $j > m + 1$, then $(r_j';(\theta_j',\xi_j')) = (r_j;(\theta_j,\xi_j))$.

  Assume that $T$ is an execution.
  We must show that $\gact{(m + 1, m)}{T}$ is also maximal.
  If it is infinite, then we are done.
  If it ends at some $M_n'$, then $M_n' = M_n$ by the above observation because $m < m + 1 \leq n$.
  Because $T$ is maximal, no rules are applicable to $M_n$, so no rules are applicable to $M_n'$.
  We conclude that $\gact{(m + 1, m)}{T}$ is maximal.

  Assume now that $T$ is fair.
  We show that $\gact{(m + 1, m)}{T}$ is also fair.
  Consider some $r(\theta)$ applicable to $M_j'$.
  We proceed by case analysis on $j < m$, $j = m$, and $j > m$ to show that $r(\theta)$ appears as some $r_k'(\theta_k')$ with $k > j$.
  \begin{proofcases}
  \item[$j < m$]
    By the above observations, $M_j' = M_j$.
    Because $T$ is fair, there exists a $k' > j$ such that $r_{k'}(\theta_{k'}) \equiv r(\theta)$.
    Because $j < m$ and $j < k$, it follows that $k = (m + 1, m)(k') > j$.
    So $r(\theta)$ appears as $r_k'(\theta_k')$ after $M_j'$ as desired.
  \item[$j = m$]
    By \cref{prop:fair-mult-rewr:3}, $r(\theta)$ is applicable to $M_{m + 1}$.
    Because $T$ is fair, there exists a $k > m + 1$ such that $r_k(\theta_k) \equiv r(\theta)$.
    So $r_k'(\theta_k') = r_k(\theta_k) \equiv (r, \theta)$.
  \item[$j > m$]
    Because $T$ is fair, there exists a $k > j$ such that $r_k(\theta_k) \equiv r(\theta)$.
    Because $k > j > m$, $k > m + 1$, so $r_k'(\theta_k') = r_k(\theta_k) \equiv (r, \theta)$.
  \end{proofcases}
  We conclude that $\gact{(m + 1, m)}{T}$ is fair whenever $T$ is fair.
\end{proof}

\begin{proposition}
  \label{prop:fair-mult-rewr:7}
  Let $\mc{R}$ be interference-free from $M_0$ and let $T = (M_0,(r_i;(\theta_i,\xi_i))_{i \in I})$ be a trace, an execution, or a fair trace.
  For all cycles $(m + k, \dotsc, m + 1, m) \in S_I$ with $k \geq 1$, if $r_{(m + k)}(\theta_{m + k})$ is applicable to $M_{m - 1}$, then $\gact{(m + k, \dotsc, m + 1, m)}{T}$ is respectively a trace, an execution, or a fair trace.
  It is equal to $T$ after the $(m + k)$-th step.
\end{proposition}

\begin{proof}
  By induction on $k$.
  If $k = 1$, then we are done by \cref{lemma:fair-mult-rewr:1} and non-interference.
  Assume the result for some $k'$, and consider the case $k = k' + 1$.
  Because $\mc{R}$ is interference-free from $M_0$, it follows that if $r_{(m + k)}(\theta_{m + k})$ is applicable to $M_{m - 1}$, then it is also applicable to $M_m$.
  By the induction hypothesis, $T' = \gact{(m + k, \dotsc, m + 1)}{T}$ is respectively a trace, an execution, or a fair trace.
  By \cref{lemma:fair-mult-rewr:1}, so is $T'' = \gact{(m, m + 1)}{T'}$, for the $(m + 1)$-th step in $T'$ is $r_{(m + k)}(\theta_{m + k})$, and it is assumed to be applicable to $M_{m - 1}$.
  The transposition $(m + 1, m)$ does not alter $T'$ beyond the $(m + k)$-th step, so $T''$ still agrees with $T$ after the $(m + k)$-th step.
  Observe that
  \[
    T'' = \gact{(m + 1, m)}{T'} = \gact{\left((m + k, m) \circ (m + k, \dotsc, m + 1)\right)}{T} = \gact{(m + k, \dotsc, m + 1, m)}{T},
  \]
  so we are done.
  We note that the second equality in the above sequence is subtle: the transposition $(m + 1, m)$ is relative to the ordering of rules in $T'$.
  It becomes $(m + k, m)$ on the right of the equality because the $(m + 1)$-th step in $T'$ is the $(m + k)$-th step of $T$.
\end{proof}

We conclude that finite permutations preserve fairness:

\begin{proposition}
  \label{prop:fair-mult-rewr:6}
  Let $\mc{R}$ be interference-free from $M_0$ and let $T = (M_0,(r_i;(\theta_i,\xi_i))_{i \in I})$ be a trace.
  Let $\sigma \in S_I$ be a finite permutation, \ie, assume that there exists an $n \in I$ such that $\sigma(i) = i$ for all $i > n$.
  Further assume that $\gact{\sigma}{T}$ is a trace.
  If $T$ is an execution or a fair trace, then $\gact{\sigma}{T}$ is respectively an execution or a fair trace.
  The traces $\gact{\sigma}{T}$ and $T$ are equal after $n$-th step.
\end{proposition}

\begin{proof}
  Informally, the approach is to decompose $\sigma$ into a finite composition of cyclic permutations of the form $(\sigma(m), \dotsc, m + 1, m)$.
  \Cref{prop:fair-mult-rewr:7} ensures that each of these cycles preserves the desired properties.

  Let $m$ be minimal in $I$ such that $\sigma(m) \neq m$; if no such $m$ exists, \ie, if $\sigma$ is the identity, then set $m = n$.
  We proceed by strong induction on $d = n - m$.
  If $d = 0$, then $\gact{\sigma}{T} = T$ and we are done.
  Assume the result for some $d'$, and consider the case $d = d' + 1$.
  By minimality of $m$, the traces $T$ and $\gact{\sigma}{T}$ are equal up until the multiset $M_{m - 1}$.
  By the assumption that $T$ and $\gact{\sigma}{T}$ are both traces, it follows that $r_m(\theta_m)$ and $r_{\sigma(m)}(\theta_{\sigma(m)})$ are both applicable to $M_{m - 1}$.
  By minimality, it also follows that $m < \sigma(m)$.
  By \cref{prop:fair-mult-rewr:7}, $T' = \gact{(\sigma(m), \dotsc, m + 1, m)}{T}$ is also respectively a trace, an execution, or a fair trace.
  By the same \namecref{prop:fair-mult-rewr:7}, it also agrees with $T$ on all multisets after the $\sigma(m)$-th.
  Since $\sigma(m) < n$, it follows that $T'$ and $T$ are equal after the $n$-th step.
  The trace $T'$ agrees with $\gact{\sigma}{T}$ on at least the first $m$ steps.
  This decreases $d$ by at least one.
  We conclude the result by the strong induction hypothesis on $\gact{(\sigma(m), \dotsc, m + 1, m)}{T}$ and $\gact{\sigma}{T}$.
\end{proof}

Next, we show that infinite permutations preserve fairness.
To do so, we use the following \namecref{lemma:main:9} to reduce arguments about infinite permutations to arguments about finite permutations.
Intuitively, it decomposes any infinite permutation $\sigma$ on $\N$ into the composition of a permutation $\tau$ that only permutes natural numbers less than some $\chi_\sigma(n)$, and of an infinite permutation $\rho$ that only permutes natural numbers greater than $\chi_\sigma(n)$.
The mechanics of the decomposition are best understood by means of a picture.
We refer the reader to \cref{fig:ssos-fairness/prop-fair-trac:1} for an illustration, where we note that $\chi_\sigma(0) = 1$ and $\chi_\sigma(1) = 2$.

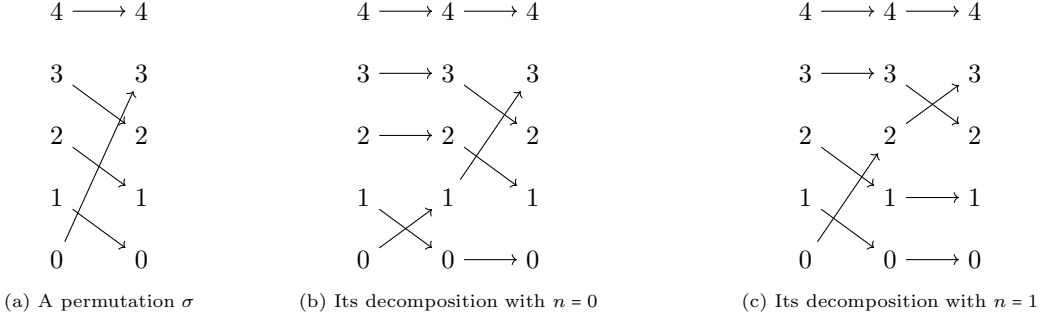
\begin{figure}
  \centering
  \begin{subfigure}{0.2\textwidth}
    \centering
    \begin{tikzpicture}[node distance={1em and 2em},auto]
      \node at (0,0) (l0) {$0$};
      \node (l1) [above=of l0] {$1$};
      \node (l2) [above=of l1] {$2$};
      \node (l3) [above=of l2] {$3$};
      \node (l4) [above=of l3] {$4$};

      \node (r0) [right=of l0] {$0$};
      \node (r1) [above=of r0] {$1$};
      \node (r2) [above=of r1] {$2$};
      \node (r3) [above=of r2] {$3$};
      \node (r4) [above=of r3] {$4$};

      \draw[->] (l0) to (r3);
      \draw[->] (l1) to (r0);
      \draw[->] (l2) to (r1);
      \draw[->] (l3) to (r2);
      \draw[->] (l4) to (r4);
    \end{tikzpicture}
    \caption{A permutation $\sigma$}
  \end{subfigure}
  \begin{subfigure}{0.35\textwidth}
    \centering
    \begin{tikzpicture}[node distance={1em and 2em},auto]
      \node at (0,0) (l0) {$0$};
      \node (l1) [above=of l0] {$1$};
      \node (l2) [above=of l1] {$2$};
      \node (l3) [above=of l2] {$3$};
      \node (l4) [above=of l3] {$4$};

      \node (m0) [right=of l0] {$0$};
      \node (m1) [above=of m0] {$1$};
      \node (m2) [above=of m1] {$2$};
      \node (m3) [above=of m2] {$3$};
      \node (m4) [above=of m3] {$4$};

      \node (r0) [right=of m0] {$0$};
      \node (r1) [above=of r0] {$1$};
      \node (r2) [above=of r1] {$2$};
      \node (r3) [above=of r2] {$3$};
      \node (r4) [above=of r3] {$4$};

      \draw[->] (l0) to (m1);
      \draw[->] (m1) to (r3);
      \draw[->] (l1) to (m0);
      \draw[->] (m0) to (r0);
      \draw[->] (l2) to (m2);
      \draw[->] (m2) to (r1);
      \draw[->] (l3) to (m3);
      \draw[->] (m3) to (r2);
      \draw[->] (l4) to (m4);
      \draw[->] (m4) to (r4);
    \end{tikzpicture}
    \caption{Its decomposition with $n = 0$}
  \end{subfigure}
  \begin{subfigure}{0.35\textwidth}
    \centering
    \begin{tikzpicture}[node distance={1em and 2em},auto]
      \node at (0,0) (l0) {$0$};
      \node (l1) [above=of l0] {$1$};
      \node (l2) [above=of l1] {$2$};
      \node (l3) [above=of l2] {$3$};
      \node (l4) [above=of l3] {$4$};

      \node (m0) [right=of l0] {$0$};
      \node (m1) [above=of m0] {$1$};
      \node (m2) [above=of m1] {$2$};
      \node (m3) [above=of m2] {$3$};
      \node (m4) [above=of m3] {$4$};

      \node (r0) [right=of m0] {$0$};
      \node (r1) [above=of r0] {$1$};
      \node (r2) [above=of r1] {$2$};
      \node (r3) [above=of r2] {$3$};
      \node (r4) [above=of r3] {$4$};

      \draw[->] (l0) to (m2);
      \draw[->] (m2) to (r3);
      \draw[->] (l1) to (m0);
      \draw[->] (m0) to (r0);
      \draw[->] (l2) to (m1);
      \draw[->] (m1) to (r1);
      \draw[->] (l3) to (m3);
      \draw[->] (m3) to (r2);
      \draw[->] (l4) to (m4);
      \draw[->] (m4) to (r4);
    \end{tikzpicture}
    \caption{Its decomposition with $n = 1$}
  \end{subfigure}
  \caption{An illustration of decompositions of $\sigma$ given by \cref{lemma:main:9}}
  \label{fig:ssos-fairness/prop-fair-trac:1}
\end{figure}

\begin{lemma}
  \label{lemma:main:9}
  For all $n \in \N$ and permutations $\sigma : \N \to \N$, set $\chi_\sigma(n) = \max_{k \leq n} \sigma^{-1}(k)$.
  Then there exist permutations $\tau_n, \rho_n : \N \to \N$ such that $\sigma = \rho_n \circ \tau_n$, $\tau_n(k) = k$ for all $k > \chi_\sigma(n)$, and $\rho_n(k) = k$ for all $k \leq n$.
\end{lemma}

\begin{proof}
  Let $s_1 < \dotsb < s_m$ be the elements of $\{ 0, 1, \dotsc, \chi_\sigma(n) \} \setminus \sigma^{-1}(\{ 0, \dotsc, n \})$.
  Explicitly, these are the natural numbers less than $\chi_\sigma(n)$ whose image under $\sigma$ is greater than $n$.

  Let $\tau$ be given by
  \[
    \tau(k) =
    \begin{cases}
      \sigma(k) & \sigma(k) \leq n\\
      n + j & k = s_j\\
      k & k > \chi_\sigma(n).
    \end{cases}
  \]
  Intuitively, $\tau$ acts as $\sigma$ on elements whose image is less than $n$; it then ``stacks'' the remaining elements less than $\chi_\sigma(n)$ in order on top of $n$ (we refer the reader to \cref{fig:ssos-fairness/prop-fair-trac:1} for this spatial intuition); and it fixes elements greater than $\chi_\sigma(n)$.

  Let $\rho$ be given by
  \[
    \rho(k) =
    \begin{cases}
      k & k \leq n\\
      \sigma(s_{k - n}) & n < k \leq \chi_\sigma(n)\\
      \sigma(k) & k > \chi_\sigma(n).
    \end{cases}
  \]
  Intuitively, $\rho$ takes the ``stacked'' elements and recovers their original value before applying $\sigma$; it directly applies $\sigma$ to those greater than $\chi_\sigma(n)$.

  We show that $\sigma = \rho \circ \tau$.
  Let $k \in \N$ be arbitrary.
  We proceed by case analysis:
  \begin{proofcases}
  \item[$\sigma(k) \leq n$] Then $\rho(\tau(k)) = \rho(\sigma(k))) = \sigma(k)$.
  \item[$k = s_j$] Then $\rho(\tau(k)) = \rho(n + j) = \sigma(s_{(n + j) - n}) = \sigma(s_j) = \sigma(k)$.
  \item[$k > \chi_\sigma(n)$] Then $\rho(\tau(k)) = \rho(k) = \sigma(k)$.
  \end{proofcases}
  The function $\tau$ is clearly total and an isomorphism.
  Because $\sigma$ is an isomorphism and $\sigma = \rho \circ \tau$, it follows that $\rho$ is also an isomorphism.
  So $\tau$ and $\rho$ are two permutations with the desired property.
\end{proof}

\begin{lemma}
  \label{lemma:ssos-fairness/prop-fair-trac:1}
  Assume that $\mc{R}$ is interference-free from $M_0$, that $T$ is a fair trace from $M_0$, and that $\gact{\sigma}{T}$ is a permutation of $T$.
  For all $n$, $\gact{\tau_n}{T}$ is a fair trace, where $\sigma = \rho_n \circ \tau_n$ is the decomposition given by \cref{lemma:main:9}.
\end{lemma}

\begin{proof}
  This proof is analogous to the proof of \cref{prop:fair-mult-rewr:6}.
  Informally, we decompose $\tau_n$ as a potentially empty composition of cyclic permutations such that each successive cyclic permutation increases the length of the prefix of the trace that agrees with $\gact{\tau_n}{T}$ and preserves the relative ordering of the $s_i$ described in the proof of \cref{lemma:main:9}.

  Let $m$ be minimal such that $\tau_n(m) \neq m$; if no such $m$ exists, then $\gact{\tau_n}{T} = T$ and we are done.
  Otherwise, the $\sigma(m)$-th step $r_{\sigma(m)}(\theta_{\sigma(m)})$ of $T$ is applicable to $M_{m - 1}$.
  Indeed, $\sigma(m) = \tau_n(m)$, and $\gact{\sigma}{T}$ is a trace by assumption.
  By \cref{prop:fair-mult-rewr:7}, $T' = \gact{(\sigma(m), \dotsc, m + 1, m)}{T}$ is a fair trace.
  The trace $T'$ agrees with $\gact{\tau_n}{T}$ on at least the first $m$ steps.
  Moreover, the relative ordering of the $s_i$ in $T'$ is the same as in $T$.

  Iterating this procedure results in a trace $T''$ that agrees with $\gact{\tau_n}{T}$ on the first $n$ steps.
  Indeed, this procedure terminates because after each iteration, the number of steps in the first $n$ that disagree decreases by one.
  The resulting trace $T''$ also agrees with $\gact{\tau_n}{T}$ on all steps after the $(\chi_\sigma(n))$-th.
  Indeed, for each cycle $(\sigma(m), \dotsc, m + 1, m)$, $\sigma(m) \leq \chi_\sigma(n)$.
  Because both $\tau_n$ and the above procedure preserves the relative ordering of the $s_i$, and both result in permutations, it follows that their images agree on all of the steps between the $n$-th and the $\chi_\sigma(n)$-th.
  So $T''$ and $\gact{\tau_n}{T}$ are equal.
  Because $T''$ is fair, we conclude that $\gact{\tau_n}{T}$ is fair.
\end{proof}

\begin{corollary}
  \label{prop:main:6}\varindex{fairness effects of@ permutation}{1!234}\varindex{permutation effects on@ fairness}{1!234}
  Fairness is invariant under permutation, that is, if $\mc{R}$ is interference-free from $M_0$, $T$ is a fair trace from $M_0$, and $\Sigma = \gact{\sigma}{T}$ is a permutation of $T$, then $\Sigma$ is also fair.
\end{corollary}

\begin{proof}
  Let $T = (M_0, (t_i;\delta_i)_i)$, and let $\Sigma$ be the trace $M_0 = \Sigma_0 \xrightarrow{(t_{\sigma(1)};\delta_{\sigma(1)})} \Sigma_1 \xrightarrow{(t_{\sigma(2)};\delta_{\sigma(2)})} \cdots$.
  Consider some rule $r \in \mc{R}$ and $\theta$ such that $r(\theta)$ is applicable to $\Sigma_i$.
  We must show that there exists a $j$ such that $\sigma(j) > \sigma(i)$, $t_{\sigma(j)}(\theta_{\sigma(j)}) \equiv r(\theta)$.

  Let the factorization $\sigma = \rho \circ \tau$ be given by \cref{lemma:main:9} for $n = \sigma(i)$.
  The trace $\gact{\tau}{T}$ is fair by \cref{lemma:ssos-fairness/prop-fair-trac:1}.
  By construction of $\tau$, $\gact{\tau}{T}$ and $\Sigma$ agree on the first $n$ steps and $n + 1$ multisets.
  By fairness, there exists a $k > \sigma(i)$ such that the $k$-th step in $\gact{\tau}{T}$ is $r(\theta)$.
  By construction of $\rho$, $\rho(k) > \sigma(i)$, so this step appears after $\Sigma_i$ in $\Sigma$ as desired.
  We conclude that $\Sigma$ is fair.
\end{proof}

It is not the case that every permutation of the steps of a fair trace is a fair trace: it could fail to be a trace.
\Cref{prop:main:6} simply states that if the result of permuting the steps of a fair trace is a trace, then that trace is fair.

\Cref{prop:main:6} established that permutations preserve fairness.
Relatedly, all fair traces from a given multiset are permutations of each other.
To show this, we construct a potentially infinite sequence of permutations.
We use the following lemma to compose them:

\begin{lemma}
  \label{lemma:main:2}
  Let $(\sigma_n)_{n \in I}$ be a family of bijections on $I$ such that for all $m < n$,
  \[
    (\sigma_n \circ \dotsb \circ \sigma_1)(m) = (\sigma_m \circ \dotsb \circ \sigma_1)(m).
  \]
  Let $\sigma : I \to I$ be given by $\sigma(m) = (\sigma_m \circ \dotsb \circ \sigma_1)(m)$.
  Then $\sigma$ is injective, but need not be surjective.
\end{lemma}

\begin{proof}
  Let $m, n \in I$ be arbitrary such that $\sigma(m) = \sigma(n)$.
  Assume without loss of generality that $m \leq n$.
  Observe that
  \[
    \sigma(m) = (\sigma_m \circ \dotsb \circ \sigma_0)(m) = (\sigma_n \circ \dotsb \circ \sigma_m \circ \dotsb \circ \sigma_0)(m)
  \]
  and
  \[
    \sigma(n) = (\sigma_n \circ \dotsb \circ \sigma_m \circ \dotsb \circ \sigma_0)(n).
  \]
  Because $\sigma_0,\dotsc,\sigma_n$ are all bijections, so is their composition.
  It follows that $m = n$, so $\sigma$ is injective.

  To see that $\sigma$ need not be surjective, consider the family $\sigma_n = (0, n)$ for $n \geq 1$.
  Then $\sigma(n) = n + 1$ for all $n$.
  It follows that $0$ is not in the image of $\sigma$.
\end{proof}

Recall from \cref{sec:ssos-fairness:weak-strong-fairness} that, given an über fair trace $T$ and an instantiation $t(\tau)$ applicable to its $i$-th multiset, $\upsilon_T(i,t,\tau)$ is the least $j > i$ such that the $j$-th step of $T$ is equivalent to $t(\tau)$.
The following \namecref{lemma:main:1} is a special case of \cref{prop:fair-mult-rewr:7}:

\begin{lemma}
  \label{lemma:main:1}
  Let $\mc{R}$ be interference-free from $M_0$ and $T$ a fair execution from $M_0$.
  If $t(\tau)$ is applicable to $M_0$, then $\gact{(\upsilon_T(0, t, \tau), \dotsc, 0)}{T}$ is a permutation of $T$ with $t(\tau)$ as its first step, and it is a fair execution.
\end{lemma}

\begin{proposition}
  \label{prop:ssos-fairness/prop-fair-trac:4}\varindex{fairness effects of@ permutation}{1!234}\varindex{permutation effects on@ fairness}{1!234}
  If $\mc{R}$ is interference-free from $M_0$, then all fair executions from $M_0$ are permutations of each other.
\end{proposition}

\begin{proof}
  Consider traces $R = (R_0, (r_i; (\theta_i, \xi_i))_{i \in I})$ and $T = (T_0, (t_j; (\tau_j, \zeta_j))_{j \in J})$ where $R_0 = M_0 = T_0$.
  We construct a sequence of permutations $\sigma_0, \sigma_1, \dotsc$, where $\Phi_0 = R$ and the step $\Phi_{n + 1} = \gact{\sigma_{n + 1}}{\Phi_n}$ is given by \cref{lemma:main:1} such that $\Phi_{n + 1}$ agrees with $T$ on the first $n + 1$ steps.
  We then assemble these permutations $\sigma_n$ into an injection $\sigma$ using \cref{lemma:main:2}; fairness ensures that it is a surjection.
  We have $T = \gact{\sigma}{R}$ by construction.

  The construction is as follows.
  We write $\Sigma_0 = R$ horizontally, and then down from $R_0$, we write $T$.
  Let $\Gamma_{n + 1}$ be the fair trace obtained by applying \cref{lemma:main:1} to $\Sigma_n$ and $(t_{n + 1};(\tau_{n + 1},\zeta_{n + 1}))$; it is a permutation of $\Sigma_n$.
  Without loss of generality, we assume that the fresh constant substitutions $\zeta_{n + 1}$ and the corresponding $\xi_k$ in $\Sigma_n, \dotsc, \Sigma_0$ are equal; refreshing both $S$ and $T$ makes this possible.
  Let $\Sigma_{n + 1}$ be the tail of $\Gamma_{n + 1}$ starting at $T_{n + 1}$; we write it horizontally to the right of $T_{n + 1}$.
  We get the following picture:
  \[
    \begin{tikzcd}[column sep=3.5em, ampersand replacement=\&]
      R_0 = T_0
      \ar[r, "{(r_1;(\theta_1,\xi_1))}"]
      \ar[d, swap, "{(t_1;(\tau_1,\zeta_1))}"]
      \&
      R_1
      \ar[r, "{(r_2;(\theta_2,\xi_2))}"]
      \&
      R_2
      \ar[r, "{(r_3;(\theta_3,\xi_3))}"]
      \&
      R_3
      \ar[r, "{(r_4;(\theta_4,\xi_4))}"]
      \&
      \cdots
      \&[-2.5em]
      \Sigma_0 = R
      \\
      T_1
      \ar[d, swap, "{(t_2;(\tau_2,\zeta_2))}"]
      \ar[r]
      \&
      \Sigma_{11}
      \ar[r]
      \&
      \Sigma_{12}
      \ar[r]
      \&
      \Sigma_{13}
      \ar[r]
      \&
      \cdots
      \&
      \Sigma_1
      \\
      T_2
      \ar[d, swap, "{(t_3;(\tau_3,\zeta_3))}"]
      \ar[r]
      \&
      \Sigma_{21}
      \ar[r]
      \&
      \Sigma_{22}
      \ar[r]
      \&
      \Sigma_{23}
      \ar[r]
      \&
      \cdots
      \&
      \Sigma_2
      \\
      T_3
      \ar[d, swap, "{(t_4;(\tau_4,\zeta_4))}"]
      \ar[r]
      \&
      \Sigma_{31}
      \ar[r]
      \&
      \Sigma_{32}
      \ar[r]
      \&
      \Sigma_{33}
      \ar[r]
      \&
      \cdots
      \&
      \Sigma_3
      \\
      \vdots
      \&
      \&
      \&
      \&
      \&
      \vdots
    \end{tikzcd}
  \]
  Set $\Phi_0 = \Sigma_0$, and for each $n > 0$, let $\Phi_n$ be the trace given by the trace $T_0 \to \cdots \to T_n$ followed by the trace $\Sigma_n$.
  Each of the permutations $\mu_n : \Sigma_n \to \Gamma_{n + 1}$ determines a permutation $\sigma_n$ from $\Phi_n$ to $\Phi_{n + 1}$ that fixes the first $n$ steps.
  The family of these permutations satisfies the hypothesis of \cref{lemma:main:2}, and gives us an injection $\sigma : I \to J$.
  It is a surjection by construction: $t_n(\tau_n)$ appears as some $r_k(\theta_k)$ by fairness, and $\sigma(k) = n$.
  To see that $\gact{\sigma}{R} = T$, it is sufficient to observe that for all $n$, $\gact{\sigma}{R}$ and $T$ agree on the first $n$ steps.
\end{proof}

Sometimes we are only interested in the set of facts that appear in a trace \(T\), \ie, in its support \(\supp(T)\).
In these cases, it we may not care about individual traces, so long as they have the same support.
This motivates the following notion of trace equivalence:

\begin{definition}
  \label{def:ssos-fairness/prop-fair-trac:1}
  Two traces $T = (M_0;(r_i,\delta_i)_I)$ and $T'$ are \defin{union-equivalent}\DIndex{union-equivalence}\varindex{{multiset rewriting system} trace union-equivalence}{1!2!3}[|defin] if $T'$ can be refreshed to a trace $[\eta]T'$ such that $\supp(T) = \supp([\eta]T')$.
\end{definition}

\begin{lemma}
  \label{lemma:obssem:6}
  If $T$ is a permutation of $S$, then $T$ and $S$ are union-equivalent.
\end{lemma}

\begin{proof}
  Consider a trace $(M_0,(r_i;\delta_i)_i)$.
  For all $n$, each fact in $M_n$ appears either in $M_0$ or in the result of some rule $r_i$ with $i \leq n$.
  Traces $T$ and $S$ start from the same multiset and have the same rule instantiations.
  It follows that they are union-equivalent.
\end{proof}

\Cref{cor:ssos-fairness/prop-fair-trac:1} will be key to showing in \cref{sec:sill-obs-equiv:observ-comm} that processes have unique observed communications:

\begin{corollary}
  \label{cor:ssos-fairness/prop-fair-trac:1}
  If $\mc{R}$ is interference-free from $M$, then all fair executions from $M$ are union-equivalent.
\end{corollary}

\defrule{F-app}{F-App}{
  \jtypef{\Psi}{MN}{\sigma}
}{
  \jtypef{\Psi}{M}{\tau \to \sigma}
  &
  \jtypef{\Psi}{N}{\tau}
}
\defrule{F-fix}{F-Fix}{
  \jtypef{\Psi}{\tFix{x}{M}}{\tau}
}{
  \jtypef{\Psi,x:\tau}{M}{\tau}
}
\defrule{F-fun}{F-Fun}{
  \jtypef{\Psi}{\lambda x : \tau.M}{\tau \to \sigma}
}{
  \jtypef{\Psi, x:\tau}{M}{\sigma}
}
\defrule{F-var}{F-Var}{
  \jtypef{\Psi, x: \tau}{x}{\tau}
}{}
\defrule{F-zero}{F-Z}{
  \jtypef{\Psi}{0}{\Tnat}
}{}
\defrule{F-succ}{F-S}{
  \jtypef{\Psi}{s(M)}{\Tnat}
}{
  \jtypef{\Psi}{M}{\Tnat}
}
\defrule{F-hole}{F-Hole}{
  \jtypef{\Gamma}{\ctxh{}{\Gamma}{\sigma}}{\sigma}
}{}
\defrule{I-proc}{I-\{\}}{
  \jtypef{\Psi}{\tProc{a}{P}{\overline{a_i}}}{\Tproc{a:A}{\overline{a_i:A_i}}}
}{
  \jtypem{\Psi}{\overline{a_i:A_i}}{P}{a}{A}
}
\defrule{L-rhon}{$\rho^-$L}{
  \jtypem{\Psi}{\Delta, a : \Trec{\alpha}{A}}{\tSendU{a}{P}}{c}{C}
}{
  \jtypem{\Psi}{\Delta, a : \subst{\Trec{\alpha}{A}}{\alpha}{A}}{P}{c}{C}
  &
  \cdot \vdash \jisst[-]{\Trec{\alpha}{A}}
}
\defrule{L-rhop}{$\rho^+$L}{
  \jtypem{\Psi}{\Delta, a : \Trec{\alpha}{A}}{\tRecvU{a}{P}}{c}{C}
}{
  \jtypem{\Psi}{\Delta, a : \subst{\Trec{\alpha}{A}}{\alpha}{A}}{P}{c}{C}
  &
  \cdot \vdash \jisst[+]{\Trec{\alpha}{A}}
}
\defrule{L-tamp}{$\Tamp$L}{
  \jtypem{\Psi}{\Delta,a:{\Tamp\{l : A_l\}}_{l \in L}}{\tSendL{a}{k}{P}}{c}{C}
}{
  \jtypem{\Psi}{\Delta,a:A_k}{P}{c}{C}
  &
  (k \in L)
}
\defrule{L-tand}{$\Tand{}{}$L}{
  \jtypem{\Psi}{\Delta, a:\Tand{\tau}{A}}{\tRecvV{x}{a}{P}}{c}{C}
}{
  \jtypem{\Psi,x:\tau}{\Delta, a:A}{P}{c}{C}
}
\defrule{L-tds}{$\Tds{}$L}{
  \jtypem{\Psi}{\Delta,a : \Tds A}{\tRecvS{a}{P}}{c}{C}
}{
  \jtypem{\Psi}{\Delta,a : A}{P}{c}{C}
}
\defrule{L-timp}{$\Timp{}{}$L}{
  \jtypem{\Psi}{\Delta,a : \Timp{\tau}{A}}{\tSendV{a}{M}{P}}{c}{C}
}{
  \jtypef{\Psi}{M}{\tau}
  &
  \jtypem{\Psi}{\Delta, a : A}{P}{c}{C}
}
\defrule{L-tlolly}{$\Tlolly$L}{
  \jtypem{\Psi}{\Delta, b : B, a : B \Tlolly A}{\tSendC{a}{b}{P}}{c}{C}
}{
  \jtypem{\Psi}{\Delta,a : A}{P}{c}{C}
}
\defrule{L-tot}{$\Tot$L}{
  \jtypem{\Psi}{\Delta, a : B \Tot A}{\tRecvC{b}{a}{P}}{c}{C}
}{
  \jtypem{\Psi}{\Delta, a : A, b : B}{P}{c}{C}
}
\defrule{L-tplus}{$\Tplus$L}{
  \jtypem{\Psi}{\Delta,a:{\Tplus\{l : A_l\}}_{l \in L}}{\tCase{a}{\left\{ l \Rightarrow P_l \right\}_{l \in L}}}{c}{C}
}{
  \jtypem{\Psi}{\Delta,a:A_l}{P_l}{c}{C}
  &
  (\forall l \in L)
}
\defrule{L-tus}{$\Tus{}$L}{
  \jtypem{\Psi}{\Delta,a : \Tus A}{\tSendS{a}{P}}{c}{C}
}{
  \jtypem{\Psi}{\Delta,a : A}{P}{c}{C}
}
\defrule{L-tu}{$\Tu$L}{
  \jtypem{\Psi}{\Delta, a : \Tu}{\tWait{a}{P}}{c}{C}
}{
  \jtypem{\Psi}{\Delta}{P}{c}{C}
}
\defrule{R-rhon}{$\rho^-$R}{
  \jtypem{\Psi}{\Delta}{\tRecvU{a}{P}}{a}{\Trec{\alpha}{A}}
}{
  \jtypem{\Psi}{\Delta}{P}{a}{\subst{\Trec{\alpha}{A}}{\alpha}{A}}
  &
  \cdot \vdash \jisst[-]{\Trec{\alpha}{A}}
}
\defrule{R-rhop}{$\rho^+$R}{
  \jtypem{\Psi}{\Delta}{\tSendU{a}{P}}{a}{\Trec{\alpha}{A}}
}{
  \jtypem{\Psi}{\Delta}{P}{a}{\subst{\Trec{\alpha}{A}}{\alpha}{A}}
  &
  \cdot \vdash \jisst[+]{\Trec{\alpha}{A}}
}
\defrule{R-tamp}{$\Tamp$R}{
  \jtypem{\Psi}{\Delta}{\tCase{a}{\left\{l \Rightarrow P_l\right\}_{l \in L}}}{a}{{\Tamp\{l :A_l \}}_{l \in L}}
}{
  \jtypem{\Psi}{\Delta}{P_l}{a}{A_l}
  &
  (\forall l \in L)
}
\defrule{R-tand}{$\Tand{}{}$R}{
  \jtypem{\Psi}{\Delta}{\tSendV{a}{M}{P}}{a}{\Tand{\tau}{A}}
}{
  \jtypef{\Psi}{M}{\tau}
  &
  \jtypem{\Psi}{\Delta}{P}{a}{A}
}
\defrule{R-tds}{$\Tds{}$R}{
  \jtypem{\Psi}{\Delta}{\tSendS{a}{P}}{a}{\Tds A}
}{
  \jtypem{\Psi}{\Delta}{P}{a}{A}
}
\defrule{R-timp}{$\Timp{}{}$R}{
  \jtypem{\Psi}{\Delta}{\tRecvV{x}{a}{P}}{a}{\Timp{\tau}{A}}
}{
  \jtypem{\Psi,x:\tau}{\Delta}{P}{a}{A}
}
\defrule{R-tlolly}{$\Tlolly$R}{
  \jtypem{\Psi}{\Delta}{\tRecvC{b}{a}{P}}{a}{B \Tlolly A}
}{
  \jtypem{\Psi}{\Delta, b : B}{P}{a}{A}
}
\defrule{R-tot}{$\Tot$R}{
  \jtypem{\Psi}{\Delta, b : B}{\tSendC{a}{b}{P}}{a}{B \Tot A}
}{
  \jtypem{\Psi}{\Delta}{P}{a}{A}
}
\defrule{R-tplus}{$\Tplus$R}{
  \jtypem{\Psi}{\Delta}{\tSendL{a}{k}{P}}{a}{\Tplus {\{l:A_l\}}_{l \in L}}
}{
  \jtypem{\Psi}{\Delta}{P}{a}{A_k}
  &
  (k \in L)
}
\defrule{R-tus}{$\Tus{}$R}{
  \jtypem{\Psi}{\Delta}{\tRecvS{a}{P}}{a}{\Tus{A}}
}{
  \jtypem{\Psi}{\Delta}{P}{a}{}
}
\defrule{R-tu}{$\Tu$R}{
  \jtypem{\Psi}{\cdot}{\tClose a}{a}{\Tu}
}{}
\defrule{T-nat}{T-\(\N\)}{
  \jftype{\Xi}{\Tnat}
}{
}
\defrule{T-proc}{T\{\}}{
  \jftype{\Xi}{\Tproc{a_0:A_0}{a_1:A_1,\dotsc,a_n:A_n}}
}{
  \jstype{\Xi}{A_i}
  &
  (0\leq i \leq n)
}
\defrule{T-rhop}{C$\rho^+$}{
  \jstype[+]{\Xi}{\Trec{\alpha}{A}}
}{
  \jstype[+]{\Xi, \jisst[+]{\alpha}}{A}
}
\defrule{T-rhopn}{C$\rho_n^+$}{
  \jstype[+]{\Xi}{\Trecn{n}{\alpha}{A}}
}{
  \jstype[+]{\Xi, \jisst[+]{\alpha}}{A}
}
\defrule{T-rhopz}{C$\rho_0^+$}{
  \jstype[+]{\Xi}{\Trecn{0}{\alpha}{A}}
}{
  \jstype[+]{\Xi, \jisst[+]{\alpha}}{A}
}
\defrule{T-rhopsn}{C$\rho_{n + 1}^+$}{
  \jstype[+]{\Xi}{\Trecn{n + 1}{\alpha}{A}}
}{
  \jstype[+]{\Xi, \jisst[+]{\alpha}}{A}
  &
  \jstype[+]{\Xi}{\Trecn{n}{\alpha}{A}}
}
\defrule{T-rhon}{C$\rho^-$}{
  \jstype[-]{\Xi}{\Trec{\alpha}{A}}
}{
  \jstype[-]{\Xi, \jisst[-]{\alpha}}{A}
}
\defrule{T-rhonn}{C$\rho_n^-$}{
  \jstype[-]{\Xi}{\Trecn{n}{\alpha}{A}}
}{
  \jstype[-]{\Xi, \jisst[-]{\alpha}}{A}
}
\defrule{T-rhonz}{C$\rho_0^-$}{
  \jstype[-]{\Xi}{\Trecn{0}{\alpha}{A}}
}{
  \jstype[-]{\Xi, \jisst[-]{\alpha}}{A}
}
\defrule{T-rhonsn}{C$\rho_{n + 1}^-$}{
  \jstype[-]{\Xi}{\Trecn{n + 1}{\alpha}{A}}
}{
  \jstype[-]{\Xi, \jisst[-]{\alpha}}{A}
  &
  \jstype[-]{\Xi}{\Trecn{n}{\alpha}{A}}
}
\defrule{T-tamp}{C$\Tamp$}{
  \Xi\vdash\jisst[-]{{\Tamp\{l :A_l \}}_{l \in L}}
}{
  \Xi\vdash\jisst[-]{A_l}
  &
  (\forall l \in L)
}
\defrule{T-tand}{C$\Tand{}{}$}{
  \Xi\vdash\jisst[+]{\Tand{\tau}{A}}
}{
  \jftype{\Xi}{\tau}
  &
  \Xi\vdash\jisst[+]{A}
}
\defrule{T-tds}{C$\Tds{}$}{
  \Xi\vdash\jisst[+]{\Tds A}
}{
  \Xi\vdash\jisst[-]{A}
}
\defrule{T-timp}{C$\Timp{}{}$}{
  \Xi\vdash\jisst[-]{\Timp{\tau}{A}}
}{
  \jftype{\Xi}{\tau}
  &
  \Xi\vdash\jisst[-]{A}
}
\defrule{T-tlolly}{C$\Tlolly$}{
  \Xi\vdash\jisst[-]{B \Tlolly A}
}{
  \Xi\vdash\jisst[+]{B}
  &
  \Xi\vdash\jisst[-]{A}
}
\defrule{T-to}{T$\to$}{
  \jftype{\Xi}{\tau \to \sigma}
}{
  \jftype{\Xi}{\tau}
  &
  \jftype{\Xi}{\sigma}
}
\defrule{T-tot}{C$\Tot$}{
  \Xi\vdash\jisst[+]{A \Tot B}
}{
  \Xi\vdash\jisst[+]{A}
  &
  \Xi\vdash\jisst[+]{B}
}
\defrule{T-tplus}{C$\Tplus$}{
  \Xi\vdash\jisst[+]{{\Tplus\{l : A_l\}}_{l \in L}}
}{
  \Xi\vdash\jisst[+]{A_l}
  &
  (\forall l \in L)
}
\defrule{T-tus}{C$\Tus{}$}{
  \Xi\vdash\jisst[-]{\Tus A}
}{
  \Xi\vdash\jisst[+]{A}
}
\defrule{T-tu}{C$\Tu$}{
  \Xi\vdash\jisst[+]{\Tu}
}{}
\defrule{T-var}{CVar}{
  \Xi,\jisst[p]{\alpha}\vdash\jisst[p]{\alpha}
}{}
\defrule{E-proc}{E-\{\}}{
  \jtypem{\Psi}{\overline{a_i:A_i}}{\tProc{a}{M}{\overline a_i}}{a}{A}
}{
  \jtypef{\Psi}{M}{\Tproc{a:A}{\overline{a_i:A_i}}}
}
\defrule{cut}{Cut}{
  \jtypem{\Psi}{\Delta_1,\Delta_2}{\tCut{a}{P}{Q}}{c}{C}
}{
  \jtypem{\Psi}{\Delta_1}{P}{a}{A}
  &
  \jtypem{\Psi}{a:A,\Delta_2}{Q}{c}{C}
}
\defrule{fwdn}{Fwd${}^-$}{
  \jtypem{\Psi}{a:A}{\tFwdN{a}{b}}{b}{A}
}{
  \jstype[-]{\cdot}{A}
}
\defrule{fwdp}{Fwd${}^+$}{
  \jtypem{\Psi}{a:A}{\tFwdP{a}{b}}{b}{A}
}{
  \jstype[+]{\cdot}{A}
}
\defrule{O-bot}{O-$\bot$}{
  \jtoc{T}{\bot}{c}{A}
}{
  \jttp{T}{c}{A}
}
\defrule{O-z-bot}{O\({}_0\)-$\bot$}{
  \jtoc[0]{T}{\bot}{c}{A}
}{
  \jttp{T}{c}{A}
}
\defrule{O-n-bot}{O\({}_{n + 1}\)-$\bot$}{
  \jtoc[n + 1]{T}{\bot}{c}{A}
}{
  \jttp{T}{c}{A}
}
\defrule{O-tamp}{O-$\Tamp$}{
  \jtoc{T}{(l,v)}{c}{\Tamp \{ l : A_l \}_{l \in L}}
}{
  \jmsg{d}{\mSendLN{c}{l}{d}} \in \mc{T}
  &
  \jtoc{T}{v}{d}{A_l}
  &
  \jttp{T}{c}{\Tamp \{ l : A_l \}_{l \in L}}
}
\defrule{O-tand}{O-$\Tand{}{}$}{
  \jtoc{T}{(\cval f, v)}{c}{\Tand{\tau}{A}}
}{
  \jmsg{c}{\mSendVP{c}{f}{d}} \in \mc{T}
  &
  \jtoc{T}{v}{d}{A}
  &
  \jttp{T}{c}{\Tand{\tau}{A}}
}
\defrule{O-n-tand}{O\({}_{n + 1}\)-$\Tand{}{}$}{
  \jtoc[n + 1]{T}{(\cval f, v)}{c}{\Tand{\tau}{A}}
}{
  \jmsg{c}{\mSendVP{c}{f}{d}} \in \mc{T}
  &
  \jtoc[n]{T}{v}{d}{A}
  &
  \jttp{T}{c}{\Tand{\tau}{A}}
}
\defrule{O-tds}{O-$\Tds{}$}{
  \jtoc{T}{(\cshift, v)}{c}{\Tds{A}}
}{
  \jmsg{c}{\mSendSP{c}{d}} \in \mc{T}
  &
  \jtoc{T}{v}{d}{A}
}
\defrule{O-timp}{O-$\Timp{}{}$}{
  \jtoc{T}{(\cval f, v)}{c}{\Timp{\tau}{A}}
}{
  \jmsg{d}{\mSendVN{c}{f}{d}} \in \mc{T}
  &
  \jtoc{T}{v}{d}{A}
  &
  \jttp{T}{c}{\Timp{\tau}{A}}
}
\defrule{O-tlolly}{O-$\Tlolly$}{
  \jtoc{T}{(u,v)}{c}{A \Tlolly B}
}{
  \jmsg{d}{\mSendCN{c}{a}{d}} \in \mc{T}
  &
  \jtoc{T}{u}{a}{A}
  &
  \jtoc{T}{v}{d}{B}
}
\defrule{O-tot}{O-$\Tot$}{
  \jtoc{T}{(u,v)}{c}{A \Tot B}
}{
  \jmsg{c}{\mSendCP{c}{a}{d}} \in \mc{T}
  &
  \jtoc{T}{u}{a}{A}
  &
  \jtoc{T}{v}{d}{B}
}
\defrule{O-n-tot}{O\({}_{n + 1}\)-$\Tot$}{
  \jtoc[n + 1]{T}{(u,v)}{c}{A \Tot B}
}{
  \jmsg{c}{\mSendCP{c}{a}{d}} \in \mc{T}
  &
  \jtoc[n]{T}{u}{a}{A}
  &
  \jtoc[n]{T}{v}{d}{B}
}
\defrule{O-tplus}{O-$\Tplus$}{
  \jtoc{T}{(l,v)}{c}{\Tplus \{ l : A_l \}_{l \in L}}
}{
  \jmsg{c}{\mSendLP{c}{l}{d}} \in \mc{T}
  &
  \jtoc{T}{v}{d}{A_l}
  &
  \jttp{T}{c}{\Tplus \{ l : A_l \}_{l \in L}}
}
\defrule{O-rhon}{O-$\rho^-$}{
  \jtoc{T}{(\cunfold, v)}{c}{\Trec{\alpha}{A}}
}{
  \jmsg{d}{\mSendUN{c}{d}} \in \mc{T}
  &
  \jtoc{T}{v}{d}{[\Trec{\alpha}{A}/\alpha]A}
}
\defrule{O-rhop}{O-$\rho^+$}{
  \jtoc{T}{(\cunfold, v)}{c}{\Trec{\alpha}{A}}
}{
  \jmsg{c}{\mSendUP{c}{d}} \in \mc{T}
  &
  \jtoc{T}{v}{d}{[\Trec{\alpha}{A}/\alpha]A}
}
\defrule{O-tu}{O-$\Tu$}{
  \jtoc{T}{\cclose}{c}{\Tu}
}{
  \jmsg{c}{\mClose{c}} \in \mc{T}
}
\defrule{O-n-tu}{O\({}_{n + 1}\)-$\Tu$}{
  \jtoc[n + 1]{T}{\cclose}{c}{\Tu}
}{
  \jmsg{c}{\mClose{c}} \in \mc{T}
}
\defrule{O-tus}{O-$\Tus{}$}{
  \jtoc{T}{(\cshift, v)}{c}{\Tus{A}}
}{
  \jmsg{d}{\mSendSN{c}{d}} \in \mc{T}
  &
  \jtoc{T}{v}{d}{A}
}
\defrule{C-bot}{C-$\bot$}{
  \jsynt{\bot}{A}
}{
  \jstype{\cdot}{A}
}
\defrule{C-tu}{C-$\Tu$}{
  \jsynt{\cclose}{\Tu}
}{
}
\defrule{C-tplus}{C-$\Tplus$}{
  \jsynt{(k,v_k)}{\Tplus \{ l : A_l \}_{l \in L}}
}{
  \jsynt{v_k}{A_k}
  &
  (k \in L)
}
\defrule{C-tamp}{C-$\Tamp$}{
  \jsynt{(k,v_k)}{\Tamp \{ l : A_l \}_{l \in L}}
}{
  \jsynt{v_k}{A_k}
  &
  (k \in L)
}
\defrule{C-tot}{C-$\Tot$}{
  \jsynt{(v,v')}{A \Tot B}
}{
  \jsynt{v}{A}
  &
  \jsynt{v'}{B}
}
\defrule{C-tlolly}{C-$\Tlolly$}{
  \jsynt{(v,v')}{A \Tlolly B}
}{
  \jsynt{v}{A}
  &
  \jsynt{v'}{B}
}
\defrule{C-rho}{C-$\rho$}{
  \jsynt{(\cunfold, v)}{\Trec{\alpha}{A}}
}{
  \jsynt{v}{[\Trec{\alpha}{A}/\alpha]A}
}
\defrule{C-tds}{C-$\Tds{}$}{
  \jsynt{(\cshift, v)}{\Tds A}
}{
  \jsynt{v}{A}
}
\defrule{C-tus}{C-$\Tus{}$}{
  \jsynt{(\cshift, v)}{\Tus A}
}{
  \jsynt{v}{A}
}
\defrule{C-tand}{C-$\Tand{}{}$}{
  \jsynt{(\cval f, v)}{\Tand{\tau}{A}}
}{
  \jtypef{\cdot}{f}{\tau}
  &
  \jsynt{v}{A}
}
\defrule{C-timp}{C-$\Timp{}{}$}{
  \jsynt{(\cval f, v)}{\Timp{\tau}{A}}
}{
  \jtypef{\cdot}{f}{\tau}
  &
  \jsynt{v}{A}
}
\defrule{Ct-p-hole}{P-Hole}{%
  \jtypem{\Gamma}{\Lambda}{\ctxh{}{\Gamma;\Lambda}{b:B}}{b}{B}
}{
}
\defrule{Ct-hole-cut-l}{Hole-Cut-L}{%
  \jtypem{\cdot}{\Delta_1,\Delta_2}{\tCut{b}{\ctxh{O}{\Delta}{a:A}}{P}}{c}{C}
}{
  \jtypem{\cdot}{\Delta_1}{\ctxh{O}{\Delta}{a:A}}{b}{B}
  &
  \jtypem{\cdot}{\Delta_2,b : B}{P}{c}{C}
}
\defrule{Ct-hole-cut-r}{Hole-Cut-R}{%
  \jtypem{\cdot}{\Delta_1,\Delta_2}{\tCut{b}{P}{\ctxh{O}{\Delta}{a:A}}}{c}{C}
}{
  \jtypem{\cdot}{\Delta_1}{P}{b}{B}
  &
  \jtypem{\cdot}{b : B,\Delta_2}{\ctxh{O}{\Delta}{a:A}}{c}{C}
}
\defrule{CS-bot}{CS-$\bot$}{
  \jsynt{\bot \commsim w}{A}
}{
  \jsynt{w}{A}
}
\defrule{CS-tu}{CS-$\Tu$}{
  \jsynt{\cclose \commsim \cclose}{\Tu}
}{
}
\defrule{CS-tplus}{CS-$\Tplus$}{
  \jsynt{(k,v_k) \commsim (k,w_k)}{\Tplus \{ l : A_l \}_{l \in L}}
}{
  \jsynt{v_k \commsim w_k}{A_k}
  &
  (k \in L)
}
\defrule{CS-tamp}{CS-$\Tamp$}{
  \jsynt{(k,v_k) \commsim (k,w_k)}{\Tamp \{ l : A_l \}_{l \in L}}
}{
  \jsynt{v_k \commsim w_k}{A_k}
  &
  (k \in L)
}
\defrule{CS-tot}{CS-$\Tot$}{
  \jsynt{(v,v') \commsim (w,w')}{A \Tot B}
}{
  \jsynt{v \commsim w}{A}
  &
  \jsynt{v' \commsim w'}{B}
}
\defrule{CS-tlolly}{CS-$\Tlolly$}{
  \jsynt{(v,v') \commsim (w,w')}{A \Tlolly B}
}{
  \jsynt{v \commsim w}{A}
  &
  \jsynt{v' \commsim w'}{B}
}
\defrule{CS-rho}{CS-$\rho$}{
  \jsynt{(\cunfold, v) \commsim (\cunfold, w)}{\Trec{\alpha}{A}}
}{
  \jsynt{v \commsim w}{[\Trec{\alpha}{A}/\alpha]A}
}
\defrule{CS-tds}{CS-$\Tds{}$}{
  \jsynt{(\cshift, v) \commsim (\cshift, w)}{\Tds A}
}{
  \jsynt{v \commsim w}{A}
}
\defrule{CS-tus}{CS-$\Tus{}$}{
  \jsynt{(\cshift, v) \commsim (\cshift, w)}{\Tus A}
}{
  \jsynt{v \commsim w}{A}
}
\defrule{CS-tand}{CS-$\Tand{}{}$}{
  \jsynt{(\cval f, v) \commsim (\cval f', w)}{\Tand{\tau}{A}}
}{
  \jtypef{\cdot}{f \leqslant f'}{\tau}
  &
  \jsynt{v \commsim w}{A}
}
\defrule{CS-timp}{CS-$\Timp{}{}$}{
  \jsynt{(\cval f, v) \commsim (\cval f', w)}{\Timp{\tau}{A}}
}{
  \jtypef{\cdot}{f \leqslant f'}{\tau}
  &
  \jsynt{v \commsim w}{A}
}
\defrule{conf-h}{Conf-H}{%
  \jcfgti{\Lambda}{\cdot}{\ctxh{}{\Lambda}{\Xi}}{\Xi}
}{
}
\defrule{conf-m}{Conf-M}{
  \jcfgti[\Sigma]{\Delta}{\cdot}{\jmsg{c}{m}}{(c : A)}
}{
  \jtypem{\cdot}{\Delta}{m}{c}{A}
}
\defrule{conf-p}{Conf-P}{
  \jcfgti[\Sigma]{\Delta}{\cdot}{\jproc{c}{P}}{(c : A)}
}{
  \jtypem{\cdot}{\Delta}{P}{c}{A}
}
\defrule{conf-c}{Conf-C}{
  \jcfgti[\Sigma,\check \Pi,\Sigma']{\Gamma\Lambda}{\Iota_1\Pi\Iota_2}{\mc{C},\mc{D}}{\Phi\Xi}
}{
  \jcfgti[\Sigma,\check \Pi]{\Gamma}{\Iota_1}{\mc{C}}{\Phi\Pi}
  &
  \jcfgti[\check \Pi,\Sigma']{\Pi\Lambda}{\Iota_2}{\mc{D}}{\Xi}
}
\defrule{EV-fun}{EV-Fun}{
  \fneval{\lambda x : \tau.M}{\lambda x : \tau.M}
}{
}
\defrule{EV-proc}{EV-Proc}{
  \fneval{\tProc{c_0}{P}{\overline{c_i}}}{\tProc{c_0}{P}{\overline{c_i}}}
}{
}
\defrule{EV-fix}{EV-Fix}{
  \fneval{\tFix{x}{M}}{v}
}{
  \fneval{\subst{\tFix{x}{M}}{x}{M}}{v}
}
\defrule{EV-app}{EV-App}{
  \fneval{MN}{v}
}{
  \fneval{M}{\lambda x : \tau . M'}
  &
  \fneval{N}{w}
  &
  \fneval{\subst{w}{x}{M'}}{v}
}
\defrule{EV-zero}{EV-Zero}{
  \fneval{0}{0}
}{}
\defrule{EV-succ}{EV-Succ}{
  \fneval{\ms{s}(M)}{\ms{s}(v)}
}{
  \fneval{M}{v}
}
\defrule{EV-pfix}{EV-fix$^{n + 1}$}{
  \fneval{\tpFix{n + 1}{x}{M}}{v}
}{
  \fneval{\subst{\tpFix{n}{x}M}{x}{M}}{v}
}
\defrule{F-pfix}{F-Fix$^n$}{
  \jtypef{\Psi}{\tpFix{n}{x}{M}}{\tau}
}{
  \jtypef{\Psi, x : \tau}{M}{\tau}
}

\section{Background on Polarized SILL}
\label{cha:sill-background}

The Polarized SILL programming language~\cite{toninho_2013:_higher_order_proces_funct_session, pfenning_griffith_2015:_polar_subst_session_types,} cohesively integrates functional computation and message-passing concurrent computation.
Its concurrent computation layer arises from a proofs-as-processes correspondence between intuitionistic linear logic and the session-typed $\pi$-calculus~\cite{caires_pfenning_2010:_session_types_intuit_linear_propos}.
We give an overview of its statics and dynamics in \cref{sec:sill-background:overview-statics,sec:sill-background:overview-dynamics} before presenting the language in \cref{sec:sill-background:typing-mult-rewr}.
We highlight its key properties in \cref{sec:sill-background:key-props}.
In \cref{sec:sill-obs-equiv:relat-equiv}, we describe general properties of relations on its programs.

\subsection{Overview of Statics}
\label{sec:sill-background:overview-statics}

Processes are computational agents that interact with their environment solely through communication.
In Polarized SILL, communication happens over named channels, which we can intuitively think of as wires that carry messages.\DIndex{channel}
Moreover, communication on channels is bidirectional: in general, a process can both send and receive communications along the same channel.
Each channel has an associated session type $A$.
Session types~\cite{takeuchi_1994:_inter_based_languag,honda_1993:_types_dyadic_inter}\DIndex{{session type}} specify communication protocols, \ie, rules for communicating along channels.
Equivalently, we can think of session types as classifying communications, analogously to how data types classify values.
A channel's session type then specifies which communications are permitted on that channel.
The type system for processes ensures that communication on a channel of type $A$ respects the protocol specified by the session type $A$.

Processes in Polarized SILL are organized according to a client-server architecture, and we can think of session types as describing services\index{service|see{session type}} provided or used along channels.
A process $P$ always \defin{provides}\DIndex<{provided channel} or is a server for a service $A$ on a channel $c$, and it \defin{uses}\DIndex<{used channel} or is a client of zero or more services $A_i$ on channels $c_i$.
We write \(c : A\) to mean that the channel\varindex{session type of@ a@ channel}{12!~345}[|defin]\DIndex<{session-typed channel} \(c\) has type \(A\).
The used services form a linear context $\Delta = c_1 : A_1, \dotsc, c_n : A_n$.\varindex{context of@ channels}{1!~23}[|defin]\Index{linear context}
The process $P$ can use values from the functional layer.
These are abstracted by a structural context\Index{structural context} $\Psi$ of functional variables.
These data are captured by the inductively defined judgment $\jtypem{\Psi}{\Delta}{P}{c}{A}$\glsadd{jtypem}\DIndex{process typing judgment}.
We say that the process $P$ is \defin{closed}\DIndex{closed process} if it does not depend on any free variables, \ie, if $\jtypem{\cdot}{\Delta}{P}{c}{A}$.
This judgment is both generic\Index{generic judgment} and parametric\Index{parametric judgment}:\footnote{We refer the reader to \cite[\S~1.2]{harper_2016:_pract_found_progr_languag} or \cite[\S~2.5]{kavanagh_2021:_commun_based_seman} for background on generic and parametric judgments.} it is closed under renaming of channel names in $\Delta, c : A$, and it is closed under renaming and substitution of functional variables in $\Psi$.

At any given point in a computation, communication flows in a single direction on a channel ${c : A}$.
The direction of communication is determined by the \emph{polarity}\index{session type!polarity|defin}\index{polarity|see{session type, polarity}} of the type $A$, where session types are partitioned as \defin{positive}\DIndex<{positive {session type}} or \defin{negative}\DIndex<{negative {session type}}~\cite{pfenning_griffith_2015:_polar_subst_session_types}.
Consider a process judgment ${\jtypem{\Psi}{\Delta}{P}{c_0}{A_0}}$.
Communication on positively typed channels flows from left-to-right in this judgment: if $A_0$ is positive, then $P$ can only send on $c_0$, while if $A_i$ is positive for $1 \leq i \leq n$, then $P$ can only receive on $c_i$.
Symmetrically, communication on negatively typed channels flows from right-to-left in the judgment.
Bidirectional communication arises from the fact that the type of a channel evolves over the course of a computation, sometimes becoming positive, sometimes becoming negative.
We write $\jisst[+]{A}$\glsadd{jisst} to mean $A$ is positive and $\jisst[-]{A}$ to mean $A$ is negative.
Most session types have a polar-dual session type, where the direction of the communication is reversed.

Open session types are given by the inductively defined judgment $\jstype[p]{\Xi}{A}$\glsadd{jisft}, where $\Xi$ is a structural context\Index{structural context} of polarized type variables $\jisst[p_i]{\alpha_i}$ and $p, p_i \in \{{-},{+}\}$.\varindex{context of@ session-typed variables}{1!~234}[|defin]\varindex{session type judgment}{12!~3 3!1-2~}[|defin]\varindex{session type variable}{12!~3 3!1-2~}[|defin]
We abbreviate the judgment as $\jstype{\Xi}{A}$ when the polarity is unambiguous.
The session type \(A\) is closed\DIndex{closed {session type}} if it does not depend on any free variables, \ie, if \(\jstype{\cdot}{A}\).

The functional layer is the simply-typed $\lambda$-calculus with a fixed-point operator and a call-by-value evaluation semantics.
A judgment $\jtypef{\Psi}{M}{\tau}$\glsadd{jtypef} means the functional term $M$ has functional type $\tau$ under the structural context $\Psi$ of functional variables $x_i : \tau_i$.\DIndex{{functional term} typing judgment}\varindex{context of@ functional variables}{1!~234}[|defin]\DIndex{{functional term} variable}
This judgment's inductive definition is standard.
We say that the term \(M\) is closed\DIndex{closed {functional term}} if it does not depend on any free variables, \ie, if \(\jtypef{\cdot}{M}{\tau}\).
We use the judgment $\jftype{\Xi}{\tau}$\glsadd{jisft} to mean that $\tau$ is a functional type depending on polarized type variables $\jisst[p_i]{\alpha_i}$.\DIndex{{functional type} judgment}
The type \(\tau\) is closed\DIndex{closed {functional type}} if it does not depend on any free variables.
New is the base type $\Tproc{a:A}{\overline{a_i:A_i}}$ of quoted processes, where we abbreviate ordered lists using an overline.

We draw attention to the fundamental difference between \textit{variables} and \textit{channel names}.
A functional variable $x : \tau$ in a context $\Psi$ stands for a value of type $\tau$.
A channel name in $\Delta,c:A$ is a symbol\varindex{channel names are symbols symbol}{1!~234 5!12}: it stands not for a value, but for a channel of typed bidirectional communications.
In particular, channel names can only be renamed; unlike functional variables, nothing can be substituted for a channel name.

\subsection{Overview of Dynamics}
\label{sec:sill-background:overview-dynamics}

The operational behaviour of processes in Polarized SILL is defined by a substructural operational semantics~\cite{simmons_2012:_subst_logic_specif} in the form of a multiset rewriting system.
This multiset rewriting system uses three different kinds of facts.
The two most commonly encountered facts involve processes and messages.
The process fact $\jproc{c}{P}$\glsadd{jproc}\DIndex{process fact} means that the closed process $P$ provides a channel $c$.
The message fact $\jmsg{c}{m}$\glsadd{jmsg}\DIndex{message fact} means that the message process $m$ provides a channel $c$.
Message processes represent single messages or pieces of data sent on a channel, and they are closed processes written in a restricted fragment of the process language.
Process communication is asynchronous: processes send messages without synchronizing with recipients.
Messages sent on a given channel are received in order.
However, there is no global ordering on sent messages: messages sent on different channels can be received out of order.

The behaviour of the functional layer is specified by the set $\mb{F}$ of persistent evaluation facts $\jeval{M}{v}$\glsadd{jeval}\DIndex{evaluation fact}, where $\jeval{M}{v}$ if and only if the closed term $M$ evaluates to the value $v$ under the standard call-by-value semantics.
Explicitly, \(\jeval{M}{v}\) if and only if \(\fneval{M}{v}\)\glsadd{fneval}\DIndex{evaluation judgment}, where \(\fneval{M}{v}\) is the usual evaluation semantics.
It is inductively defined in \cref{fig:sill-background:2}.
We write $\fnval{v}$\glsadd{fnval} if $v$ is a value, \ie, if $\fneval{v}{v}$.
The fact \(\jeval{M}{v}\) captures an evaluation relation instead of a transition relation because we only ever observe values from the process layer, and we never need to observe its individual steps.
For conciseness, we do not mention this set of facts in our multisets and instead treat it implicitly.

\begin{figure}
  \begin{gather*}
    \getrule{EV-fun}
    \qquad
    \getrule{EV-app}
    \\
    \getrule{EV-proc}
    \qquad
    \getrule{EV-fix}
  \end{gather*}
  \caption{Big-step semantics underlying Polarized SILL's functional layer}
  \label{fig:sill-background:2}
\end{figure}

Recall from \cref{sec:ssos-fairness:mult-rewr-syst:first-order-multiset-persistence} that a multiset-in-context \(\msinc{\Sigma}{\persfnt{\Pi}, M}\) is formed of a set \(\persfnt{\Pi}\) of persistent facts, a multiset \(M\) of ephemeral facts, and a signature \(\Sigma\) extending an initial signature \(\Sigma_i\) to contain all symbols appearing in \(\persfnt{\Pi}, M\).
We let \(\mc{C}\) range over multisets formed of process, message, and evaluation facts.
For consistency with the literature, we call a multiset-in-context \(\msinc{\Sigma}{\mc{C}}\) in a process trace a \defin{configuration}\DIndex{configuration}\Index>{multiset-in-context configuration}.
We leave the initial signature \(\Sigma_i\) of process and fact symbols implicit when writing configurations.
A \defin{process trace}\DIndex{process trace} is a trace from the initial configuration of a process.
The \defin{initial configuration}\DIndex<{initial configuration} of a well-typed process $\jtypem{\cdot}{c_1:A_1,\dotsc,c_n:A_n}{P}{c_0}{A_0}$ is the multiset-in-context
\[
  \msinc{c_0,\dotsc,c_n}{\jproc{c_0}{P}},
\]
where the multiset $\mb{F}$ of $\jeval{M}{v}$ facts is implicitly present.
A \defin{fair execution}\varindex{fair process execution}{2!1~3 3!12~}[|defin] of ${\jtypem{\cdot}{\Delta}{P}{c}{A}}$ is a weakly fair execution from its initial configuration.
By \cref{prop:sill-background-dyn-prop-typed-config:1,prop:ssos-fairness/prop-fair-trac:1}, the multiset rewriting system specifying Polarized SILL is non-overlapping and interference-free from initial configurations.
It then follows by \cref{prop:ssos-fairness/prop-fair-trac:2,prop:sill-background-dyn-prop-typed-config:1} that all weakly fair executions also strongly fair and über fair.

The substructural operational semantics maintains several invariants.
Chief among these is a type preservation property (\cref{prop:sill-background:1}), where each configuration appearing in a trace is well-typed, and where each multiset rewrite rule preserves the type of the configuration.
Concretely, we introduce a type system for configurations inspired by one due to \textcite[\S~4.4]{gommerstadt_2018:_session_typed_concur_contr}.
It assigns a session type to each free channel appearing in a configuration.
The preservation property states that the types of channels not free in the active multiset of a rule remain unchanged in the result, and that the types of external channels remain unchanged.
This approach is in contrast to the one taken by \textcite{kavanagh_2020:_subst_obser_commun_seman}, which conservatively extended the underlying substructural operational semantics to track typing information at runtime.
Advantageously, our approach preserves the distinction between operational rules and typing concerns, and it requires no changes to the original substructural operational semantics.

The typing judgment $\jcfgti[\Sigma]{\Gamma}{\Iota}{\mc{C}}{\Delta}$\glsadd{jcfgti}\DIndex{configuration typing judgment} means that the configuration $\msinc{\Sigma}{\mc{C}}$ \textit{uses}\Index<{used channel} the channels in $\Gamma$, \textit{provides}\Index<{provided channel} the channels in $\Delta$, and has internal channels\DIndex<{internal channel} $\Iota$.
Here, $\Gamma = \gamma_1 : A_1, \dotsc, \gamma_n : A_n$, $\Delta = \delta_1 : B_1, \dotsc, \delta_m : B_m$, and $\Iota = \iota_1 : C_1, \dotsc, \iota_k : C_k$ are linear contexts\Index{linear context} of session-typed channel names, with $n, k \geq 0$ and $m \geq 1$.
Write $\check\Gamma$\glsadd{check} for the list $\gamma_1, \dotsc, \gamma_n$ of channel names appearing in $\Gamma$.
The judgment $\jcfgti[\Sigma]{\Gamma}{\Iota}{\mc{C}}{\Delta}$ is well-formed only if the channel names in $\check \Gamma, \check \Delta, \check \Iota$ are pairwise distinct and $\check \Gamma, \check \Delta, \check \Iota \subseteq \Sigma$.
This judgment is parametric in $\Sigma$: it is closed under renaming of channel names in \(\Sigma\), and we can freely extend \(\Sigma\) with new channel names.
We usually leave $\Sigma$ implicit and write $\jcfgti{\Gamma}{\Iota}{\mc{C}}{\Delta}$ or \(\jcfgt{\Gamma}{\mc{C}}{\Delta}\) for $\jcfgti[\Sigma]{\Gamma}{\Iota}{\mc{C}}{\Delta}$.
We call the pair $(\Gamma, \Delta)$ the \defin{interface}\DIndex>{configuration interface}\index{interface|see{configuration, interface}} of $\mc{C}$.
In contrast to processes, configurations can provide multiple channels.
This is to allow for applications like run-time monitoring~\cite{gommerstadt_2018:_session_typed_concur_contr}.
To avoid overloading commas in interfaces and elsewhere, we often use juxtaposition to denote the concatenation of contexts of session-typed channels, \ie, we write \(\Gamma\Delta\) for the concatenation \(\gamma_1 : A_1, \dotsc, \gamma_n : A_n, \delta_1 : B_1, \dotsc, \delta_m : B_m\) of \(\Gamma\) and \(\Delta\).
We remark that the multiset $\mb{F}$ of $\jeval{M}{v}$ facts is implicitly contained in $\mc{C}$ in every judgment $\jcfgti[\Sigma]{\Gamma}{\Iota}{\mc{C}}{\Delta}$.

The judgment \( \jcfgti[\Sigma]{\Gamma}{\Iota}{\mc{C}}{\Delta} \) is inductively defined by the rules \getrn{conf-m}, \getrn{conf-p}, and \getrn{conf-c}:
\[
  \getrule{conf-m}
  \quad
  \getrule{conf-p}
\]
\[
  \getrule{conf-c}
\]
The rules \getrn{conf-m} and \getrn{conf-p} lift closed messages and processes to message and process facts, while preserving their used and provided channels.
In \getrn{conf-m}, we assume that $m$ ranges over ``message processes'' $m^+$ and $m^-_{b,c}$.
These message processes are a restricted class of processes defined in \cref{eq:sill-background-typing-mult-rewr:1,eq:sill-background-typing-mult-rewr:2}.
The composition\DIndex{configuration composition} rule \getrn{conf-c} is a ``parallel composition plus hiding'' operation (\cf~\cite[pp.~20f.]{milner_1980:_calcul_commun_system}).
It composes two configurations \(\mc{C}\) and \(\mc{D}\) so that they communicate along some common (but potentially empty) collection of channels \(\Pi\).
These channels are then hidden from external view: they do not appear in the interface \((\Gamma\Lambda, \Phi\Xi)\), but instead appear in the composition's context \(\Iota_1 \Pi \Iota_2\) of internal channels.
Without loss of generality, we assume that $\Sigma \cap \Sigma' = \emptyset$ (so $\check I_1 \cap \check I_2 = \emptyset$) when composing $\mc{C}$ and $\mc{D}$.
This requirement ensures that the internal channels in $\mc{C}$ do not interfere with those in $\mc{D}$, and vice-versa.
Because the hypotheses are parametric in $\Sigma$ and $\Sigma'$, we can always rename those channels to ensure that this is the case.

We can recognize \getrn{conf-c} as a composition operator in a pluricategory.\footnote{Pluricategories~\cite[Definition~2.1.11]{kavanagh_2021:_commun_based_seman} generalize polycategories~\cite{szabo_1975:_polyc} to allow for composition along multiple objects.
  In our case, the pluricategory's objects are session-typed channels, and its morphisms \(\Gamma \to \Delta\) are typed configuration \(\jcfgt{\Gamma}{\mc{C}}{\Delta}\).}
We make this observation only to draw attention to the fact that the rule \getrn{conf-c} determines an associative and partially commutative partial composition operator.

\subsection{Typing and Multiset Rewriting Rules}
\label{sec:sill-background:typing-mult-rewr}

In this section, we give the typing rules that inductively define well-typed functional terms and processes.
In each case, we also give the associated multiset rewriting rules.
With a few exceptions and for conciseness, we only give the rules for processes that provide positive session types.
All rules can be found in \cref{sec:sill-background:compl-list-mult}.

\subsubsection{Manipulating Channels}
\label{sec:sill-background:manip-chann}

The forwarding process $\tFwdP{b}{a}$ forwards all messages between channels $a$ and $b$ of the same positive type.\DIndex<{forwarding process}
The process $\tFwdN{b}{a}$ is the dual for channels of negative type.
We remark that the syntax reflects the direction in which messages flow.
Though some presentations use a single forwarding process for both polarities, it is useful for practical and semantic concerns to syntactically differentiate between forwarding positive communications and forwarding negative communications.
\begin{gather}
  \getrule{fwdp}
  \quad
  \getrule{fwdn}
  \nonumber
\end{gather}
These processes act on messages $m^+$\glsadd{posmsg} and $m^-_{b,c}$\glsadd{negmsg} travelling in the positive and negative directions.
These messages are respectively \defin{message processes}\DIndex{message process} given by \cref{eq:sill-background-typing-mult-rewr:1,eq:sill-background-typing-mult-rewr:2}.
Their meaning will be explained below.
The letters \(a\), \(b\), \(c\), and \(d\) range over channel names, \(l\) ranges over labels, and \(v\) ranges over functional values.
\begin{align}
  m^+ &\Coloneqq \mSendVP{a}{v}{d} \mid \mSendSP{a}{d} \mid \mSendLP{a}{l}{d} \mid \mSendCP{a}{b}{d}\nonumber\\
      &\;\;\mid \mSendUP{a}{d} \mid \mClose{a}\label{eq:sill-background-typing-mult-rewr:1}\\
  m^-_{b,c} &\Coloneqq \mSendVN{b}{v}{c} \mid \mSendSN{b}{c} \mid \mSendLN{a}{l}{c} \mid \mSendCN{a}{b}{c}\nonumber\\
      &\;\;\mid \mSendUN{a}{c}\label{eq:sill-background-typing-mult-rewr:2}
\end{align}
The subscripts on $m^-_{b,c}$ serve to indicate which channel names appear in the message fact, and they ensure that \cref{eq:sill:msr-fwdn} is only applicable when the process fact and the message fact have a common channel.
The operational behaviour is given by \cref{eq:sill:msr-fwdp,eq:sill:msr-fwdn}.
Properly speaking, there is an instance of these rules for each different kind of message $m^+$ and $m^-_{b,c}$.
We implicitly universally quantify on all free channel names appearing in these multiset rewriting rules:
\begin{gather}
  \label[msr]{eq:sill:msr-fwdp}
  \jmsg{a}{m^+},
\jproc{b}{\tFwdP{a}{b}}
\to
\jmsg{b}{\subst{b}{a}{m^+}}
\\
  \label[msr]{eq:sill:msr-fwdn}
  \jproc{b}{\tFwdN{a}{b}},
\jmsg{c}{m^-_{b,c}}
\to
\jmsg{c}{\subst{a}{b}{m^-_{b,c}}}

\end{gather}

Process composition $\tCut{a}{P}{Q}$\DIndex{process composition} captures Milner's ``parallel composition plus hiding'' operation~\cite[pp.~20f.]{milner_1980:_calcul_commun_system}.
It spawns processes $P$ and $Q$ that communicate over a shared private channel~$a$ of type $A$.
Those familiar with the \(\pi\)-calculus may like to think of this syntax as analogous to the \(\pi\)-calculus process \((\nu a)(P \mid Q)\).
\begin{gather}
  \getrule{cut}
  \nonumber\\
  \label[msr]{eq:sill:msr-cut}
  \jproc{c}{\tCut{a}{P}{Q}} \to \exists b. \jproc{b}{\subst{b}{a}{P}}, \jproc{c}{\subst{b}{a}{Q}}

\end{gather}
We remark that in \eqref{eq:sill:msr-cut} we do not quantify over the channel name $a$.
This is because the channel name $a$ is a bound in processes $P$ and $Q$.

Processes can close channels of type $\Tu$\glsadd{T-tu}.\varindex{{session type} {unit type}}{1!2}[|defin]
To do so, the process $\tClose a$ sends a ``close message'' over the channel $a$ and terminates.
The process $\tWait{a}{P}$ blocks on $a$ until it receives the close message and then continues as $P$.
\begin{gather}
  \getrule{T-tu}
  \quad
  \getrule{R-tu}
  \quad
  \getrule{L-tu}
  \nonumber\\
  \label[msr]{eq:sill:msr-tu-l}
  \jmsg{a}{\mClose{a}},
\jproc{c}{\tWait{a}{P}}
\to
\jproc{c}{P}
\\
  \label[msr]{eq:sill:msr-tu-r}
  \jproc{a}{\tClose{a}}
\to
\jmsg{a}{\mClose{a}}

\end{gather}
The positive type $\Tu$ does not have a negative dual.
It would require a detached process with no client.

\begin{example}
  \label{ex:sill-background-typing-mult-rewr:6}
  The following process is, informally speaking, equivalent to the forwarding process \(\jtypem{\cdot}{a : \Tu}{\tFwdP{a}{b}}{b}{\Tu}\) for channels of type \(\Tu\):
  \[
    \jtypem{\cdot}{a : \Tu}{\tWait{a}{\tClose{b}}}{b}{\Tu}.
  \]
  Indeed, if no close message arrives on \(a\), then neither process does anything.
  If a close message arrives on \(a\), then both processes send a close message on \(b\) and terminate:
  \begin{align*}
    \jmsg{a}{\mClose{a}}, \jproc{b}{\tFwdP{a}{b}} &\msstep \jmsg{b}{\mClose{b}}\\
    \jmsg{a}{\mClose{a}}, \jproc{b}{\tWait{a}{\tClose{b}}} &\mssteps \jmsg{b}{\mClose{b}}
  \end{align*}
\end{example}

Processes can send and receive channels over channels.
The protocol $B \Tot A$ prescribes transmitting a channel of type $B$ followed by communication of type $A$.
The process $\tSendC{a}{b}{P}$ sends the channel $b$ over the channel $a$ and continues as $P$.
The process $\tRecvC{b}{a}{P}$ receives a channel over $a$, binds it to the name $b$, and continues as $P$.
In particular, the channel name \(b\) is bound in \(P\) in $\tRecvC{b}{a}{P}$.
To ensure a queue-like structure for messages on $a$, we generate a fresh channel name $d$ in \cref{eq:sill:msr-tot-r} for the ``continuation channel'' that will carry subsequent communications.
Operationally, we rename $a$ in $P$ to the continuation channel $d$ carrying the remainder of the communications.
\begin{gather}
  \getrule{T-tot}
  \nonumber
  \\
  \getrule{R-tot}
  \quad
  \getrule{L-tot}
  \nonumber
  \\
  \label[msr]{eq:sill:msr-tot-r}
  \jproc{a}{\tSendC{a}{b}{P}}
\to
\exists d.
\jproc{d}{\subst{d}{a}{P}},
\jmsg{a}{\mSendCP{a}{b}{d}}

  \\
  \label[msr]{eq:sill:msr-tot-l}
  \jmsg{a}{\mSendCP{a}{e}{d}},
\jproc{c}{\tRecvC{b}{a}{P}}
\to
\jproc{c}{\subst{e,d}{b,a}{P}}

\end{gather}

\begin{example}
  \label{ex:sill-background-typing-mult-rewr:7}
  The following closed process \(P\) sends a channel \(a : \Tu\) over \(b : \Tu \Tot \Tu\), and closes \(b\):
  \[
    \jtypem{\cdot}{a : \Tu}{\tSendC{b}{a}{\tClose{b}}}{b}{\Tu \Tot \Tu}
  \]
  The following closed process \(Q\) receives a channel of type \(\Tu\) on \(b\) and binds it to the name \(d\).
  Then it waits for a close message on \(b\) and forwards \(d\) over \(c\):
  \[
    \jtypem{\cdot}{b : \Tu \Tot \Tu}{\tRecvC{d}{b}{\tWait{b}{\tFwdP{d}{c}}}}{c}{\Tu}.
  \]
  Their composition
  \[
    \jtypem{\cdot}{a : \Tu}{\tCut{b}{\left(\tSendC{b}{a}{\tClose{b}}\right)}{\left(\tRecvC{d}{b}{\tWait{b}{\tFwdP{d}{c}}}\right)}}{c}{\Tu}
  \]
  spawns \(P\) and \(Q\) and eventually forwards the channel \(a\) over \(c\):
  \begin{align*}
    &\jproc{c}{\tCut{b}{P}{Q}}\\
    &\msstep \jproc{b_1}{\tSendC{b_1}{a}{\tClose{b_1}}}, \jproc{c}{\tRecvC{d}{b_1}{\tWait{b_1}{\tFwdP{d}{c}}}}\\
    &\mssteps \jproc{b_2}{\tClose{b_2}}, \jproc{c}{\tWait{b_2}{\tFwdP{a}{c}}}\\
    &\mssteps \jproc{c}{\tFwdP{a}{c}}.\qedhere
  \end{align*}
\end{example}

The protocol \(B \Tlolly A\) is the negative dual of \(B \Tot A\).
Recall that polar-dual session types prescribe the same kind of communications, but in opposite directions.
In this case, where a provider of type $B \Tot A$ \emph{sends} a channel of type $B$, a provider of type $B \Tlolly A$ \emph{receives} a channel of type $B$.
The negative dual $B \Tlolly A$ of $B \Tot A$ is subtle: for linear-logical reasons, the polarities of $A$ and $B$ differ.
Indeed, though $B \Tlolly A$ and $A$ are negative, $B$ must be positive:
\begin{gather}
  \getrule{T-tlolly}
  \nonumber
  \\
  \getrule{R-tlolly}
  \quad
  \getrule{L-tlolly}
  \nonumber
  \\
  \label[msr]{eq:sill:msr-tlolly-r}
  \jproc{a}{\tRecvC{b}{a}{P}},
\jmsg{d}{\mSendCN{a}{e}{d}}
\to
\jproc{d}{\subst{e,d}{b,a}{P}}

  \\
  \label[msr]{eq:sill:msr-tlolly-l}
  \jproc{c}{\tSendC{a}{b}{P}}
\to
\exists d.
\jmsg{d}{\mSendCN{a}{b}{d}},
\jproc{c}{\subst{d}{a}P}

\end{gather}

\subsubsection{Functional Programming and Value Transmission}
\label{sec:sill-background:funct-progr-value}

The only base types in the functional layer are the types $\Tproc{a_0:A_0}{a_1:A_1,\dotsc,a_n:A_n}$ of quoted processes.\footnote{Polarized SILL can straightforwardly be extended to support other base types.}
These types are formed by the rule \rn{T\{\}}\glsadd{T-proc}.
The functional layer also supports function types \(\tau \to \sigma\).\varindex{functional type function type}{12!34}[|defin]
These are formed by the rule \rn{T$\to$}\glsadd{T-to}.
We assume that types are closed whenever they appear in a typing judgment for terms or processes.
\[
  \getrule{T-proc}
  \quad
  \getrule{T-to}
\]

Most of the introduction and elimination rules\varindex{{functional term} {introduction and elimination rules}}{1!2} for functional terms are standard.
New is the introduction rule \rn{I-\{\}} for quoted processes.
It encapsulates a process \(P\) as a value \(\tProcQ{a}{P}{\overline{a_i}}\) of quoted process type, where we abbreviate ordered lists using an overline.
Again, functional terms are not associated with any multiset rewriting rules: their operational behaviour is captured by the relation $\jeval{M}{v}$.\Index{fixed point}
\begin{gather*}
  \getrule{F-var}
  \quad
  \getrule{F-fix}
  \\
  \getrule{F-fun}
  \quad
  \getrule{F-app}
  \\
  \getrule{I-proc}
\end{gather*}

The elimination form for quoted process terms \(M\) is the process \(\tProcU{a}{M}{\overline{a_i}}\).
Operationally, this form evaluates the quoted process term \(M\) to a value \(v\), and then spawns the associated quoted process.
By the canonical forms property~\cite[Proposition~5.8.2]{kavanagh_2021:_commun_based_seman}, $v$ will always be a quoted process $\tProcQ{a}{P}{\overline{a_i}}$.
\[
  \getrule{E-proc}
\]
\begin{equation}
  \label[msr]{eq:sill:msr-e-qt}
  \jeval{M}{\tProc{a}{P}{\overline{a_i}}}, \jproc{a}{\tProc{a}{M}{\overline{a_i}}} \msstep \jproc{a}{P}

\end{equation}

The elimination rule \getrn{E-proc} differs from the original elimination rule given by \textcite[354]{toninho_2013:_higher_order_proces_funct_session}.
There, the elimination rule was a monadic bind similar to the \rn{Cut} rule below, where a quoted process could only be unquoted if it was composed with a continuation process.
Though the original rule has the advantage of enforcing a monadic discipline on the interaction between the functional and process layers, it complicates writing and reasoning about recursive processes.
Indeed, one cannot directly make a recursive tail call, but must instead always compose the recursive call with a continuation process.
We lose nothing by not requiring a continuation process: the original rule is a derived rule in our setting, and our rule can be defined as syntactic sugar in the original version of SILL.

\begin{example}
  \label{ex:sill-background-typing-mult-rewr:2}
  Recursive processes\Index{recursive process} are implemented using the functional layer's fixed point operator.
  An important recursive process is the \defin{divergent process}\DIndex<{divergent process} $\jtypem{\Psi}{\Delta}{\Omega}{c}{A}$\glsadd{diverge}, which exists for all $\Psi$ and $\Delta, {c : A}$.
  Let $\Omega'$ be the term given by $\jtypef{\Psi}{\tFix{\omega}{\tProc{c}{\tProc{c}{\omega}{\check \Delta}}{\check \Delta}}}{\Tproc{c : A}{\Delta}}$.
  Observe that $\jeval{\Omega'}{\tProc{c}{\tProc{c}{\Omega'}{\check \Delta}}{\check \Delta}}$.
  The process $\Omega$ is given by:
  \[
    \jtypem{\Psi}{\Delta}{\tProc{c}{\Omega'}{\check \Delta}}{c}{A}.
  \]
  By \cref{eq:sill:msr-e-qt}, $\jproc{c}{\Omega} \msstep \jproc{c}{\Omega}$.
\end{example}

\begin{remark}
  \label{rem:sill-background-typing-mult-rewr:5}
  \Cref{ex:sill-background-typing-mult-rewr:2} illustrates the subtle interplay between recursion and linearity.
  Though the linearity of channel contexts ensures that no channels are discarded, linearity cannot guarantee that all channels are used in the presence of recursion.
  Indeed, the divergent process \(\Omega\) communicates on none of its channels.
\end{remark}

Functional values can be sent over channels of type $\Tand{\tau}{A}$.\varindex{{session type} {value transmission}}{1!2}[|defin]
This positive protocol, formed by the rule \rn{C$\Tand{}{}$}\glsadd{T-tand}, specifies that the sent value has type $\tau$ and that subsequent communication has type $A$.
The process $\tSendV{a}{M}{P}$ evaluates the term $M$ to a value $v$, sends $v$ over the channel $a$, and continues as $P$.
The process $\tRecvV{x}{a}{Q}$ receives a value $v$ on $a$ and continues as $\subst{v}{x}{Q}$.
These behaviours are captured by multiset rewrite rules \labelcref{eq:sill:msr-tand-r,eq:sill:msr-tand-l}.
\begin{gather}
  \getrule{T-tand}
  \nonumber\\
  \begin{adjustbox}{max width=\textwidth}
    \getrule{R-tand}
    \quad
    \getrule{L-tand}
  \end{adjustbox}\nonumber\\
  \label[msr]{eq:sill:msr-tand-r}
  \jeval{M}{v}, \jproc{a}{\tSendV{a}{M}{P}} \to \exists d . \jproc{d}{\subst{d}{a}{P}}, \jmsg{a}{\mSendVP{a}{v}{d}}
\\
  \label[msr]{eq:sill:msr-tand-l}
  \jmsg{a}{\mSendVP{a}{v}{d}},\jproc{c}{\tRecvV{x}{a}{P}} \to \jproc{c}{\subst{d,v}{a,x}{P}}

\end{gather}

\begin{example}
  \label{ex:sill-background-typing-mult-rewr:8}
  The following process sends values \(v_1 : \tau_1\), \(v_2 : \tau_2\), and \(v_3 : \tau_3\) on the channel \(a\) before divergently providing a service of type \(A\):
  \[
    \jtypem{\cdot}{\cdot}{\tSendV{a}{v_1}{\tSendV{a}{v_2}{\tSendV{a}{v_3}{\Omega}}}}{a}{\Tand{\tau_1}{\left(\Tand{\tau_2}{\left(\Tand{\tau_3}{A}\right)}\right)}}.
  \]
  The queue-like structure of message processes \(\mSendVP{a}{v}{d}\) is analogous to the queues to \cref{ex:fair-mult-rewr:1}.
  Combined with \cref{eq:sill:msr-tand-l}, it ensures that messages are received in order.
  Indeed, consider the following trace of the above process:
  \begin{align*}
    &\jproc{a}{\tSendV{a}{v_1}{\tSendV{a}{v_2}{\tSendV{a}{v_3}{\Omega}}}}\\
    &\msstep \jproc{b}{\tSendV{a}{v_2}{\tSendV{a}{v_3}{\Omega}}}, \jmsg{a}{\mSendVP{a}{v_1}{b}}\\
    &\msstep \jproc{c}{\tSendV{c}{v_3}{\Omega}}, \jmsg{b}{\mSendVP{b}{v_2}{c}}, \jmsg{a}{\mSendVP{a}{v_1}{b}}\\
    &\msstep \left(\jproc{d}{\Omega}, \jmsg{c}{\mSendVP{c}{v_3}{d}}, \jmsg{b}{\mSendVP{b}{v_2}{c}},\right.\\
    &\qquad\qquad\qquad\left.\jmsg{a}{\mSendVP{a}{v_1}{b}}\right).
  \end{align*}
  Though multisets do not impose an order on their elements, any process \(P\) using the channel \(a\) will receive the values \(v_1\), \(v_2\), and \(v_3\) in that order.
  This is because \(P\) cannot know the name of the continuation channel \(b\) carrying \(v_2\) until it has received the value \(v_1\).
  Similarly, it cannot know the name of the continuation channel \(c\) carrying \(v_3\) until it has received the value \(v_2\).
\end{example}

\begin{example}
  \label{ex:sill-background-typing-mult-rewr:5}
  The following closed process \(P\) receives a function \(f\) of type \(\tau \to \sigma\) and a value \(x\) of type \(\tau\) on the channel \(a\).
  It evaluates \(f(x)\) and sends the corresponding value of type \(\sigma\) on \(b\), before forwarding \(a\) to \(b\):
  \[
    \jtypem{\cdot}{a : \Tand{(\tau \to \sigma)}{(\Tand{\tau}{A})}}{\tRecvV{f}{a}{\tRecvV{x}{a}{\tSendV{b}{f(x)}{\tFwdP{a}{b}}}}}{b}{\Tand{\sigma}{A}}
  \]
  If we take \(\tau = \sigma = (\delta \to \delta)\) for some type \(\delta\), then the following execution shows that sending \(P\) the process the values \(\lambda x : \tau . x\) and \(\lambda x : \delta . x\) on \(a\) causes it to send the value \(\lambda x : \delta . x\) on \(b\):
  \begin{align*}
    &\jmsg{c}{\mSendVP{b}{(\lambda x : \delta . x)}{d}}, \jmsg{a}{\mSendVP{a}{(\lambda x : \tau . x)}{c}}, \jproc{b}{P}\\
    &\msstep \jmsg{c}{\mSendVP{b}{(\lambda x : \delta . x)}{d}}, \jproc{b}{\tRecvV{x}{c}{\tSendV{b}{(\lambda x : \tau.x)x}{\tFwdP{c}{b}}}}\\
    &\msstep \jproc{b}{\tSendV{b}{(\lambda x : \tau.x)(\lambda x : \delta .x)}{\tFwdP{d}{b}}}\\
    &\msstep \jproc{e}{\tFwdP{d}{e}}, \jmsg{b}{\mSendVP{c}{(\lambda x : \delta . x)}{e}}\qedhere
  \end{align*}
\end{example}

\begin{example}
  \label{ex:sill-background-typing-mult-rewr:3}
  Polarized SILL supports remote code execution.
  The following process receives a quoted process $p$ of type $\Tproc{c : C}{a : B}$ over a channel $a$.
  After receipt of $p$, the channel $a$ has type $B$.
  The continuation process unquotes $p$ to provide a channel $c$ of type $C$ using the channel $a$:
  \[
    \jtypem{\Psi}{a : \Tand{\Tproc{c : C}{a : B}}{B}}{\tRecvV{p}{a}{\tProc{c}{p}{a}}}{c}{C}.\qedhere
  \]
\end{example}

The protocol $\Timp{\tau}{A}$ is the negative dual of $\Tand{\tau}{A}$.

\subsubsection{Choices}
\label{sec:sill-background:choices}

Processes can choose between services.\varindex{{session type} {choice type}}{1!2}[|defin]\index{internal choice|see{session type, choice type}}
An internal choice type  $\Tplus \{ l : A_l \}_{l \in L}$\glsadd{T-tplus} prescribes a choice between session types $\{ A_l \}_{l \in L}$ ($L$ finite).
The process $\tSendL{a}{k}{P}$ chooses to provide the service $A_k$ by sending the label $k$ on $a$, and then continues as $P$.
The process $\tCase{a}{\left\{l \Rightarrow P_l\right\}_{l \in L}}$ blocks until it receives a label $k$ on $a$ and then continues as $P_k$.
\begin{gather}
  \getrule{T-tplus}
  \qquad
  \getrule{R-tplus}
  \nonumber
  \\
  \getrule{L-tplus}
  \nonumber
  \\
  \label[msr]{eq:sill:msr-tplus-r}
  \jproc{a}{\tSendL{a}{k}{P}}
\to {}
\exists d.
\jproc{d}{\subst{d}{a}{P}},
\jmsg{a}{\mSendLP{a}{k}{d}}

  \\
  \label[msr]{eq:sill:msr-tplus-l}
  \jmsg{a}{\mSendLP{a}{k}{d}},
\jproc{c}{\tCase{a}{\left\{l \Rightarrow P_l\right\}_{l \in L}}}
\to
\jproc{c}{\subst{d}{a}{P_k}}

\end{gather}
The polar dual of the internal choice type $\Tplus \{l:A_l\}_{l \in L}$ is the external choice\index{external choice|see{session type, choice type}} type $\Tamp \{l:A_l\}_{l \in L}$\glsadd{T-tamp}.

\begin{example}
  \label{ex:sill-background-typing-mult-rewr:10}
  The following process \(P\) diverges if it receives the label \(\mt{div}\) on \(a\), and it forwards \(a\) to \(b\) if it receives the label \(\mt{fwd}\):
  \(
  \jtypem{\cdot}{a : \Tplus \{ \mt{div} : A, \mt{fwd} : B \}}{\tCase{a}{\{\mt{div} \Rightarrow \Omega \mid \mt{fwd} \Rightarrow \tFwdP{a}{b}\}}}{b}{B} \).
  Indeed, if \(a\) carries the label \(\mt{div}\), then \( \jmsg{a}{\tSendL{a}{\mt{div}}{c}}, \jproc{b}{P} \msstep \jproc{b}{\Omega} \).
  If \(a\) carries the label \(\mt{fwd}\), then \( \jmsg{a}{\tSendL{a}{\mt{fwd}}{c}}, \jproc{b}{P} \msstep \jproc{b}{\tFwdP{c}{b}} \).
\end{example}

\subsubsection{Shifts in Polarity}
\label{sec:sill-background:shifts-polarity}

Process communication is asynchronous.\varindex{{session type} {polarity shift}}{1!2}[|defin]\index{synchronization|see{session type, polarity shift}}
Synchronization on a channel is encoded using ``polarity shifts''~\cite{pfenning_griffith_2015:_polar_subst_session_types}.
The positive protocol $\Tds A$\glsadd{T-tds} prescribes a synchronization (a shift message) followed by communication satisfying the negative type $A$.
The process $\tSendS{a}{P}$ signals that it is ready to receive on $a$ by sending a ``shift message'' on $a$, and continues as $P$.
The process $\tRecvS{a}{P}$ blocks until it receives the shift message and continues as~$P$.
\begin{gather}
  \getrule{T-tds}
  \nonumber
  \\
  \getrule{R-tds}
  \quad
  \getrule{L-tds}
  \nonumber
  \\
  \label[msr]{eq:sill:msr-tds-r}
  \jproc{a}{\tSendS{a}{P}}
\to
\exists d.
\jproc{d}{\subst{d}{a}{P}},
\jmsg{a}{\mSendSP{a}{d}}

  \\
  \label[msr]{eq:sill:msr-tds-l}
  \jmsg{a}{\mSendSP{a}{d}},
\jproc{c}{\tRecvS{a}{P}}
\to
\jproc{c}{\subst{d}{a}{P}}

\end{gather}
The dual of the positive type $\Tds A$ is the negative type $\Tus A$\glsadd{T-tus}.

At first glance, polarity shifts may appear to be special cases of choice types.
The key difference is that the direction of communication does not change with choice types: the label and subsequent communications travel in the same direction.
In contrast, the shift message and subsequent communications travel in opposite directions.

\subsubsection{Recursive Types}
\label{sec:sill-background:recursive-types}

The recursive type $\Trec{\alpha}{A}$\glsadd{T-trec} prescribes an ``unfold'' message followed by communication of type $[\Trec{\alpha}{A}/\alpha]A$.\DIndex{recursive {session type}}
To ensure that unfolding a recursive type is well-defined, we require that the type variable $\alpha$ have the same polarity as the recursive type.
The process $\tSendU{a}{P}$ sends an unfold message and continues as $P$.
The process $\tRecvU{a}{P}$ receives an unfold message and continues as $P$.
\begin{gather}
  \getrule{T-var}
  \quad
  \getrule{T-rhop}
  \nonumber
  \\
  \getrule{R-rhop}
  \nonumber
  \\
  \getrule{L-rhop}
  \nonumber
  \\
  \label[msr]{eq:sill:msr-rhop-r}
  \jproc{a}{\tSendU{a}{P}}
\to
\exists d.
\jproc{d}{\subst{d}{a}{P}},
\jmsg{a}{\mSendUP{a}{d}}

  \\
  \label[msr]{eq:sill:msr-rhop-l}
  \jmsg{a}{\mSendUP{a}{d}},
\jproc{c}{\tRecvU{a}{P}}
\to
\jproc{c}{\subst{d}{a}{P}}

\end{gather}

\begin{example}
  \label{ex:sill-background-typing-mult-rewr:1}
  The set of conatural numbers is given by \(\N \cup \{ \omega \}\), where \(\N\) is the usual set of natural numbers, and \(\omega\) corresponds to a countably infinite stack of successors \(s(s(s(\cdots)))\).
  The protocol $\mt{conat} = \Trec{\alpha}{\Tplus \{ \mt{z} : \Tu, \mt{s} : \alpha \}}$ encodes conatural numbers.
  Indeed, a communication is either an infinite sequence of successor labels $\mt{s}$, or some finite number of $\mt{s}$ labels followed by the zero label $\mt{z}$ and termination.
  The following recursive process \(\mt{omega}\) outputs \(\omega\) on $o$:
  \[
    \jtypem{\cdot}{\cdot}{\tFix{\omega}{\tSendU{o}{\tSendL{\mt{s}}{o}{\omega}}}}{o}{\mt{conat}}.
  \]
  It has an infinite fair execution where for $n \geq 1$, the \((3n-2)\)-th, \((3n-1)\)-th, and \(3n\)-th rules are respectively instantiations of rules \eqref{eq:sill:msr-tot-l}, \eqref{eq:sill:msr-rhop-r}, and \eqref{eq:sill:msr-tot-r}.
\end{example}

\begin{example}
  \label{ex:sill-background-typing-mult-rewr:4}
  Consider the type $\mt{bits} = \Trec{\beta}{\Tplus \{ \mt{0} : \beta, \mt{1} : \beta \}}$ of bit streams.
  Its communications consist of potentially infinite sequences of $\cunfold$ messages interleaved with labels $\mt{0}$ and $\mt{1}$.
  The following process receives a bit stream on $i$, flips its bits, and outputs the result on $o$:
  \begin{align*}
    \jtypem{\Psi}{i : \mt{bits}}{%
    \tProc{o}{%
    \tFix{f}{%
    \tProc{o}{%
    &\tRecvU{i}{\\
    &\tSendU{o}{\\
    &\tCase{i}{\{\,\mt{0} \Rightarrow \tSendL{o}{\mt{1}}{\tProc{o}{f}{i}}\\
    &\phantom{\tCase{i}{}}\mid \mt{1} \Rightarrow \tSendL{o}{\mt{0}}{\tProc{o}{f}{i}}\}}\\
    &}%
      }%
      }{i}%
      }%
      }{x}%
      }{o}{\mt{bits}}\qedhere
  \end{align*}
\end{example}

We refer the reader to \cite{toninho_2013:_higher_order_proces_funct_session, kavanagh_2021:_commun_based_seman} for further example processes.

\subsection{Key Properties of Polarized SILL}
\label{sec:sill-background:key-props}

We highlight some of the key properties of Polarized SILL that will be used in later portions of this paper.
We refer the reader to \cite[\bibstring{chapter}~5]{kavanagh_2021:_commun_based_seman} for a full development of these properties.

\subsubsection{Static Properties of Processes and Terms}
\label{sec:stat-prop-proc}

The set $\freecn(P)$ of \defin{free channel names} in a process $P$ is inductively defined in the usual way on the structure of $P$.
Recall that the only process formers that bind channel names are $\tCut{a}{P}{Q}$, which binds $a$ in $P$ and $Q$, and $\tRecvC{a}{b}{P}$, which binds $a$ in $P$.
The following \namecref{prop:sill-background:2} follows by induction on the typing derivation:

\begin{proposition}
  \label{prop:sill-background:2}
  If $\jtypem{\Psi}{c_1 : A_1, \dotsc, c_n : A_n}{P}{c_0}{A_0}$, then $\freecn(P) = \{ c_0, \dotsc, c_n \}$.
\end{proposition}

Using polarity, we can statically partition free channels into the sets $\outcn(P)$ of \defin{output channel names}\DIndex<{output channel} and $\inpcn(P)$ of \defin{input channel names}\DIndex<{input channel}.
Intuitively, $c \in \outcn(P)$ if the next time $P$ communicates on $c$, it sends a message on $c$; the meaning of $c \in \inpcn(P)$ is symmetric.
Explicitly:

\begin{definition}
  \label{def:sill-background-stat-prop-terms:1}
  Assume $\jtypem{\Psi}{c_1 : A_1, \dotsc, c_n : A_n}{P}{c_0}{A_0}$.
  Let $\outcn(P)$\glsadd{outcn} and $\inpcn(P)$\glsadd{inpcn} be the least subsets of $\freecn(P)$ such that:
  \begin{enumerate}
  \item $c_0 \in \outcn(P)$ if and only if $A_0$ is positive;
  \item $c_0 \in \inpcn(P)$ if and only if $A_0$ is negative;
  \item for $1 \leq i \leq n$, $c_i \in \outcn(P)$ if and only if $A_i$ is negative;
  \item for $1 \leq i \leq n$, $c_i \in \inpcn(P)$ if and only if $A_i$ is positive.\qedhere
  \end{enumerate}
\end{definition}

A \defin{substitution} ${\jcmf{\sigma}{\Phi}{x_1 : \tau_1, \dotsc, x_n : \tau_n}}$ is a list $\sigma$ of terms $N_1, \dotsc, N_n$ satisfying $\jtypef{\Phi}{N_i}{\tau_i}$ for all $1 \leq i \leq n$.
We write $\apprs{\sigma}{M}$ for the simultaneous substitution $\subst{\vec N}{\vec x}{M}$.
The typing judgments for terms and processes are closed under substitution.

\subsubsection{Static Properties of Configurations}

The free and bound channel names in a configuration are given by the unions of the free and bound channel names that appear in its process and message facts.
Input and output channels of configurations are a direct adaptation of \cref{def:sill-background-stat-prop-terms:1}:

\begin{definition}
  \label{def:sill-background-stat-prop-typed-config:3}
  Assume $\jcfgti{c_1 : C_1, \dotsc, c_n : C_n}{\Iota}{\mc{C}}{a_0 : A_0, \dotsc, a_m : A_m}$.
  Let the \defin{output channels}\DIndex<{output channel} $\outcn(\mc{C})$\glsadd{outcn} and \defin{input channels}\DIndex<{input channel} $\inpcn(\mc{C})$\glsadd{inpcn} of $\mc{C}$ be the least subsets of $\freecn(\mc{C})$ such that:
  \begin{enumerate}
  \item for $0 \leq i \leq m$, $a_i \in \outcn(\mc{C})$ if and only if $A_i$ is positive;
  \item for $0 \leq i \leq m$, $a_i \in \inpcn(\mc{C})$ if and only if $A_i$ is negative;
  \item for $1 \leq i \leq n$, $c_i \in \outcn(\mc{C})$ if and only if $C_i$ is negative;
  \item for $1 \leq i \leq n$, $c_i \in \inpcn(\mc{C})$ if and only if $C_i$ is positive.\qedhere
  \end{enumerate}
\end{definition}

\begin{remark}
  Unlike the free channels of processes, the free channels of typed configurations are not partitioned as input and output channels, but as input, output, and internal channels.
  Recall that the internal channels of $\jcfgti{\Gamma}{\Iota}{\mc{C}}{\Delta}$ are those in $\Iota$.
\end{remark}

It is also useful to specify the channel on which a message was sent---its ``carrier channel''---as well as its continuation channel.

\begin{definition}
  \label{def:sill-background-stat-prop-typed-config:1}
  The \defin[channel!carrier]{carrier channel} $\carrcn(\jmsg{a}{m})$ and \defin[channel!continuation]{continuation channel} $\contcn(\jmsg{a}{m})$ and of a message fact $\jmsg{a}{m}$, if defined, are:
  \begin{align*}
    \carrcn(\jmsg{a}{\mClose{a}}) &= a & \contcn(\jmsg{a}{\mClose{a}})\ \text{undefined}\\
    \carrcn(\jmsg{a}{\mSendCP{a}{b}{d}}) &= a &     \contcn(\jmsg{a}{\mSendCP{a}{b}{d}}) &= d\\
    \carrcn(\jmsg{d}{\mSendCN{a}{b}{d}}) &= a &     \contcn(\jmsg{d}{\mSendCN{a}{b}{d}}) &= d\\
    \carrcn(\jmsg{a}{\mSendLP{a}{k}{d}}) &= a &     \contcn(\jmsg{a}{\mSendLP{a}{k}{d}}) &= d\\
    \carrcn(\jmsg{d}{\mSendLN{a}{k}{d}}) &= a &     \contcn(\jmsg{d}{\mSendLN{a}{k}{d}}) &= d\\
    \carrcn(\jmsg{a}{\mSendVP{a}{v}{d}}) &= a &     \contcn(\jmsg{a}{\mSendVP{a}{v}{d}}) &= d\\
    \carrcn(\jmsg{d}{\mSendVN{a}{v}{d}}) &= a &     \contcn(\jmsg{d}{\mSendVN{a}{v}{d}}) &= d\\
    \carrcn(\jmsg{a}{\mSendSP{a}{d}}) &= a &    \contcn(\jmsg{a}{\mSendSP{a}{d}}) &= d\\
    \carrcn(\jmsg{d}{\mSendSN{a}{d}}) &= a &     \contcn(\jmsg{d}{\mSendSN{a}{d}}) &= d\\
    \carrcn(\jmsg{a}{\mSendUP{a}{d}}) &= a &    \contcn(\jmsg{a}{\mSendUP{a}{d}}) &= d\\
    \carrcn(\jmsg{d}{\mSendUN{a}{d}}) &= a &    \contcn(\jmsg{d}{\mSendUN{a}{d}}) &= d\qedhere
  \end{align*}
\end{definition}

\begin{remark}
  \label{rem:sill-background-stat-prop-typed-config:1}
  The carrier channel of a message fact is always an output channel, but the converse need not be true.
  Consider for example the message fact $\jmsg{a}{\mSendSP{a}{d}}$.
  Its carrier channel is $a$, but its output channels are $a$ and its continuation channel $d$.
\end{remark}

\begin{proposition}[Subformula Property]\varindex{configuration {subformula property}}{1!2 2}
  \label{prop:sill-background-stat-prop-typed-config:3}
  If\/ $\jcfgti{\Gamma'}{\Iota'}{\mc{E}}{\Delta'}$ appears in the derivation of $\jcfgti{\Gamma}{\Iota}{\mc{F}}{\Delta}$, then $\Gamma' \subseteq \Gamma \Iota$, $\Iota' \subseteq \Iota$, and $\Delta' \subseteq \Iota \Delta$.
\end{proposition}

\begin{proof}
  By induction on the derivation of $\jcfgti{\Gamma}{\Iota}{\mc{F}}{\Delta}$.
\end{proof}

If two multisets in a configuration share a channel, then it is an internal channel:

\begin{lemma}
  \label{lemma:sill-background-properties-traces:5}
  If\/ $\jcfgti{\Psi}{\Iota}{\mc{E},\mc{F}}{\Theta}$ and $c \in \freecn(\mc{E}) \cap \freecn(\mc{F})$, then $c \in \check \Iota$.
\end{lemma}

\begin{proof}
  By induction on the derivation of $\jcfgti{\Psi}{\Iota}{\mc{E},\mc{F}}{\Theta}$.
\end{proof}

Every configuration can be decomposed as the composition of independent ``simply branched'' configurations with no common channels.
The absence of common channels implies that these simply branched subconfigurations do not interact during executions.
We will use this fact to reduce proofs about arbitrary configurations to proofs about simply branched configurations.

\begin{definition}
  \label{def:sill-obs-equiv/observ-comm:9}
  A configuration $\jcfgti{\Gamma}{\Iota}{\mc{C}}{\Delta}$ is \defin{simply branched}\DIndex<{{simply branched} configuration} if it has a derivation in which every instance of the rule
  \[
    \getrule{conf-c}
  \]
  has exactly one channel in the context $\Pi$.\qedhere
\end{definition}

\begin{proposition}[Simply Branched Decomposition]
  \label{prop:sill-obs-equiv/observ-comm-equiv:7}
  Every configuration $\jcfgt{\Gamma}{\mc{C}}{d_0 : D_0, \dotsc, d_n : D_n}$ is the composition $\jcfgti{\Gamma_0, \dotsc, \Gamma_n}{\Iota_0, \dotsc, \Iota_n}{\mc{C}_0, \dotsc, \mc{C}_n}{d_0 : D_0, \dotsc, d_n : D_n}$ of simply-branched configurations $\jcfgti{\Gamma_i}{\Iota_i}{\mc{C}_i}{d_i : D_i}$ for $0 \leq i \leq n$.
\end{proposition}

\begin{proof}
  By induction on the derivation of $\jcfgt{\Gamma}{\mc{C}}{d_0 : D_0, \dotsc, d_n : D_n}$.
\end{proof}

We can characterize simply branched configurations by looking at their provided channels:

\begin{proposition}
  \label{prop:sill-obs-equiv/observ-comm:3}
  A configuration $\jcfgti{\Gamma}{\Iota}{\mc{C}}{\Delta}$ is simply branched if and only if $\Delta$ contains exactly one channel.
\end{proposition}

\begin{proof}
  Sufficiency follows by induction on the derivation of $\jcfgt{\Gamma}{\mc{C}}{\Delta}$.
  \Cref{prop:sill-obs-equiv/observ-comm-equiv:7} implies necessity.
\end{proof}

The following \namecref{lemma:sill-background-properties-traces:9} shows that we can ``swap out'' any subset of a multiset with one that has the same interface, even if their internal channels differ.
We use the notation $\mc{J} \derivable{}{} \mc{J}'$ to mean that there is a (hypothetical) derivation of $\mc{J}'$ assuming the hypothesis $\mc{J}$.

\begin{lemma}[Replacement Property]
  \label{lemma:sill-background-properties-traces:9}
  Assume $\jcfgti{\Gamma'}{\Iota_F}{\mc{F}}{\Theta'} \derivable{}{} \jcfgti{\Gamma}{\Iota_E \Iota_F}{\mc{E},\mc{F}}{\Theta}$.
  If\/ $\jcfgti{\Gamma'}{\Iota_G}{\mc{G}}{\Theta'}$ is well-formed and $\check \Iota_G$ is disjoint from $\check \Gamma, \check \Iota_E, \check \Theta$, then $\jcfgti{\Gamma'}{\Iota_G}{\mc{G}}{\Theta'} \derivable{}{} \jcfgti{\Gamma}{\Iota_E \Iota_G}{\mc{E},\mc{G}}{\Theta}$.
\end{lemma}

\begin{proof}
  By induction on the hypothetical derivation ${\jcfgti{\Gamma'}{\Iota_F}{\mc{F}}{\Theta'} \derivable{}{} \jcfgti{\Gamma}{\Iota_E \Iota_F}{\mc{E},\mc{F}}{\Theta}}$.
\end{proof}

Configuration contexts are configuration contexts with a hole:

\begin{definition}
  \label{def:sill-obs-equiv/relat-equiv:1}
  A \defin[context!typed!for configurations|defin]{(typed) configuration context}\DIndex{configuration context} $\jcfgti{\Gamma}{\Iota}{\ctxh{\mc{C}}{\Lambda}{\Xi}}{\Delta}$\glsadd{confctx} is a configuration formed by the rules of \cref{sec:sill-background:overview-dynamics} plus one instance of the axiom \getrn{conf-h}:
  \[
    \getrule{conf-h}
  \]
  Consider a configuration context $\jcfgti{\Gamma}{\Iota}{\ctxh{\mc{C}}{\Lambda}{\Xi}}{\Delta}$ and a configuration $\jcfgti{\Lambda}{\Iota'}{\mc{D}}{\Xi}$ such that $\check \Iota'$ is disjoint from $\check \Gamma,\check \Iota,\check \Delta$.
  The result of ``plugging'' $\mc{D}$ into the hole is the configuration $\jcfgti{\Gamma}{\Iota,\Iota'}{\ctxh[\mc{D}]{\mc{C}}{\Lambda}{\Xi}}{\Delta}$ given by \cref{lemma:sill-background-properties-traces:9}, where we replace the axiom \getrn{conf-h} by the derivation of $\jcfgti{\Lambda}{\Iota'}{\mc{D}}{\Xi}$ and thread the added internal channels $\Iota'$ through the derivation.
\end{definition}

\begin{remark}
  We can always plug a configuration in a hole with a matching interface in \cref{def:sill-obs-equiv/relat-equiv:1} by suitably renaming the internal channel names in $\Iota'$.
\end{remark}

\subsubsection{Dynamic Properties of Typed Configurations}

The type-preservation property for Polarized SILL states that its substructural operational semantics preserves configuration interfaces and that it does not change the types of internal channel names.
We formulate our preservation result in terms of configuration contexts.\Index{configuration context}
In particular, our formulation makes explicit the fact that the type of the stationary multiset (seen as a configuration context) does not change, and the fact that the active multiset and its replacement have the same interface.

\begin{proposition}[Preservation]\varindex{{preservation property} for@ configuration configurations}{3!1 1!~24}[|defin]
  \label{prop:sill-background:1}
  Assume $\jcfgti[\Sigma]{\Gamma}{\Iota}{\mc{C}}{\Delta}$.
  If $\msinc{\Sigma}{\mc{C}} \msstep \msinc{\Sigma'}{\mc{C}'}$ by some rule instance $\mc{E} \msstep \mc{E}'$, then there exist $\Psi \subseteq \Gamma \Iota$, $\Iota_L \subseteq \Iota$, and $\Theta \subseteq \Iota \Delta$ such that
  \begin{enumerate}
  \item $\jcfgti{\Psi}{\Iota_L}{\mc{E}}{\Theta}$,
  \item $\jcfgti{\Psi}{\Iota_R}{\mc{E}'}{\Theta}$ for some $\Iota_R$ whose channel names are disjoint from those in $\Gamma\Iota_L\Delta$,
  \item $\mc{C}$ is given by $\jcfgti[\Sigma]{\Gamma}{\Iota_L \Iota'}{\ctxh[\mc{E}]{\mc{D}}{\Psi}{\Theta}}{\Delta}$ for some configuration context $\jcfgti{\Gamma}{\Iota'}{\ctxh{\mc{D}}{\Psi}{\Theta}}{\Delta}$ and some $\Iota'$, and
  \item $\mc{C}'$ is given by $\jcfgti[\Sigma, \check \Iota_R]{\Gamma}{\Iota_R \Iota'}{\ctxh[\mc{E}']{\mc{D}}{\Psi}{\Theta}}{\Delta}$ and $\Sigma' = \Sigma, \check \Iota_R$.
  \end{enumerate}
\end{proposition}

As a corollary of \cref{prop:sill-background:1}, we know that in a trace $T = (\mc{C}_0, (r_i;\delta_i)_i)$ from a well-typed configuration $\jcfgti[\Sigma]{\Gamma}{\Iota_0}{\mc{C}_0}{\Delta}$, there exist for all $i$ some $\Sigma_i$ and $\Iota_i$ such that $\jcfgti[\Sigma_i]{\Gamma}{\Iota_i}{\mc{C}_i}{\Delta}$.
Indeed, $\Sigma_i$ and $\Iota_i$ are given by induction on $n$, where each step is given by an application of \cref{prop:sill-background:1}.
We also know that every free channel name appearing in a trace has an associated session type, a fact that we will use repeatedly when reasoning about traces.
\Cref{def:sill-background-properties-traces:1,cor:sill-background-properties-traces:1} capture this relationship between traces, channels, and session types.
We use colour to indicate the \defin[judgment!mode]{modes} of use of judgments, where we designate \inm{inputs} in blue and \outm{outputs} in red.
The set of \defin{free channel names in a trace} \(T = (\mc{C}_0, (r_i;\delta_i)_i)\) is \(\freecn(T) = \bigcup_i \freecn(\mc{C}_i)\).

\begin{definition}
  \label{def:sill-background-properties-traces:1}
  Let $T = (\mc{C}_0, (r_i;\delta_i)_i)$ be a trace from ${\jcfgti[\Sigma]{\Gamma}{\Iota_0}{\mc{C}_0}{\Delta}}$.
  We write ${\jttp{T}{c}{A}}$\glsadd{jttp} to mean that $c : A$ appears in $\Gamma$, $\Delta$, or $\Iota_i$ for some $i$, where \(\Iota_i\) is given by recursion on \(i\) and \cref{prop:sill-background:1}.
\end{definition}

\begin{proposition}
  \label{cor:sill-background-properties-traces:1}
  Let $T = (\mc{C}_0, (r_i;\delta_i)_i)$ be a trace from $\jcfgt{\Gamma}{\mc{C}_0}{\Delta}$.
  The judgment $\jttp{T}{c}{A}$ is a total function from free channel names in \(T\) to session types, \ie, for all \(c \in \freecn(T)\), there exists unique \(A\) such that \(\jttp{T}{c}{A}\).
\end{proposition}

\Cref{prop:main:3} is reminiscent of a type inversion principle (contrast it with, \eg, \cite[Lemma~8.2.2]{pierce_2002:_types_progr_languag}).
It implies that types assigned by $\jttp{T}{c}{A}$\varindex{session type of@ a@ channel}{12!~345}\Index<{session-typed channel} are consistent with those for message and process facts appearing in the trace $T$:

\begin{proposition}[Inversion Principle]
  \label{prop:main:3}
  Let $T = (\mc{C}_0, (r_i;\delta_i)_i)$ be a trace from $\jcfgti[\Sigma_0]{\Gamma}{\Iota_0}{\mc{C}_0}{\Delta_0}$, and let $\Sigma_i$ and $\Iota_i$ be given by induction on $i$ and \cref{prop:sill-background:1}.
  For all $n$, if $\jproc{c_0}{P} \in \mc{C}_n$, then
  \begin{enumerate}
  \item $c_0 \in \freecn(P)$;
  \item for all $c_i \in \freecn(P)$, then $\jttp{T}{c_i}{A_i}$ for some $A_i$; and
  \item where $\freecn(P) = \{ c_0, \dotsc, c_m \}$, we have $\jtypem{\cdot}{c_1 : A_1, \dotsc, c_m : A_m}{P}{c_0}{A_0}$.
  \end{enumerate}
  If $\jmsg{c_0}{m} \in \mc{C}_n$, then
  \begin{itemize}
  \item if $m = \mClose{c}$, then $\jttp{T}{c}{\Tu}$;
  \item if $m = \mSendLP{c}{l_j}{d}$, then $\jttp{T}{c}{\Tplus \{ l_i : A_i \}_{i \in I}}$ for some $A_i$ ($i \in I$), and $\jttp{T}{d}{A_j}$ for some $j \in I$;
  \item if $m = \mSendLN{c}{l_j}{d}$, then $\jttp{T}{c}{\Tamp \{ l_i : A_i \}_{i \in I}}$ for some $A_i$ ($i \in I$), and $\jttp{T}{d}{A_j}$ for some $j \in I$;
  \item if $m = \mSendCP{c}{a}{b}$, then $\jttp{T}{c}{A \Tot B}$, $\jttp{T}{a}{A}$, and $\jttp{T}{b}{B}$ for some $A$ and $B$;
  \item if $m = \mSendCN{c}{a}{b}$, then $\jttp{T}{c}{A \Tlolly B}$, $\jttp{T}{a}{A}$, and $\jttp{T}{b}{B}$ for some $A$ and $B$;
  \item if $m = \mSendVP{c}{v}{d}$, then $\jttp{T}{c}{\Tand{\tau}{A}}$ and $\jttp{T}{d}{A}$ for some $A$ and $\tau$ such that $\jtypef{\cdot}{v}{\tau}$;
  \item if $m = \mSendVN{c}{v}{d}$, then $\jttp{T}{c}{\Timp{\tau}{A}}$ and $\jttp{T}{d}{A}$ for some $A$ and $\tau$ such that $\jtypef{\cdot}{v}{\tau}$;
  \item if $m = \mSendUP{c}{d}$, then $\jttp{T}{c}{\Trec{\alpha}{A}}$ and $\jttp{T}{d}{[\Trec{\alpha}{A}/\alpha]A}$ for some $\jstype[+]{\alpha}{A}$; and
  \item if $m = \mSendUN{c}{d}$, then $\jttp{T}{c}{\Trec{\alpha}{A}}$ and $\jttp{T}{d}{[\Trec{\alpha}{A}/\alpha]A}$ for some $\jstype[-]{\alpha}{A}$.
  \end{itemize}
\end{proposition}

\Cref{cor:sill-background-dyn-prop-typed-config:1} states that each channel in a trace appears as the carrier channel of at most one message fact.
In \cref{sec:sill-obs-equiv:observ-comm:comm-single-chan}, this \namecref{cor:sill-background-dyn-prop-typed-config:1} will imply that we observe a unique communication on each channel.

\begin{proposition}
  \label{cor:sill-background-dyn-prop-typed-config:1}
  Let $T = (\mc{C}_0, (r_i;\delta_i)_i)$ be a trace from $\jcfgti{\Gamma}{\Iota_0}{\mc{C}_0}{\Delta_0}$.
  For all $j \leq k$, if $\jmsg{c_j}{m_j} \in \mc{C}_j$ and $K \in \mc{C}_k$, then $\carrcn(\jmsg{c_j}{m_j}) \notin \outcn(K)$ or $K = \jmsg{c_j}{m_j}$.
\end{proposition}

Finally, we apply the results of \cref{sec:ssos-fairness:prop-fair-trac} to Polarized SILL.
We use the fact that its multiset rewriting system is non-overlapping on well-typed configurations to characterize fair executions of processes.

\begin{proposition}
  \label{prop:sill-background-dyn-prop-typed-config:1}
  If\/ $\jcfgti{\Gamma}{\Iota}{\mc{C}}{\Delta}$, then the MRS given by \cref{sec:sill-background:compl-list-mult} is non-overlapping on $\mc{C}$.
\end{proposition}

\begin{proof}
  It is sufficient to show that if $s_1(\phi_1)$ and $s_2(\phi_2)$ are distinct instantiations applicable to $\mc{C}$, then $F_1(\phi_1)$ and $F_2(\phi_2)$ are disjoint multisets: $F_1(\phi_1) \cap F_2(\phi_2) = \emptyset$.
  Indeed, if this is the case and $s_1(\phi_1), \dotsc, s_k(\phi_k)$ are the distinct rule instantiations applications to $\mc{C}$, then $F_1(\phi_1), \dotsc, F_k(\phi_k)$ are all pairwise-disjoint multisets.
  It follows that $F_1(\phi_1),\dotsc,F_k(\phi_k) \subseteq \mc{C}$, so the overlap in $\mc{C}$ is empty: $\olap{\mc{C}}{F_1(\phi_1),\dotsc,F_k(\phi_k)} = \emptyset$.

  To show that $F_1(\phi_1)$ and $F_2(\phi_2)$ are disjoint, we suppose to the contrary.
  A case analysis on the possible facts in the intersection implies that $s_1(\phi_1)$ and $s_2(\phi_2)$ are equivalent instantiations, a contradiction.
\end{proof}

\begin{corollary}\varindex{fairness fair execution}{1!23}\varindex{{multiset rewriting system} execution}{1!2}
  \label{cor:sill-background-dyn-prop-typed-config:2}
  Every configuration $\jcfgti{\Gamma}{\Iota}{\mc{C}}{\Delta}$ has a fair execution.
  Its fair executions are all permutations of each other and they are all union-equivalent.
\end{corollary}

\begin{proof}
  By \cref{prop:sill-background:1,prop:sill-background-dyn-prop-typed-config:1}, $\mc{P}$ is non-overlapping from $\mc{C}$.
  By \cref{prop:ssos-fairness/prop-fair-trac:1}, this implies that it commutes from $\mc{C}$, so a fair execution exists by \cref{prop:ssos-fairness/prop-fair-trac:3}.
  All of its fair executions are permutations of each other by \cref{prop:ssos-fairness/prop-fair-trac:4}.
  They are union-equivalent by \cref{cor:ssos-fairness/prop-fair-trac:1}.
\end{proof}

\begin{corollary}
  \label{cor:sill-background-properties-traces:2}
  Every process $\jtypem{\cdot}{\Delta}{P}{c}{A}$ has a fair execution.
  Its fair executions are all permutations of each other and they are all union-equivalent.
\end{corollary}

\begin{proof}
  Immediate by \cref{cor:sill-background-dyn-prop-typed-config:2} with the initial configuration ${\jcfgti{\Delta}{\cdot}{\jproc{c}{P}}{c : A}}$.
\end{proof}

\subsection{Type-Indexed Relations}
\label{sec:sill-obs-equiv:relat-equiv}

Our ultimate goal is to relate programs that are equivalent or that somehow approximate each other.
We define various desirable properties for relations on programs and configurations.

Polarized SILL and its configurations do not have unicity of typing, and processes could be equivalent at one type but not at another.
Accordingly, we would like our relations to be type-indexed:

\begin{definition}
  \label{def:sill-obs-equiv/bisimulations:1}
  \defin{Type-indexed (binary) relations}\DIndex<{type-indexed relation} are families of relations indexed by typing sequents.
  Explicitly:
  \begin{enumerate}
  \item A \defin{type-indexed relation $\relfnt{R}$ on configurations}\DIndex<{{type-indexed relation} configuration} is a family of relations $(\relfnt{R}_{\Delta \vdash \Phi})_{\Delta,\Phi}$ where $(\mc{C}, \mc{D}) \in \relfnt{R}_{\Delta \vdash \Phi}$ only if $\jcfgt{\Delta}{\mc{C}}{\Phi}$ and $\jcfgt{\Delta}{\mc{D}}{\Phi}$.
    In this case, we write $\jtrelc{\relfnt{R}}{\Delta}{\mc{C}}{\mc{D}}{\Phi}$\glsadd{jtrelc}.
  \item A \defin{type-indexed relation $\relfnt{R}$ on processes}\DIndex<{{type-indexed relation} process} is a family of relations $(\relfnt{R}_{\Psi; \Delta \vdash c : A})_{\Psi,\Delta,c : A}$ where $(P, Q) \in \relfnt{R}_{\Psi, \Delta \vdash c : A}$ only if $\jtypem{\Psi}{\Delta}{P}{c}{A}$ and $\jtypem{\Psi}{\Delta}{Q}{c}{A}$.
    In this case, we write $\jtrelp{\relfnt{R}}{\Psi}{\Delta}{P}{Q}{c}{A}$\glsadd{jtrelp}.
  \item A \defin{type-indexed relation $\relfnt{R}$ on terms}\DIndex<{{type-indexed relation} {functional term}} is a family of relations $(\relfnt{R}_{\Psi \vdash \tau})_{\Psi,\tau}$ where $(M, N) \in \relfnt{R}_{\Psi \vdash \tau}$ only if $\jtypef{\Psi}{M}{\tau}$ and $\jtypef{\Psi}{N}{\tau}$.
    In this case, we write $\jtrelf{\relfnt{R}}{\Psi}{M}{N}{\tau}$\glsadd{jtrelf}.
  \end{enumerate}

  Type-indexed relations are assumed to satisfy the exchange, renaming, and weakening structural properties whenever their underlying judgments do.
  The \defin{renaming}\varindex{{structural property} renaming}{1!2 2!1} property for type-indexed relations on configurations is subtle because we elided internal channels.
  Explicitly, it is the property:
  \begin{itemize}
  \item If $\jtrelc{\relfnt{R}}{\Gamma}{\mc{C}}{\mc{D}}{\Delta}$, $\jcfgti{\Gamma}{\Iota_1}{\mc{C}}{\Delta}$, and $\jcfgti{\Gamma}{\Iota_2}{\mc{D}}{\Delta}$, and $\rename{\sigma}{\Gamma\Delta}{\Gamma'\Delta'}$, $\rename{\sigma_1}{\Iota_1}{\Iota_1'}$, and $\rename{\sigma_2}{\Iota_2}{\Iota_2'}$ are renamings, then $\jtrelc{\relfnt{R}}{\Gamma'}{\apprs{\sigma,\sigma_1}{\mc{C}}}{\apprs{\sigma,\sigma_2}{\mc{D}}}{\Delta'}$.\qedhere
  \end{itemize}
\end{definition}

We will study the effects of running ``equivalent programs'' in various program contexts.
Contexts are programs with holes:

\begin{definition}
  \label{def:sill-obs-equiv/relat-equiv:3}
  A \defin[context!typed!for processes|defin]{(typed) process context}\DIndex{process context} $\jtypem{\Psi}{\Delta}{\ctxh{C}{\Gamma;\Lambda}{b:B}}{a}{A}$\glsadd{procctx} is a process formed by the rules of \cref{sec:sill-den-sem:rules-proc-form} plus one instance of the axiom \getrn{Ct-p-hole}:
  \[
    \getrule{Ct-p-hole}
  \]
  Given a process context $\jtypem{\Psi}{\Delta}{\ctxh{C}{\Gamma;\Lambda}{b:B}}{a}{A}$ and a process $\jtypem{\Gamma}{\Lambda}{P}{b}{B}$, the result of ``plugging'' $P$ into the hole is the term $\jtypem{\Psi}{\Delta}{\ctxh[P]{C}{\Gamma;\Lambda}{b:B}}{a}{A}$ obtained by replacing the axiom \getrn{Ct-p-hole} by the derivation of $\jtypem{\Gamma}{\Lambda}{P}{b}{B}$.
\end{definition}

We most often work only with closed processes, and listing empty functional contexts becomes tiresome.
Consequently, we write $\ctxh{C}{\Lambda}{b : B}$ for $\ctxh{C}{{\cdot}\,;\,\Lambda}{b : B}$.

\begin{definition}
  \label{def:sill-obs-equiv/bisimulations:2}
  A type-indexed relation is \defin{contextual}\varindex{context contextual relation}{1!23 3!2~}[|defin]  if it is closed under contexts.
  Explicitly:
  \begin{enumerate}
  \item A type-indexed relation $\relfnt{R}$ on configurations is contextual if $\jtrelc{\relfnt{R}}{\Lambda}{\mc{C}}{\mc{D}}{\Xi}$ implies that\\
    \( \jtrelc{\relfnt{R}}{\Gamma}{\ctxh[\mc{C}]{\mc{E}}{\Lambda}{\Xi}}{\ctxh[\mc{D}]{\mc{E}}{\Lambda}{\Xi}}{\Delta} \) for all $\jcfgt{\Gamma}{\ctxh{\mc{E}}{\Lambda}{\Xi}}{\Delta}$.
  \item A type-indexed relation $\relfnt{R}$ on processes is contextual if $\jtrelp{\relfnt{R}}{\Psi}{\Delta}{P}{Q}{c}{A}$ implies that\\
    \( {\jtrelp{\relfnt{R}}{\Psi'}{\Delta'}{\ctxh[P]{C}{\Psi; \Delta}{c:A}}{\ctxh[Q]{C}{\Psi; \Delta}{c:A}}{b}{B}} \)
    for all $\jtypem{\Psi'}{\Delta'}{\ctxh{C}{\Psi; \Delta}{c:A}}{b}{B}$.\qedhere
  \end{enumerate}
\end{definition}

\begin{definition}
  \label{def:sill-obs-equiv/relat-equiv:5}
  The \defin{contextual interior}\varindex{relation contextual interior context}{1!23 4!23}[|defin] $\mkcg{\relfnt{R}}$\glsadd{mkcg} of a type-indexed relation $\relfnt{R}$ is the greatest contextual type-indexed relation contained in $\relfnt{R}$.
\end{definition}

\begin{lemma}
  \label{lemma:sill-obs-equiv/observ-comm-equiv:2}
  Taking the contextual interior of a relation is a monotone operation, and it preserves arbitrary intersections.
\end{lemma}

Contextual preorders are called precongruences:

\begin{definition}
  \label{def:sill-obs-equiv/bisimulations:4}
  A typed relation $\relfnt{R}$ on configurations is a \defin{precongruence}\DIndex/{precongruence relation} if:
  \begin{enumerate}
  \item each relation in the family is a preorder; and
  \item the relation respects composition: if $\jtrelc{\relfnt{R}}{\Gamma}{\mc{C}}{\mc{C}'}{\Phi\Pi}$ and $\jtrelc{\relfnt{R}}{\Pi\Lambda}{\mc{D}}{\mc{D}'}{\Xi}$,\\
    then $\jtrelc{\relfnt{R}}{\Gamma\Lambda}{\mc{C},\mc{D}}{\mc{C}',\mc{D}'}{\Phi\Xi}$.
  \end{enumerate}
  It is a \defin{congruence}\DIndex/{congruence relation} if it is also an equivalence relation.\qedhere
\end{definition}

Congruence relations are desirable because they let us ``replace equals by equals''.
The following proposition is standard:

\begin{proposition}
  \label{prop:sill-obs-equiv/bisimulations:2}
  A typed equivalence relation $\relfnt{R}$ on configurations is a precongruence if and only if it is a contextual preorder.
\end{proposition}

\begin{proof}
  It is obvious that every precongruence is contextual.
  To show that every contextual preorder is a precongruence, assume $\jtrelc{\relfnt{R}}{\Gamma}{\mc{C}}{\mc{C}'}{\Phi\Pi}$ and $\jtrelc{\relfnt{R}}{\Pi\Lambda}{\mc{D}}{\mc{D}'}{\Xi}$.
  By contextuality, $\jtrelc{\relfnt{R}}{\Gamma\Lambda}{\mc{C},\mc{D}}{\mc{C}',\mc{D}}{\Phi\Xi}$ and $\jtrelc{\relfnt{R}}{\Gamma\Lambda}{\mc{C}',\mc{D}}{\mc{C}',\mc{D}'}{\Phi\Xi}$.
  By transitivity, $\jtrelc{\relfnt{R}}{\Gamma\Lambda}{\mc{C},\mc{D}}{\mc{C}',\mc{D}'}{\Phi\Xi}$.
\end{proof}

We can use contextual interiors and \cref{prop:sill-obs-equiv/bisimulations:2} to extract precongruences from preorders (\cf~\cite[Theorem~7.5]{milner_1980:_calcul_commun_system}).
\Cref{prop:sill-obs-equiv/bisimulations:2} and the definitions of precongruence and congruence translate from configurations to processes in the obvious way.

\section{Observed Communication Semantics}
\label{sec:sill-obs-equiv:observ-comm}

A longstanding idea in concurrency theory is that processes can only interact with their environments through communication, and that we can only observe systems by communicating with them.
Indeed, as far back as 1980, \Textcite[2]{milner_1980:_calcul_commun_system} wrote ``we suppose that the only way to observe a [concurrent] system is to communicate with it''.

In this \namecref{sec:sill-obs-equiv:observ-comm}, we make the above intuitions mathematically rigorous by giving Polarized SILL an \textit{observed communication semantics}.
Observed communication semantics, introduced by \Textcite{atkey_2017:_obser_commun_seman_class_proces}, define the meaning of a process to be the communications observed on its channels.

In \cref{sec:sill-obs-equiv:observ-comm:sess-typed-comm}, we make the notion of a session-typed communication explicit.
We endow session-typed communications with a notion of approximation.
This approximation will be used later to relate various notions of equivalence.
We also characterize infinite communications by their finite approximations.

In \cref{sec:sill-obs-equiv:observ-comm:comm-single-chan} we show how to observe communications on free channels of configurations.
We first do so using a coinductively defined judgment that observes communications on fair executions.
We show that the choice of fair trace does not matter.
We also show that we can instead consider only finite prefixes of fair executions.
The communications observed on the finite prefixes approximate and determine those observed on the complete fair execution.
We generalize from single channels to collections of channels in \cref{sec:sill-obs-equiv:observ-comm:comm-mult-chans}.

We view these three sections as the three steps of a recipe for defining observed communication semantics for session-typed languages.
The notion of session-typed communication given in \cref{sec:sill-obs-equiv:observ-comm:sess-typed-comm} is language independent: it describes what it means to be a session-typed communication \emph{in general}.
The constructions of \cref{sec:sill-obs-equiv:observ-comm:comm-single-chan,sec:sill-obs-equiv:observ-comm:comm-mult-chans} are necessarily language-dependent: they specify how to observe (language-independent) communications in (language-dependent) process executions.
To define an observed communication semantics for a different session-typed language then consists of extending \cref{sec:sill-obs-equiv:observ-comm:sess-typed-comm} to handle any new kinds of session-typed communications, followed by showing how to map process executions to session-typed communications.

\subsection{Session-Typed Communications \textit{qua} Communications}
\label{sec:sill-obs-equiv:observ-comm:sess-typed-comm}

We begin by defining session-typed communications.
Let a communication $v$ be a (potentially infinite) tree generated by the following grammar, where $k$ ranges over labels and $f$ ranges over functional values such that $\jtypef{\cdot}{f}{\tau}$ for some $\tau$.
We explain these communications $v$ below when we associate them with session types.
\begin{align*}
  v,v' &\Coloneqq \bot & \text{empty communication}\\
       &\mathrel{\mid} \cclose & \text{close message}\\
       &\mathrel{\mid} (\cunfold, v) &\text{unfolding message}\\
       &\mathrel{\mid} (k, v) & \text{choice message}\\
       &\mathrel{\mid} (v, v') & \text{channel message}\\
       &\mathrel{\mid} (\cshift, v) &\text{shift message}\\
       &\mathrel{\mid} (\cval{f}, v) &\text{functional value message}
\end{align*}

The judgment $\jsynt{v}{A}$\glsadd{jsynt} means that the syntactic communication $v$ has closed type $A$.
It is coinductively defined by the following rules:
\begin{gather*}
  \getrule{C-bot}
  \qquad
  \getrule{C-tu}
  \qquad
  \getrule{C-rho}
  \\
  \getrule{C-tplus}
  \qquad
  \getrule{C-tamp}
  \\
  \getrule{C-tot}
  \qquad
  \getrule{C-tlolly}
  \\
  \getrule{C-tds}
  \qquad
  \getrule{C-tus}
  \\
  \getrule{C-tand}
  \qquad
  \getrule{C-timp}
\end{gather*}
Every closed session type $A$ has an empty communication $\bot$ representing the absence of communication of that type.
The communication $\cclose$ represents the close message.
A communication of type $\Tplus \{ l : A_l \}_{l \in L}$ or $\Tamp \{ l : A_l \}_{l \in L}$ is a label $k \in L$ followed by a communication $v_k$ of type $A_k$, whence the communication $(k, v_k)$.
Though by itself the communication $(k, v_k)$ does not capture the direction in which the label $k$ travelled, this poses no problem to our development: we almost never consider communications without an associated session type, and the polarity of the type specifies the direction in which $k$ travels.
We cannot directly observe channels, but we can observe communications over channels.
Consequently, we observe a communication of type $A \Tot B$ or $A \Tlolly B$ as a pair $(v, v')$ of communications $v$ of type $A$ and $v'$ of type $B$.
This is analogous to the semantics of $A \Tot B$ in the ``folklore'' relational semantics of classical linear logic proofs~\cite{atkey_2017:_obser_commun_seman_class_proces,barr_1991:_auton_categ_linear_logic}.
A communication of type $\Trec{\alpha}{A}$ is an unfold message followed by a communication of type $[\Trec{\alpha}{A}/\alpha]A$.
A communication of type $\Tand{\tau}{A}$ or $\Timp{\tau}{A}$ is a value $f$ of type $\tau$ followed by a communication of type $A$.

In \cref{cha:sill-obs-equiv}, we will relate processes based on their observed communications.
To do so, we must first be able to relate communications.
We expect relations on communications to be ``type-indexed'':

\begin{definition}
  \label{def:sill-obs-equiv/observ-comm:10}
  A \defin{type-indexed relation $\relfnt{R}$ on communications} is a family of relations $(\relfnt{R}_{A})_A$ indexed by session types $A$, where $(v, w) \in \relfnt{R}_A$ only if $\jsynt{v}{A}$ and $\jsynt{w}{A}$.
  In this case, we write $\jsynt{v \mathrel{\relfnt{R}} w}{A}$\glsadd{jtrelcomms}.
\end{definition}

Given some relation $\leqslant$ on terms, we can endow session-typed communications with a notion of simulation $\commsim[\leqslant]$.
Intuitively, $u \commsim[\leqslant] w$\glsadd{commsim} means that $u$ approximates $w$, or that $w$ carries at least as much information as $u$.
Functional values aside, it suggests that $u$ is a potentially incomplete version of $w$.
In this regard, it is analogous to the ordering on domains of lazy natural numbers~\cite{freyd_1990:_recur_types_reduc, escardo_1993:_lazy_natur_number}.
Though we intend for $\leqslant$ to be a preorder, it is not required to be one.
This relaxation is for purely technical reasons: it simplifies the task of relating $\commsim[\leqslant]$ to other relations on communications.

\begin{definition}
  \label{def:sill-obs-equiv/observ-comm:7}
  Let $\leqslant$ be a type-indexed relation on terms.
  \defin{Communication simulation $\commsim$ modulo $\leqslant$}\glsadd{commsim} is the largest type-indexed family $\commsim[\leqslant]$ of relations $(\commsim_A)_A$ on session-typed communications defined by the following rules.
  When $\leqslant$ is clear from context, we write $\commsim$ for $\commsim[\leqslant]$.
  \begin{gather*}
    \getrule{CS-bot}
    \quad
    \getrule{CS-tu}
    \quad
    \getrule{CS-rho}
    \\
    \getrule{CS-tplus}
    \qquad
    \getrule{CS-tamp}
    \\
    \getrule{CS-tot}
    \qquad
    \getrule{CS-tlolly}
    \\
    \getrule{CS-tds}
    \qquad
    \getrule{CS-tus}
    \\
    \getrule{CS-tand}
    \qquad
    \getrule{CS-timp}\qedhere
  \end{gather*}
\end{definition}

\begin{remark}
  We do not ask for $\commsim$ to be a partial order.
  This is because antisymmetry forces communication equivalence (defined in \cref{def:sill-obs-equiv/observ-comm:11} below as the intersection of \(\commsim\) and \(\opr{\commsim}\)) for communications of type $\Tand{\tau}{A}$ to hold only when the transmitted values are equal on the nose.
  This is too fine of an equivalence: we would like to allow communications of type $\Tand{\tau}{A}$ to be ``equivalent'' whenever the values of type $\tau$ are in some sense ``equivalent'', without insisting that that equivalence be syntactic equality.
\end{remark}

\begin{proposition}
  \label{prop:sill-obs-equiv/observ-comm-equiv:3}
  The function $\commsim[({-})]$ is monotone, $\omega$-continuous, and $\omega$-cocontinuous.
  The relation $\commsim[\leqslant]$ is respectively reflexive or transitive whenever $\leqslant$ is reflexive or transitive.
  It is a type-indexed relation.
\end{proposition}

\begin{proof}
  We begin by showing that the function is well-defined.
  Let $\mc{R}$ be the complete lattice of all type-indexed relations on session-typed communications, and let $\mc{F}$ be the complete lattice of all type-indexed relations on functional terms.
  For each $\relfnt{F} \in \mc{F}$, the above rules define a rule functional $\Phi(\relfnt{F}, {-}) : \mc{R} \to \mc{R}$.
  It is $\omega$-cocontinuous by \cite[Theorem~2.9.4]{sangiorgi_2012:_introd_bisim_coind}.
  It extends to a monotone function $\Phi : \mc{F} \times \mc{R} \to \mc{R}$.
  We observe that $\commsim[{(-)}]$ is given by $\opc{\left(\sfix{\left(\opc{\Phi}\right)}\right)} : \mc{F} \to \mc{R}$.
  Indeed, the greatest fixed point of $\Phi(\relfnt{F}, {-})$ is the initial fixed point $\sfix{\left(\opc{\Phi}\right)}(\relfnt{F})$ of $\opc{\Phi}(\relfnt{F}, {-}) : \opc{\mc{R}} \to \opc{\mc{R}}$, where $\opc{\Phi} : \opc{\mc{F}} \times \opc{\mc{R}} \to \opc{\mc{R}}$, and $\sfix{\left(\opc{\Phi}\right)} : \opc{\mc{F}} \to \opc{\mc{R}}$ is given by \cite[Proposition~4.3.1]{kavanagh_2021:_commun_based_seman}.
  By the same proposition, $\commsim[{(-)}]$ is monotone, $\omega$-continuous, and $\omega$-cocontinuous.

  We use the coinduction proof principle to show that $\commsim$ is reflexive.
  Let $\Delta$ be the identity relation on session-typed communications.
  A case analysis on the rules shows that $\Delta \subseteq \Phi({\leqslant}, \Delta)$.
  Because $\commsim[\leqslant]$ is the greatest post-fixed point of $\Phi({\leqslant}, {-})$, we conclude that it contains $\Delta$, \ie, that it is reflexive.

  Assume now that $\leqslant$ is a preorder.
  We use the same technique to show that $\commsim[\leqslant]$ is a preorder.
  Now let $\reltransc{\commsim}$ be the transitive closure of $\commsim[\leqslant]$.
  Recall that the transitive closure $\reltransc{\relfnt{R}}$ of a relation $\relfnt{R}$ can be calculated by
  \[
    {\reltransc{\relfnt{R}}} = \bigcup_{n = 1}^\infty {\relfnt{R}^n},
  \]
  where $\relfnt{R}^n$ is the $n$-fold composition of $\relfnt{R}$ with itself.
  The functional $\Phi(\leqslant, {-})$ is $\omega$-continuous by \cite[Exercise~2.9.2]{sangiorgi_2012:_introd_bisim_coind}.
  In particular, this implies that
  \[
    \Phi({\leqslant}, {\reltransc{\relfnt{R}}}) = \bigcup_{n = 1}^\infty \Phi({\leqslant}, {\relfnt{R}^n}).
  \]
  Thus, to show that ${\reltransc{\relfnt{R}}} \subseteq \Phi({\leqslant}, {\reltransc{\relfnt{R}}})$, it is sufficient to show that ${\relfnt{R}^n} \subseteq \Phi({\leqslant}, {\relfnt{R}^n})$ for all $n$.
  Recall that $\commsim$ is the greatest post-fixed point of $\Phi({\leqslant}, {\relfnt{R}^n})$.
  This means that to show that $\commsim$ is transitive, \ie, ${\reltransc{\commsim}} \subseteq {\commsim}$, it is sufficient to show that ${\reltransc{\commsim}} \subseteq \Phi({\leqslant}, {\reltransc{\commsim}})$.
  We proceed by case analysis on $n$ to show that ${\commsim^n} \subseteq \Phi({\leqslant}, {\commsim^n})$.
  The case $n = 1$ is immediate by definition of $\commsim$ as the greatest fixed point of $\Phi({\leqslant}, {-})$.
  We now show the case $n = m + 1$.
  Assume that $\jsynt{u \commsim^n w}{A}$ because $\jsynt{u \commsim v}{A}$ and $\jsynt{v \commsim^m w}{A}$.
  We show that $\jsynt{u \mathrel{\Phi({\leqslant}, {\commsim^n})} w}{A}$.
  We proceed by case analysis on the rule that formed $\jsynt{u \commsim v}{A}$, giving several illustrative cases:
  \begin{proofcases}
  \item[\getrn{CS-bot}] Then $u = \bot$, and $\jsynt{w}{A}$ because we assumed that $\commsim$ was a type-indexed relation on session-typed communications.
    So $\jsynt{u \mathrel{\Phi({\leqslant}, {\commsim^n})} w}{A}$ thanks to $\getrn{CS-bot}$.
  \item[\getrn{CS-tplus}] Then $A = \Tplus \{ l : A_l \}_{l \in L}$, $u = (k, u')$, and $v = (k, v')$ for some $u'$ and $v'$.
    By hypothesis, $\jsynt{u' \commsim v'}{A_k}$.
    And induction on $m$ reveals that $w = (k, w')$ for some $\jsynt{v' \mathrel{\commsim^m} w'}{A_k}$.
    So $\jsynt{u' \mathrel{\commsim^n} w'}{A_k}$.
    By \getrn{CS-tplus}, we then get $\jsynt{u \mathrel{\Phi({\leqslant}, {\commsim^n})} w}{A}$ as desired.
  \item[\getrn{CS-tand}] Then $A = \Tand{\tau}{B}$, $u = (\cval f, u')$, and $v = (\cval g, v')$ for some $f$, $g$, $u'$, and $v'$.
    By hypothesis, $\jsynt{u' \commsim v'}{B}$ and $\jtrelf{\leqslant}{\cdot}{f}{g}{\tau}$.
    And induction on $m$ reveals that $w = (\cval h, w')$ for some $\jsynt{v' \mathrel{\commsim^m} w'}{B}$ and $\jtrelf{\leqslant^m}{\cdot}{g}{h}{\tau}$.
    So $\jtrelf{\leqslant}{\cdot}{f}{h}{\tau}$ and $\jsynt{u' \mathrel{\commsim^n} w'}{B}$.
    By \getrn{CS-tand}, we then get $\jsynt{u \mathrel{\Phi({\leqslant}, {\commsim^n})} w}{A}$ as desired.\qedhere
  \end{proofcases}
\end{proof}

\begin{definition}
  \label{def:sill-obs-equiv/observ-comm:11}
  Let $\equiv$ be a type-indexed relation on terms.
  \defin{Communication equivalence $\commeq$ modulo $\equiv$}, written $\commeq[\equiv]$\glsadd{commeq}, is given by $\jsynt{v \commeq[\equiv] w}{A}$ if and only if both $\jsynt{v \commsim[\equiv] w}{A}$ and $\jsynt{w \commsim[\equiv] v}{A}$.
  When $\equiv$ is clear from context, we write $\commeq$ for $\commeq[\equiv]$.
\end{definition}

\begin{proposition}
  \label{prop:sill-obs-equiv/observ-comm:1}
  Communication equivalence $\commeq[\equiv]$ is a type-indexed relation.
  It is an equivalence relation whenever $\equiv$ is a preorder.
\end{proposition}

\begin{proof}
  It follows from \cref{prop:sill-obs-equiv/observ-comm-equiv:3} that it is type-indexed.
  Assume now that $\equiv$ is a preorder.
  By \cref{prop:sill-obs-equiv/observ-comm-equiv:3}, $\commsim[\equiv]$ is a preorder.
  Then by definition, $\commeq[\equiv]$ is the intersection of a preorder and its opposite.
  But in general ${<} \cap {\opr{<}}$ is an equivalence relation whenever $<$ is a preorder.
  We conclude that $\commeq[\equiv]$ is an equivalence relation.
\end{proof}

Communication equivalence modulo~$=$ holds if and only if two communications are equal on the nose:

\begin{proposition}
  \label{lemma:sill-obs-equiv/observ-comm:3}
  For all $A$, $\jsynt{u \commeq[=] v}{A}$ if and only if $u = v$.
\end{proposition}

\begin{proof}
  Necessity is immediate by reflexivity of $\commeq[=]$.
  Sufficiency comes from recognizing $\commeq[=]$ as the notion of bisimulation given by the coinductive definition of $\jsynt{w}{A}$, and that by \cite[Theorem~2.7.2]{jacobs_rutten_2012:_introd_coalg_coind}, bisimilar elements of the terminal coalgebra are equal.
\end{proof}

\begin{proposition}
  \label{prop:sill-obs-equiv/observ-comm-equiv:13}
  ``Communication simulation modulo'' and ``communication equivalence modulo'' are related by the identity $\left({\commsim[\leqslant]} \cap {\opr{(\commsim[\leqslant])}}\right) = \left({\commeq[({\leqslant} \cap {\opr{\leqslant}})]}\right)$.
\end{proposition}

\begin{proof}
  Let $\Phi$ be the functional defining $\commsim[({-})]$, and set $I = {\leqslant} \cap {\opr{\leqslant}}$.
  Observe for all relations $\relfnt{X}$, $\relfnt{Y}$, $\relfnt{Z}$ that $\Phi(\relfnt{X} \cap \relfnt{Y}, \relfnt{Z}) = \Phi(\relfnt{X}, \relfnt{Z}) \cap \Phi(\relfnt{Y}, \relfnt{Z})$.
  We compute, where we use the syntax $\nu X. F(X)$ for the greatest fixed point of $F$, that:
  \begin{align*}
    & \commeq[I]\\
    &= \left({\commsim[I]}\right) \cap \opr{\left(\commsim[I]\right)}\\
    &= \left(\nu \relfnt{V}.\Phi(I, \relfnt{V})\right) \cap \opr{\left( \nu \relfnt{V}. \Phi(I, \relfnt{V}) \right)}\\
    &= \left(\nu \relfnt{V}.\Phi(I, \relfnt{V})\right) \cap  \nu \relfnt{V}. \opr{\left(\Phi(I, \relfnt{V}) \right)}\\
    &= \nu \relfnt{V}.\Phi({\leqslant}, \relfnt{V}) \cap \Phi({\opr\leqslant}, \relfnt{V}) \cap \opr{\left(\Phi({\leqslant}, \relfnt{V}) \right)} \cap \opr{\left(\Phi({\opr\leqslant}, \relfnt{V}) \right)},
      \shortintertext{and analogously,}
    & \left({\commsim[\leqslant]} \cap {\opr{(\commsim[\leqslant])}}\right)\\
    &= \nu \relfnt{V}. \Phi({\leqslant}, \relfnt{V}) \cap \opr{\left(\Phi({\leqslant}, \relfnt{V}) \right)}.
  \end{align*}
  To show  that the fixed points are equal, it is sufficient to show that they have the same post-fixed points.
  Set
  \begin{align*}
    L(\relfnt{V}) &= \Phi({\leqslant}, \relfnt{V}) \cap \Phi({\opr\leqslant}, \relfnt{V}) \cap \opr{\left(\Phi({\leqslant}, \relfnt{V}) \right)} \cap \opr{\left(\Phi({\opr\leqslant}, \relfnt{V}) \right)},\\
    R(\relfnt{V}) &= \Phi({\leqslant}, \relfnt{V}) \cap \opr{\left(\Phi({\leqslant}, \relfnt{V}) \right)}.
  \end{align*}
  Clearly every post-fixed point of $L$ is a post-fixed point of $R$: $L(\relfnt{V}) \subseteq R(\relfnt{V})$ for all $\relfnt{V}$.
  Conversely, assume that $\relfnt{V}$ is a post-fixed point of $R$, \ie, $\relfnt{V} \subseteq R(\relfnt{V})$.
  We show that $\relfnt{V}$ is a post-fixed point of $L$ by showing that $R(\relfnt{V}) \subseteq L(\relfnt{V})$.
  Assume that $\jsynt{v \mathrel{R(\relfnt{V})} w}{A}$.
  Then $\jsynt{v \mathrel{\Phi({\leqslant}, \relfnt{V})} w}{A}$.
  A case analysis then reveals that $\jsynt{v \mathrel{\opr{\left(\Phi({\leqslant}, \relfnt{V})\right)}} w}{A}$ by the same rule.
  The result follows by case analysis on this rule.
  We give a few illustrative cases.
  \begin{proofcases}
  \item[\getrn{CS-bot}] Then $v = w = \bot$, and a straightforward check gives $\jsynt{\bot \mathrel{L(\relfnt{V})} \bot}{A}$.
  \item[\getrn{CS-tplus}] Then $v = (k, v')$, $w = (k, w')$, $\jsynt{v' \mathrel{\relfnt{V}} w'}{A_k}$, and $\jsynt{w' \mathrel{\relfnt{V}} v'}{A_k}$.
    A straightforward check again gives $\jsynt{v \mathrel{L(\relfnt{V})} w}{A}$.
  \item[\getrn{CS-tand}] Then $A = \Tand{\tau}{B}$, $v = (\cval f, v')$, $w = (\cval g, w')$, $f \leqslant g$, $g \leqslant f$, $\jsynt{v' \mathrel{\relfnt{V}} w'}{B}$, and $\jsynt{w' \mathrel{\relfnt{V}} v'}{B}$.
    A straightforward check gives that $\jsynt{v \mathrel{L(\relfnt{V})} w}{A}$.\qedhere
  \end{proofcases}
\end{proof}

We now show that communications are uniquely determined by their finite approximations.
This opens the door to reasoning about $\commsim[\leqslant]$ using inductive techniques.
We will also use sequences of finite approximations to construct infinite observed communications ``in the limit'' in \cref{sec:sill-obs-equiv:observ-comm:comm-single-chan}.
We can think of height \(n\) approximations as truncating communications \(w\) (viewed as trees) at height \(n\):

\begin{definition}
  \label{def:sill-obs-equiv/barb-cont-congr:5}
  The \defin{height $n$ approximation} $\capxn{w}{n}$\glsadd{capxn} of a communication $w$ is defined by induction on $n$ and recursion on $w$:
  \begin{align*}
    \capxn{w}{0} &= \bot & \capxn{\bot}{n + 1} &= \bot\\
    \capxn{\cclose}{n + 1} &= \cclose & \capxn{(\cval f, v)}{n + 1} &= (\cval f, \capxn{v}{n})\\
    \capxn{(k, v)}{n + 1} &= (k, \capxn{v}{n}) & \capxn{(u, v)}{n + 1} &= (\capxn{u}{n}, \capxn{v}{n})\\
    \capxn{(\cshift, v)}{n + 1} &= (\cshift, \capxn{v}{n}) & \capxn{(\cunfold, v)}{n + 1} &= (\cunfold, \capxn{v}{n})\qedhere
  \end{align*}
\end{definition}

This approximation does not affect the typing:

\begin{proposition}
  \label{prop:sill-obs-equiv/barb-cont-congr:10}
  If $\jsynt{w}{A}$, then $\jsynt{\capxn{w}{n}}{A}$ for all $n$.
\end{proposition}

\begin{proof}
  By induction on $n$.
  The base case is immediate.
  The inductive step follows by a case analysis on the rule used to form $\jsynt{w}{A}$.
\end{proof}

The height \(n\) approximations of a communication form an ascending chain:

\begin{proposition}
  \label{prop:sill-obs-equiv/barb-cont-congr:12}
  For all reflexive $\leqslant$, all $n$, and all $\jsynt{w}{A}$, $\jsynt{\capxn{w}{n} \commsim[\leqslant] \capxn{w}{n + 1}}{A}$.
\end{proposition}

\begin{proof}
  By induction on $n$.
  The base case is given by \getrn{CS-bot}.
  The inductive step is given by case analysis on $\jsynt{w}{A}$.
  Reflexivity of $\leqslant$ is required for the cases $A = \Tand{\tau}{B}$ and $A = \Timp{\tau}{B}$.
\end{proof}

Height \(n\) approximations also characterize communication simulation:

\begin{proposition}
  \label{prop:sill-obs-equiv/barb-cont-congr:11}
  For all $\jsynt{w}{A}$ and $\jsynt{u}{A}$, $\jsynt{u \commsim[\leqslant] w}{A}$ if and only if, for all $n$, ${\jsynt{\capxn{u}{n} \commsim[\leqslant] w}{A}}$.
\end{proposition}

\begin{proof}
  We proceed by induction on $n$ to show that for all $n \in \N$ and for all $\jsynt{w}{A}$ and $\jsynt{u}{A}$, $\jsynt{u \commsim[\leqslant] w}{A}$ implies $\jsynt{\capxn{u}{n} \commsim[\leqslant] w}{A}$.
  The base case is immediate by \getrn{CS-bot}.
  Assume the result for some $n$.
  We show that $\jsynt{\capxn{u}{n + 1} \commsim[\leqslant] w}{A}$ by case analysis on the rule used to form $\jsynt{u \commsim[\leqslant] w}{A}$.
  We give two illustrative cases; the rest follow by analogy.
  \begin{proofcases}
  \item[\getrn{CS-bot}] Then $u = \bot$ and $\capxn{u}{n + 1} = \bot$.
    We are done by \getrn{CS-bot}.
  \item[\getrn{CS-tand}] Then $A = \Tand{\tau}{B}$, $u = (\cval f, u')$, and $w = (\cval g, w')$ with $\jtrelf{\leqslant}{\cdot}{f}{g}{\tau}$ and $\jsynt{u' \commsim[\leqslant] w'}{B}$.
    By definition, $\capxn{u}{n + 1} = (\cval f, \capxn{u'}{n})$.
    By the induction hypothesis, $\jsynt{\capxn{u'}{n} \commsim[\leqslant] w'}{B}$.
    By \getrn{CS-tand}, $\jsynt{(\cval f, \capxn{u'}{n}) \commsim[\leqslant] (\cval g, w')}{A}$ as desired.
  \end{proofcases}

  To show the converse, let $T$ be the set of triples $\{ (u, w, A) \mid \forall n \in \N \;.\; \jsynt{\capxn{u}{n} \commsim[\leqslant] w}{A} \}$.
  We want to show that if $(u, w, A) \in T$, then $\jsynt{u \commsim[\leqslant] w}{A}$.
  By the coinduction proof principle~\cite[49]{sangiorgi_2012:_introd_bisim_coind}, it is sufficient to show that $T$ is ``closed backwards'' under the rules defining $\commsim[\leqslant]$.
  Let $(u, v, A) \in T$ be arbitrary.
  We proceed by case analysis on $u$ and $A$ to show that there is a rule whose conclusion is $(u, v, A)$ and whose premises are in $T$.
  If $u = \bot$, then we are done by \getrn{CS-bot}, so assume that $u \neq \bot$.
  We proceed by case analysis on $A$.
  We show two cases; the rest follow by analogy.
  \begin{proofcases}
  \item[$A = \Tu$]
    The only possible value for $u$ and $v$ is $u = v = \cclose$.
    So $(u, w, \Tu) \in T$ by \getrn{CS-tu}.
  \item[$A = \Tand{\tau}{B}$] Then $u = (\cval f, u')$ for some value $\jtypef{\cdot}{f}{\tau}$ and some $\jsynt{u'}{B}$, and $v = (\cval g, v')$ for some value $\jtypef{\cdot}{g}{\tau}$ and some $\jsynt{v'}{B}$.
    By assumption, $\jsynt{\capxn{(\cval f, u')}{n} \commsim[\leqslant] (\cval g, v')}{A}$ for all $n$.
    By inversion, for all $n = m + 1 \geq 1$, the last rule in the derivation must have been \getrn{CS-tand} with $\jsynt{\capxn{u'}{m} \commsim[\leqslant] v'}{B}$ and with its side condition $\jtrelf{\leqslant}{\cdot}{f}{g}{\tau}$ satisfied.
    So $(u', w', B) \in T$.
    Then $(u, w, A) \in T$ by \getrn{CS-tand} with the premise $(u', w', B) \in T$ and the side condition $\jtrelf{\leqslant}{\cdot}{f}{g}{\tau}$.\qedhere
  \end{proofcases}
\end{proof}

\begin{definition}
  \label{def:sill-obs-equiv/observ-comm:12}
  Where $\leqslant$ is a type-indexed preorder on functional values, let $\commbr[\leqslant]{A}$\glsadd{commbr} be the set of communications $\jsynt{w}{A}$ ordered by the preorder $\commsim[\leqslant]$.
\end{definition}

\begin{corollary}
  \label{cor:sill-obs-equiv/barb-cont-congr:1}
  If $\leqslant$ is a preorder, then for all $\jsynt{w}{A}$, $w$ is a\footnote{In contrast to least upper bounds in partial orders, least upper bounds in preorders are not necessarily unique.} least upper bound of $\left(\capxn{w}{n}\right)_{n \in \N}$ in the preorder $\commbr[\leqslant]{A}$.
\end{corollary}

\subsection{Session-Typed Communications on Single Channels}
\label{sec:sill-obs-equiv:observ-comm:comm-single-chan}

In this section, we show how to observe session-typed communications $\jsynt{v}{A}$ on a single channel $c$ in a trace $T$.
We capture these observations using a coinductively defined judgment $\jtoc{T}{v}{c}{A}$\glsadd{jtoc}.
This judgment defines a total function from free channel names in $T$ to session-typed communications $\jsynt{v}{A}$.
We show that the type of the observed communications agrees with the statically determined type of the channel, \ie, that $\jtoc{T}{v}{c}{A}$ implies $\jttp{T}{c}{A}$.\footnote{Recall from \cref{def:sill-background-properties-traces:1} that every free channel in \(T\) has a type, and that this type is given by the type-preservation theorem.}
We will also show that the communication observed on $c$ is \textit{independent} of the choice of execution $T$, provided that $T$ is a fair execution.
Finally, we show that the communications observed from finite prefixes of $T$ both approximate and determine the observations on the entirety of $T$.

Given a trace $T = (\mc{C}_0, (r_i;(\theta_i,\xi_i))_{i})$, we write $\mc{T}$ for the support of $T$, that is, $x \in \mc{T}$ if and only if $x \in \mc{C}_i$ for some $i$.
The judgment $\jtoc{T}{v}{c}{A}$ is coinductively defined by the following rules, \ie, it is the largest set of triples $(v,c,A)$ closed under the rules \getrn{O-bot} through \getrn{O-timp} below.

We observe no communications on a channel $c$ if and only if $c$ carried no message.
Subject to the side condition that $c \neq \carrcn(\jmsg{d}{m})$ for all $\jmsg{d}{m} \in \mc{T}$, we have the rule
\[
  \getrule{O-bot} \quad \text{\raisebox{0.6em}{whenever \(\forall \jmsg{d}{m} \in \mc{T} .\; c \neq \carrcn(\jmsg{d}{m})\).}}
\]
We observe a close message on $c$ if and only if the close message was sent on $c$:
\[
  \getrule{O-tu}
\]
We observe label transmission as labelling communications on the continuation channel.
We rely on the judgment $\jttp{T}{c}{\Tplus \{ l : A_l \}_{l \in L}}$ or $\jttp{T}{c}{\Tamp \{ l : A_l \}_{l \in L}}$ to determine the type of $c$:
\begin{gather*}
  \getrule{O-tplus}
  \\
  \getrule{O-tamp}
\end{gather*}
As described above, we observe channel transmission as pairing of communications:
\begin{gather*}
  \getrule{O-tot}
  \\
  \getrule{O-tlolly}
\end{gather*}
We observe the unfold and shift messages directly:
\begin{gather*}
  \getrule{O-rhop}
  \\
  \getrule{O-rhon}
  \\
  \getrule{O-tds}
  \\
  \getrule{O-tus}
\end{gather*}
Finally, we observe functional values:
\begin{gather*}
  \getrule{O-tand}
  \\
  \getrule{O-timp}
\end{gather*}

We set out to show that for any process trace $T$, the judgment $\jtoc{T}{v}{c}{A}$ defines a total function from channel names $c$ in $T$ to session-typed communications $\jsynt{v}{A}$.

We begin by showing that if $\jtoc{T}{v}{c}{A}$, then this session-typed communication $\jsynt{v}{A}$ is unique.
We use a bisimulation approach and follow standard techniques to define bisimulations for $\jtoc{T}{v}{c}{A}$.
We interpret the premises the rules defining $\jtoc{T}{v}{c}{A}$ that are not of the form $\jtoc{T}{w}{d}{B}$ as side conditions, giving an instance of the rule for each such premise.
For example, the rule \getrn{O-tplus} should be seen as a family of rules \rn{O-$\Tplus$-$c$-$d$-$l$}, where we have a rule
\[
  \infer[\rn{O-$\Tplus$-$c$-$d$-$l$}]{
    \jtoc{T}{(l,v)}{c}{\Tplus \{ l : A_l \}_{l \in L}}
  }{
    \jtoc{T}{v}{d}{A_l}
  }
\]
for each $\jmsg{c}{\mSendLP{c}{l}{d}} \in \mc{T}$ such that $\jttp{T}{c}{\Tplus \{ l : A_l \}_{l \in L}}$.
A symmetric binary relation $\relfnt{R}$ on observed communications in $T$ is a bisimulation if:
\begin{itemize}
\item if $(\jtoc{T}{\bot}{c}{A},\jtoc{T}{w}{c}{A'}) \in \relfnt{R}$ and $\jtoc{T}{\bot}{c}{A}$ by the instance of \getrn{O-bot} for $\jttp{T}{c}{A}$, then $w = \bot$ and $A' = A$;
\item if $(\jtoc{T}{\cclose}{c}{\Tu},\jtoc{T}{w}{c}{A}) \in \relfnt{R}$, then $w = \cclose$ and $A = \Tu$;
\item if $(\jtoc{T}{(l,v)}{c}{\Tplus \{ l : A_l \}_{l \in L}},\jtoc{T}{w}{c}{A}) \in \relfnt{R}$ and $\jtoc{T}{(l,v)}{c}{\Tplus \{ l : A_l \}_{l \in L}}$ by the instance of \getrn{O-tplus} for $\jmsg{c}{\mSendLP{c}{l}{d}} \in \mc{T}$, then $w = (l,v')$ for some $v'$, $A = \Tplus \{ l : A_l \}_{l \in L}$, and $(\jtoc{T}{v}{d}{A_l}, \jtoc{T}{w}{d}{A_l}) \in \relfnt{R}$;
\item if $(\jtoc{T}{(l,v)}{c}{\Tamp \{ l : A_l \}_{l \in L}},\jtoc{T}{w}{c}{A}) \in \relfnt{R}$ and $\jtoc{T}{(l,v)}{c}{\Tamp \{ l : A_l \}_{l \in L}}$ by the instance of \getrn{O-tamp} for $\jmsg{c}{\mSendLN{c}{l}{d}} \in \mc{T}$, then $w = (l, v')$ for some $v'$, $A = \Tamp \{ l : A_l \}_{l \in L}$, and $(\jtoc{T}{v}{d}{A_l}, \jtoc{T}{w}{d}{A_l}) \in \relfnt{R}$;
\item if $(\jtoc{T}{(u,v)}{c}{A \Tot B}, \jtoc{T}{w}{c}{C}) \in \relfnt{R}$ and $\jtoc{T}{(u,v)}{c}{A \Tot B}$ was formed by the instance of \getrn{O-tot} for $\jmsg{c}{\mSendCP{c}{a}{d}} \in \mc{T}$, then $v = (u', v')$ and $C = A' \Tot B'$ for some $u', v', A', B'$, and $(\jtoc{T}{u}{a}{A},\jtoc{T}{u'}{a}{A'}) \in \relfnt{R}$ and $(\jtoc{T}{v}{d}{A}, \jtoc{T}{v'}{d}{A'}) \in \relfnt{R}$;
\item if $(\jtoc{T}{(u,v)}{c}{A \Tlolly B}, \jtoc{T}{w}{c}{C}) \in \relfnt{R}$ and $\jtoc{T}{(u,v)}{c}{A \Tlolly B}$ was formed by the instance of \getrn{O-tlolly} for $\jmsg{c}{\mSendCN{c}{a}{d}} \in \mc{T}$, then $v = (u', v')$ and $C = A' \Tlolly B'$ for some $u', v', A', B'$, and $(\jtoc{T}{u}{a}{A},\jtoc{T}{u'}{a}{A'}) \in \relfnt{R}$ and $(\jtoc{T}{v}{d}{A}, \jtoc{T}{v'}{d}{A'}) \in \relfnt{R}$;
\item if $(\jtoc{T}{(\cunfold, v)}{c}{\Trec{\alpha}{A}},\jtoc{T}{w}{c}{B}) \in \relfnt{R}$ and $\jtoc{T}{(\cunfold, v)}{c}{\Trec{\alpha}{A}}$ was formed by the instance of \getrn{O-rhop} for $\jmsg{c}{\mSendUP{c}{d}} \in \mc{T}$, then $w = (\cunfold, v')$ for some $v'$ and $\Trec{\alpha'}{A'}$ such that $(\jtoc{T}{v}{d}{[\Trec{\alpha}{A}/\alpha]A},\jtoc{T}{v'}{d}{[\Trec{\alpha'}{A'}/\alpha']A'}) \in \relfnt{R}$;
\item if $(\jtoc{T}{(\cunfold, v)}{c}{\Trec{\alpha}{A}},\jtoc{T}{w}{c}{B}) \in \relfnt{R}$ and $\jtoc{T}{(\cunfold, v)}{c}{\Trec{\alpha}{A}}$ was formed by the instance of \getrn{O-rhon} for $\jmsg{c}{\mSendUN{c}{d}} \in \mc{T}$, then $w = (\cunfold, v')$ for some $v'$ and $\Trec{\alpha'}{A'}$ such that $(\jtoc{T}{v}{d}{[\Trec{\alpha}{A}/\alpha]A},\jtoc{T}{v'}{d}{[\Trec{\alpha'}{A'}/\alpha']A'}) \in \relfnt{R}$;
\item if $(\jtoc{T}{(\cshift, v)}{c}{\Tds{A}},\jtoc{T}{w}{c}{B}) \in \relfnt{R}$ and $\jtoc{T}{(\cshift, v)}{c}{\Tds{A}}$ was formed by the instance of \getrn{O-tds} for $\jmsg{c}{\mSendSP{c}{d}} \in \mc{T}$, then $w = (\cshift, v')$ and $B = \Tds{B'}$ for some $v'$ and $B'$ such that $(\jtoc{T}{v}{d}{A'},\jtoc{T}{v'}{d}{B'}) \in \relfnt{R}$;
\item if $(\jtoc{T}{(\cshift, v)}{c}{\Tds{A}},\jtoc{T}{w}{c}{B}) \in \relfnt{R}$ and $\jtoc{T}{(\cshift, v)}{c}{\Tus{A}}$ was formed by the instance of \getrn{O-tus} for $\jmsg{c}{\mSendSN{c}{d}} \in \mc{T}$, then $w = (\cshift, v')$ and $B = \Tus{B'}$ for some $v'$ and $B'$ such that $(\jtoc{T}{v}{d}{A'},\jtoc{T}{v'}{d}{B'}) \in \relfnt{R}$;
\item if $(\jtoc{T}{(\cval f, v)}{c}{\Tand{\tau}{A}},\jtoc{T}{w}{c}{B}) \in \relfnt{R}$ and $\jtoc{T}{(\cval f, v)}{c}{\Tand{\tau}{A}}$ was formed by the instance of \getrn{O-tand} for $\jmsg{c}{\mSendVP{c}{f}{d}} \in \mc{T}$, then $w = (\cval f, v')$ and $B = \Tand{\tau}{B'}$ for some $v'$ and $B'$ such that $(\jtoc{T}{v}{d}{A'},\jtoc{T}{v'}{d}{B'}) \in \relfnt{R}$;
\item if $(\jtoc{T}{(\cval f, v)}{c}{\Timp{\tau}{A}},\jtoc{T}{w}{c}{B}) \in \relfnt{R}$ and $\jtoc{T}{(\cval f, v)}{c}{\Timp{\tau}{A}}$ was formed by the instance of \getrn{O-timp} for $\jmsg{c}{\mSendVN{c}{f}{d}} \in \mc{T}$, then $w = (\cval f, v')$ and $B = \Timp{\tau}{B'}$ for some $v'$ and $B'$ such that $(\jtoc{T}{v}{d}{A'},\jtoc{T}{v'}{d}{B'}) \in \relfnt{R}$.
\end{itemize}

\begin{proposition}[Uniqueness]
  \label{prop:obssem:4}
  If $T$ is a trace from a well-typed configuration, then for all channels $c$, if $\jtoc{T}{v}{c}{A}$ and $\jtoc{T}{w}{c}{B}$, then $v = w$ and $A = B$.
  If $\jtoc{T}{v}{c}{A}$, then its derivation is unique.
\end{proposition}

\begin{proof}
  Fix some trace \(T\) and let $\relfnt{R}$ be the relation
  \[
    \relfnt{R} = \{ (\jtoc{T}{v}{c}{A}, \jtoc{T}{w}{c}{B}) \mid \exists v, w, c, A, B . \jtoc{T}{v}{c}{A} \land \jtoc{T}{w}{c}{B} \}.
  \]
  We show that it is a bisimulation.
  Let $(\jtoc{T}{v}{c}{A}, \jtoc{T}{w}{c}{B}) \in \relfnt{R}$ be arbitrary.
  It follows from \cref{cor:sill-background-dyn-prop-typed-config:1,cor:sill-background-properties-traces:1} that at most one rule is applicable to form a judgment of the form $\jtoc{T}{\cdot}{c}{\cdot}$ (with $c$ fixed), so $\jtoc{T}{v}{c}{A}$ and $\jtoc{T}{w}{c}{B}$ were both formed by the same rule.
  We proceed by case analysis on this rule.
  We only give a few illustrative cases; the rest will follow by analogy.
  \begin{proofcases}
  \item[\getrn{O-bot}] The conclusions are equal, so we are done.
  \item[\getrn{O-tot} for $\jmsg{c}{\mSendCP{c}{a}{d}}$] Then there exist $r,r',u,u',C,C',D,D'$ such that $v = (r, u)$, $w = (r', u')$, $A = C \Tot D$, $B = C' \Tot D'$, $\jtoc{T}{r}{a}{C}$, $\jtoc{T}{r'}{a}{C'}$, $\jtoc{T}{u}{d}{D}$, and $\jtoc{T}{u'}{d}{D'}$.
    But $(\jtoc{T}{r}{a}{C}, \jtoc{T}{r'}{a}{C'}) \in \relfnt{R}$ and $(\jtoc{T}{u}{d}{D}, \jtoc{T}{u'}{d}{D'}) \in \relfnt{R}$, so we are done.
  \item[\getrn{O-tand} for $\jmsg{c}{\mSendVP{c}{f}{d}}$] Then there exist $v', w', A', B'$ such that $v = (\cval f, v')$, $w = (\cval f, w')$, $A = \Tand{\tau}{A'}$, $B = \Tand{\tau}{B'}$, $\jtoc{T}{v'}{d}{A'}$, and $\jtoc{T}{w'}{d}{B'}$.
    But $(\jtoc{T}{v'}{d}{A'}, \jtoc{T}{w'}{d}{B'}) \in \relfnt{R}$, so we are done.
  \end{proofcases}
  It follows that $\relfnt{R}$ is a bisimulation.

  Consider arbitrary $\jtoc{T}{v}{c}{A}$ and $\jtoc{T}{w}{c}{B}$.
  They are related by $\relfnt{R}$, so they are bisimilar.
  By \cite[Theorem~2.7.2]{jacobs_rutten_2012:_introd_coalg_coind}, bisimilar elements of the terminal coalgebra are equal.
  It follows that $v = w$ and $A = B$.

  To see that the derivation of $\jtoc{T}{v}{c}{A}$ is unique, recall from above that at most one rule is applicable to form a judgment of the form $\jtoc{T}{\cdot}{c}{\cdot}$ (with $c$ fixed).
  Because each rule has only judgments of this form as its hypotheses, it follows that at each step in the derivation, exactly one rule instance can be applied to justify a given hypothesis.
  So the derivation is unique.
\end{proof}

Next, we show that an observed communication \(v\) exists for every channel \(c\) appearing in a trace.
This fact is not immediately obvious: we must construct this potentially infinite \(v\) and show that it satisfies \(\jtoc{T}{v}{c}{A}\).
We construct this potentially infinite \(v\) as a least upper bound of a sequence of height \(n\) approximations \(v_n\).
To do so, we introduce an auxiliary judgment \(\jtoc[n]{T}{v}{c}{A}\).
Its intended meaning is that \(v\) is the height \(n\) approximation of the communication on \(c\), \ie, if \(\jtoc{T}{w}{c}{A}\), then \(\jtoc[n]{T}{\capxn{w}{n}}{c}{A}\).
It is defined by the following \(n\)-indexed family of rules.
We split \getrn{O-bot} into two rules.
The first, \getrn{O-z-bot}, truncates an observation.
The second, \getrn{O-n-bot}, is the direct analog of \getrn{O-bot}.
\begin{gather*}
  \getrule{O-z-bot}\\
  \getrule{O-n-bot} \quad \text{\raisebox{0.6em}{whenever \(\forall \jmsg{d}{m} \in \mc{T} .\; c \neq \carrcn(\jmsg{d}{m})\).}}
\end{gather*}
We modify the remaining rules such that conclusions become judgment of the form \(\jtoc[n + 1]{T}{v}{c}{A}\), and any hypotheses become judgments of the form \(\jtoc[n]{T}{v}{c}{A}\).
For example:
\begin{gather*}
  \getrule{O-n-tu}\\
  \getrule{O-n-tot}
\end{gather*}
By decrementing the superscript with each rule application, we ensure that the derivation has finite height.
\Cref{prop:sill-obs-equiv-main-obs:1} states that \(\jtoc[n]{T}{v}{c}{A}\) captures the intended semantics:

\begin{proposition}
  \label{prop:sill-obs-equiv-main-obs:1}
  If \(T\) is a trace from \(\jcfgt{\Gamma}{\mc{C}}{\Delta}\), then for all \(n \in \N\) and free channel names \(c \in \freecn(T)\):
  \begin{enumerate}
  \item If \(\jttp{T}{c}{A}\), then there exists a unique \(v_n\) such that \(\jtoc[n]{T}{v_n}{c}{A}\).\label{item:sill-obs-equiv-main-obs:1}
  \item If \(\jtoc{T}{v}{c}{A}\), then \( \jtoc[n]{T}{\capxn{v}{n}}{c}{A} \).
  \end{enumerate}
  Moreover, the following type agreement properties are satisfied for all \(c \in \freecn(T)\):
  \begin{enumerate}[resume]
  \item For all \(n \in \N\), if \(\jtoc[n]{T}{v_n}{c}{A}\), then \(\jsynt{v_n}{A}\) and \(\jttp{T}{c}{A}\).\label{item:sill-obs-equiv-main-obs:2}
  \item If \(\jtoc[n]{T}{v_n}{c}{A_n}\) and \(\jtoc[m]{T}{v_m}{c}{A_m}\) for \(n, m \in \N\), then \(A_n = A_m\).
  \end{enumerate}
\end{proposition}

\begin{proof}
  We proceed by induction on \(n\) to show \cref{item:sill-obs-equiv-main-obs:1,item:sill-obs-equiv-main-obs:2}.
  In particular, we show that \(\jttp{T}{c}{A}\) if and only if \(\jtoc[n]{T}{v_n}{c}{A}\) for some \(v_n\), and in this case \(v_n\) is unique and satisfies \(\jsynt{v_n}{A}\).

  In the base case \(n = 0\), assume first that \(\jttp{T}{c}{A}\).
  There is a unique rule forming judgments \(\jtoc[0]{T}{c}{v_0}{A}\), namely, \getrn{O-z-bot}.
  It follows that \(\jtoc[0]{T}{c}{\bot}{A}\) and that \(\bot\) is the unique such \(v_0\).
  Moreover, \(\jsynt{\bot}{A}\) by \getrn{C-bot}.
  Conversely, if \(\jtoc[0]{T}{c}{v_0}{A}\) for some \(v_0\), then it must have been by \getrn{O-z-bot}.
  But \(\jttp{T}{c}{A}\) is its hypothesis, so we are done.

  Assume now that the result holds for some \(n\).
  We split into two cases, depending whether or not \(c\) is the carrier channel of some message fact \(\jmsg{a_0}{m_0} \in \mc{T}\).

  Assume first that it is not the case, \ie, that \( c \neq \carrcn(\jmsg{a_0}{m_0}) \) for all \( \jmsg{a_0}{m_0} \in \mc{T} \).
  The proof is analogous to the proof in the base case, except that we use \getrn{O-n-bot} instead of \getrn{O-z-bot}.

  Assume next that \(c = \carrcn(\jmsg{a_0}{m_0})\) for some \(\jmsg{a_0}{m_0} \in \mc{T}\).
  This message fact is unique with this property by \cref{cor:sill-background-dyn-prop-typed-config:1}.
  Assume that \(\jttp{T}{c}{A}\).
  A case analysis on the message fact reveals that exactly one rule is applicable to form the judgment \(\jtoc[n + 1]{T}{v_{n + 1}}{c}{A}\).
  From these observations, we deduce the existence and uniqueness of \(v_{n + 1}\).
  We give two possible cases; the remaining cases will follow by analogy.
  \begin{proofcases}
  \item[\(\jmsg{c}{\mClose{c}}\)] By \cref{prop:main:3}, \(A = \Tu\).
    By \getrn{O-n-tu}, \(\jtoc[n + 1]{T}{\cclose}{c}{\Tu}\), and \(\jsynt{\cclose}{\Tu}\) by \getrn{C-tu}.
  \item[\(\jmsg{c}{\mSendCP{c}{a}{d}}\)] By \cref{prop:main:3}, \(A = A' \Tot B\), \(\jttp{T}{a}{A'}\), and \(\jttp{T}{d}{B}\).
    By the induction hypothesis, there exist unique communications \(u\) and \(v\) such that \(\jtoc[n]{T}{u}{a}{A'}\) and \(\jtoc[n]{T}{v}{d}{B}\), and they satisfy \(\jsynt{u}{A'}\) and \(\jsynt{v}{B}\).
    Then \(\jtoc[n + 1]{T}{(u, v)}{c}{A \Tot B}\) by \getrn{O-n-tot}, and \(\jsynt{(u, v)}{A' \Tot B}\) by \getrn{C-tot}.
  \end{proofcases}
  Conversely, assume that \(\jtoc[n + 1]{T}{v_n}{c}{A}\).
  We proceed by case analysis on the rule used to form this judgment.
  We again give only two cases:
  \begin{proofcases}
  \item[\getrn{O-n-tu}] Then the judgment is \(\jtoc[n + 1]{T}{\cclose}{c}{\Tu}\), and \(\jmsg{c}{\mClose{c}} \in \mc{T}\) by hypothesis.
    We conclude \(\jttp{T}{c}{\Tu}\) by \cref{prop:main:3}.
  \item[\getrn{O-n-tot}] Then the judgment is \(\jtoc[n + 1]{T}{(u, v)}{c}{B \Tot D}\).
    By hypothesis, \(\jmsg{c}{\mSendCP{c}{b}{d}} \in \mc{T}\) for some \(b\) and \(d\), \(\jtoc[n]{T}{b}{u}{B}\), and \(\jtoc[n]{T}{d}{v}{D}\).
    By the induction hypothesis, \(\jttp{T}{b}{B}\) and \(\jttp{T}{d}{D}\).
    By \cref{prop:main:3}, \(\jttp{T}{c}{B' \Tot D'}\) for some \(B'\) and \(D'\) such that \(\jttp{T}{b}{B'}\) and \(\jttp{T}{d}{D'}\).
    But \(B = B'\) and \(D = D'\) by \cref{cor:sill-background-properties-traces:1}, so \(\jttp{T}{c}{B \Tot D}\) as desired.
  \end{proofcases}
  This completes the induction.

  Next, we proceed by induction on \(n\) to show that for all \(n\) and \(c\), if \(\jtoc{T}{v}{c}{A}\), then \(\jtoc[n]{T}{\capxn{v}{n}}{c}{A}\).
  The base case \(n = 0\) is immediate by \getrn{O-z-bot} and the fact that \(\capxn{v}{0} = \bot\).
  Assume the result for some \(n\).
  We proceed by case analysis on the rule used to form \(\jtoc{T}{v}{c}{A}\).
  We give a few model cases; the remaining cases follow by analogy.
  \begin{proofcases}
  \item[\getrn{O-bot}] The judgment is \(\jtoc{T}{\bot}{c}{A}\).
    Observe that \(\capxn{\bot}{n + 1} = \bot\).
    We are done by \getrn{O-n-bot}.
  \item[\getrn{O-tu}] The judgment is \(\jtoc{T}{\cclose}{c}{\Tu}\).
    Observe that \(\capxn{\cclose}{n + 1} = \cclose\).
    We are done by \getrn{O-n-tu}.
  \item[\getrn{O-tot}] The judgment is \(\jtoc{T}{(u, v)}{c}{B \Tot D}\).
    By hypothesis, \(\jmsg{c}{\mSendCP{c}{b}{d}} \in \mc{T}\) for some \(b\) and \(d\), \(\jtoc{T}{b}{u}{B}\), and \(\jtoc{T}{d}{v}{D}\).
    By the induction hypothesis, \(\jtoc[n]{T}{\capxn{u}{n}}{b}{B}\) and \(\jtoc[n]{T}{\capxn{v}{n}}{d}{D}\).
    Observe that \(\capxn{(u, v)}{n + 1} = (\capxn{u}{n}, \capxn{v}{n})\).
    We are done by \getrn{O-n-tot}.
  \end{proofcases}
  This completes the induction.

  Finally, we show that if \(\jtoc[n]{T}{v_n}{c}{A_n}\) and \(\jtoc[m]{T}{v_m}{c}{A_m}\) for \(n, m \in \N\), then \(A_n = A_m\).
  By \cref{item:sill-obs-equiv-main-obs:2}, \(\jttp{T}{c}{A_n}\) and \(\jttp{T}{c}{A_m}\).
  We conclude \(A_n = A_m\) by \cref{cor:sill-background-properties-traces:1}.
\end{proof}

Using \cref{lemma:sill-obs-equiv/observ-comm:3,prop:sill-obs-equiv/observ-comm-equiv:13}, we recognize partially ordered sets \(\commbr[{=}]{A}\) of communications of type \(A\) as complete partial orders.
In particular, every ascending chain \(v_n \commsim[{=}] v_{n + 1}\) in \(\commbr[{=}]{A}\) has a least upper bound \(\dirsup_{n \in N} v_n\) in \(\commbr[{=}]{A}\).
We are finally in a position to show the existence of observed communications.

\begin{proposition}[Existence]
  \label{prop:main:11}
  If $T$ is a trace from $\jcfgt{\Gamma}{\mc{C}}{\Delta}$, then for all $c$, if $\jttp{T}{c}{A}$, then $\jtoc{T}{\dirsup_{n \in \N} v_n}{c}{A}$ where \(\jtoc[n]{T}{v_n}{c}{A}\) for \(n \in \N\) and \(\dirsup_{n \in N} v_n\) is computed in \(\commbr[{=}]{A}\).
\end{proposition}

\begin{proof}
  We use the coinduction proof principle.
  Set
  \[
    \mc{O} = \Set*{ \left(\,\dirsup_{n \in \N} v_n, c, A\right) \given c \in \freecn(T), \jttp{T}{c}{A}, \forall n \in \N \exists v_n \;.\; \jtoc[n]{T}{v_n}{c}{A} }
  \]
  where the suprema are computed relative to \(\commsim[{=}]\).
  By \cref{prop:sill-obs-equiv-main-obs:1} and the remark preceding \cref{cor:sill-background-properties-traces:1}, there exists a unique element \((v,c,A) \in \mc{O}\) for all free channel names \(c\) appearing in \(T\).
  Let \(\Phi\) be the rule functional defining the judgment \(\jtoc{T}{v}{c}{A}\).
  To show that \(\jtoc{T}{v}{c}{A}\) for all \((v,c,A) \in \mc{O}\), it is sufficient to show that \(\mc{O}\) is a post-fixed point of \(\Phi\).
  Let \((v,c,A) \in \mc{O}\) be arbitrary.
  We proceed by case analysis on the fact \(\jmsg{a_0}{m_0} \in \mc{T}\) with \(\carrcn(\jmsg{a_0}{m_0}) = c\).
  If it exists, then it is unique with this property by \cref{cor:sill-background-dyn-prop-typed-config:1}.
  We give a few representative cases:
  \begin{proofcases}
  \item[(non-existence)] Assume first that there is no such message fact.
    Then \((v, c, A) = (\bot, c, A)\) because \(\jtoc[n + 1]{T}{\bot}{c}{A}\) by \getrn{O-n-bot} for all \(n\), and \(\dirsup_{n \in \N} \bot = \bot \).
    But \((\bot, c, A) \in \Phi(\mc{O})\) by \getrn{O-bot}, so we are done.
  \item[\(\jmsg{c}{\mClose{c}}\)]
    Then \((v, c, A) = (\cclose, c, \Tu)\) because \(\jtoc[n+1]{T}{\cclose}{c}{\Tu}\) by \getrn{O-n-tu} for all \(n\), and \(\dirsup_{n \in \N} \cclose = \cclose\).
    But \((\cclose, c, \Tu) \in \Phi(\mc{O})\) by \getrn{O-tu}, so we are done.
  \item[\(\jmsg{c}{\mSendCP{c}{b}{d}}\)]
    Then \( (v, c, A) = \left(\dirsup (u_n,w_n), c, B \Tot D\right) \) because \(\jtoc[n + 1]{T}{(u_n, w_n)}{B \Tot D}{c}\) by \getrn{O-n-tot} for all \(n\), where \(\jtoc[n]{T}{u_n}{B}{b}\) and \(\jtoc[n]{T}{w_n}{d}{D}\).
    Observe that \(\left(\dirsup u_n, b, B\right) \in \mc{O}\) and \(\left(\dirsup w_n, d, D\right) \in \mc{O} \).
    This implies that \(\left(\left(\dirsup u_n, \dirsup w_n\right), c, B \Tot D\right) \in \Phi(\mc{O})\) by \getrn{O-tot}.
    The result then follows from the fact that \(\dirsup (u_n, w_n) = \left(\dirsup u_n, \dirsup w_n\right)\) in \(\commbr[{=}]{B \Tot D}\), where \(\dirsup u_n\) and \(\dirsup w_n\) are computed in \(\commbr[{=}]{B}\) and \(\commbr[{=}]{D}\), respectively.
  \item[\(\jmsg{c}{\mSendVP{c}{f}{d}}\)]
    The corresponding modified rule is:
    \[
      \getrule{O-n-tand}
    \]
    Then \((v, c, A) = \left(\dirsup (\cval f, u_n), c, \Tand{\tau}{D}\right)\) because \(\jtoc[n]{T}{u_n}{D}{d}\) for all \(n\).
    Observe that \(\left(\dirsup u_n, d, D\right) \in \mc{O}\), so \(\left(\left(\cval f, \dirsup u_n\right), c, \Tand{\tau}{D}\right) \in \Phi(\mc{O})\) by \getrn{O-tand}.
    But \(\dirsup (\cval f, u_n) = \left(\cval f, \dirsup u_n\right)\), so we are done.\qedhere
  \end{proofcases}
\end{proof}

The following \namecref{theorem:main:2} is an immediate corollary of \cref{prop:main:11,prop:sill-obs-equiv-main-obs:1,prop:obssem:4}.
It states $\jtoc{T}{v}{c}{A}$ defines a total function from free channel names $c$ in $T$ to session-typed communications $\jsynt{v}{A}$.
It also gives a semantic account of session fidelity.
Session fidelity is a language property guaranteeing that well-typed processes respect the communication behaviour specified by their session types.
It is reflected in the fact that the session-typed communication \(\jsynt{v}{A}\) observed on \(c\) agrees with \(c\)'s statically determined type \(A\).

\begin{theorem}
  \label{theorem:main:2}
  If $T$ is a trace from $\jcfgt{\Gamma}{\mc{C}}{\Delta}$, then for all $c \in \freecn(T)$,
  \begin{enumerate}
  \item there exists a unique \(v\) such that \(\jtoc{T}{v}{c}{A}\);
  \item if \(\jtoc{T}{v}{c}{A}\), then \(\jsynt{v}{A}\);
  \item $\jtoc{T}{v}{c}{A}$ if and only if $\jttp{T}{c}{A}$.
  \end{enumerate}
\end{theorem}

The following \namecref{theorem:main:1} captures the confluence property typically enjoyed by SILL-style languages.
In particular, it implies that the communications observed on free channels of well-typed configurations do not depend on the choice of fair execution.

\begin{theorem}
  \label{theorem:main:1}
  Let $T$ and $T'$ be fair executions of $\jcfgti{\Gamma}{\Iota}{\mc{C}}{\Delta}$.
  For all $c : A \in \Gamma \Iota \Delta$, if $\jtoc{T}{v}{c}{A}$ and $\jtoc{T'}{w}{c}{A}$, then $v = w$.
\end{theorem}

\begin{proof}
  Assume $\jtoc{T}{v}{c}{A}$ and $\jtoc{T'}{w}{c}{A}$.
  By \cref{cor:sill-background-properties-traces:2}, traces $T$ and $T'$ are union-equivalent, \ie, $\mc{T} = \mc{T}'$.
  It easily follows that $\jtoc{T'}{w}{c}{A}$ if and only if $\jtoc{T}{w}{c}{A}$.
  So $v = w$ by \cref{theorem:main:2}.
\end{proof}

\Cref{theorem:main:1} crucially depends on fairness.
Indeed, without fairness a process could have infinitely many observations.
To see this, let $\Omega$ be the divergent process given by \cref{ex:sill-background-typing-mult-rewr:2} and let $B$ be given by
\[
  \jtypem{\cdot}{a : \Tu}{\tFix{p}{\tSendU{b}{\tSendL{b}{l}{p}}}}{b}{\Trec{\beta}{\Tplus\{l:\beta\}}}
\]
\Cref{eq:sill:msr-cut} is the first step of any execution of their composition $\jtypem{\cdot}{\cdot}{\tCut{a}{\Omega}{B}}{b}{\Trec{\beta}{\Tplus\{l:\beta\}}}$.
It spawns $\Omega$ and $B$ as separate processes.
Without fairness, an execution could then consist exclusively of applications of \cref{eq:sill:msr-tot-l} to $\Omega$.
This would give the observed communication $\bot$ on $b$.
Alternatively, $B$ could take finitely many steps, leading to observations where $b$ is a tree of correspondingly finite height.
Fairness ensures that $B$ and $\Omega$ both take infinitely many steps, leading to the unique observation $(\ms{unfold}, (l, (\ms{unfold}, \dotsc)))$ on $b$.

This notion of observed communication scales to support language extensions.
Indeed, for each new session type one first defines its corresponding session-typed communications.
Then, one specifies how to observe message facts $\jmsg{c}{m}$ in a trace as communications.

So far, we have assumed that whenever we want to observe a process's communications on a channel \(c\), that we already have a complete trace \(T\) of that process.
Indeed, the coinductively defined judgment \(\jtoc{T}{v}{c}{A}\) is defined in terms of the set of all message facts that appear in \(T\).
Though this assumption is mathematically expedient, it can be unrealistic when considering concrete non-terminating programs.
In these cases, we might want to observe communications in a step-by-step manner over the course of an execution.
Fortunately, communications observed in prefixes of a trace approximate and uniquely determine the observed communications of the complete trace.
Given some trace $T$, let $T^n$ be its prefix of $n$ steps.
Concretely, we show that \(\jtoc{T}{v}{c}{A}\) is the least upper bound of the ascending chain \(\jsynt{v_n}{A}\) (relative to \(\commsim[=]\)), where for each \(n\), $\jtoc{T^n}{v_n}{c}{A}$ is the communication observed on \(c\) in the first \(n\) steps of \(T\).
We begin by showing that the $v_n$ form an ascending chain, \ie, that observing communications is monotone in the length of a trace.

\begin{proposition}
  \label{prop:sill-obs-equiv/observ-comm:2}
  For all $n$ and all channels $c$, if $\jtoc{T^n}{v_n}{c}{A}$ and $\jtoc{T^{n + 1}}{v_{n + 1}}{c}{A}$, then $\jsynt{v_n \commsim[=] v_{n + 1}}{A}$.
\end{proposition}

\begin{proof}
  All communications observed from a finite prefix are finite.
  We proceed by strong induction on the size of the derivation of $\jtoc{T^n}{v_n}{c}{A}$.
  The base cases for $\jsynt{v_n}{A}$ are given by the axioms:
  \begin{proofcases}
  \item[\getrn{O-bot}]
    Then $v_n = \bot$.
    We are done by \getrn{CS-bot}.
  \item[\getrn{O-tu}]
    Then $v_n = \cclose$.
    Extending the trace by a single step does not affect the observed communication on $c$, so $v_{n + 1} = \cclose$ as well.
    We are done by \getrn{CS-tu}.
  \end{proofcases}
  Now we proceed to the inductive step.
  We show several illustrative cases; the remainder are analogous.
  \begin{proofcases}
  \item[\getrn{O-tot}]
    Then $\jtoc{T^n}{(u_n, w_n)}{c}{A \Tot B}$ because $\jmsg{c}{\mSendCP{c}{a}{d}}$ appears in $T^n$, and $\jtoc{T^n}{u_n}{a}{A}$ and $\jtoc{T^n}{w_n}{d}{B}$.
    Let $u_{n + 1}$ and $w_{n + 1}$ be given by $\jtoc{T^{n + 1}}{u_{n + 1}}{a}{A}$ and $\jtoc{T^{n + 1}}{w_{n + 1}}{d}{B}$.
    By the induction hypothesis, $\jsynt{u_n \commsim[=] u_{n + 1}}{A}$ and $\jsynt{w_n \commsim[=] w_{n + 1}}{B}$.
    By \getrn{O-tot}, $\jtoc{T^{n+ 1}}{(u_{n + 1}, w_{n + 1})}{c}{A \Tot B}$.
    We are done by \getrn{CS-tot}.
  \item[\getrn{O-tand}]
    Then $\jtoc{T^n}{(\cval f, w_n)}{c}{\Tand{\tau}{A}}$ because $\jmsg{c}{\mSendVP{c}{f}{d}}$ appears in $T^n$, and $\jtoc{T^n}{w_n}{d}{A}$.
    Let $w_{n + 1}$ be given by $\jtoc{T^{n + 1}}{w_{n + 1}}{d}{A}$.
    By the induction hypothesis,  $\jsynt{w_n \commsim[=] w_{n + 1}}{A}$.
    By \getrn{O-tand}, $\jtoc{T^{n + 1}}{(\cval f, w_{n + 1})}{c}{\Tand{\tau}{A}}$.
    We are done by \getrn{CS-tand}.\qedhere
  \end{proofcases}
\end{proof}

\begin{proposition}
  \label{prop:sill-obs-equiv/observ-comm:4}
  Let $T$ be a trace, and let $v$ be given by $\jtoc{T}{v}{c}{A}$.
  For each $n$, let $T^n$ be the prefix of length $n$ of $T$, and let $v_n$ be given by $\jtoc{T^n}{v_n}{c}{A}$.
  Then $v = \dirsup_n v_n$ in $\commbr[=]{A}$, the partially ordered set of communications of type \(A\) ordered by \(\commsim[=]\).
\end{proposition}

\begin{proof}
  It is sufficient to show that:
  \begin{enumerate}
  \item for all $l$, there exists an $m$ such that $\capxn{v}{l} \commsim[=] v_m$; and\label{item:sill-obs-equiv/observ-comm:2}
  \item for all $m$, there exists an $u$ such that $v_m \commsim[=] \capxn{v}{u}$.\label{item:sill-obs-equiv/observ-comm:8}
  \end{enumerate}
  Indeed, \cref{item:sill-obs-equiv/observ-comm:2} implies that every upper bound of the $v_m$ is an upper bound of the $\capxn{v}{m}$, while \cref{item:sill-obs-equiv/observ-comm:8} implies that every upper bound of the $\capxn{v}{m}$ is an upper bound of the $v_m$.
  So $\jsynt{\dirsup_m \capxn{v}{m} \commeq[=] \dirsup_m v_m}{A}$.
  By \cref{lemma:sill-obs-equiv/observ-comm:3}, this is an equality.
  By \cref{cor:sill-obs-equiv/barb-cont-congr:1}, $\dirsup_m \capxn{v}{m} = v$.
  So $v = \dirsup_n v_n$ as desired.

  We begin with \cref{item:sill-obs-equiv/observ-comm:2}.
  Consider some $l$, and let $m$ be any $m$ such that all of the messages appearing in the derivation of $\capxn{v}{l}$ appear in $T^m$.
  For \cref{item:sill-obs-equiv/observ-comm:8}, consider some $m$, and let $u$ be height of $v_m$ plus one.\qedhere
\end{proof}

We now state a few basic properties of observations on single channels.
First, we show that observed communications are entirely determined by message facts deemed ``observable''.
This fact helps us relate the observed communications of different configurations, and we will use it later when defining observed communication-based equivalences on configurations.

\begin{definition}
  \label{def:sill-obs-equiv/observ-comm:13}
  Let $T$ be a trace of some configuration $\mc{C}$.
  A message fact $\jmsg{a}{m}$ is \defin{observable from $c$ in $T$} if appears in the derivation of $\jtoc{T}{v}{c}{A}$.
\end{definition}

The following proposition captures the intuitive fact that observed communications on $c$ are entirely determined by the set of message facts observable from $c$.
Recall from \cref{sec:ssos-fairness:mult-rewr-syst:semant-irrel-fresh} that a refreshing a trace involves renaming its fresh constants in a semantics-preserving manner.

\begin{proposition}
  \label{prop:sill-obs-equiv/observ-comm:5}
  Let $T$ and $T'$ be fair traces of $\mc{C}$ and $\mc{C}'$, respectively.
  For all $c$, let $v$ and $v'$ be given by $\jtoc{T}{v}{c}{A}$ and $\jtoc{T'}{v'}{c}{A}$.
  If $T$ can be refreshed to $T''$ such that $T''$ and $T'$ have the same sets of observable message facts from $c$, then $v = v'$.
\end{proposition}

\begin{proof}
  Immediate from the fact that $\jtoc{T}{v}{c}{A}$ is entirely determined by the set of messages observable from $c$, and that it is invariant under refreshing of channel names.
\end{proof}

Configurations with no common channels do not interfere with each other:

\begin{proposition}[Non-Interference]
  \label{prop:sill-obs-equiv/observ-comm:11}
  Consider multisets $\jcfgt{\Gamma}{\mc{C}}{\Delta}$ and $\jcfgt{\Phi}{\mc{D}}{\Xi}$ with disjoint sets of free channels, \ie, such that $\jcfgt{\Gamma\Phi}{\mc{C},\mc{D}}{\Delta\Xi}$ is well formed.
  Then for all fair traces $T$ of $\mc{C}, \mc{D}$ and $T'$ of $\mc{C}$ and all $c \in \freecn(\mc{C})$, $\jtoc{T}{v}{c}{A}$ if and only if $\jtoc{T'}{v}{c}{A}$.
\end{proposition}

\begin{proof}
  We claim that every fair trace $T$ of $\mc{C}, \mc{D}$ induces a fair trace $T'$ of $\mc{C}$.
  Explicitly, we use preservation (\cref{prop:sill-background:1}) to show that
  \begin{enumerate}
  \item every multiset in $T$ is of the form $\jcfgt{\Gamma\Phi}{\mc{C}',\mc{D}'}{\Delta\Xi}$ for some $\jcfgt{\Gamma}{\mc{C}'}{\Delta}$ and $\jcfgt{\Phi}{\mc{D}'}{\Xi}$;
  \item every step in $T$ is of the form $\ctxh[\mc{D}']{\mc{C}'}{\Phi}{\Xi} \msstep \ctxh[\mc{D}'']{\mc{C}'}{\Phi}{\Xi}$ or $\ctxh[\mc{C}']{\mc{D'}}{\Gamma}{\Delta} \msstep \ctxh[\mc{C}'']{\mc{D}'}{\Gamma}{\Delta}$.
  \end{enumerate}
  It is sufficient to show that these properties hold for every finite prefix of $T$.
  We do so by induction on the number $n$ of steps in the prefix.
  The result is immediate when $n = 0$.
  Assume the result for some $n$, then the last multiset is of the form $\jcfgti{\Gamma\Phi}{\Iota}{\mc{C}',\mc{D}'}{\Delta\Xi}$ for some $\jcfgt{\Gamma}{\mc{C}'}{\Delta}$ and $\jcfgt{\Phi}{\mc{D}'}{\Xi}$.
  Assume some rule instance $r(\theta)$ is applicable to $\mc{C}',\mc{D}'$.
  If its active multiset intersects with both $\mc{C}'$ and $\mc{D}'$, then a case analysis on the rules reveals that it contains a message fact.
  A case analysis on the rules (or, equivalently, \cite[Lemma~5.9.5]{kavanagh_2021:_commun_based_seman}) then implies that $\mc{C}'$ and $\mc{D}'$ have a free channel in common, a contradiction.
  So the active multiset of $r(\theta)$ is contained in $\mc{C}'$ or in $\mc{D}'$.
  We are done by preservation.

  Taking the subsequence of steps of the form $\ctxh[\mc{C}']{\mc{D'}}{\Gamma}{\Delta} \msstep \ctxh[\mc{C}'']{\mc{D}'}{\Gamma}{\Delta}$ gives a trace $T'$ of $\mc{C}$.
  It is fair because $T$ is fair.
  Let $c \in \freecn(\mc{C})$ be arbitrary.
  Every message fact observable from $c$ in $T$ is observable from $c$ in $T'$, and vice-versa.
  It follows that $\jtoc{T}{v}{c}{A}$ if and only if $\jtoc{T'}{v}{c}{A}$.
  The choice of trace $T'$ for $\mc{C}$ does not matter by \cref{theorem:main:1}.
\end{proof}

To summarize, in this section we coinductively defined a judgment \(\jtoc{T}{v}{c}{A}\) that observes a unique sessions-typed communication \(\jsynt{v}{A}\) for each free channel \(c\) in a trace \(T\) of a well-typed configuration \({\jcfgt{\Gamma}{\mc{C}}{\Delta}}\).
The observed communication \(\jsynt{v}{A}\) is independent of the choice of fair execution \(T\), reflecting a confluence property enjoyed by SILL-like languages.
Moreover, its type \(A\) agrees with the of \(c\) statically determined by the preservation theorem.
We also characterized observed communications on channels using their finite approximations.
In the next section, we use these results to define observations on collections of channels.

\subsection{Observed Communications of Configurations}
\label{sec:sill-obs-equiv:observ-comm:comm-mult-chans}

We use \cref{theorem:main:1,theorem:main:2} to define observations on channels in a configuration, independently of the fair execution:

\begin{definition}
  \label{def:sill-obs-equiv/observ-comm:5}
  Given $\Psi \subseteq \Gamma\Iota\Delta$, the \defin{observed communication on $\Psi$}\glsadd{obsbr} of $\jcfgti{\Gamma}{\Iota}{\mc{C}}{\Delta}$ is the tuple
  \[
    \obsbr{\jcfgti{\Gamma}{\Iota}{\mc{C}}{\Delta}}_\Psi = (c : v_c)_{c \in \check \Psi}
  \]
  of observed communications, where $\jtoc{T}{v_c}{c}{A}$ for $c : A \in \Psi$ for some fair execution $T$ of $\jcfgti{\Gamma}{\Iota}{\mc{C}}{\Delta}$.
  If $c \in \check \Psi$, then we occasionally write $\obsbr{\jcfgti{\Gamma}{\Iota}{\mc{C}}{\Delta}}_\Psi(c)$ for the communication $v_c$ observed on $c$.
\end{definition}

\Cref{def:sill-obs-equiv/observ-comm:5} is well-defined.
Indeed, a fair execution $T$ exists by \cref{cor:sill-background-properties-traces:2}, and the observed communications do not depend on the choice of $T$ by \cref{theorem:main:1}.
The communications $v$ such that $\jtoc{T}{v}{c}{A}$ exist and are unique by \cref{theorem:main:2}.

\Cref{def:sill-obs-equiv/observ-comm:5} is simplified by the fact that the multiset rewriting rules defining Polarized SILL are non-overlapping, so its fair executions are union-equivalent.
Indeed, this fact was used in the proof of \cref{theorem:main:1} to show that the communication observed on a channel is independent of the fair execution.
In language extensions whose rules are overlapping, \eg, those involving non-deterministic choice, observed communications will be sets of tuples instead of single tuples.
This is because fair executions need not be union-equivalent in the presence of overlapping rules, so each equivalence class of union-equivalent executions may give a different observation.
We conjecture that using sets of tuples should pose no significant difficulty to the theory.

Generally, we deem internal channels to be private and unobservable, and we only interact with configurations over their interfaces.
However, \cref{def:sill-obs-equiv/observ-comm:5} allows us to observe communications on a configuration's internal channels.
This will let us observe communications between configurations and experiments when defining ``internal'' observational equivalence in \cref{cha:sill-obs-equiv}.

\begin{example}
  \label{ex:sill-obs-equiv-main-obs:1}
  Recall the process \(\jtypem{\cdot}{\cdot}{\mt{omega}}{o}{\mt{conat}}\) from \cref{ex:sill-background-typing-mult-rewr:1} that sends an infinite sequence of successor labels \(\mt{s}\), and the divergent process \(\jtypem{\cdot}{o : \mt{conat}}{\Omega}{a}{\Tu}\) from \cref{ex:sill-background-typing-mult-rewr:2}.
  Consider the configuration \(\jcfgti{\cdot}{o : \mt{conat}}{\jproc{o}{\mt{omega}}, \jproc{a}{\Omega}}{a : \Tu}\).
  Its observed communication is
  \[
    \obsbr{\jcfgti{\cdot}{o : \mt{conat}}{\jproc{o}{\mt{omega}}, \jproc{a}{\Omega}}{a : \Tu}}_{a : \Tu, o : \mt{conat}} = (a : \bot, o : (\cunfold, (\mt{s}, (\cunfold, \dotsc)))).
  \]
  It reflects the fact that, thanks to fairness, both processes make progress, even though they are non-terminating.
  If we dropped the fairness assumption, then infinitely many observations on the channel \(o\) would be possible.
  Indeed, without fairness we could step \(\mt{omega}\) only finitely many times and then step \(\Omega\) forever.
  This observed communication also reflects Polarized SILL's asynchronous communication semantics.
  Indeed, though \(\Omega\) never synchronizes or attempts to receive communications on \(o\), the process \(\mt{omega}\) sends \(\Omega\) infinitely many messages on \(o\).
\end{example}

\begin{example}
  \label{ex:sill-obs-equiv-main-obs:3}
  Our observed communication semantics lets us differentiate sent channels in a name-independent manner.
  Indeed, consider processes
  \begin{gather*}
    \jtypem{\cdot}{a : \Tu, b : \Tu}{\tSendC{c}{a}{\Omega}}{c}{\Tu \Tot \Tu},\\
    \jtypem{\cdot}{a : \Tu, b : \Tu}{\tSendC{c}{b}{\Omega}}{c}{\Tu \Tot \Tu}.
  \end{gather*}
  We can differentiate these two processes by composing them with communication partners on \(a\) and \(b\) and then observing communication on \(c\).
  Indeed, contrast the observed communications
  \begin{align*}
    \obsbr{\jcfgti{\cdot}{a : \Tu, b : \Tu}{\jproc{a}{\tClose{a}}, \jproc{b}{\Omega}, \jproc{c}{\tSendC{c}{a}{\Omega}}}{c : \Tu \Tot \Tu}}_{c : \Tu \Tot \Tu} &= (c : (\cclose, \bot)),\\
    \obsbr{\jcfgti{\cdot}{a : \Tu, b : \Tu}{\jproc{a}{\tClose{a}}, \jproc{b}{\Omega}, \jproc{c}{\tSendC{c}{b}{\Omega}}}{c : \Tu \Tot \Tu}}_{c : \Tu \Tot \Tu} &= (c : (\bot, \bot)).
  \end{align*}
  In any fair execution of \( \jproc{a}{\tClose{a}}, \jproc{b}{\Omega}, \jproc{c}{\tSendC{c}{d}{\Omega}} \), the process \(\tClose a\) will send the close message over \(a\), while no messages are ever sent on \(b\).
  The observed communication \((\cclose, \bot)\) makes it apparent that the channel \(a\) was sent over \(c\).
  In contrast, the observation \((\bot, \bot)\) makes it apparent no communication occurred on the channel sent over \(c\), and we deduce that the sent channel must have been~\(b\).
\end{example}

\begin{example}
  \label{ex:sill-obs-equiv-main-obs:2}
  Observed communications do not take into account the order in which a process sends on channels.
  For example, the following configurations have the same observed communications $(a : (l, \bot), b : (r, \bot))$ on $a$ and $b$, even though they send on $a$ and on $b$ in different orders:
  \begin{gather*}
    \jcfgt{a : \Tamp \{ l : \Tus{\Tu} \}}{\jproc{b}{\tSendL{a}{l}{\tSendL{b}{r}{\Omega}}}}{b : \Tplus \{r : \Tu\}}\\
    \jcfgt{a : \Tamp \{ l : \Tus{\Tu} \}}{\jproc{b}{\tSendL{b}{r}{\tSendL{a}{l}{\Omega}}}}{b : \Tplus \{r : \Tu\}}.
  \end{gather*}
  The order in which channels are used is not reflected in observations for several reasons.
  First, messages are only ordered on a per-channel basis, and messages sent on different channels can arrive out of order.
  Second, each channel has a unique pair of endpoints, and the \rn{Conf-C} rule organizes processes in a forest-like structure (\cf \cref{prop:sill-obs-equiv/observ-comm-equiv:7}).
  This means that two configurations communicating with a configuration $\mc{C}$ at the same time cannot directly communicate with each other to compare the order in which $\mc{C}$ sent them messages.
  In other words, configurations cannot distinguish the ordering of messages on different channels.
\end{example}

We lift communication simulations and equivalences to tuples of communications component-wise:
\begin{align*}
  (c : v_c)_{c : C \in \Gamma} \commsim[\leqslant] (c : w_c)_{c : C \in \Gamma} &\iff \forall c : C \in \Gamma \;.\; \jsynt{v_c \commsim[\leqslant] w_c}{C},\\
  (c : v_c)_{c : C \in \Gamma} \commeq[\equiv] (c : w_c)_{c : C \in \Gamma} &\iff \forall c : C \in \Gamma \;.\; \jsynt{v_c \commeq[\equiv] w_c}{C}.
\end{align*}

\section{Communication-Based Testing Equivalences}
\label{cha:sill-obs-equiv}

We adopt an extensional view of process equivalence, where we say that two processes are equivalent if we cannot differentiate them through experimentation.
Recall that communication is our sole means of interacting with processes, and that communication their only phenomenon.
This suggests that processes should be deemed equivalent if, whenever we subject them to ``communicating experiments'', we observe equivalent communications.
We introduce ``observation systems'' as a general framework for experimentation on communicating systems, and we use them to define observational equivalence.
We study observation systems in general, and we consider three natural choices of observation systems.
We study these three systems in turn in \cref{sec:total-obs-equiv:total-observ,sec:sill-obs-equiv:intern-observ,sec:sill-obs-equiv:extern-observ}.
We show that one of them induces an observational equivalence that coincides with barbed congruence in \cref{sec:sill-obs-equiv:extern-observ}.
Because our substructural operational semantics and observed communication semantics are defined on configurations and not on processes, we define observation systems and observational equivalence on configurations instead of on processes.
We show how to restrict observational equivalences to processes in \cref{sec:sill-obs-equiv:proc-equiv}.

We follow \textcites[\bibstring{chapter}~2]{milner_1980:_calcul_commun_system}[65]{hoare_1985:_commun_sequen_proces} in identifying experimenting agents with processes themselves (in our setting, with configuration contexts).
Concretely, an \defin{experiment} on configurations with interface \((\Gamma, \Delta)\) is a pair \((\mc{E}, \Psi)\) where $\jcfgti{\Lambda}{\Iota}{\ctxh{\mc{E}}{\Gamma}{\Delta}}{\Xi}$ is a configuration context and $\Psi \subseteq \Lambda\Iota\,\Xi$ is a subset of the channels free in \(\ctxh{\mc{E}}{\Gamma}{\Delta}\).
The context \(\mc{E}\) is the experimenting agent, while \(\Psi\) is the collection of channels on which we observe communications.
We subject a configuration \(\jcfgt{\Gamma}{\mc{C}}{\Delta}\) to an experiment \((\mc{E}, \Psi)\) by executing \(\ctxh[\mc{C}]{\mc{E}}{\Gamma}{\Delta}\) and making the observation \(\obsbr{\jcfgt{\Lambda}{\ctxh[\mc{C}]{\mc{E}}{\Gamma}{\Delta}}{\Xi}}_\Psi\).
Configurations \(\jcfgt{\Gamma}{\mc{C}}{\Delta}\) and \(\jcfgt{\Gamma}{\mc{D}}{\Delta}\) are equivalent according to \((\mc{E}, \Psi)\) (relative to some preorder \(\leqslant\)) if
\[
  \obsbr{\jcfgt{\Lambda}{\ctxh[\mc{C}]{\mc{E}}{\Gamma}{\Delta}}{\Xi}}_\Psi \commeq[\leqslant] \obsbr{\jcfgt{\Lambda}{\ctxh[\mc{D}]{\mc{E}}{\Gamma}{\Delta}}{\Xi}}_\Psi.
\]
Intuitively, this means that we cannot observe any differences in communication on channels \(\Psi\) when we compose \(\mc{C}\) and \(\mc{D}\) with \(\mc{E}\).

A single experiment is typically insufficient to discriminate between configurations.
As a result, we instead use collections of experiments.
This approach builds on the ``testing equivalences'' framework introduced by De Nicola and Hennessy~\cite{denicola_hennessy_1984:_testin_equiv_proces,hennessy_1983:_synch_async_exper_proces,denicola_1985:_testin_equiv_fully}.
Roughly speaking, their framework subjects processes to collections experiments that could potentially succeed, and it deems two processes to be equivalent if they succeed the same experiments.
Their notion of experimental success was based on observing a ``success'' state.
Instead of observing states, our experiments involve observing communications.
They called a collection of experiments and success states a ``computational system''~\cite[Definition~2.1.1]{hennessy_1983:_synch_async_exper_proces}.
We adapt computational systems to our communication-based setting as follows:

\begin{definition}
  \label{def:sill-obs-equiv/main-equiv:1}
  An \defin{observation system} $\mc{S}$ on configurations is a pair $(\mc{X}, {\leqslant})$, where
  \begin{itemize}
  \item $\mc{X} = (\mc{X}_{\Gamma \vdash \Delta})_{\Gamma, \Delta}$ is a type-indexed family of sets, where $\mc{X}_{\Gamma \vdash \Delta}$ is a set of experiments $(\mc{E}, \Psi)$, where $\jcfgti{\Lambda}{\Iota}{\ctxh{\mc{E}}{\Gamma}{\Delta}}{\Xi}$ is a context and $\Psi \subseteq \Lambda\Iota\,\Xi$ is a subset of the channels free in \(\ctxh{\mc{E}}{\Gamma}{\Delta}\);
  \item $\leqslant$ is a type-indexed relation on terms.
  \end{itemize}
  We assume that $\mc{X}$ is closed under exchange: if $\Gamma'$ and $\Delta'$ are permutations of $\Gamma$ and $\Delta$, respectively, then $\mc{X}_{\Gamma' \vdash \Delta'} = \mc{X}_{\Gamma \vdash \Delta}$.
  We also assume that $\mc{X}$ is closed under renaming.\qedhere
\end{definition}

In our development, we assume that the experiments $\mc{E}$ of observation systems are configuration contexts written in the same language as the configurations on which we are experimenting.
However, this is not a necessary assumption: experiments could be written in any language, so long as its multiset rewriting semantics does not interfere with the hypotheses underlying the observed communication semantics of \cref{sec:sill-obs-equiv:observ-comm}.

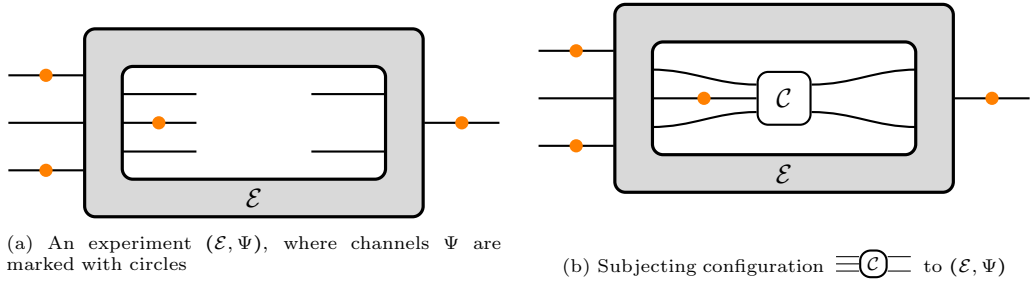
\begin{figure}
  \centering
  \begin{subfigure}{0.4\linewidth}
    \centering
    \begin{tikzpicture}[scale=0.5, thick, baseline]
      \draw [ctxt, very thick, even odd rule] (0,0) rectangle (9,5) (1,1) rectangle (8,4);
      \node [fit={(0,0) (9,5)}, inner sep=0] (outer) {};
      \node [fit={(1,1) (8,4)}, inner sep=0] (inner) {};
      \node [fit={(3,1) (6,4)}, inner sep=0] (cinner) {};
      \node [fit={(-2,0) (11,5)}, inner sep=0] (couter) {};
      \node at (4.5,0.5) {\(\mc{E}\)};

      \draw (inner.wnw) -- (inner.wnw -| cinner.west);
      \draw (inner.west) -- (inner.west -| cinner.west) node [obs] {};
      \draw (inner.wsw) -- (inner.wsw -| cinner.west);
      \draw (inner.ene) -- (inner.ene -| cinner.east);
      \draw (inner.ese) -- (inner.ese -| cinner.east);

      \draw (outer.wnw) -- (outer.wnw -| couter.west) node [obs] {};
      \draw (outer.west) -- (outer.west -| couter.west);
      \draw (outer.wsw) -- (outer.wsw -| couter.west) node [obs] {};
      \draw (outer.east) -- (outer.east -| couter.east) node [obs] {};
    \end{tikzpicture}
    \caption{An experiment \((\mc{E}, \Psi)\), where channels \(\Psi\) are marked with circles}
  \end{subfigure}
  \quad
  \begin{subfigure}{0.4\linewidth}
    \centering
    \begin{tikzpicture}[scale=0.5, thick, baseline=-0.25cm]
      \draw [ctxt, very thick, even odd rule] (0,0) rectangle (9,5) (1,1) rectangle (8,4);
      \node [fit={(0,0) (9,5)}, inner sep=0] (outer) {};
      \node [fit={(1,1) (8,4)}, inner sep=0] (inner) {};
      \node [fit={(3,1) (6,4)}, inner sep=0] (cinner) {};
      \node [fit={(-2,0) (11,5)}, inner sep=0] (couter) {};
      \node at (4.5,0.5) {\(\mc{E}\)};

      \node [node, minimum height=2em, minimum width=2em] (proc) at (4.5, 2.5) {\( \mc{C} \)};

      \draw (inner.wnw) to[out=0, in=180] (proc.wnw);
      \draw (inner.west) -- (proc.west) node [obs] {};
      \draw (inner.wsw) to[out=0, in=180] (proc.wsw);
      \draw (inner.ene) to[out=180, in=0] (proc.ene);
      \draw (inner.ese) to[out=180, in=0] (proc.ese);

      \draw (outer.wnw) -- (outer.wnw -| couter.west) node [obs] {};
      \draw (outer.west) -- (outer.west -| couter.west);
      \draw (outer.wsw) -- (outer.wsw -| couter.west) node [obs] {};
      \draw (outer.east) -- (outer.east -| couter.east) node [obs] {};
    \end{tikzpicture}
    \caption{Subjecting configuration\; \tikz[baseline=-0.1cm]{
        \node [node] (c) {\(\mc{C}\)};
        \node [fit={($(c.south west)+(-0.3,-0.3)$) ($(c.north east)+(0.3,0.3)$)}, inner sep=0cm] (rect) {};
        \useasboundingbox (rect);
        \draw (c.ene) -- (c.ene -| rect.east);
        \draw (c.ese) -- (c.ese -| rect.east);
        \draw (c.wnw) -- (c.wnw -| rect.west);
        \draw (c.west) -- (c.west -| rect.west);
        \draw (c.wsw) -- (c.wsw -| rect.west);
      }\; to \((\mc{E}, \Psi)\)}
  \end{subfigure}
  \caption{Pictorial representations of experiments}
  \label{fig:sill-obs-equiv-main-equiv:1}
\end{figure}

To help build intuition about various observation systems, it will be useful to represent experiments pictorially using diagrams like those in \cref{fig:sill-obs-equiv-main-equiv:1}.
In particular, we represent channels as wires and contexts \(\ctxh{\mc{E}}{\Gamma}{\Delta}\) as configurations with a hole.
We mark observed channel \(\Psi\) with circles.
Subjecting a configuration \(\mc{C}\) to an experiment \(\mc{E}\) can be thought of as ``plugging'' \(\mc{C}\) into the hole and eavesdropping on marked wires.

Intuitively, given an observation system \(\mc{S}\), two configurations are observationally \(\mc{S}\)-similar if subjecting them to any experiment in \(\mc{S}\) produces similar communications.
Recall from \cref{def:sill-obs-equiv/relat-equiv:5} that the contextual interior \(\mkcg{\relfnt{R}}\) of a relation \(\relfnt{R}\) is the largest contextual relation contained in \(\mc{R}\).

\begin{definition}
  \label{def:sill-obs-equiv/main-equiv:2}
  Let $\mc{S} = (\mc{X}, {\leqslant})$ be an observation system and $\leqslant$ a preorder.
  \defin{Observational $\mc{S}$-simulation} is the type-indexed relation $\sobssim{\mc{S}}$\glsadd{sobssim} on configurations such that $\jtrelc{\sobssim{\mc{S}}}{\Gamma}{\mc{C}}{\mc{D}}{\Delta}$ if and only if for all experiments $(\jcfgt{\Lambda}{\ctxh{\mc{E}}{\Gamma}{\Delta}}{\Xi},\; \Psi) \in \mc{X}_{\Gamma \vdash \Delta}$,
  \[
    \obsbr{\jcfgt{\Lambda}{\ctxh[\mc{C}]{\mc{E}}{\Gamma}{\Delta}}{\Xi}}_\Psi \commsim[\leqslant] \obsbr{\jcfgt{\Lambda}{\ctxh[\mc{D}]{\mc{E}}{\Gamma}{\Delta}}{\Xi}}_\Psi.
  \]
  In this case, we say that $\mc{C}$ and $\mc{D}$ are observationally $\mc{S}$-similar.
  We call $\sobsprec{\mc{S}}$ \defin{observational $\mc{S}$-precongruence}.
\end{definition}

Observational $\mc{S}$-simulation is a preorder by \cref{prop:sill-obs-equiv/observ-comm-equiv:3}.
It follows from \cref{prop:sill-obs-equiv/bisimulations:2} that \(\sobsprec{\mc{S}}\) is indeed a precongruence, and that it is the greatest precongruence contained in \(\sobssim{\mc{S}}\).

\begin{example}
  \label{ex:sill-obs-equiv-main-equiv:1}
  We recast the observations of \cref{ex:sill-obs-equiv-main-obs:2} using observation systems and observational simulations.
  Consider an observation system \(\mc{S} = (\mc{X}, {\leqslant})\) such that
  \[
    \mc{X}_{\cdot \vdash c : \Tu \Tot \Tu} = \Set*{ \left(\ctxh{\left(\jproc{a}{\tClose{a}}, \jproc{b}{\Omega}\right)}{a : \Tu, b : \Tu}{c : \Tu \Tot \Tu},\; c : \Tu \Tot \Tu\right) }.
  \]
  Then \(\jtrelc{\sobssim{\mc{S}}}{\cdot}{\jproc{c}{\tSendC{c}{b}{\Omega}}}{\jproc{c}{\tSendC{c}{a}{\Omega}}}{c : \Tu \Tot \Tu}\) by \cref{ex:sill-obs-equiv-main-obs:2} and \((\bot, \bot) \commsim[\leqslant] (\cclose, \bot)\).
  However, it is not the case that \(\jtrelc{\sobssim{\mc{S}}}{\cdot}{\jproc{c}{\tSendC{c}{a}{\Omega}}}{\jproc{c}{\tSendC{c}{b}{\Omega}}}{c : \Tu \Tot \Tu}\).
\end{example}

We define observational $\mc{S}$-equivalence analogously to \(\mc{S}\)-simulation:

\begin{definition}
  \label{def:sill-obs-equiv/main-equiv:6}
  Let $\mc{S} = (\mc{X}, {\leqslant})$ be an observation system and $\leqslant$ a preorder.
  \defin{Observational $\mc{S}$-equivalence} is the type-indexed relation $\sobseq{\mc{S}}$\glsadd{sobseq} on configurations such that $\jtrelc{\sobseq{\mc{S}}}{\Gamma}{\mc{C}}{\mc{D}}{\Delta}$ if and only if for all experiments $(\jcfgt{\Lambda}{\ctxh{\mc{E}}{\Gamma}{\Delta}}{\Xi}, \Psi) \in \mc{X}_{\Gamma \vdash \Delta}$,
  \[
    \obsbr{\jcfgt{\Lambda}{\ctxh[\mc{C}]{\mc{E}}{\Gamma}{\Delta}}{\Xi}}_\Psi \commeq[\leqslant] \obsbr{\jcfgt{\Lambda}{\ctxh[\mc{D}]{\mc{E}}{\Gamma}{\Delta}}{\Xi}}_\Psi.
  \]
  In this case, we say that $\mc{C}$ and $\mc{D}$ are observationally $\mc{S}$-equivalent.
  We call $\sobscong{\mc{S}}$ \defin{observational $\mc{S}$-congruence}.
\end{definition}

Observational $\mc{S}$-equivalence is an equivalence relation by \cref{prop:sill-obs-equiv/observ-comm:1}.
Apart from instances in which we want to emphasize the symmetry of a result, we will not consider observational $\mc{S}$-equivalence.
This does not restrict the applicability of our results in light of \cref{prop:sill-obs-equiv/observ-comm-equiv:10}, and the general fact that ${\leqslant} \cap {\opr{\leqslant}}$ is an equivalence relation whenever $\leqslant$ is a preorder.

\begin{proposition}
  \label{prop:sill-obs-equiv/observ-comm-equiv:10}
  Let $\mc{S} = (\mc{X}, {\leqslant})$ be an observation system.
  $\mc{S}$-simulation and $\mc{S}$-equivalence and their (pre)congruences are related as follows:
  \begin{enumerate}
  \item $\jtrelc{\sobseq{\mc{S}}}{\Gamma}{\mc{C}}{\mc{D}}{\Xi}$ if and only if $\jtrelc{\sobssim{\mc{S}}}{\Gamma}{\mc{C}}{\mc{D}}{\Xi}$ and $\jtrelc{\sobssim{\mc{S}}}{\Gamma}{\mc{D}}{\mc{C}}{\Xi}$;
  \item $\jtrelc{\sobscong{\mc{S}}}{\Gamma}{\mc{C}}{\mc{D}}{\Xi}$ if and only if $\jtrelc{\sobsprec{\mc{S}}}{\Gamma}{\mc{C}}{\mc{D}}{\Xi}$ and $\jtrelc{\sobsprec{\mc{S}}}{\Gamma}{\mc{D}}{\mc{C}}{\Xi}$.
  \end{enumerate}
\end{proposition}

Though $\mc{S}$-simulation and $\mc{S}$-equivalence are defined using contexts, it is important to note that, unlike contextual-equivalence-style relations, they are stable under language extension.
Indeed, extending the process language with new language constructs does not affect $\mc{X}$'s ability to discriminate between pre-existing programs.
However, $\mc{S}$-precongruence and $\mc{S}$-congruence need not be stable under language extension.
This is because subjecting previously precongruent configurations to new language constructs could have effects that are discernible by $\mc{X}$.

There are three natural\footnote{We use ``natural'' in the same sense as in ``natural transformation'', where the choice of observed channels $\Psi$ is uniform across all contexts.} families of observation systems, and each reflects a different outlook on program testing.
The first is an ``external'' notion of experimentation, where we imagine experiments as black boxes into which we place configurations, and where the result of an experiment is reported on its exterior channels.
Processes are then similar if they produce similar reports under each experiment.
Because values in Polarized SILL are unobservable (they are all of function or quoted process type), we do not differentiate between values observed on experiments' exterior channels.
This gives:

\begin{definition}
  \label{def:sill-obs-equiv/main-equiv:3}
  The \defin{external observation system} $E$ is given by $(\mc{X}^E, \relfnt{U})$, where $\relfnt{U}$\glsadd{unirel} is the universal relation and $\mc{X}^E$ is the family
  \[
    \mc{X}^E_{\Gamma \vdash \Delta} = \Set{ (\ctxh{\mc{E}}{\Gamma}{\Delta}, \Lambda\Xi) \given \jcfgt{\Lambda}{\ctxh{\mc{E}}{\Gamma}{\Delta}}{\Xi} }.\qedhere
  \]
\end{definition}

\begin{figure}
  \centering
  \begin{subfigure}{0.4\linewidth}
    \centering
    \begin{tikzpicture}[scale=0.5, thick, baseline]
      \draw [ctxt, very thick, even odd rule] (0,0) rectangle (9,5) (1,1) rectangle (8,4);
      \node [fit={(0,0) (9,5)}, inner sep=0] (outer) {};
      \node [fit={(1,1) (8,4)}, inner sep=0] (inner) {};
      \node [fit={(3,1) (6,4)}, inner sep=0] (cinner) {};
      \node [fit={(-2,0) (11,5)}, inner sep=0] (couter) {};
      \node at (4.5,0.5) {\(\mc{E}\)};

      \draw (inner.wnw) -- (inner.wnw -| cinner.west);
      \draw (inner.west) -- (inner.west -| cinner.west);
      \draw (inner.wsw) -- (inner.wsw -| cinner.west);
      \draw (inner.ene) -- (inner.ene -| cinner.east);
      \draw (inner.ese) -- (inner.ese -| cinner.east);

      \draw (outer.wnw) -- (outer.wnw -| couter.west) node [obs] {};
      \draw (outer.west) -- (outer.west -| couter.west) node [obs] {};
      \draw (outer.wsw) -- (outer.wsw -| couter.west) node [obs] {};
      \draw (outer.east) -- (outer.east -| couter.east) node [obs] {};
    \end{tikzpicture}
    \caption{A representative experiment \((\mc{E}, \Psi) \in \mc{X}^E\)}
  \end{subfigure}
  \quad
  \begin{subfigure}{0.4\linewidth}
    \centering
    \begin{tikzpicture}[scale=0.5, thick, baseline]
      \draw [ctxt, very thick, even odd rule] (0,0) rectangle (9,5) (1,1) rectangle (8,4);
      \node [fit={(0,0) (9,5)}, inner sep=0] (outer) {};
      \node [fit={(1,1) (8,4)}, inner sep=0] (inner) {};
      \node [fit={(3,1) (6,4)}, inner sep=0] (cinner) {};
      \node [fit={(-2,0) (11,5)}, inner sep=0] (couter) {};
      \node at (4.5,0.5) {\(\mc{E}\)};

      \draw (inner.wnw) -- (inner.wnw -| cinner.west) node [obs] {};
      \draw (inner.west) -- (inner.west -| cinner.west) node [obs] {};
      \draw (inner.wsw) -- (inner.wsw -| cinner.west) node [obs] {};
      \draw (inner.ene) -- (inner.ene -| cinner.east) node [obs] {};
      \draw (inner.ese) -- (inner.ese -| cinner.east) node [obs] {};

      \draw (outer.wnw) -- (outer.wnw -| couter.west);
      \draw (outer.west) -- (outer.west -| couter.west);
      \draw (outer.wsw) -- (outer.wsw -| couter.west);
      \draw (outer.east) -- (outer.east -| couter.east);
    \end{tikzpicture}
    \caption{A representative experiment \((\mc{E}, \Psi) \in \mc{X}^I\)}
  \end{subfigure}
  \caption{Pictorial representations of experiments in external and internal observation systems}
  \label{fig:sill-obs-equiv-main-equiv:2}
\end{figure}
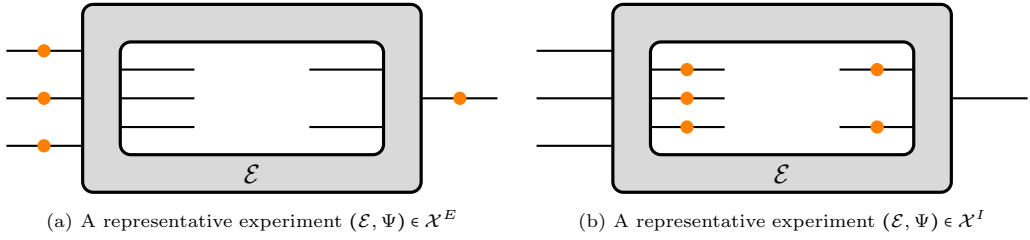

Alternatively, we could adopt an ``internal'' view of experimentation, where experiments question their subjects and we observe their answers.
This approach is reminiscent of the process equivalence \textcite{darondeau_1982:_enlar_defin_compl} gave to a calculus inspired by CCS.
This internal view is well-suited to synchronous settings like Classical Processes~\cite{wadler_2014:_propos_as_session} where processes cannot communicate unless we give them communication partners.
Accordingly, it is the approach \textcite{atkey_2017:_obser_commun_seman_class_proces} took when defining process equivalence for CP.
\Cref{fig:sill-obs-equiv-main-equiv:2} pictorially contrasts the internal and external views of experimentation.

\begin{definition}
  \label{def:sill-obs-equiv/main-equiv:5}
  The \defin{internal observation system} $I$ is given by $(\mc{X}^I, \relfnt{U})$, where $\relfnt{U}$ is the universal relation and $\mc{X}^I$ is the family
  \[
    \mc{X}^I_{\Gamma \vdash \Delta} = \Set{ (\ctxh{\mc{E}}{\Gamma}{\Delta}, \Gamma\Delta) \given \jcfgt{\Lambda}{\ctxh{\mc{E}}{\Gamma}{\Delta}}{\Xi} }.\qedhere
  \]
\end{definition}

The final approach takes a ``total'' view on experimentation, where we observe communications on all channels in the experimental context.
We use it strictly as a technical tool for relating observation systems.
We will see that total observational simulation implies both internal and external observational simulations.

\begin{definition}
  \label{def:sill-obs-equiv/main-equiv:4}
  The \defin{(strict) total observation system} $T$ is given by $(\mc{X}^T, {=})$, where $\mc{X}^T$ is the family
  \[
    \mc{X}^T_{\Gamma \vdash \Delta} = \Set{ (\ctxh{\mc{E}}{\Gamma}{\Delta}, \Lambda \Iota \Xi) \given \jcfgti{\Lambda}{\Iota}{\ctxh{\mc{E}}{\Gamma}{\Delta}}{\Xi} }.\qedhere
  \]
\end{definition}

We respectively call $\sobssim{E}$\glsadd{sobssime}, $\sobssim{I}$\glsadd{sobssimi}, and $\sobssim{T}$\glsadd{sobssimt} \defin{external}, \defin{internal}, and \defin{total observational simulation}.
Analogously, we respectively call \(\sobseq{E}\), \(\sobseq{I}\), and \(\sobseq{T}\) \defin{external}, \defin{internal}, and \defin{total observational equivalence}.

Before considering general properties of simulations, we consider some examples of similar and dissimilar configurations.
Our examples use total observational equivalence, but we can deduce external and internal equivalence by simply observing fewer channels (this fact is made explicit by \cref{prop:sill-obs-equiv/total-observ:1}).

\begin{example}
  \label{ex:sill-obs-equiv-main-equiv:4}
  We revisit \cref{ex:sill-obs-equiv-main-obs:2}.
  Write \(\mc{L}\) and \(\mc{R}\), respectively, for the configurations
  \begin{gather*}
    \jcfgt{a : \Tamp \{ l : \Tus{\Tu} \}}{\jproc{b}{\tSendL{a}{l}{\tSendL{b}{r}{\Omega}}}}{b : \Tplus \{r : \Tu\}}\\
    \jcfgt{a : \Tamp \{ l : \Tus{\Tu} \}}{\jproc{b}{\tSendL{b}{r}{\tSendL{a}{l}{\Omega}}}}{b : \Tplus \{r : \Tu\}}.
  \end{gather*}
  We show that they are total observationally equivalent.
  To do so, consider some arbitrary experiment \((\ctxh{\mc{E}}{a : \Tamp \{ l : \Tu\}}{b : \Tplus \{r : \Tu\}}, \Gamma\Lambda\Delta) \in \mc{X}^T_{a : \Tamp \{ l : \Tu \} \vdash b : \Tplus \{ r : \Tu \}}\) where \(\jcfgti{\Gamma}{\Lambda}{\ctxh{\mc{E}}{}{}}{\Delta}\) (we elide the annotations on the hole for clarity).
  We must show that \( \obsbr{\jcfgti{\Gamma}{\Lambda}{\ctxh[\mc{L}]{\mc{E}}{}{}}{\Delta}}_{\Gamma\Lambda\Delta} \sobseq{\relfnt{U}} \obsbr{\jcfgti{\Gamma}{\Lambda}{\ctxh[\mc{R}]{\mc{E}}{}{}}{\Delta}}_{\Gamma\Lambda\Delta} \).
  Observed communications are independent of our choice of fair execution, so we construct fair executions of our choosing.
  Consider the fair executions of \(\ctxh[\mc{L}]{\mc{E}}{}{}\) and \(\ctxh[\mc{R}]{\mc{E}}{}{}\) that respectively start
  \begin{gather*}
    \ctxh[\mc{L}]{\mc{E}}{}{} \msstep \ctxh[\jmsg{a'}{\mSendLN{a}{l}{a'}}, \jproc{b}{\tSendL{b}{r}{\Omega}}]{\mc{E}}{}{}
    \msstep \ctxh[\jmsg{a'}{\mSendLN{a}{l}{a'}}, \jproc{b'}{\Omega}, \jmsg{b}{\mSendLP{b}{r}{b'}}]{\mc{E}}{}{}\\
    \ctxh[\mc{R}]{\mc{E}}{}{} \msstep \ctxh[\jproc{b'}{\tSendL{b}{r}{\Omega}}, \jmsg{b}{\mSendLP{b}{r}{b'}}]{\mc{E}}{}{}
    \msstep \ctxh[\jmsg{a'}{\mSendLN{a}{l}{a'}}, \jproc{b'}{\Omega}, \jmsg{b}{\mSendLP{b}{r}{b'}}]{\mc{E}}{}{}
  \end{gather*}
  and have the same arbitrary fair tail.
  Both executions have the same message facts, and so induce the same observed communications on all channels (including all channels in \(\Gamma \Lambda \Delta\)).
  This gives the result.
\end{example}

\begin{example}
  \label{ex:sill-obs-equiv-main-equiv:3}
  We revisit \cref{ex:sill-background-typing-mult-rewr:6}.
  We show that configurations \(\jcfgt{a : \Tu}{\jproc{b}{\tFwdP{a}{b}}}{b : \Tu}\) and \(\jcfgt{a : \Tu}{\jproc{b}{\tWait{a}{\tClose{b}}}}{b : \Tu}\) are total observationally equivalent.
  Consider some arbitrary experiment \((\ctxh{\mc{E}}{a : \Tu}{b : \Tu}, \Gamma\Lambda\Delta) \in \mc{X}^T_{a : \Tu \vdash b : \Tu}\) where \(\jcfgti{\Gamma}{\Lambda}{\ctxh{\mc{E}}{a : \Tu}{b : \Tu}}{\Delta}\).
  We must show that
  \begin{equation}
    \label{eq:sill-obs-equiv-main-equiv:1}
    \obsbr{\jcfgti{\Gamma}{\Lambda}{\ctxh[\jproc{b}{\tFwdP{a}{b}}]{\mc{E}}{a : \Tu}{b : \Tu}}{\Delta}}_{\Gamma\Lambda\Delta} =
    \obsbr{\jcfgti{\Gamma}{\Lambda}{\ctxh[\jproc{b}{\tWait{a}{\tClose{b}}}]{\mc{E}}{a : \Tu}{b : \Tu}}{\Delta}}_{\Gamma\Lambda\Delta}.
  \end{equation}
  A case analysis on the multiset rewriting rules shows that the configurations in the hole step if and only if \(\mc{E} \mssteps \mc{E}', \jmsg{a}{\mClose{a}}\) for some \(\mc{E}'\).
  Assume first that this is the case.
  Then
  \begin{align*}
    &\ctxh[\jproc{b}{\tFwdP{a}{b}}]{\mc{E}}{}{} \mssteps \mc{E}', \jmsg{a}{\mClose{a}}, \jproc{b}{\tFwdP{a}{b}} \msstep \mc{E}', \jmsg{b}{\mClose{b}}\text{ and }\\
    &\ctxh[\jproc{b}{\tWait{a}{\tClose{b}}}]{\mc{E}}{}{} \mssteps \mc{E}', \jmsg{a}{\mClose{a}}, \jproc{b}{\tWait{a}{\tClose{b}}} \msstep {} \\
    &\qquad\qquad\qquad\qquad{} \msstep \mc{E}', \jproc{b}{\tClose{b}} \msstep \mc{E}', \jmsg{b}{\mClose{b}}.
  \end{align*}
  We remark that both traces have the same message facts, so we conclude \cref{eq:sill-obs-equiv-main-equiv:1}.
  Otherwise, if \(\mc{E}\) never steps to a configuration containing \(\jmsg{a}{\mClose{a}}\), then the configurations on both sides of \cref{eq:sill-obs-equiv-main-equiv:1} have equal traces (modulo the unused facts).
  We again conclude \cref{eq:sill-obs-equiv-main-equiv:1}.
\end{example}

Internal and external observational simulation do not coincide because of value transmission and our choice to compare functional values using the universal relation $\unirel$.
As described above, we use $\unirel$ because of philosophical objections to inspecting values of function type.
However, this means that we cannot discriminate between sent values until we try to use them.
In the internal observational simulation case, this use necessarily occurs after we have already observed the value as part of a communication between configuration and experiment, while in the external case we can observe the effects of experiments using the value.
This distinction is illustrated in the proof of \cref{prop:sill-obs-equiv-main-equiv:1}.
We conjecture that replacing \(\unirel\) by a suitable refinement might cause the two simulations to coincide.

\begin{proposition}
  \label{prop:sill-obs-equiv-main-equiv:1}
  Internal and external observational simulation do not coincide.
\end{proposition}

\begin{proof}
  We construct a counter-example.
  Let the processes $P$ and $Q$ respectively be:
  \begin{align*}
    \jtypem{\cdot}{\cdot}{&\tSendV{c}{(\lambda x : \tau.x)}{\tClose c}}{c}{\Tand{(\tau \to \tau)}{\Tu}},\\
    \jtypem{\cdot}{\cdot}{&\tSendV{c}{(\lambda x : \tau. \tFix{y}{y})}{\tClose c}}{c}{\Tand{(\tau \to \tau)}{\Tu}}.
  \end{align*}
  Observe that for all experiments $(\jcfgt{\Lambda}{\ctxh{\mc{E}}{}{c : \Tand{(\tau \to \tau)}{\Tu}}}{\Xi}, c) \in \mc{X}^I$,
  \begin{align*}
    &\obsbr{\jcfgt{\Lambda}{\ctxh[\jproc{c}{P}]{\mc{E}}{}{c : \Tand{(\tau \to \tau)}{\Tu}}}{\Xi}}_c\\
    &= (c : (\cval{\lambda x : \tau.x}, \cclose))\\
    &\commsim[\unirel] (c : (\cval{\lambda x : \tau.\tFix{y}{y}}, \cclose))\\
    &= \obsbr{\jcfgt{\Lambda}{\ctxh[\jproc{c}{Q}]{\mc{E}}{}{c : \Tand{(\tau \to \tau)}{\Tu}}}{\Xi}}_c.
  \end{align*}
  It follows that \(\jtrelc{\sobssim{I}}{\cdot}{\jproc{c}{P}}{\jproc{c}{Q}}{c : \Tand{(\tau \to \tau)}{\Tu}}\).

  We show that it is not the case that $\jtrelc{\sobssim{E}}{\cdot}{\jproc{c}{P}}{\jproc{c}{Q}}{c : \Tand{(\tau \to \tau)}{\Tu}}$.
  Take $\tau$ to be $\rho \to \rho$ for some $\rho$, and consider the process $R$ given by
  \[
    \jtypem{\cdot}{c : \Tand{(\tau \to \tau)}{\Tu}}{\tRecvV{x}{c}{\tSendV{b}{(x(\lambda z : \rho.z))}{\tFwdP{c}{b}}}}{b}{\Tand{\tau}{\Tu}}.
  \]\
  Set $\ctxh{\mc{C}}{}{c : \Tand{(\tau \to \tau)}{\Tu}} = \ctxh{}{}{c : \Tand{(\tau \to \tau)}{\Tu}}, \jproc{b}{R}$.
  The experiment $(\ctxh{\mc{C}}{}{b : \Tand{\tau}{\Tu}}, b) \in \mc{X}^E$ can differentiate $\jproc{c}{P}$ and $\jproc{c}{Q}$.
  Intuitively, this is because $(\lambda x : \tau.x)(\lambda z : \rho.z)$ will converge in the first case, but the term $(\lambda x : \tau. \tFix{y}{y})(\lambda z : \rho.z)$ will diverge in the second.
  Explicitly,
  \begin{align*}
    &\obsbr{\jcfgt{\cdot}{\ctxh[\jproc{c}{P}]{\mc{C}}{}{}}{b : \Tand{(\tau \to \tau)}{\Tu}}}_b\\
    &= (b : (\cval (\lambda z : \rho.z), \cclose))\\
    &\ncommsim[\unirel] (b : \bot)\\
    &= \obsbr{\jcfgt{\cdot}{\ctxh[\jproc{c}{Q}]{\mc{C}}{}{}}{b : \Tand{(\tau \to \tau)}{\Tu}}}_b.
  \end{align*}
  This gives the result.
\end{proof}

Our observation system framework lets us prove general properties about observational simulations.
For example, we can use the fact that experiments are written in the same language as configurations to give an easy check that observational simulations are precongruences.
We say that an experiment set $\mc{X}$ is \defin{closed under composition with contexts} if for all contexts $\jcfgt{\Lambda}{\ctxh{\mc{C}}{\Gamma}{\Delta}}{\Xi}$ and all experiments $(\ctxh{\mc{E}}{\Lambda}{\Xi}, \Psi) \in \mc{X}$, there exists a $\Psi' \supseteq \Psi$ such that $(\ctxh[{\ctxh{\mc{C}}{\Gamma}{\Delta}}]{\mc{E}}{\Lambda}{\Xi}, \Psi') \in \mc{X}$.

\begin{proposition}
  \label{prop:sill-obs-equiv/main-equiv:2}
  Let $\mc{S} = (\mc{X}, {\leqslant})$ be an observation system.
  If $\mc{X}$ is closed under composition with contexts, then $\sobssim{\mc{S}}$ is a precongruence, \ie, $\jtrelc{\sobssim{\mc{S}}}{\Gamma}{\mc{C}}{\mc{D}}{\Delta}$ if and only if $\jtrelc{\mkcg{\left(\sobssim{\mc{S}}\right)}}{\Gamma}{\mc{C}}{\mc{D}}{\Delta}$.
\end{proposition}

\begin{proof}
  Assume that $\jtrelc{\sobssim{\mc{S}}}{\Gamma}{\mc{C}}{\mc{D}}{\Delta}$.
  Let $\jcfgt{\Lambda}{\ctxh{\mc{F}}{\Gamma}{\Delta}}{\Xi}$ be an arbitrary context and let $(\jcfgt{\Phi}{\ctxh{\mc{E}}{\Lambda}{\Xi}}{\Pi}, \Psi) \in \mc{X}$ be an arbitrary experiment.
  We must show that
  \[
    \obsbr{\jcfgt{\Phi}{\ctxh[{\ctxh[\mc{C}]{\mc{F}}{\Gamma}{\Delta}}]{\mc{E}}{\Lambda}{\Xi}}{\Pi}}_\Psi
    \commsim[\leqslant]
    \obsbr{\jcfgt{\Phi}{\ctxh[{\ctxh[\mc{D}]{\mc{F}}{\Gamma}{\Delta}}]{\mc{E}}{\Lambda}{\Xi}}{\Pi}}_\Psi.
  \]
  By closure under composition with contexts, $(\jcfgt{\Phi}{\ctxh[{\ctxh{\mc{F}}{\Gamma}{\Delta}}]{\mc{E}}{\Lambda}{\Xi}}{\Pi}, \Psi') \in \mc{X}$ for some $\Psi' \supseteq \Psi$.
  Then \(\jtrelc{\sobssim{\mc{S}}}{\Gamma}{\mc{C}}{\mc{D}}{\Delta}\) implies
  \[
    \obsbr{\jcfgt{\Phi}{\ctxh[{\ctxh[\mc{C}]{\mc{F}}{\Gamma}{\Delta}}]{\mc{E}}{\Lambda}{\Xi}}{\Pi}}_{\Psi'}
    \commsim[\leqslant]
    \obsbr{\jcfgt{\Phi}{\ctxh[{\ctxh[\mc{D}]{\mc{F}}{\Gamma}{\Delta}}]{\mc{E}}{\Lambda}{\Xi}}{\Pi}}_{\Psi'}.
  \]
  The result is immediate from the fact that $\Psi' \supseteq \Psi$.
\end{proof}

\begin{corollary}
  \label{prop:sill-obs-equiv/total-observ:3}
  Total observational simulation $\tocssim$ and external observational simulation $\eocssim$ are precongruences.
\end{corollary}

Internal, external, and total observational simulations and equivalences are defined by universal quantification on infinite sets of experiments.
However, to establish similarity or equivalence, it is sufficient to quantify over a restricted subset of experiments, a result reminiscent of Milner's ``context lemma''~\cite{milner_1977:_fully_abstr_model_typed}.
The experiments in this restricted collection are those with simply branched configuration contexts.
Simply branched configuration contexts are the context analogue of simply branched configurations (\cref{def:sill-obs-equiv/observ-comm:9}), where every instance of \getrn{conf-c} composes configurations along exactly one channel.
The definition of simply branched contexts is subtle: given a simply branched configuration context, the result of filling its hole should again be simply branched.
However, this need not always be the case for contexts satisfying the definition of a simply branched configuration: the context $\ctxh{}{\Gamma}{a : A, b : B}$ satisfies \cref{def:sill-obs-equiv/observ-comm:9}, but \cref{prop:sill-obs-equiv/observ-comm:3} implies that for no $\mc{C}$ is $\ctxh[\mc{C}]{}{\Gamma}{a : A, b : B}$ simply branched.
Instead, we use the characterization of \cref{prop:sill-obs-equiv/observ-comm:3} to define simply branched configuration contexts:

\begin{definition}
  \label{def:sill-background-relat-equiv:2}
  A configuration context $\jcfgt{\Lambda}{\ctxh{\mc{B}}{\Gamma}{\Delta}}{\Xi}$ is \defin{simply branched}\varindex{configuration context {simply branched}}{1!~2!3= 2!1~!3=}[|defin] if $\Xi$ contains exactly one channel.
\end{definition}

We will use the following \namecref{lemma:sill-obs-equiv/total-observ:1} to replace arbitrary experiments with simply branched experiments that induce the same observations:

\begin{lemma}
  \label{lemma:sill-obs-equiv/total-observ:1}
  Let $(\jcfgti{\Lambda}{\Iota}{\ctxh{\mc{E}}{\Gamma}{\Delta}}{\Xi}, \Psi)$ be an arbitrary experiment.
  There exists a simply branched context $\jcfgti{\Lambda}{\Iota\,\Xi}{\ctxh{\mc{B}}{\Gamma}{\Delta}}{a : \Tu}$ such that for all $\jcfgt{\Gamma}{\mc{C}}{\Delta}$,
  \begin{align*}
    \obsbr{\jcfgti{\Lambda}{\Iota\,\Xi}{\ctxh[\mc{C}]{\mc{B}}{\Gamma}{\Delta}}{a : \Tu}}_{\Psi} &= \obsbr{\jcfgti{\Lambda}{\Iota}{\ctxh[\mc{C}]{\mc{E}}{\Gamma}{\Delta}}{\Xi}}_\Psi,\\
    \obsbr{\jcfgti{\Lambda}{\Iota\,\Xi}{\ctxh[\mc{C}]{\mc{B}}{\Gamma}{\Delta}}{a : \Tu}}_{a} &= (a : \bot).
  \end{align*}
\end{lemma}

\begin{proof}
  Recall from \cref{ex:sill-background-typing-mult-rewr:2} that there exists a divergent process $\jtypem{\cdot}{\Xi}{\Omega}{c}{\Tu}$.
  Let $\jcfgti{\Lambda}{\Iota\,\Xi}{\ctxh{\mc{B}}{\Gamma}{\Delta}}{a : \Tu}$ be given by the composition $\ctxh{\mc{E}}{\Gamma}{\Delta}, \jproc{c}{\Omega}$.
  Let $\jcfgt{\Gamma}{\mc{C}}{\Delta}$ be arbitrary.
  Let $T$ be an arbitrary fair trace of $\jcfgti{\Lambda}{\Iota}{\ctxh[\mc{C}]{\mc{E}}{\Gamma}{\Delta}}{\Xi}$.
  Let $T'$ be the trace of $\jcfgti{\Lambda}{\Iota\,\Xi}{\ctxh[\mc{C}]{\mc{B}}{\Gamma}{\Delta}}{a : \Tu}$ given by interleaving each step of $T$ with an application of \cref{eq:sill:msr-e-qt} to $\jproc{c}{\Omega}$.
  It is fair.
  The two traces are union equivalent, so they have the same observed communications on all channels.
  In particular, no message is sent on $a$, so $\jtoc{T'}{a}{\bot}{\Tu}$.
\end{proof}

\Cref{cor:sill-obs-equiv/main-equiv:1} implies that to show configurations are internal or total observationally similar, it is sufficient to consider only experiments with simply branched contexts:

\begin{proposition}
  \label{cor:sill-obs-equiv/main-equiv:1}
  Let $\mc{S} = (\mc{X}, {\leqslant})$ range over observation systems $I$ and $T$.
  Let $\mc{S}_B = (\mc{X}^B, {\leqslant})$ be its restriction to simply branched contexts, where $\mc{X}^B = \Set{ (\mc{E}, \Psi) \in \mc{X} \given \mc{E} \text{ is simply branched}\, }$.
  Then ${\jtrelc{\sobssim{\mc{S}}}{\Gamma}{\mc{C}}{\mc{D}}{\Delta}}$ if and only if $\jtrelc{\sobssim{\mc{S}_B}}{\Gamma}{\mc{C}}{\mc{D}}{\Delta}$.
\end{proposition}

\begin{proof}
  Sufficiency is immediate for both properties: every simply branched experiment is an experiment.
  To see necessity, assume that $\jtrelc{\sobssim{\mc{S}_B}}{\Gamma}{\mc{C}}{\mc{D}}{\Delta}$.
  Let $(\jcfgt{\Lambda}{\ctxh{\mc{E}}{\Gamma}{\Delta}}{\Xi}, C) \in \mc{X}$ be an arbitrary experiment.
  We must show that
  \[
    \obsbr{\jcfgt{\Lambda}{\ctxh[\mc{C}]{\mc{E}}{\Gamma}{\Delta}}{\Xi}}_C \commsim[\leqslant] \obsbr{\jcfgt{\Lambda}{\ctxh[\mc{D}]{\mc{E}}{\Gamma}{\Delta}}{\Xi}}_C.
  \]
  By \cref{lemma:sill-obs-equiv/total-observ:1}, there exists a simply branched experiment $(\jcfgt{\Lambda}{\ctxh{\mc{B}}{\Gamma}{\Delta}}{a : \Tu}, D) \in \mc{X}$ with $C \subseteq D$.
  In either case,
  \[
    \obsbr{\jcfgt{\Lambda}{\ctxh[\mc{C}]{\mc{B}}{\Gamma}{\Delta}}{a : \Tu}}_D \commsim[\leqslant] \obsbr{\jcfgt{\Lambda}{\ctxh[\mc{D}]{\mc{B}}{\Gamma}{\Delta}}{a : \Tu}}_D
  \]
  by assumption.
  Because $C \subseteq D$,
  \[
    \obsbr{\jcfgt{\Lambda}{\ctxh[\mc{C}]{\mc{B}}{\Gamma}{\Delta}}{a : \Tu}}_C \commsim[\leqslant] \obsbr{\jcfgt{\Lambda}{\ctxh[\mc{D}]{\mc{B}}{\Gamma}{\Delta}}{a : \Tu}}_C.
  \]
  By the \namecref{lemma:sill-obs-equiv/total-observ:1},
  \[
    \obsbr{\jcfgt{\Lambda}{\ctxh[\mc{C}]{\mc{E}}{\Gamma}{\Delta}}{\Xi}}_C \commsim[\leqslant] \obsbr{\jcfgt{\Lambda}{\ctxh[\mc{D}]{\mc{E}}{\Gamma}{\Delta}}{\Xi}}_C.\qedhere
  \]
\end{proof}

Recall from \cref{prop:sill-obs-equiv/bisimulations:2} that a preorder on configurations is a precongruence if and only if it is contextual.
By restricting our attention to simply branched contexts (instead of all contexts), we get restricted notion of contextual relation that we call ``simply branched contextual'', and a similarly restricted notion of contextual interior.
Our context lemma (\cf~\cite[6--7]{milner_1977:_fully_abstr_model_typed}) states that, in the cases of external and total simulations, this restricted notion of contextual interior coincides with the usual notion of contextual interior.
This reduces our proof burden when showing that configurations are precongruent: we only need to quantify over simply branched contexts instead of over all contexts.
Explicitly:

\begin{definition}
  \label{def:sill-background-relat-equiv:3}
  A typed relation $\relfnt{R}$ on configurations is \defin{simply branched contextual} if $\jtrelc{\relfnt{R}}{\Gamma}{\mc{C}}{\mc{D}}{\Delta}$ implies that $\jtrelc{\relfnt{R}}{\Phi}{\ctxh[\mc{C}]{\mc{B}}{\Gamma}{\Delta}}{\ctxh[\mc{D}]{\mc{B}}{\Gamma}{\Delta}}{c : C}$ for all simply branched contexts $\jcfgt{\Phi}{\ctxh{\mc{B}}{\Gamma}{\Delta}}{c : C}$.
\end{definition}

\begin{definition}
  \label{def:sill-background-relat-equiv:1}
  The \defin{simply branched contextual interior}\varindex{relation simply-branched {contextual interior} context}{1!3!2= 4!3!2=}[|defin] $\mkbcg{\relfnt{R}}$\glsadd{mkbcg} of a typed relation $\relfnt{R}$ on configurations is the greatest simply branched contextual typed relation contained in $\relfnt{R}$.
\end{definition}

\begin{proposition}[Context Lemma]
  \label{prop:sill-obs-equiv/extern-observ:6}
  Let $\mc{S} = (\mc{X}, {\leqslant})$ range over observation systems $E$ and $T$.
  Then $\jtrelc{\sobsprec{\mc{S}}}{\Gamma}{\mc{C}}{\mc{D}}{a : A}$ if and only if $\jtrelc{\mkbcg{\left(\sobssim{\mc{S}}\right)}}{\Gamma}{\mc{C}}{\mc{D}}{a : A}$.
\end{proposition}

\begin{proof}
  Sufficiency is immediate: every simply branched context is a context.
  To see necessity, assume that $\jtrelc{\mkbcg{\left(\sobssim{\mc{S}}\right)}}{\Gamma}{\mc{C}}{\mc{D}}{a : A}$, \ie, that
  \begin{equation}
    \label{eq:sill-obs-equiv/main-equiv:1}
    \obsbr{\jcfgt{\Lambda}{\ctxh[{\ctxh[\mc{C}]{\mc{B}}{\Gamma}{a : A}}]{\mc{E}}{\Phi}{b : B}}{\Xi}}_C
    \commsim[\leqslant]
    \obsbr{\jcfgt{\Lambda}{\ctxh[{\ctxh[\mc{D}]{\mc{B}}{\Gamma}{a : A}}]{\mc{E}}{\Phi}{b : B}}{\Xi}}_C
  \end{equation}
  for all simply branched contexts $\jcfgt{\Phi}{\ctxh{\mc{B}}{\Gamma}{a : A}}{b : B}$ and experiments $(\jcfgt{\Lambda}{\ctxh{\mc{E}}{\Phi}{b : B}}{\Xi}, C) \in \mc{X}$.
  Let $\jcfgti{\Phi}{\Iota_F}{\ctxh{\mc{F}}{\Gamma}{a : A}}{\Psi}$ be an arbitrary context, and let $(\jcfgti{\Lambda}{\Iota_E}{\ctxh{\mc{E}}{\Phi}{\Psi}}{\Xi}, C) \in \mc{X}$ be an arbitrary experiment.
  We must show that
  \begin{equation}
    \label{eq:sill-obs-equiv/main-equiv:2}
    \obsbr{\jcfgt{\Lambda}{\ctxh[{\ctxh[\mc{C}]{\mc{F}}{\Gamma}{a : A}}]{\mc{E}}{\Phi}{\Psi}}{\Xi}}_C
    \commsim[\leqslant]
    \obsbr{\jcfgt{\Lambda}{\ctxh[{\ctxh[\mc{D}]{\mc{F}}{\Gamma}{a : A}}]{\mc{E}}{\Phi}{\Psi}}{\Xi}}_C.
  \end{equation}
  The context $\ctxh[{\ctxh{\mc{F}}{\Gamma}{a : A}}]{\mc{E}}{\Phi}{\Psi}$ appears as the experiments $(\ctxh[{\ctxh{\mc{F}}{\Gamma}{a : A}}]{\mc{E}}{\Phi}{\Psi}, \Lambda\Xi) \in \mc{X}^E$ and $(\ctxh[{\ctxh{\mc{F}}{\Gamma}{a : A}}]{\mc{E}}{\Phi}{\Psi}, \Lambda\Iota_F\Iota_E\Xi) \in \mc{X}^T$.
  Instantiating \eqref{eq:sill-obs-equiv/main-equiv:1} with $\mc{B} = \ctxh{}{\Gamma}{a : A}$ gives the result.
\end{proof}

In this paper, we focus on external observational simulation.
It is better behaved than internal observational simulation (it is a precongruence).
It is also, in some sense, easier to work with.
This is because we can use the fact that we never observe input on external channels.
As a result, the observed communications are simpler: they are never bidirectional.

In \cref{sec:total-obs-equiv:total-observ}, we show that total observational equivalence is closed under execution.
We show in \cref{sec:sill-obs-equiv:intern-observ} that internal observational precongruence implies external observational precongruence.
In \cref{sec:sill-obs-equiv:extern-observ}, we develop external observational simulation.
We show that external observational congruence coincides with barbed congruence.
\Cref{fig:sill-obs-equiv/extern-observ:1} of \cref{sec:sill-obs-equiv:extern-observ:relations} summarizes the relationships between these different relations on configurations.
We show in \cref{sec:sill-obs-equiv:proc-equiv} how to relate relations on configurations and relations on processes.
This will give certain precongruences for processes.

\subsection{Total Observations for Configurations}
\label{sec:total-obs-equiv:total-observ}

Total observational equivalence is useful for showing properties that hold for $\mc{S}$-simulations and $\mc{S}$-equivalences in general.
This is because total observational equivalence is the finest notion of equivalence based on observed communications.
We start by relating observational simulations in general:

\begin{proposition}
  \label{prop:sill-obs-equiv/total-observ:2}
  Let $\mc{S}_i = (\mc{X}_i, {\leqslant}_i)$ for $i = 1, 2$ be observation systems.
  Then $\jtrelc{\sobssim{\mc{S}_1}}{\Gamma}{\mc{C}}{\mc{D}}{\Delta}$ implies $\jtrelc{\sobssim{\mc{S}_2}}{\Gamma}{\mc{C}}{\mc{D}}{\Delta}$ whenever both:
  \begin{enumerate}
  \item for all $(\jcfgt{\Lambda}{\ctxh{\mc{E}}{\Gamma}{\Delta}}{\Xi}, \Psi) \in \mc{X}_2$, there exists $\Psi' \supseteq \Psi$ such that $(\jcfgt{\Lambda}{\ctxh{\mc{E}}{\Gamma}{\Delta}}{\Xi}, \Psi') \in \mc{X}_1$, and
  \item ${\leqslant}_1 \subseteq {\leqslant}_2$.
  \end{enumerate}
\end{proposition}

\begin{proof}
  Assume that $\jtrelc{\sobssim{\mc{S}_1}}{\Gamma}{\mc{C}}{\mc{D}}{\Delta}$.
  Let $(\jcfgt{\Lambda}{\ctxh{\mc{E}}{\Gamma}{\Delta}}{\Xi}, \Psi) \in \mc{X}_2$ be arbitrary.
  We must show that
  \[
    \obsbr{\jcfgt{\Lambda}{\ctxh[\mc{C}]{\mc{E}}{\Gamma}{\Delta}}{\Xi}}_{\Psi} \commsim[{\leqslant_2}] \obsbr{\jcfgt{\Lambda}{\ctxh[\mc{D}]{\mc{E}}{\Gamma}{\Delta}}{\Xi}}_{\Psi}.
  \]
  By assumption, $(\jcfgt{\Lambda}{\ctxh{\mc{E}}{\Gamma}{\Delta}}{\Xi}, \Psi') \in \mc{X}_1$ for some $\Psi' \supseteq \Psi$.
  This implies that
  \[
    \obsbr{\jcfgt{\Lambda}{\ctxh[\mc{C}]{\mc{E}}{\Gamma}{\Delta}}{\Xi}}_{\Psi'} \commsim[{\leqslant_1}] \obsbr{\jcfgt{\Lambda}{\ctxh[\mc{D}]{\mc{E}}{\Gamma}{\Delta}}{\Xi}}_{\Psi'}.
  \]
  Because $\Psi \subseteq \Psi'$,
  \[
    \obsbr{\jcfgt{\Lambda}{\ctxh[\mc{C}]{\mc{E}}{\Gamma}{\Delta}}{\Xi}}_{\Psi} \commsim[{\leqslant_1}] \obsbr{\jcfgt{\Lambda}{\ctxh[\mc{D}]{\mc{E}}{\Gamma}{\Delta}}{\Xi}}_{\Psi}.
  \]
  By monotonicity (\cref{prop:sill-obs-equiv/observ-comm-equiv:3}),
  \[
    \obsbr{\jcfgt{\Lambda}{\ctxh[\mc{C}]{\mc{E}}{\Gamma}{\Delta}}{\Xi}}_{\Psi} \commsim[{\leqslant_2}] \obsbr{\jcfgt{\Lambda}{\ctxh[\mc{D}]{\mc{E}}{\Gamma}{\Delta}}{\Xi}}_{\Psi}.\qedhere
  \]
\end{proof}

\begin{corollary}
  \label{prop:sill-obs-equiv/total-observ:1}
  Let $T = (\mc{X}^T, {=})$ be the total observation system (\cref{def:sill-obs-equiv/main-equiv:4}).
  For all observation systems $\mc{S} = (\mc{X}, {\leqslant})$,
  \begin{itemize}
  \item if $\jtrelc{\sobssim{T}}{\Gamma}{\mc{C}}{\mc{D}}{\Delta}$, then $\jtrelc{\sobssim{\mc{S}}}{\Gamma}{\mc{C}}{\mc{D}}{\Delta}$;
  \item if $\jtrelc{\sobssim{(\mc{X}^T, {\leqslant})}}{\Gamma}{\mc{C}}{\mc{D}}{\Delta}$, then $\jtrelc{\sobssim{\mc{S}}}{\Gamma}{\mc{C}}{\mc{D}}{\Delta}$.
  \end{itemize}
\end{corollary}

Total observational equivalence is closed under multiset rewriting:

\begin{proposition}
  \label{prop:sill-obs-equiv/observ-comm:8}
  If $\jcfgti{\Gamma}{\mc{C}}{\Delta}$ and $\mc{C} \msstep \mc{C}'$, then $\jtrelc{\strobsc}{\Gamma}{\mc{C}}{\mc{C}'}{\Delta}$.
\end{proposition}

\begin{proof}
  Consider an arbitrary experiment $(\jcfgti{\Phi}{\Iota}{\ctxh{\mc{E}}{\Gamma}{\Delta}}{\Lambda}, \Phi\Iota\Lambda)$, and let $T$ be a fair trace of $\ctxh[\mc{C}]{\mc{E}}{\Gamma}{\Delta}$.
  By fairness and preservation, we can assume without loss of generality that the first step of $T$ is $\ctxh[\mc{C}]{\mc{E}}{\Gamma}{\Delta} \msstep \ctxh[\mc{C}']{\mc{E}}{\Gamma}{\Delta}$.
  If every message fact in $\mc{C}$ appears in $\mc{C}'$, then we are done.
  Indeed, the tail $T'$ of $T$ is a fair trace of $\ctxh[\mc{C}']{\mc{E}}{\Gamma}{\Delta}$ by the fair tail property (\cref{prop:main:4}).
  Both $T$ and $T'$ have the same sets of message facts, so they induce the same observations for each channel in $\Phi, \Iota, \Delta$.

  Now assume that $\mc{C} \msstep \mc{C}'$ consumes a message fact, \ie, that there is some message fact $\jmsg{c}{m} \in \mc{C}$ that is not in $\mc{C}'$.
  We must show that it is not observable from any channel $d$ in $\Phi, \Iota, \Lambda$ in $T$, \ie, that it does not appear in any derivation of $\jtoc{T}{u}{d}{A}$ for $d$ in $\Phi, \Iota, \Lambda$.

  A case analysis on the rules defining $\jtoc{T}{u}{d}{A}$ shows that if $\jtoc{T}{v}{a}{B}$ appears as a premise of a rule, then
  \begin{enumerate}
  \item the rule is due to a $\jmsg{b}{m}$ with $a \in \freecn(\jmsg{b}{m})$, and
  \item the conclusion of the rule is of the form $\jtoc{T}{w}{c}{C}$, where $c = \carrcn(\jmsg{b}{m})$.
  \end{enumerate}

  Suppose to the contrary that $\jmsg{c}{m}$ is observable from some $d$ in $\Phi, \Iota, \Lambda$ in $T$.
  We proceed by induction on the height $h$ of $\jmsg{c}{m}$ in the derivation of observed communication.
  Set $a = \carrcn(\jmsg{c}{m})$.
  \begin{proofcases}
  \item[$h = 0$] Then $\jmsg{c}{m}$ is observable because $\carrcn(\jmsg{c}{m}) = d$ is in $\Phi, \Iota, \Lambda$.
    We have $\jmsg{c}{m} \notin \mc{C}'$ only if $\jmsg{c}{m}$ was in the active portion of the rule used to make the step.
    However, a case analysis on the rule (equivalently, \cite[Lemma~5.9.5]{kavanagh_2021:_commun_based_seman}) implies that $d \in \freecn(\jproc{b}{P})$ for some $\jproc{b}{P} \in \mc{C}$.
    This implies that $d$ is an internal channel of $\mc{C}$ by \cref{lemma:sill-background-properties-traces:5}, which in turn implies that $d$ is not in $\Phi, \Iota, \Gamma, \Delta, \Lambda$.
    This is a contradiction.
  \item[$h = h' + 1$] Assume the result for $h'$.
    Then $\jmsg{c}{m}$ appears at height $h$ in the derivation, and there is a $\jmsg{b}{m'}$ at height $h'$ in the derivation such that $a \in \freecn(\jmsg{b}{m'})$.
    Because $\jmsg{c}{m} \in \mc{C}$ but $\jmsg{c}{m} \notin \mc{C}'$, we know that $\jmsg{c}{m}$ was in the active multiset of the rule used to make the step $\mc{C} \msstep \mc{C}'$.
    The same argument as the base case implies that $a$ was an internal channel and that it does not appear free on the right side of the rule.
    By preservation, it follows that $a$ is not free in $\ctxh[\mc{C}']{\mc{E}}{\Gamma}{\Delta}$.
    It follows that $a$ cannot appear free in $\jmsg{b}{m'}$, a contradiction.

    To see that $a$ cannot appear free in $\jmsg{b}{m'}$, we consider two cases:
    \begin{proofcases}
    \item[$\jmsg{b}{m'} \in \mc{C}$] A case analysis on the rules shows that we also have $\jmsg{b}{m'} \in \mc{C}'$, a contradiction of $a \notin \freecn(\ctxh[\mc{C}']{\mc{E}}{\Gamma}{\Delta})$.
    \item[$\jmsg{b}{m'} \notin \mc{C}$] Then $\jmsg{b}{m'}$ must appear in some $\mc{C}''$ such that $\ctxh[\mc{C}']{\mc{E}}{\Gamma}{\Delta} \mssteps \mc{C}''$.
      But each free channel in $\jmsg{b}{m'} \in \mc{C}''$ is either already in $\ctxh[\mc{C}']{\mc{E}}{\Gamma}{\Delta}$, or it is freshly generated, so not in $\ctxh[\mc{C}]{\mc{E}}{\Gamma}{\Delta}$ or $\ctxh[\mc{C}']{\mc{E}}{\Gamma}{\Delta}$.
      Both of these possibilities contradict the assumption that $a \in \freecn(\jmsg{b}{m'})$.\qedhere
    \end{proofcases}
  \end{proofcases}
\end{proof}

The following proposition shows that forwarding has no observable effect on communications, and that it acts only to rename channels:

\begin{proposition}
  \label{prop:sill-obs-equiv/observ-comm-equiv:2}
  For all $\jcfgt{\Gamma}{\mc{C}}{\Delta, c : C}$ and $\jcfgt{\Gamma, c : A}{\mc{A}}{\Delta}$, respectively,
  \begin{enumerate}
  \item if $C$ is positive, then $\jtrelc{\strobsc}{\Gamma}{\subst{d}{c}{\mc{C}}}{\mc{C}, \jproc{d}{\tFwdP{c}{d}}}{\Delta, d : C}$;
  \item if $A$ is positive, then $\jtrelc{\strobsc}{\Gamma, d : A}{\subst{d}{c}{\mc{A}}}{\jproc{c}{\tFwdP{d}{c}}, \mc{A}}{\Delta}$;
  \item if $A$ is negative, then $\jtrelc{\strobsc}{\Gamma, d : A}{\subst{d}{c}{\mc{A}}}{\jproc{c}{\tFwdN{d}{c}}, \mc{A}}{\Delta}$;
  \item if $C$ is negative, then $\jtrelc{\strobsc}{\Gamma}{\subst{d}{c}{\mc{C}}}{\mc{C}, \jproc{d}{\tFwdN{c}{d}}}{\Delta, d : C}$.
  \end{enumerate}
\end{proposition}

\begin{proof}
  Assume first that $\jcfgt{\Gamma}{\mc{C}}{\Delta, c : C}$ and that $C$ is positive.
  Let $(\ctxh{\mc{E}}{\Gamma}{\Delta, c : C}, D)$ be an arbitrary experiment.
  Then $\ctxh[\mc{C}]{\mc{E}}{\Gamma}{\Delta, c : C} \mssteps \ctxh[ \jmsg{c}{m}]{\mc{E}'}{\Gamma'}{\Delta', c : C}$ for some $\mc{E}'$ if and only if
  \[
    \ctxh[\subst{d}{c}{\mc{C}}]{\left(\subst{d}{c}{\mc{E}}\right)}{\Gamma}{\Delta, d : C} \mssteps \ctxh[\jmsg{d}{\subst{d}{c}{m}}]{(\subst{d}{c}{\mc{E}'})}{\Gamma'}{\Delta', d : C},
  \]
  and an induction shows that this holds if and only if
  \[
    \ctxh[\mc{C}, \jproc{d}{\tFwdP{c}{d}}]{\left(\subst{d}{c}{\mc{E}}\right)}{\Gamma}{\Delta, d : C} \mssteps \ctxh[\jmsg{c}{m}, \jproc{d}{\tFwdP{c}{d}}]{\left(\subst{d}{c}{\mc{E}'}\right)}{\Gamma'}{\Delta', d : C}.
  \]
  But this last multiset in turn steps to $\ctxh[\jmsg{d}{\subst{d}{c}{m}}]{(\subst{d}{c}{\mc{E}'})}{\Gamma'}{\Delta', d : C}$ by \cref{eq:sill:msr-fwdp}:
  \[
    \ctxh[\jmsg{c}{m}, \jproc{d}{\tFwdP{c}{d}}]{\left(\subst{d}{c}{\mc{E}'}\right)}{\Gamma'}{\Delta', d : C} \msstep \ctxh[\jmsg{d}{\subst{d}{c}{m}}]{(\subst{d}{c}{\mc{E}'})}{\Gamma'}{\Delta', d : C}.
  \]
  So if this collection of logical equivalences hold, we are done by \cref{prop:sill-obs-equiv/observ-comm:8} and the fact that $\strobsc$ is reflexive.

  If $\jmsg{c}{m}$ appears in no fair trace of $\ctxh[\mc{C}]{\mc{E}}{\Gamma}{\Delta, c : C}$, then an induction shows that every trace of $\ctxh[\mc{C}]{\mc{E}}{\Gamma}{\Delta, c : C}$ is a trace of $\ctxh[\mc{C}, \jproc{d}{\tFwdP{c}{d}}]{\left(\subst{d}{c}{\mc{E}}\right)}{\Gamma}{\Delta, d : C}$ (modulo the presence of the forwarding process), and that they have the same sets of message facts.
  So they induce the same observations on all channels and we are done.

  The remaining cases are analogous.
\end{proof}

\subsection{Internal Observations for Configurations}
\label{sec:sill-obs-equiv:intern-observ}

\Textcite[79]{atkey_2017:_obser_commun_seman_class_proces} states without proof that his internal-style observational equivalence is a congruence.
Internal observational equivalence is not a congruence in our setting because of value transmission and our use of \(\unirel\) to relate functional values, \ie, for the same reasons internal and external observational simulation do not coincide (see the remarks preceding \cref{prop:sill-obs-equiv-main-equiv:1}).
We conjecture that replacing \(\unirel\) by a suitable refinement would cause internal observational equivalence to be a congruence.

\begin{proposition}
  \label{prop:sill-obs-equiv/intern-observ:2}
  Internal observational simulation (equivalence) is not a precongruence (congruence).
\end{proposition}

\begin{proof}
  We revisit the counter-example used to prove \cref{prop:sill-obs-equiv-main-equiv:1}.
  Let the processes $P$, $Q$, and \(R\) be
  \begin{align*}
    \jtypem{\cdot}{\cdot}{&\tSendV{c}{(\lambda x : \tau.x)}{\tClose c}}{c}{\Tand{(\tau \to \tau)}{\Tu}},\\
    \jtypem{\cdot}{\cdot}{&\tSendV{c}{(\lambda x : \tau. \tFix{y}{y})}{\tClose c}}{c}{\Tand{(\tau \to \tau)}{\Tu}}\\
    \jtypem{\cdot}{c : \Tand{(\tau \to \tau)}{\Tu}}{&\tRecvV{x}{c}{\tSendV{b}{(x(\lambda z : \rho.z))}{\tFwdP{c}{b}}}}{b}{\Tand{\tau}{\Tu}}
  \end{align*}
  respectively, and set $\ctxh{\mc{C}}{}{c : \Tand{(\tau \to \tau)}{\Tu}} = \ctxh{}{}{c : \Tand{(\tau \to \tau)}{\Tu}}, \jproc{b}{R}$.
  If internal observational simulation were a precongruence, then
  \[
    \obsbr{\jcfgt{\Lambda}{\ctxh[{\ctxh[\jproc{c}{P}]{\mc{C}}{\cdot}{c : \Tand{(\tau \to \tau)}{\Tu}}}]{\mc{E}}{}{b : \Tand{\tau}{\Tu}}}{\Xi}}_b \commsim[\unirel] \obsbr{\jcfgt{\Lambda}{\ctxh[{\ctxh[\jproc{c}{Q}]{\mc{C}}{\cdot}{c : \Tand{(\tau \to \tau)}{\Tu}}}]{\mc{E}}{}{b : \Tand{\tau}{\Tu}}}{\Xi}}_b
  \]
  for all experiments $(\jcfgt{\Lambda}{\ctxh{\mc{E}}{}{b : \Tand{(\tau \to \tau)}{\Tu}}}{\Xi}, b) \in \mc{X}^I$.
  However, by the proof of \cref{prop:sill-obs-equiv-main-equiv:1}, the experiment $(\ctxh{}{}{b : \Tand{\tau}{\Tu}}, b) \in \mc{X}^I$ can differentiate $\ctxh[\jproc{c}{P}]{\mc{C}}{\cdot}{c : \Tand{(\tau \to \tau)}{\Tu}}$ and $\ctxh[\jproc{c}{Q}]{\mc{C}}{\cdot}{c : \Tand{(\tau \to \tau)}{\Tu}}$.
  So $\iocssim$ is not a precongruence.
  It follows from \cref{prop:sill-obs-equiv/observ-comm-equiv:10} that $\icommeq$ is not a congruence.
\end{proof}

Internal-style precongruences imply their total counterparts:

\begin{proposition}
  \label{prop:sill-obs-equiv/intern-observ:1}
  Let $\mc{X}^I$ and $\mc{X}^T$ be the internal and total observation systems given by \cref{def:sill-obs-equiv/main-equiv:5} and \cref{def:sill-obs-equiv/main-equiv:4}, respectively.
  Then for all $\leqslant$, if $\jtrelc{\sobsprec{(\mc{X}^I, {\leqslant})}}{\Gamma}{\mc{C}}{\mc{D}}{\Delta}$, then $\jtrelc{\sobsprec{(\mc{X}^T, {\leqslant})}}{\Gamma}{\mc{C}}{\mc{D}}{\Delta}$.
\end{proposition}

\begin{proof}
  Let $\jcfgt{\Lambda}{\ctxh{\mc{F}}{\Gamma}{\Delta}}{\Xi}$ be an arbitrary context and $(\jcfgti{\Phi}{\Iota}{\ctxh{\mc{E}}{\Lambda}{\Xi}}{\Psi}, \Phi\Iota\Psi) \in \mc{X}^T$ an arbitrary experiment.
  We must show that
  \[
    \obsbr*{\jcfgt{\Phi}{\ctxh[{\ctxh[\mc{C}]{\mc{F}}{\Gamma}{\Delta}}]{\mc{E}}{\Lambda}{\Xi}}{\Psi}}_{\Phi\Iota\Psi}
    \commsim[\leqslant]
    \obsbr*{\jcfgt{\Phi}{\ctxh[{\ctxh[\mc{D}]{\mc{F}}{\Gamma}{\Delta}}]{\mc{E}}{\Lambda}{\Xi}}{\Psi}}_{\Phi\Iota\Psi}.
  \]
  This is the case if and only if for all $c : C \in \Phi,\Iota,\Xi$,
  \begin{equation}
    \label{eq:sill-obs-equiv/intern-observ:1}
    \obsbr*{\jcfgt{\Phi}{\ctxh[{\ctxh[\mc{C}]{\mc{F}}{\Gamma}{\Delta}}]{\mc{E}}{\Lambda}{\Xi}}{\Psi}}_{\Phi\Iota\Psi}(c)
    \commsim[\leqslant]
    \obsbr*{\jcfgt{\Phi}{\ctxh[{\ctxh[\mc{D}]{\mc{F}}{\Gamma}{\Delta}}]{\mc{E}}{\Lambda}{\Xi}}{\Psi}}_{\Phi\Iota\Psi}(c).
  \end{equation}
  Fix some arbitrary such $c : C$.
  Induction on $\jcfgti{\Phi}{\Iota}{\ctxh{\mc{E}}{\Lambda}{\Xi}}{\Psi}$ gives a decomposition of $\mc{E}$ as a composition of contexts $\ctxh{\mc{E}}{\Lambda}{\Xi} = \ctxh[{\ctxh{\mc{E}''}{\Lambda}{\Xi}}]{\mc{E}'}{\Lambda'}{\Xi'}$ such that $c : C \in \Lambda',\Xi'$.
  Observe that the composition $\ctxh[{\ctxh{\mc{F}}{\Gamma}{\Delta}}]{\mc{E}''}{\Lambda}{\Xi}$ is again a context, and that $(\ctxh{\mc{E}'}{\Lambda'}{\Xi'}, \Lambda'\Xi') \in \mc{X}^I$.
  Then by assumption,
  \[
    \obsbr*{\jcfgt{\Phi}{\ctxh[{\left(\ctxh[{\ctxh[\mc{C}]{\mc{F}}{\Gamma}{\Delta}}]{\mc{E}''}{\Lambda}{\Xi}\right)}]{\mc{E}'}{\Lambda'}{\Xi'}}{\Psi}}_{\Lambda'\Xi'}
    \commsim[\leqslant]
    \obsbr*{\jcfgt{\Phi}{\ctxh[{\left(\ctxh[{\ctxh[\mc{D}]{\mc{F}}{\Gamma}{\Delta}}]{\mc{E}''}{\Lambda}{\Xi}\right)}]{\mc{E}'}{\Lambda'}{\Xi'}}{\Psi}}_{\Lambda'\Xi'}.
  \]
  Because $c : C \in \Lambda', \Xi'$, this implies \eqref{eq:sill-obs-equiv/intern-observ:1} and we are done.
\end{proof}

Combining \cref{prop:sill-obs-equiv/intern-observ:1,prop:sill-obs-equiv/total-observ:1}, we conclude that internal observational precongruence implies external observational precongruence:

\begin{corollary}
  \label{cor:sill-obs-equiv/intern-observ:1}
  If $\jtrelc{\iocsprec}{\mc{C}}{\mc{D}}{\Delta}$, then $\jtrelc{\eocsprec}{\Gamma}{\mc{C}}{\mc{D}}{\Delta}$.
\end{corollary}

\subsection{External Observations for Configurations}
\label{sec:sill-obs-equiv:extern-observ}

We show that external observational precongruence coincides with weak barbed precongruence.
We first show some general properties about observations on external channels.

\begin{proposition}
  \label{prop:sill-obs-equiv/observ-comm:6}
  Assume that\/ $\jcfgt{\Gamma}{\mc{C}}{\Delta}$.
  We observe no communication on its input channels, \ie,  $\obsbr{\jcfgt{\Gamma}{\mc{C}}{\Delta}}(c) = \bot$ for all $c \in \inpcn(\mc{C})$.
\end{proposition}

\begin{proof}
  Let $c \in \inpcn(\mc{C})$ be arbitrary, let $T$ be a fair execution, and let $\mc{T}$ be the union of all facts appearing in $T$.
  Suppose to the contrary that some $\jmsg{d}{m} \in \mc{T}$ has $c$ as its carrier channel.
  Then by \cref{rem:sill-background-stat-prop-typed-config:1}, $c$ is an output channel of $\jmsg{d}{m}$.
  By preservation, the subformula property (\cref{prop:sill-background-stat-prop-typed-config:3}) and \cref{def:sill-background-stat-prop-typed-config:3}, $c$ must also be an output channel of $\jcfgt{\Gamma}{\mc{C}}{\Delta}$.
  But the sets of input and output channels are disjoint, so this is a contradiction.
  It follows that no message fact in $\mc{T}$ has $c$ as its carrier, so $v_c = \bot$ by \getrn{O-bot}.
\end{proof}

We can relate the observed communications of different configurations using their externally observable messages.
Recall from \cref{def:sill-obs-equiv/observ-comm:13} that a message fact $\jmsg{a}{m}$ is observable from $c$ in a trace $T$ if appears in the derivation of $\jtoc{T}{v}{c}{A}$.

\begin{definition}
  \label{def:sill-obs-equiv/observ-comm:6}
  Let \(T\) be a trace of $\jcfgt{\Gamma}{\mc{C}}{\Delta}$.
  A message fact $\jmsg{a}{m}$ is \defin{externally observable in $T$} if it is observable from some $c \in \check \Gamma, \check \Delta$ in $T$.
\end{definition}

\begin{definition}
  \label{def:sill-obs-equiv/extern-observ:5}
  Two configurations $\jcfgt{\Gamma}{\mc{C}}{\Delta}$ and $\jcfgt{\Gamma}{\mc{C}'}{\Delta}$ have the same externally observable message facts if for some fair traces $T$ and $T'$ of $\mc{C}$ and $\mc{C}'$, respectively, for all channels $c \in \check \Gamma, \check \Delta$, the sets of messages observable from $c$ in $T$ and in $T'$ are equal.
\end{definition}

\begin{proposition}
  \label{prop:sill-obs-equiv/observ-comm:10}
  If $\jcfgt{\Gamma}{\mc{C}}{\Delta}$ and $\jcfgt{\Gamma}{\mc{D}}{\Delta}$ have the same sets of externally observable message facts, then $\obsbr{\jcfgt{\Gamma}{\mc{C}}{\Delta}} = \obsbr{\jcfgt{\Gamma}{\mc{D}}{\Delta}}$.
\end{proposition}

\begin{proof}
  This is an immediate corollary of \cref{prop:sill-obs-equiv/observ-comm:5}.
\end{proof}

\subsubsection{Barbed Simulation and Precongruence}
\label{sec:sill-obs-equiv:extern-observ:barbs}

Barbed bisimulations and congruences~\cite{milner_sangiorgi_1992:_barbed_bisim, sangiorgi_1992:_expres_mobil_proces_algeb} are the canonical notion of equivalence for process calculi.
A barb is an observation predicate $\barb{}$ defined on terms in a calculus that specifies the most basic behavioural observable: the ability to perform an observable action.
When defining the barb predicate,
\begin{quote}
  the global observer [\ldots] can also recognize the production of an observable action, but in this case he cannot see neither the identity of the action produced nor the state reached. \cite[691]{milner_sangiorgi_1992:_barbed_bisim}
\end{quote}
Concretely, we follow \textcite[\S~3.2]{sangiorgi_1992:_expres_mobil_proces_algeb} and define barbs on a per-channel basis: the predicate $\barb[a]{(\cdot)}$ specifies the ability to perform an observable action on channel $a$.

This minimalist approach to defining barbs contrasts with some recent approaches~\cite{yoshida_2007:_linear_bisim, toninho_2015:_logic_found_session_concur_comput, kokke_2019:_better_late_than_never,} whose barbs distinguish between different kinds of actions.
For example, \textcite{yoshida_2007:_linear_bisim} found it necessary to observe which label was sent to ensure that barbed bisimulation was a congruence.
Following an argument by \textcite{atkey_2017:_obser_commun_seman_class_proces}, \Textcite{kokke_2019:_better_late_than_never} extended barbs to observe how a program used or disposed of non-linear resources (the minimalist approach can be recovered by dropping exponentials from their calculus).
By using a different barb for each kind of typed communication, \textcite[\S~6.2]{toninho_2015:_logic_found_session_concur_comput} was able to give a binary logical relation that was consistent by construction with barbed equivalence.

We prefer the minimalist approach for its conceptual simplicity and generality: it is calculus agnostic.
Instead of modifying the concept of a barb to ensure that barbed bisimulation is a congruence, we follow the original approach and extract ``barbed congruences'' from barbed bisimulations using contextual interiors~(\cf~\cites[Definition~8]{milner_sangiorgi_1992:_barbed_bisim}[Definition~3.2.6]{sangiorgi_1992:_expres_mobil_proces_algeb}).

\begin{definition}
  \label{def:sill-obs-equiv/barb-cont-congr:1}
  A \defin{barb} is the channel-indexed predicate $\barb[a]{(\cdot)}$\glsadd{barb} on processes and configurations inductively defined by:\\
  \newcommand{\tabitem}{\hspace{\leftMargin}\llap{\textbullet}~~}
  \hspace{\leftmargin}\begin{tabular}{ll}
    \tabitem $\barb[a]{\tClose{a}}$ & \tabitem $\barb[a]{\jproc{c}{P}}$ whenever $\barb[a]{P}$\\
    \tabitem $\barb[a]{\tSendL{a}{k}{P}}$ & \tabitem $\barb[a]{\jmsg{a}{m}}$ whenever $\barb[a]{m}$ \\
    \tabitem $\barb[a]{\tSendS{a}{P}}$ &  \tabitem $\barb[a]{(\jproc{b}{\tFwdN{a}{b}}, \jmsg{c}{m^-_{b,c}})}$\\
    \tabitem $\barb[a]{\tSendC{a}{b}{P}}$ &   \tabitem $\barb[b]{(\jmsg{a}{m^+}, \jproc{b}{\tFwdP{a}{b}})}$, and\\
    \tabitem $\barb[a]{\tSendU{a}{P}}$ &    \tabitem $\barb[a]{(\ctxh[\mc{D}]{\mc{C}}{\Gamma}{\Delta})}$ whenever $\barb[a]{\mc{D}}$. \\
    \tabitem $\barb[a]{\tSendV{a}{M}{P}}$ if \(\fneval{M}{v}\) for some \(v\) &
  \end{tabular}\\
  A \defin{weak barb} is the channel-indexed predicate $\wbarb[a]{(\cdot)}$\glsadd{wbarb} on configurations defined by the composition of relations $\steps\barb[a]{(\cdot)}$.
  We write $\nwbarb[a]{(\cdot)}$ for the negation of $\wbarb[a]{(\cdot)}$.
\end{definition}

Write $\relreflc{\msstep}$ for the reflexive closure of $\msstep$.

\begin{proposition}
  \label{prop:sill-obs-equiv/barb-cont-congr:6}
  For all configurations $\mc{C}$, $\barb[a]{\mc{C}}$ if and only if $\mc{C} \relreflc{\msstep} \mc{C}', \jmsg{c}{m}$ for some $\jmsg{c}{m}$ with $\carrcn(\jmsg{c}{m}) = a$.
\end{proposition}

\begin{proof}
  Sufficiency follows by a case analysis on why $\barb[a]{\mc{C}}$.
  To see necessity, assume first that $\jmsg{c}{m} \in \mc{C}$.
  Then a case analysis on $m$ gives the result.
  If $\jmsg{c}{m} \notin \mc{C}$, then a case analysis on the (non-reflexive) step gives the result.
\end{proof}

\Cref{cor:sill-obs-equiv/extern-observ:3} states that an observable action eventually occurs on \(a\) (\(\wbarb[a]{\mc{C}}\)) if and only if we observe communication on \(a\) (\(\obsbr{\jcfgt{\Gamma}{\mc{C}}{\Delta}}_a(a) \neq \bot\)):

\begin{corollary}
  \label{cor:sill-obs-equiv/extern-observ:3}
  For all $\jcfgt{\Gamma}{\mc{C}}{\Delta}$ and $a \in \freecn(\mc{C})$, $\wbarb[a]{\mc{C}}$ if and only if $\obsbr{\jcfgt{\Gamma}{\mc{C}}{\Delta}}_a(a) \neq \bot$.
\end{corollary}

\begin{proof}
  By definition, $\wbarb[a]{\mc{C}}$ if and only if \(\mc{C} \mssteps \mc{D}\) for some \(\mc{D}\) such that \(\barb[a]{\mc{D}}\).
  By \cref{prop:sill-obs-equiv/barb-cont-congr:6}, this implies that $\wbarb[a]{\mc{C}}$ if and only if \(\mc{C} \mssteps \mc{C}',\jmsg{c}{m}\) for some \(\mc{C}'\) and \(\jmsg{c}{m}\) such that \(\carrcn(\jmsg{c}{m}) = a\).
  The multiset rewriting system defining polarized SILL is non-overlapping by \cref{prop:sill-background-dyn-prop-typed-config:1}, so interference-free on \cref{prop:ssos-fairness/prop-fair-trac:1}.
  By the fair concatenation property (\cref{cor:ssos-fairness/prop-fair-trac:2}), we can extend \(\mc{C} \mssteps \mc{C}',\jmsg{c}{m}\) to a fair execution \(T\).
  A case analysis on \(\jmsg{c}{m}\) then implies \(\jtoc{T}{v}{a}{A}\) for some \(v\) and \(A\), so $\obsbr{\jcfgt{\Gamma}{\mc{C}}{\Delta}}_a(a) \neq \bot$.
  Conversely, if $\obsbr{\jcfgt{\Gamma}{\mc{C}}{\Delta}}_a(a) = v \neq \bot$, then \(\jtoc{T}{v}{a}{A}\) for some \(T\) and \(A\).
  A case analysis on the last rule used to form \(\jtoc{T}{v}{a}{A}\) shows that \(\mc{C} \mssteps \mc{C}',\jmsg{c}{m}\) for some \(\mc{C}'\) and \(\jmsg{c}{m}\) such that \(\carrcn{\jmsg{c}{m}} = a\).
  We have already show that this is logically equivalent to $\wbarb[a]{\mc{C}}$.
\end{proof}

The barbed simulation game requires that the simulating configuration match the simulated configuration's barbs:

\begin{definition}
  \label{def:sill-obs-equiv/barb-cont-congr:3}
  A typed relation $\relfnt{R}$ on configurations is a \defin{(weak) barbed simulation} if $\jtrelc{\relfnt{R}}{\Delta}{\mc{C}}{\mc{D}}{\Psi}$ implies
  \begin{enumerate}
  \item if $\jstep{\mc{C}}{\mc{C}'}$, then $\mc{D} \steps \mc{D}'$ with $\jtrelc{\relfnt{R}}{\Delta}{\mc{C}'}{\mc{D}'}{\Psi}$; and
  \item for all channels $a : A \in \Delta,\Psi$, if $\wbarb[a]{\mc{C}}$, then $\wbarb[a]{\mc{D}}$.
  \end{enumerate}
  \defin{(Weak) barbed similarity}, $\wbsim$\glsadd{wbsim}, is the largest barbed simulation.
  Two configurations $\jcfgt{\Delta}{\mc{C}}{\Psi}$ and $\jcfgt{\Delta}{\mc{D}}{\Psi}$ are \defin{(weak) barbed similar}, $\jtrelc{\wbsim}{\Delta}{\mc{C}}{\mc{D}}{\Psi}$, if $\jtrelc{\relfnt{R}}{\Delta}{\mc{C}}{\mc{D}}{\Psi}$ for some barbed simulation $\relfnt{R}$.
\end{definition}

We can define barbed bisimulation from barbed simulation in the usual manner:

\begin{definition}
  \label{def:sill-obs-equiv/extern-observ:3}
  A typed relation $\relfnt{R}$ on configurations is a \defin{(weak) barbed bisimulation} if both $\relfnt{R}$ and $\relfnt{R}^{-1}$ are barbed simulations.
  \defin{(Weak) barbed bisimilarity}, $\wbbisim$\glsadd{wbbisim}, is the largest barbed bisimulation.
  Two configurations $\jcfgt{\Delta}{\mc{C}}{\Psi}$ and $\jcfgt{\Delta}{\mc{D}}{\Psi}$ are \defin{(weak) barbed bisimilar}, $\jtrelc{\wbbisim}{\Delta}{\mc{C}}{\mc{D}}{\Psi}$, if $\jtrelc{\relfnt{R}}{\Delta}{\mc{C}}{\mc{D}}{\Psi}$ for some barbed bisimulation $\relfnt{R}$.
\end{definition}

Barbed bisimulation is an equivalence relation.
We do not develop its theory any further.

\begin{example}
  \label{ex:sill-obs-equiv/barb-cont-congr:1}
  The following two configurations are barbed bisimilar:
  \begin{align}
    &\jcfgt{\cdot}{\jproc{b}{\tCut{a}{\tClose{a}}{(\tWait{a}{\tClose{b}})}}}{b : \Tu}\label{eq:sill-obs-equiv/barb-cont-congr:2}\\
    &\jcfgt{\cdot}{\jproc{b}{\tClose{b}}}{b : \Tu}\label{eq:sill-obs-equiv/barb-cont-congr:3}
  \end{align}
  The unique execution of \eqref{eq:sill-obs-equiv/barb-cont-congr:2} is:
  \begin{align}
    &\jproc{b}{\tCut{a}{\tClose{a}}{(\tWait{a}{\tClose{b}})}}\label{eq:sill-obs-equiv/barb-cont-congr:4}\\
    &\msstep \jproc{a}{\tClose{a}}, \jproc{b}{\tWait{a}{\tClose{b}}}\label{eq:sill-obs-equiv/barb-cont-congr:5}\\
    &\msstep \jmsg{a}{\mClose{a}}, \jproc{b}{\tWait{a}{\tClose{b}}}\label{eq:sill-obs-equiv/barb-cont-congr:6}\\
    &\msstep \jproc{b}{\tClose{b}}\label{eq:sill-obs-equiv/barb-cont-congr:7}\\
    &\msstep \jmsg{b}{\mClose{b}},\label{eq:sill-obs-equiv/barb-cont-congr:8}
  \end{align}
  while the unique execution of \eqref{eq:sill-obs-equiv/barb-cont-congr:3} is:
  \begin{align}
    &\jproc{b}{\tClose{b}}\label{eq:sill-obs-equiv/barb-cont-congr:9}\\
    &\msstep \jmsg{b}{\mClose{b}}.\label{eq:sill-obs-equiv/barb-cont-congr:10}
  \end{align}
  Where the numbers refer to the configurations in the above executions, the following relation is a barbed bisimulation:
  \[
    \relfnt{R} = \{
    (\eqref{eq:sill-obs-equiv/barb-cont-congr:4}, \eqref{eq:sill-obs-equiv/barb-cont-congr:9}),
    (\eqref{eq:sill-obs-equiv/barb-cont-congr:5}, \eqref{eq:sill-obs-equiv/barb-cont-congr:9}),
    (\eqref{eq:sill-obs-equiv/barb-cont-congr:6}, \eqref{eq:sill-obs-equiv/barb-cont-congr:9}),
    (\eqref{eq:sill-obs-equiv/barb-cont-congr:7}, \eqref{eq:sill-obs-equiv/barb-cont-congr:9}),
    (\eqref{eq:sill-obs-equiv/barb-cont-congr:8}, \eqref{eq:sill-obs-equiv/barb-cont-congr:10})
    \}
  \]
  Indeed, it ensures that the two configurations remain related throughout the stepping game.
  It also satisfies the requirement that related configurations have the same barbs for channels in their interfaces: in each pair, both configurations satisfy the weak barb $\wbarb[b]{(\cdot)}$.
\end{example}

\begin{lemma}
  \label{lemma:sill-obs-equiv/barb-cont-congr:1}
  If $\jcfgt{\Gamma}{\mc{C}}{\Delta}$ and $\mc{C} \msstep \mc{C}'$, then for all $c \in \check \Gamma,\check \Delta$, we have $\wbarb[c]{\mc{C}}$ if and only if $\wbarb[c]{\mc{C}'}$.
\end{lemma}

\begin{proof}
  This is a consequence of \cref{prop:sill-obs-equiv/observ-comm:8,cor:sill-obs-equiv/extern-observ:3}.
\end{proof}

We can characterize barbed similarity using \cref{prop:sill-obs-equiv/barb-cont-congr:8}.

\begin{proposition}
  \label{prop:sill-obs-equiv/barb-cont-congr:8}
  Two configurations are barbed similar, $\jtrelc{\wbsim}{\Delta}{\mc{C}}{\mc{D}}{\Psi}$, if and only if for all $c \in \check \Delta, \check \Psi$, if $\wbarb[c]{\mc{C}}$, then $\wbarb[c]{\mc{D}}$.
\end{proposition}

\begin{proof}
  Sufficiency is obvious.
  To see necessity, let $\relfnt{R}$ be the relation given by $\{ (\mc{C}', \mc{D}') \mid \mc{C} \mssteps \mc{C}' \land \mc{D} \mssteps \mc{D}' \}$.
  It is a barbed simulation.
  Indeed, it is closed under stepping by construction.
  Moreover, if $\mc{C}' \mathrel{\relfnt{R}} \mc{D}'$ and $\wbarb[c]{\mc{C}'}$ for some $c \in \check \Delta, \check \Psi$, then $\wbarb[c]{\mc{D}'}$.
  To see that this is so, observe that if $\wbarb[c]{\mc{C}'}$, then $\wbarb[c]{\mc{C}}$, so $\wbarb[c]{\mc{D}}$ by assumption.
  Then $\wbarb[c]{\mc{D}'}$ by \cref{lemma:sill-obs-equiv/barb-cont-congr:1}.
\end{proof}

\begin{proposition}
  \label{prop:sill-obs-equiv/barb-cont-congr:2}
  Barbed similarity not a precongruence relation on configurations.
\end{proposition}

\begin{proof}
  Recall the process $\Omega$ from \cref{ex:sill-background-typing-mult-rewr:2}.
  Consider the processes $P$, $Q$, and $R$ respectively given by:
  \begin{align*}
    &\jtypem{\cdot}{a : \Tplus \{ l : \Tu, r : \Tu \}}{\tCase{a}{\{ l \Rightarrow \tWait{a}{\tClose c} \mid  r \Rightarrow \Omega \}}}{c}{\Tu},\\
    &\jtypem{\cdot}{a : \Tplus \{ l : \Tu, r : \Tu \}}{\tCase{a}{\{ l \Rightarrow \Omega \mid  r \Rightarrow \tWait{a}{\tClose c} \}}}{c}{\Tu},\\
    &\jtypem{\cdot}{\cdot}{\tSendL{a}{l}{\tClose{a}}}{a}{\Tplus \{ l : \Tu, r : \Tu \}}.
  \end{align*}
  We have the following pairs of barbed bisimilar configurations:
  \begin{align*}
    &\jtrelc{\wbbisim}{a : \Tplus \{ l : \Tu, r : \Tu \}}{\jproc{c}{P}}{\jproc{c}{Q}}{c : \Tu}\\
    &\jtrelc{\wbbisim}{\cdot}{\jproc{a}{R}}{\jproc{a}{R}}{(a : \Tplus \{ l : \Tu, r : \Tu \})}
  \end{align*}
  However, barbed similarity is not contextual (so not a congruence relation), for
  \[
    \jtrelc{\not\wbsim}{\cdot}{\jproc{a}{R},\jproc{c}{P}}{\jproc{a}{R},\jproc{c}{Q}}{c : \Tu}.
  \]
  Indeed, the left side steps to the configuration $\jproc{c}{\tClose{c}}$ which satisfies the barb $\wbarb[c]{(\cdot)}$, while right side cannot step to a configuration satisfying this barb.
\end{proof}

\begin{definition}
  \label{def:sill-obs-equiv/extern-observ:2}
  \defin{(Weak) barbed precongruence} $\mkcg{\wbsim}$ is the contextual interior of barbed similarity, \ie,\\
  ${\jtrelc{\mkcg{\wbsim}}{\Gamma}{\mc{C}}{\mc{D}}{\Xi}}$ if and only if if $\jtrelc{\wbsim}{\Delta}{\ctxh[\mc{C}]{\mc{E}}{\Gamma}{\Xi}}{\ctxh[\mc{C}]{\mc{E}}{\Gamma}{\Xi}}{\Lambda}$ for all contexts $\jcfgt{\Delta}{\ctxh{\mc{E}}{\Gamma}{\Xi}}{\Lambda}$.
\end{definition}

\subsubsection{Relating Barbed Simulations and External-style Observational Simulations}
\label{sec:sill-obs-equiv:extern-observ:relat-barbs-ext}

We show that barbed precongruence \(\mkcg{\wbsim}\) coincides with external observational simulation \(\eocssim\).
In fact, we show a more general result.
Let $E = (\mc{X}^E, {\unirel})$ be the external observation system given by \cref{def:sill-obs-equiv/main-equiv:3}, and let $\leqslant$ be a preorder such that $(\mc{X}^E, {\leqslant})$ is an observation system.
We instead relate $\mkcg{\wbsim}$ and $\mkcg{\left(\sobssim{(\mc{X}^E, {\leqslant})}\right)}$; the result for $\eocssim$ will follow as a special case.

\begin{proposition}
  \label{prop:sill-obs-equiv/barb-cont-congr:7}
  If $\jtrelc{\sobssim{(\mc{X}^E, {\leqslant})}}{\Gamma}{\mc{C}}{\mc{D}}{\Delta}$, then $\jtrelc{\wbsim}{\Gamma}{\mc{C}}{\mc{D}}{\Delta}$.
\end{proposition}

\begin{proof}
  Assume that $\jtrelc{\sobssim{(\mc{X}^E, {\leqslant})}}{\Gamma}{\mc{C}}{\mc{D}}{\Delta}$.
  Let $c \in \check \Gamma, \check \Delta$ be arbitrary.
  By \cref{cor:sill-obs-equiv/extern-observ:3}, $\wbarb[c]{\mc{C}}$ if and only if $\obsbr{\jcfgt{\Gamma}{\mc{C}}{\Xi}}(c) \neq \bot$.
  In this case, $\obsbr{\jcfgt{\Gamma}{\mc{D}}{\Delta}}(c) \neq \bot$, which holds if and only if $\wbarb[c]{\mc{D}}$.
  We conclude that $\jtrelc{\wbsim}{\Gamma}{\mc{C}}{\mc{D}}{\Delta}$ by \cref{prop:sill-obs-equiv/barb-cont-congr:8}.
\end{proof}

\begin{corollary}
  \label{cor:sill-obs-equiv/extern-observ:4}
  If $\jtrelc{\mkcg{\left(\sobssim{(\mc{X}^E, {\leqslant})}\right)}}{\Gamma}{\mc{C}}{\mc{D}}{\Delta}$, then $\jtrelc{\mkcg{\wbsim}}{\Gamma}{\mc{C}}{\mc{D}}{\Delta}$.
\end{corollary}

The converse of \cref{prop:sill-obs-equiv/barb-cont-congr:7} is false.
This is because barbs do not distinguish between sent messages, but only identify that a message was sent.
Concretely,
\[
  \jtrelc{\wbsim}{\cdot}{\jproc{c}{\tSendL{c}{\mt{0}}{\tClose{c}}}}{\jproc{c}{\tSendL{c}{\mt{1}}{\tClose{c}}}}{c : \Tplus \{ \mt{0} : \Tu, \mt{1} : \Tu \}},
\]
but
\begin{align*}
  &\obsbr{\jcfgt{\cdot}{\jproc{c}{\tSendL{c}{\mt{0}}{\tClose{c}}}}{c : \Tplus \{ \mt{0} : \Tu, \mt{1} : \Tu \}}}\\
  &= (c : (\mt{0}, \cclose))\\
  &\ncommsim[\leqslant] (c : (\mt{1}, \cclose))\\
  &= \obsbr{\jcfgt{\cdot}{\jproc{c}{\tSendL{c}{\mt{1}}{\tClose{c}}}}{c : \Tplus \{ \mt{0} : \Tu, \mt{1} : \Tu \}}}.
\end{align*}

Despite this, the converse of \cref{cor:sill-obs-equiv/extern-observ:4} holds under certain reasonable hypotheses.
Assume that $\jtrelc{\mkcg{\wbsim}}{\Gamma}{\mc{C}}{\mc{D}}{\Xi}$.
We must show that $\jtrelc{\mkcg{\left(\sobssim{(\mc{X}^E, {\leqslant})}\right)}}{\Lambda}{\mc{C}}{\mc{D}}{\Xi}$.
By \cref{prop:sill-obs-equiv/main-equiv:2}, it is sufficient to show that $\jtrelc{\sobssim{(\mc{X}^E, {\leqslant})}}{\Lambda}{\mc{C}}{\mc{D}}{\Xi}$.
We reduce this problem to the following:

\begin{problem}
  \label{question:sill-obs-equiv/extern-observ:1}
  Consider an observed communication $\obsbr{\jcfgt{\Gamma}{\mc{C}}{\Delta}} = (c : v_c)_c$.
  Can we construct an experiment context\footnote{Though they serve similar purposes, these should not be confused with the experiments of observation systems.} $\jcfgt{\widehat{\Gamma}}{\ctxh{\mc{E}}{\Gamma}{\Delta}}{\widehat{\Delta}}$ (for some $\widehat{\Gamma}$ and $\widehat{\Delta}$ determined from $\Gamma$ and $\Delta$) such that for all $\jcfgt{\Gamma}{\mc{D}}{\Delta}$, $\wbarb[\widehat{c}]{\ctxh[\mc{D}]{\mc{E}}{\Gamma}{\Delta}}$ for all $\widehat{c} : \widehat{A} \in \widehat{\Gamma}, \widehat{\Delta}$ if and only if $\jtrelc{\sobssim{(\mc{X}^E, {\leqslant})}}{\Gamma}{\mc{C}}{\mc{D}}{\Delta}$?
\end{problem}

Indeed, given a solution to \cref{question:sill-obs-equiv/extern-observ:1}, we can show that $\jtrelc{\mkcg{\wbsim}}{\Gamma}{\mc{C}}{\mc{D}}{\Xi}$ implies $\jtrelc{\sobssim{(\mc{X}^E, {\leqslant})}}{\Lambda}{\mc{C}}{\mc{D}}{\Xi}$ as follows.
Let \(\ctxh{\mc{E}}{}{}\) be given by the solution.
By reflexivity, $\jtrelc{\sobssim{(\mc{X}^E, {\leqslant})}}{\Lambda}{\mc{C}}{\mc{C}}{\Xi}$.
By construction, $\wbarb[\widehat{c}]{\ctxh[\mc{C}]{\mc{E}}{\Gamma}{\Delta}}$ for all $\widehat{c} : \widehat{A} \in \widehat{\Gamma}, \widehat{\Delta}$.
By assumption, $\jtrelc{\wbsim}{\widehat{\Gamma}}{\ctxh[\mc{C}]{\mc{E}}{\Gamma}{\Delta}}{\ctxh[\mc{D}]{\mc{E}}{\Gamma}{\Delta}}{\widehat{\Xi}}$.
So by \cref{prop:sill-obs-equiv/barb-cont-congr:8}, $\wbarb[\widehat{c}]{\ctxh[\mc{D}]{\mc{E}}{\Gamma}{\Delta}}$ for all $\widehat{c} : \widehat{A} \in \widehat{\Gamma}, \widehat{\Delta}$.
Then $\jtrelc{\sobssim{(\mc{X}^E, {\leqslant})}}{\Gamma}{\mc{C}}{\mc{D}}{\Delta}$ by construction of $\ctxh{\mc{E}}{\Gamma}{\Delta}$.

The answer to \cref{question:sill-obs-equiv/extern-observ:1} is ``no, but almost''.
A finite number of experiment contexts is insufficient in the presence of recursion and channel transmission (\cf \cite[38]{hennessy_1983:_synch_async_exper_proces}).
We will instead construct families of contexts, one for each height $n$ approximation of the communications in $(c : v_c)_c$.
These height \(n\) approximations determine the \(v_c\) by \cref{prop:sill-obs-equiv/barb-cont-congr:11}, and we will use this fact to conclude the general result.

For brevity, we write $\commsim$ for $\commsim[\leqslant]$ in the remainder of this \namecref{sec:sill-obs-equiv:extern-observ:relat-barbs-ext}.

Before giving the full construction of experiment contexts, we illustrate the approach by a sequence of examples.
In these examples, we only consider configurations of the form $\jcfgt{\cdot}{\mc{C}}{a : A}$.
Given a channel name $a$, let $\widehat a$ be a globally fresh channel name.
Let $\mathbb{Y}^+$ be the ``positive answer type'' $\Tplus \{ \mt{y} : \mb{0}^+ \}$, where $\mb{0}^+$ is the positive empty type $\Tplus \{ \}$.
We will construct non-empty sets $\mc{E}(n, \jsynt{v}{A})$ of processes of type $\jtypem{\cdot}{a : A}{E}{\widehat{a}}{\mathbb{Y}^+}$.
They will satisfy the following weakened form of the general result (\cref{prop:sill-obs-equiv/extern-observ:1}):

\begin{proposition}
  Fix some $\jsynt{v}{A}$.
  Let $\jcfgt{\cdot}{\mc{C}}{a : A}$ be arbitrary, and set $w = \obsbr{\jcfgt{\cdot}{\mc{C}}{a : A}}(a)$.
  For all $n$, $\jsynt{\capxn{v}{n + 1} \commsim w}{A}$ if and only if $\obsbr{\jcfgt{\cdot}{\mc{C}, \jproc{\widehat a}{E}}{\widehat a : \mathbb{Y}^+}}(\widehat a) = (\mt{y}, \bot)$ for all $E \in \mc{E}(n, \jsynt{v}{A})$.
\end{proposition}

In particular, by \cref{cor:sill-obs-equiv/extern-observ:3}, $\obsbr{\jcfgt{\cdot}{\mc{C}, \jproc{\widehat a}{E}}{\widehat a : \mathbb{Y}^+}}(\widehat a) = (\mt{y}, \bot)$ if and only if $\wbarb[\widehat{a}]{\left(\mc{C}, \jproc{\widehat a}{E}\right)}$.
The set $\mc{E}(n, \jsynt{v}{A})$ checks that $\jsynt{\capxn{v}{n + 1} \commsim w}{A}$ instead of $\jsynt{\capxn{v}{n} \commsim w}{A}$ because it is always the case that $\jsynt{\capxn{w}{0} \commsim v}{A}$, and doing so simplifies the definition.
These sets of processes induce sets of experiment contexts: each process \(E\) gives the configuration context \(\ctxh{}{}{a : A},\jproc{\widehat a}{E}\).
Let $Y_c$ be the process $\tSendL{c}{\mt{y}}{\Omega}$, where $\Omega$ is given at each type by \cref{ex:sill-background-typing-mult-rewr:2}.

\begin{example}
  \label{ex:sill-obs-equiv/extern-observ:1}
  Set $v = (k, \cclose)$ and $A = \Tplus \{ k : \Tu, l : \Tu \}$.
  Set $w = \obsbr{\jcfgt{\cdot}{\mc{C}}{a : \Tplus \{ k : \Tu, l : \Tu \}}}(a)$ for some configuration $\mc{C}$.

  Recall that $\capxn{v}{1} = (k, \bot)$, and observe that $\jsynt{\capxn{(k, \cclose)}{1} \commsim w}{A}$ if and only if $\mc{C}$ sent the label $k$ on $a$, so if and only if
  \[
    \obsbr{\jcfgt{}{\mc{C}, \jproc{\widehat a}{\tCase{a}{\{ k \Rightarrow Y_{\widehat a} \mid \uscore \Rightarrow \Omega \}}}}{\widehat a : \mathbb{Y}^+}}(\widehat a) = (\mt{y}, \bot).
  \]
  Take $\mc{E}(0, \jsynt{v}{A}) = \{ \tCase{a}{\{ k \Rightarrow Y_{\widehat c} \mid \uscore \Rightarrow \Omega \} } \}$.

  Now observe that $\jsynt{\capxn{v}{2} \commsim w}{A}$ if and only if
  \[
    \obsbr{\jcfgt{}{\mc{C}, \jproc{\widehat a}{\tCase{a}{\{ k \Rightarrow \tWait{c}{Y_{\widehat c}} \mid \uscore \Rightarrow \Omega \}}}}{\widehat a : \mathbb{Y}^+}}(\widehat a) = (\mt{y}, \bot).
  \]
  Take $\mc{E}(1, \jsynt{v}{A}) = \{\tCase{a}{\{ k \Rightarrow \tWait{a}{Y_{\widehat c}} \mid \uscore \Rightarrow \Omega \}}\}$.
  Because $\capxn{v}{n} = \capxn{v}{2}$ for $n \geq 2$, take $\mc{E}(n, \jsynt{v}{A}) = \mc{E}(1, \jsynt{v}{A})$ for all $n \geq 1$.
\end{example}

Next, we illustrate why channel transmission forces $\mc{E}(n, \jsynt{v}{A})$ to be a set:

\begin{example}
  \label{ex:sill-obs-equiv/extern-observ:2}
  Set $v = ((\mt{l}, \cclose), (\mt{r}, \bot))$ and $A = \left(\Tplus \{ \mt{l} : \Tu \}\right) \Tot \left(\Tplus \{ \mt{r} : \Tu \}\right)$.
  Set $w = \obsbr{\jcfgt{}{\mc{C}}{a : A}}(a)$ for some configuration $\mc{C}$.
  Clearly, $\jsynt{\capxn{v}{1} = (\bot, \bot) \commsim w}{A}$ if and only if
  \[
    \obsbr{\jcfgt{}{\mc{C}, \jproc{\widehat a}{\tRecvC{a}{c}{Y_{\widehat c}}}}{\widehat a : \mathbb{Y}^+}}(\widehat a) = (\mt{y}, \bot).
  \]
  Our task becomes harder when we consider $\capxn{v}{2} = ((\mt{l}, \bot), (\mt{r}, \bot))$: we must somehow inspect the communications on $a$ and also those on $c$, and return a result on $\widehat a : \mathbb{Y}^+$.
  We do so by using two experiment contexts.
  Indeed, $\jsynt{\capxn{v}{2} \commsim w}{A}$ if and only if both
  \begin{align*}
    \obsbr{\jcfgt{}{\mc{C}, \jproc{\widehat a}{\tRecvC{a}{c}{\tCase{c}{\{ \mt{l} \Rightarrow Y_{\widehat a} \}}}}}{\widehat a : \mathbb{Y}^+}}(\widehat a) &= (\mt{y}, \bot),\\
    \obsbr{\jcfgt{}{\mc{C}, \jproc{\widehat a}{\tRecvC{a}{c}{\tCase{a}{\{ \mt{r} \Rightarrow Y_{\widehat a} \}}}}}{\widehat a : \mathbb{Y}^+}}(\widehat a) &= (\mt{y}, \bot).
  \end{align*}
  Accordingly, we take \(\mc{E}(1, \jsynt{v}{A}) = \{ \tRecvC{a}{c}{\tCase{c}{\{ \mt{l} \Rightarrow Y_{\widehat a} \}}}, \tRecvC{a}{c}{\tCase{a}{\{ \mt{r} \Rightarrow Y_{\widehat a} \}}} \}. \)
\end{example}

Having illustrated the approach, we define the family $\mc{E}_R(n, i, r, \jsynt{v}{A})$ of experiment processes by induction on $n$ and recursion on $\jsynt{v}{A}$, where $i$ is the ``input channel'' whose communications we are examining ($a$ in the above examples), and $r$ is the ``results channel'' ($\widehat a$ in the above examples).
We can lift these experiment processes to testing contexts in the obvious manner.
The family $\mc{E}_R(n, i, r, \jsynt{v}{A})$ checks that communications $w$ on $i$ satisfy $\jsynt{\capxn{v}{n + 1} \commsim w}{A}$.

We maintain the invariant that if $E \in \mc{E}_R(n, i, r, \jsynt{v}{A})$, then $\jtypem{\cdot}{i : A}{E}{r}{\mathbb{Y}^+}$ and $E$ can be weakened\footnote{Though the type system is linear, the presence of unbounded recursion allows us to ignore channels in non-terminating processes. See \cref{rem:sill-background-typing-mult-rewr:5} for a discussion of this fact.} to $\jtypem{\cdot}{\Delta, i : A}{E}{r}{\mathbb{Y}^+}$ for all $\Delta$ not mentioning $i$ or $r$.
In particular, the processes in $\mc{E}_R(n, i, r, \jsynt{v}{A})$ always listen from left and report results on the right; we will consider the symmetric case $\mc{E}_L(n, i, r, \jsynt{v}{A})$ later.
We also maintain the invariant that $E \in \mc{E}_R(n, i, r, \jsynt{v}{A})$ if and only if $\subst{i'}{i}{E} \in \mc{E}_R(n, i', r, \jsynt{v}{A})$ for $i \neq r$ and $i' \neq r$.
\begin{align*}
  \mc{E}_R(n, i, r, \jsynt{\bot}{A}) &= \left\{ Y_{r} \right\}\\
  \mc{E}_R(n, i, r, \jsynt{\cclose}{\Tu}) &= \left\{ \tWait{i}{Y_{r}} \right\}\\
  \mc{E}_R(0, i, r, \jsynt{(k, \uscore)}{\Tplus \{ l : A_l \}_{l \in L}}) &= \left\{ \tCase{i}{\{ k \Rightarrow Y_{r} \mid \uscore \Rightarrow \Omega \}} \right\}\\
  \mc{E}_R(n + 1, i, r, \jsynt{(k, v)}{\Tplus \{ l : A_l \}_{l \in L}}) &= \left\{ \tCase{i}{\{ k \Rightarrow E \mid \uscore \Rightarrow \Omega \}} \mid E \in \mc{E}_R(n, i, r, \jsynt{v}{A_k})  \right\}\\
  \mc{E}_R(0, i, r, \jsynt{(\cunfold, \uscore)}{\Trec{\alpha}{A}}) &= \{ \tRecvU{i}{Y_r} \}\\
  \mc{E}_R(n + 1, i, r, \jsynt{(\cunfold, v)}{\Trec{\alpha}{A}}) &= \left\{ \tRecvU{i}{E} \mid E \in \mc{E}_R(n, i, r, \jsynt{v}{\subst{\Trec{\alpha}{A}}{\alpha}{A}}) \right\}\\
  \mc{E}_R(0, i, r, \jsynt{(\cshift, \uscore)}{\Tus{A}}) &= \{ \tRecvS{i}{Y_r} \}\\
  \mc{E}_R(n + 1, i, r, \jsynt{(\cshift, v)}{\Tus{A}}) &= \left\{ \tRecvS{i}{E} \mid E \in \mc{E}_R(n, i, r, \jsynt{v}{A}) \right\}\\
  \mc{E}_R(0, i, r, \jsynt{(\uscore,\uscore)}{A \Tot B}) &= \left\{ \tRecvC{\uscore}{i}{Y_r} \right\}\\
  \mc{E}_R(n + 1, i, r, \jsynt{(u,v)}{A \Tot B}) &= \left\{ \tRecvC{\uscore}{i}{E} \mid E \in \mc{E}_R(n, i, r, \jsynt{v}{B}) \right\} {} \cup\\
                                   &\qquad\qquad{} \cup \left\{ \tRecvC{a}{i}{E} \mid E \in \mc{E}_R(n, a, r, \jsynt{u}{A}) \right\}
\end{align*}

Conspicuously absent are negative protocols.
Because any provided channel with a negative protocol is an input channel, we can only observe $\bot$ on that channel by \cref{prop:sill-obs-equiv/observ-comm:6}.
Accordingly, we define:
\begin{align*}
  \mc{E}_R(n, i, r, \jsynt{v}{A^-}) &= \begin{cases}
    \left\{ Y_r \right\} & v = \bot\\
    \left\{ \Omega \right\} & v \neq \bot
  \end{cases}
\end{align*}

Also absent are sets of experiment processes for the protocol $\Tand{\tau}{A}$.
We use an oracle process to check if two transmitted values are related by $\leqslant$.

\begin{definition}
  Let $P$ be a predicate on functional values of type $\tau$.
  An \defin{oracle process for $P$} is a process $\jtypem{\cdot}{\cdot}{O}{c}{\Timp{\tau}{\Tus{\Tplus \{ \mt{tt} : \Tu, \mt{ff} : \Tu \}}}}$ that receives a functional value $w$ and a shift message\footnote{Because types are polarized, the up shift is required to ensure that the type is well-formed.} over $c$, and sends $\mt{tt}$ if $P(w)$ and $\mt{ff}$ otherwise; and closes the channel in both cases.
  Explicitly,
  \begin{gather*}
    \jproc{c}{O}, \jmsg{d}{\mSendVN{c}{w}{d}}, \jmsg{e}{\mSendSN{d}{e}} \mssteps {}\\
    {} \mssteps \exists f .
    \begin{cases}
      \jmsg{f}{\mClose{f}}, \jmsg{e}{\mSendLP{e}{\ms{tt}}{f}} & \text{if } P(w)\\
      \jmsg{f}{\mClose{f}}, \jmsg{e}{\mSendLP{e}{\ms{ff}}{f}} & \text{otherwise}
    \end{cases}\qedhere
  \end{gather*}
\end{definition}

\begin{assumption}
  \label{ass:sill-obs-equiv/extern-observ:1}
  Assume that for all $\tau$ and values $\jtypef{\cdot}{v}{\tau}$, there is an oracle $O^\leqslant_{v : \tau}$ for the predicate $\jtrelf{\leqslant}{\cdot}{v}{({-})}{\tau}$.
\end{assumption}

Using the oracle, we define:
\begin{align*}
  \mc{E}_R(0, i, r, \jsynt{(\cval f, \uscore)}{\Tand{\tau}{A}}) &= \{ \tCut{c}{O_{f:\tau}}{\tRecvV{x}{i}{\tSendV{c}{x}{\tSendS{c}{\\
                                                              &\qquad\tCase{c}{\{ \mt{tt} \Rightarrow \tWait{c}{Y_r} \mid \mt{ff} \Rightarrow \Omega \}}}}}} \}\\
  \mc{E}_R(n + 1, i, r, \jsynt{(\cval f, u)}{\Tand{\tau}{A}}) &= \{ \tCut{c}{O_{f:\tau}}{\tRecvV{x}{i}{\tSendV{c}{x}{\tSendS{c}{\\
                                                              &\qquad\tCase{c}{\{ \mt{tt} \Rightarrow \tWait{c}{E} \mid \mt{ff} \Rightarrow \Omega \}}}}}} \mid E \in \mc{E}_R(n, i, r, \jsynt{u}{A}) \}
\end{align*}

Recall from \cref{prop:sill-obs-equiv/barb-cont-congr:11} that height \(n\) approximations characterize communication simulation, \ie, that $\jsynt{v \commsim[\leqslant] w}{A}$ if and only if, for all $n$, ${\jsynt{\capxn{v}{n} \commsim[\leqslant] w}{A}}$.\footnote{In fact, we can use \cref{prop:sill-obs-equiv/barb-cont-congr:12} to strengthen this characterization: $\jsynt{v \commsim[\leqslant] w}{A}$ if and only if there exists an \(N\) such that ${\jsynt{\capxn{v}{n} \commsim[\leqslant] w}{A}}$ for all $n \geq N$.}
\Cref{prop:sill-obs-equiv/extern-observ:1} gives an analogous characterization of observed communications on \(a\) in terms of experiments in \(\mc{E}_R(n, a, \widehat a, \jsynt{v}{A})\).
Indeed, it states that $\jsynt{\capxn{v}{n + 1} \commsim w}{A}$ if and only if every experiment context \(E \in \mc{E}_R(n, a, \widehat a, \jsynt{v}{A})\) sends the ``yes'' answer \((\mt{y}, \bot)\) on the results channel \(\widehat a\).
Applying \cref{prop:sill-obs-equiv/extern-observ:1} to all \(n\) determines whether $\jsynt{v \commsim[\leqslant] w}{A}$.

\begin{proposition}
  \label{prop:sill-obs-equiv/extern-observ:1}
  Let $\jcfgt{\Gamma}{\mc{C}}{a : A}$ and $\jsynt{v}{A}$ be arbitrary.
  Set $w = \obsbr{\jcfgt{\Gamma}{\mc{C}}{\Delta, a : A}}(a)$ and let $\widehat a$ be globally fresh.
  Then for all $n$, $\jsynt{\capxn{v}{n + 1} \commsim w}{A}$ if and only if for all $E \in \mc{E}_R(n, a, \widehat a, \jsynt{v}{A})$,
  \[
    \obsbr{\jcfgt{\Gamma}{\mc{C}, \jproc{\widehat a}{E}}{\Delta, \widehat a : \mathbb{Y}^+}}(\widehat a) = (\mt{y}, \bot).
  \]
\end{proposition}

\begin{proof}
  By induction on $n$.
  Assume first that $n = 0$.
  Then we proceed by case analysis on $\jsynt{v}{A}$.
  We give the representative cases; the rest will follow by analogy.
  \begin{proofcases}
  \item[$\jsynt{\bot}{A}$] The result is immediate.
  \item[$\jsynt{v}{A^-}$] Then $w = \bot$ by \cref{prop:sill-obs-equiv/observ-comm:6}.
    If $v = \bot$, then the result is immediate.
    If $v \neq \bot$, then it is not the case that $\jsynt{v \commsim w}{A}$, and the result also follows from the definition of the divergent process $\Omega$.
  \item[$\jsynt{\cclose}{\Tu}$] Then $\capxn{\cclose}{1} = \cclose$.
    By inversion, $\jsynt{\cclose \commsim w}{A}$ if and only if $w = \cclose$.
    This is the case if and only if $\jmsg{a}{\mClose{a}}$ appears in a fair trace of $\mc{C}$.
    The only element of $\mc{E}_R(n, a, \widehat a, \jsynt{\cclose}{\Tu})$ is $\left\{ \tWait{a}{Y_{\widehat a}} \right\}$.
    The fact $\jmsg{a}{\mClose{a}}$ appears in a fair trace of $\mc{C}$ if and only if every fair trace of $\mc{C}, \jproc{\widehat a}{\tWait{a}{Y_{\widehat a}}}$ has an instantiation
    \[
      \jmsg{a}{\mClose a}, \jproc{\widehat a}{\tWait{a}{Y_{\widehat a}}} \msstep \jproc{\widehat a}{Y_{\widehat a}}
    \]
    of \cref{eq:sill:msr-tu-l}.
    By fairness, $Y_{\widehat a}$ produces the observation $(\mt{y}, \bot)$ on $\widehat a$.
    So $\jmsg{a}{\mClose{a}}$ appears in a fair trace of $\mc{C}$ if and only if $\obsbr{\jcfgt{\Gamma}{\mc{C}, \jproc{\widehat a}{\tWait{a}{Y_{\widehat a}}}}{\widehat a : \mathbb{Y}^+}}(\widehat a) = (\mt{y}, \bot)$.
    This gives the result.
  \item[$\jsynt{(k, u)}{\Tplus \{ l : A_l \}_{l \in L}}$]
    Then $\capxn{(k, u)}{1} = (k, \bot)$.
    By inversion, $\jsynt{(k, \bot) \commsim w}{A}$ if and only if $w = (k, w')$ for some $w'$.
    This is the case if and only if $\jmsg{a}{\mSendLP{a}{k}{\uscore}}$ appears in a fair trace of $\mc{C}$.
    The only element of $\mc{E}_R(0, a, \widehat a, \jsynt{(k, u)}{A}$ is $\tCase{a}{\{ k \Rightarrow Y_{\widehat a} \mid \uscore \Rightarrow \Omega \}}$.
    The aforementioned fact appears in a fair trace of $\mc{C}$ if and only if every fair trace of $\mc{C}, \jproc{\widehat a}{\tCase{a}{\{ k \Rightarrow Y_{\widehat a} \mid \uscore \Rightarrow \Omega \}}}$ has an instantiation
    \[
      \jmsg{a}{\mSendLP{a}{k}{\uscore}}, \jproc{\widehat a}{\tCase{a}{\{ k \Rightarrow Y_{\widehat a} \mid \uscore \Rightarrow \Omega \}}} \msstep \jproc{\widehat a}{Y_{\widehat a}}
    \]
    of \cref{eq:sill:msr-tplus-l}.
    The remainder is analogous to the previous case.
  \item[$\jsynt{(\cval f, u)}{\Tand{\tau}{A}}$]
    Then $\capxn{(\cval f, u)}{1} = (\cval f, \bot)$.
    By inversion, $\jsynt{(\cval f, \bot) \commsim w}{\Tand{\tau}{B}}$ if and only if $w = (\cval g, w')$ for some $g$ and $w'$ and $\jtrelf{\leqslant}{\cdot}{f}{g}{\tau}$.
    This is the case if and only if $\jmsg{a}{\mSendVP{a}{g}{\uscore}}$ appears in a fair trace of $\mc{C}$ with $g$ satisfying the above relation.
    The only element of $\mc{E}_R(0, a, \widehat a, \jsynt{(\cval f, u)}{A})$ is
    \[
      \tCut{c}{O_{f:\tau}}{\tRecvV{x}{a}{\tSendV{c}{x}{\tSendS{c}{
              \tCase{c}{\{ \mt{tt} \Rightarrow \tWait{c}{Y_r} \mid \mt{ff} \Rightarrow \Omega \}}}}}}.
    \]
    Every fair trace of $\mc{C}$ has an instantiation
    \begin{align*}
      &\jproc{\widehat a}{\tCut{c}{O_{f:\tau}}{\tRecvV{x}{a}{\tSendV{c}{x}{\tSendS{c}{\tCase{c}{\{ \mt{tt} \Rightarrow \tWait{c}{Y_r} \mid \mt{ff} \Rightarrow \Omega \}}}}}}} \msstep {}\\
      &\msstep \jproc{c'}{\subst{c'}{c}{O_{f:\tau}}}, \jproc{\widehat a}{\tRecvV{x}{a}{\\
          &\hspace{4em}\tSendV{c'}{x}{\tSendS{c'}{\tCase{c'}{\{ \mt{tt} \Rightarrow \tWait{c'}{Y_r} \mid \mt{ff} \Rightarrow \Omega \}}}}}}
    \end{align*}
    of \cref{eq:sill:msr-cut}.
    The aforementioned message fact appears in a fair trace of $\mc{C}$ if and only if the trace contains the following instantiation of \cref{eq:sill:msr-tand-l}:
    \begin{align*}
      &\left(\jmsg{a}{\mSendVP{a}{g}{\uscore}}, \jproc{\widehat a}{\tRecvV{x}{a}{ \right.\\
      &\hspace{4em}\left.\tSendV{c'}{x}{\tSendS{c'}{\tCase{c'}{\{ \mt{tt} \Rightarrow \tWait{c'}{Y_r} \mid \mt{ff} \Rightarrow \Omega \}}}}}}\right) \msstep {} \\
      &\msstep \left(\jproc{\widehat a}{\tSendV{c'}{g}{\tSendS{c'}{\tCase{c'}{\{ \mt{tt} \Rightarrow \tWait{c'}{Y_r} \mid \mt{ff} \Rightarrow \Omega \}}}}}\right),
    \end{align*}
    and of \cref{eq:sill:msr-timp-l,eq:sill:msr-tus-l} (not necessarily in immediate succession of each other):
    \begin{gather*}
      \left(\jproc{\widehat a}{\tSendV{c'}{g}{\tSendS{c'}{\tCase{c'}{\{ \mt{tt} \Rightarrow \tWait{c'}{Y_r} \mid \mt{ff} \Rightarrow \Omega \}}}}}\right) \mssteps {} \\
      {} \mssteps \left(\jmsg{d}{\mSendVN{c'}{g}{d}}, \jmsg{e}{\mSendSN{d}{e}},\right.\\
      \left.\jproc{\widehat a}{\tCase{e}{\{ \mt{tt} \Rightarrow \tWait{e}{Y_r} \mid \mt{ff} \Rightarrow \Omega \}}}\right).
    \end{gather*}
    By fairness and the definition of the oracle, the above hold if and only if the oracle takes the steps
    \begin{gather*}
      \left(\jproc{c'}{\subst{c'}{c}{O_{f:\tau}}}, \jmsg{d}{\mSendVN{c'}{g}{d}}, \jmsg{e}{\mSendSN{d}{e}}\right)
      \mssteps {}\\
      {} \mssteps \jmsg{h}{\mClose{h}}, \jmsg{c'}{\mSendLP{c'}{\ms{tt}}{h}}.
    \end{gather*}
    From here, the proof is analogous to the previous cases.
  \end{proofcases}

  Now assume that the result holds for some $n$.
  We show the inductive step $n + 1$, again by case analysis on $\jsynt{v}{A}$.
  We give the representative case; the rest follow by analogy with this case or with base cases.
  \begin{proofcases}
  \item[$\jsynt{(k, u)}{\Tplus \{l : A_l \}_{l \in L}}$]
    The elements of $\mc{E}_R(n + 1, a, \widehat a, \jsynt{(k, u)}{\Tplus \{ l : A_l \}_{l \in L}})$ are of the form $\tCase{a}{\{ k \Rightarrow E \mid \uscore \Rightarrow \Omega \}}$ for $E \in \mc{E}_R(n, a, \widehat a, \jsynt{u}{A_k})$.
    Recall that $\capxn{(k, u)}{n + 2} = (k, \capxn{u}{n + 1})$.
    By inversion, $\jsynt{(k, \capxn{u}{n + 1}) \commsim w}{\Tplus \{l : A_l \}_{l \in L}}$ if and only if $w = (k, w')$ and $\jsynt{\capxn{u}{n + 1} \commsim w'}{A_k}$.
    But $w = (k, w')$ if and only if there is an observable message $\jmsg{a}{\mSendLP{a}{k}{d}}$ in a fair trace of $\mc{C}$ for some channel $d$ with $\jtoc{T}{w'}{d}{A_k}$.
    There is a message $\jmsg{a}{\mSendLP{a}{k}{d}}$ in a fair trace of $\mc{C}$ if and only if there is such a message in a fair trace of $\mc{C}, \jproc{\widehat a}{\tCase{a}{\{ k \Rightarrow E \mid \uscore \Rightarrow \Omega \}}}$.
    Fairness implies that a fair trace of $\mc{C}, \jproc{\widehat a}{\tCase{a}{\{ k \Rightarrow E \mid \uscore \Rightarrow \Omega \}}}$ has an instantiation
    \[
      \jmsg{a}{\mSendLP{a}{k}{d}}, \jproc{\widehat a}{\tCase{a}{\{ k \Rightarrow E \mid \uscore \Rightarrow \Omega \}}}
      \msstep
      \jproc{\widehat a}{\subst{d}{a}{E}}
    \]
    of \cref{eq:sill:msr-tplus-l} if and only if $\jmsg{a}{\mSendLP{a}{k}{d}}$ appears in the trace.
    If this is the case, then we apply the induction hypothesis: $\jtoc{T}{(\mt{y}, \bot)}{\widehat a}{\mathbb{Y}^+}$ for all $\subst{d}{a}{E} \in \mc{E}_R(n, d, \widehat a, \jsynt{u}{A_k})$ if and only if $\jsynt{\capxn{u}{n + 1} \commsim w'}{A_k}$, where we recall that $w'$ is given by $\jtoc{T}{w'}{d}{A_k}$.
    It follows that
    \[
      \obsbr{\mc{C}, \jproc{\widehat a}{T}}(\widehat a) = (\mt{y}, \bot),
    \]
    for all processes $T \in \mc{E}_R(n + 1, a, \widehat a, \jsynt{(k, u)}{\Tplus \{ l : A_l \}_{l \in L}}$ if and only if both $w = (k, w')$ and $\jsynt{\capxn{u}{n + 1} \commsim w'}{A_k}$, \ie, if and only if $\jsynt{\capxn{(k, u)}{n + 2}) \commsim w}{\Tplus \{l : A_l \}_{l \in L}}$.\qedhere
  \end{proofcases}
\end{proof}

Let $\mb{0}^-$ be the negative empty type $\Tamp \{ \}$ and let $\mathbb{Y}^-$ be the ``negative answer type'' $\Tamp \{ \mt{y} : \mb{0}^- \}$.
We can dualize the definition of $\mc{E}_R$ to get a family $\mc{E}_L$ of processes that listen on the right and reports on the left.
In fact, the processes carry over unchanged:\footnote{Here, we are taking an extrinsic~\cite[\S~15.4]{reynolds_1998:_theor_progr_languag} or ``Curry-style'' view of process typing.}
\begin{align*}
  \mc{E}_L(n, i, r, \jsynt{\bot}{A}) &= \mc{E}_R(n, i, r, \jsynt{\bot}{A})\\
  \mc{E}_L(n, i, r, \jsynt{(k, v)}{\Tamp \{ l : A_l \}_{l \in L}}) &= \mc{E}_R(n, i, r, \jsynt{(k, v)}{\Tplus \{ l : A_l \}_{l \in L}})\\
  \mc{E}_L(n, i, r, \jsynt{(\cunfold, v)}{\Trec{\alpha}{A}}) &= \mc{E}_R(n, i, r, \jsynt{(\cunfold, v)}{\Trec{\alpha}{A}})\\
  \mc{E}_L(n, i, r, \jsynt{(\cshift, v)}{\Tds{A}}) &= \mc{E}_R(n, i, r, \jsynt{(\cshift, v)}{\Tus{A}})\\
  \mc{E}_L(n, i, r, \jsynt{(u,v)}{A \Tlolly B}) &= \mc{E}_R(n, r, \jsynt{(u,v)}{A \Tot B})\\
  \mc{E}_L(n, i, r, \jsynt{v}{A^+}) &= \mc{E}_R(n, i, r, \jsynt{v}{B^-})\\
  \mc{E}_L(n, i, r, \jsynt{(\cval f, u)}{\Tand{\tau}{A}}) &= \mc{E}_R(n, i, r, \jsynt{(\cval f, u)}{\Timp{\tau}{A}}).
\end{align*}

The proof of the following \namecref{prop:sill-obs-equiv/extern-observ:4} is analogous to the proof of \cref{prop:sill-obs-equiv/extern-observ:1}.

\begin{proposition}
  \label{prop:sill-obs-equiv/extern-observ:4}
  Let $\jcfgt{\Gamma, a : A}{\mc{C}}{\Delta}$ and $\jsynt{v}{A}$ be arbitrary.
  Set $w = \obsbr{\jcfgt{\Gamma, a : A}{\mc{C}}{\Delta}}(a)$ and let $\widehat a$ be globally fresh.
  Then for all $n$, $\jsynt{\capxn{v}{n + 1} \commsim w}{A}$ if and only if for all $E \in \mc{E}_L(n, a, \widehat a, \jsynt{w}{A})$,
  \[
    \obsbr{\jcfgt{\Gamma, \widehat a : \mathbb{Y}^-}{\jproc{\widehat a}{E}, \mc{C}}{\Delta}}(\widehat a) = (\mt{y}, \bot).
  \]
\end{proposition}

We combine \cref{prop:sill-obs-equiv/extern-observ:1,prop:sill-obs-equiv/extern-observ:4} to build families of experiment contexts.
Given $n$ observations $\overline{\jsynt{v_i}{A_i}/a_i}$ and $m$ observations $\overline{\jsynt{w_j}{C_j}/c_j}$ with $n \geq 0$ and $m \geq 1$, define the set of configuration contexts:
\begin{align*}
  &\mc{E}\left(n, \overline{\jsynt{v_i}{A_i}/a_i}, \overline{\jsynt{w_j}{C_j}/c_j} \right)\\
  &= \big\{ \jproc{a_1}{L_1}, \dotsc, \jproc{a_n}{L_n}, \ctxh{}{\overline{a_i : Ai}}{\overline{c_j : C_j}}, \jproc{\widehat{c_1}}{R_1}, \dotsc, \jproc{\widehat{c_m}}{R_m} \mid {} \\
  &\qquad\qquad\qquad {} \mid L_i \in \mc{E}_L\left(n, a_i, \widehat{a_i}, \jsynt{v_i}{A_i}\right), R_j \in \mc{E}_R\left(n, c_j, \widehat{c_j}, \jsynt{w_j}{C_j}\right) \big\}.
\end{align*}

\begin{proposition}
  \label{theorem:sill-obs-equiv/extern-observ:1}
  Let $\jcfgt{\Gamma}{\mc{C}}{\Delta}$ and $\jcfgt{\Gamma}{\mc{D}}{\Delta}$ be arbitrary.
  If $\Tand{\tau}{A}$ is a subphrase of a type in $\Delta$ or $\Timp{\tau}{A}$ is a subphrase of a type in $\Gamma$, then for each value $\jtypef{\cdot}{v}{\tau}$ assume the existence of an oracle process $O^\leqslant_{v : \tau}$ satisfying \cref{ass:sill-obs-equiv/extern-observ:1}.
  Let $v_i$ and $w_j$ be given by:
  \begin{align*}
    \obsbr{\jcfgt{\Gamma}{\mc{C}}{\Delta}}_{\Gamma} &= (a_i : v_i)_{a_i : A_i \in \Gamma},\\
    \obsbr{\jcfgt{\Gamma}{\mc{C}}{\Delta}}_{\Delta} &= (c_j : w_j)_{c_j : C_j \in \Delta}.
  \end{align*}
  Then
  \[
    \obsbr{\jcfgt{\Gamma}{\mc{C}}{\Delta}}_{\Gamma,\Delta} \commsim[\leqslant] \obsbr{\jcfgt{\Gamma}{\mc{D}}{\Delta}}_{\Gamma,\Delta}
  \]
  if and only if for all $n$ and $\ctxh{\mc{F}}{\Gamma}{\Delta} \in \mc{E}(n, \overline{\jsynt{v_i}{A_i}/a_i}, \overline{\jsynt{w_j}{C_j}/c_j})$,
  \[
    \obsbr*{\jcfgt{\overline{\widehat{a_i} : \mathbb{Y}^-}}{
        \ctxh[\mc{D}]{\mc{F}}{\Gamma}{\Delta}
      }{\overline{\widehat{c_j} : \mathbb{Y}^+}}}(b) = (\mt{y}, \bot)
  \]
  for all $b \in \overline{\widehat{a_i}}, \overline{\widehat{c_j}}$.
\end{proposition}

\begin{proof}
  By definition,
  \[
    \obsbr{\jcfgt{\Gamma}{\mc{C}}{\Delta}}_{\Gamma,\Delta} \commsim[\leqslant] \obsbr{\jcfgt{\Gamma}{\mc{D}}{\Delta}}_{\Gamma,\Delta}
  \]
  if and only if
  \[
    (a_i : v_i, c_j : w_j)_{a_i : A_i \in \Gamma, c_j : C_j \in \Delta} \commsim[\leqslant] \obsbr{\jcfgt{\Gamma}{\mc{D}}{\Delta}}_{\Gamma,\Delta}.
  \]
  By \cref{prop:sill-obs-equiv/barb-cont-congr:11}, this is the case if and only if for all $n$,
  \[
    (a_i : \capxn{v_i}{n}, c_j : \capxn{w_j}{n})_{a_i : A_i \in \Gamma, c_j : C_j \in \Delta} \commsim[\leqslant] \obsbr{\jcfgt{\Gamma}{\mc{D}}{\Delta}}_{\Gamma,\Delta}.
  \]
  The result then follows by fairness and \cref{prop:sill-obs-equiv/extern-observ:4,prop:sill-obs-equiv/extern-observ:1}.
\end{proof}

\Cref{theorem:sill-obs-equiv/extern-observ:1} and the analysis following \cref{question:sill-obs-equiv/extern-observ:1} imply:

\begin{corollary}
  \label{cor:sill-obs-equiv/extern-observ:5}
  For each value $\jtypef{\cdot}{v}{\tau}$ assume the existence of an oracle process $O^\leqslant_{v : \tau}$ satisfying \cref{ass:sill-obs-equiv/extern-observ:1}.
  If $\jtrelc{\mkcg{\wbsim}}{\Gamma}{\mc{C}}{\mc{D}}{\Delta}$, then $\jtrelc{\sobssim{(\mc{X}^E, {\leqslant})}}{\Gamma}{\mc{C}}{\mc{D}}{\Delta}$.
\end{corollary}

\begin{corollary}
  \label{cor:sill-obs-equiv/extern-observ:2}
  If $\jtrelc{\mkcg{\wbsim}}{\Gamma}{\mc{C}}{\mc{D}}{\Delta}$, then $\jtrelc{\eocssim}{\Gamma}{\mc{C}}{\mc{D}}{\Delta}$.
\end{corollary}

\begin{proof}
  By \cref{cor:sill-obs-equiv/extern-observ:5}.
  The oracle $O^{\unirel}_{v : \tau}$ assumed by \cref{cor:sill-obs-equiv/extern-observ:5} is given by:
  \[
    \jtypem{\cdot}{\cdot}{\tRecvV{\uscore}{c}{\tRecvS{c}{\tSendL{c}{\mt{tt}}{\tClose{c}}}}}{c}{\Timp{\tau}{\Tus{\Tplus \{ \mt{tt} : \Tu, \mt{ff} : \Tu \}}}}.\qedhere
  \]
\end{proof}

Combining \cref{cor:sill-obs-equiv/extern-observ:2,prop:sill-obs-equiv/barb-cont-congr:7} gives:

\begin{theorem}
  \label{theorem:sill-obs-equiv/extern-observ:2}
  For all $\jcfgt{\Gamma}{\mc{C}}{\Delta}$ and $\jcfgt{\Gamma}{\mc{D}}{\Delta}$, $\jtrelc{\mkcg{\wbsim}}{\Gamma}{\mc{C}}{\mc{D}}{\Delta}$ if and only if $\jtrelc{\eocssim}{\Gamma}{\mc{C}}{\mc{D}}{\Delta}$.
\end{theorem}

\subsection{Summary of Relations}
\label{sec:sill-obs-equiv:extern-observ:relations}

\Cref{fig:sill-obs-equiv/extern-observ:1} summarizes the main results for relations on configurations.
Double arrows denote implications.
Dashed arrows denote conjectured implications.
Missing arrows (when not implied by transitivity) indicate falsehoods.
We recall that: \( \mssteps \) is the reflexive, transitive closure of \( \msstep \); \( \iocssim \) is internal observational simulation; \( \iocsprec\) is internal observational precongruence; \( \tcommeq \) is total observational equivalence; \( \tocssim \) is total observational simulation; \( \tocsprec \) is total observational precongruence; \( \eocssim \) is external observational simulation; \( \eocsprec \) is external observational precongruence; \( \wbsim \) is barbed simulation; and \( \mkcg{\wbsim} \) is barbed precongruence.

\begin{figure}
  \begin{equation*}
    \begin{tikzcd}[column sep=5em, row sep=2em]
      {\mssteps}
      \ar[d, Rightarrow, "{\text{\cref*{prop:sill-obs-equiv/observ-comm:8}}}"]
      &
      \iocssim
      &
      &
      \iocsprec
      \ar[ll, Rightarrow, swap, "{\text{def.}}"]
      \ar[dd, Rightarrow, bend left, sloped, "{\text{\cref*{cor:sill-obs-equiv/intern-observ:1}}}"]
      \\
      \tcommeq
      \ar[r, Rightarrow, "{\text{\cref*{prop:sill-obs-equiv/observ-comm-equiv:10}}}"]
      &
      \tocssim
      \ar[r, Leftrightarrow, "{\text{\cref*{prop:sill-obs-equiv/total-observ:3}}}"]
      \ar[u, Rightarrow, swap, "{\text{\cref*{prop:sill-obs-equiv/total-observ:1}}}"]
      \ar[d, Rightarrow, "{\text{\cref*{prop:sill-obs-equiv/total-observ:1}}}"]
      &
      \tocsprec
      \ar[ur, Rightarrow, sloped, "{\text{monot.}}"]
      \ar[dr, Rightarrow, sloped, "{\text{monot.}}"]
      &
      \\
      &
      \eocssim
      \ar[rr, Leftrightarrow, "{\text{\cref*{prop:sill-obs-equiv/total-observ:3}}}"]
      \ar[d, Rightarrow, "{\text{\cref*{prop:sill-obs-equiv/barb-cont-congr:7}}}"]
      &
      &
      \eocsprec
      \ar[d, Leftrightarrow, swap, "{\text{\cref*{theorem:sill-obs-equiv/extern-observ:2}}}"]
      \ar[uu, Rightarrow, dashed, bend left, sloped, "{\text{conj.}}"]
      \\
      &
      \wbsim
      &
      &
      \mkcg{\wbsim}
      \ar[ll, Rightarrow, swap, "{\text{def.}}"]
    \end{tikzcd}
  \end{equation*}
  \caption{Relationship between relations of \cref{sec:sill-obs-equiv:extern-observ}}
  \label{fig:sill-obs-equiv/extern-observ:1}
\end{figure}
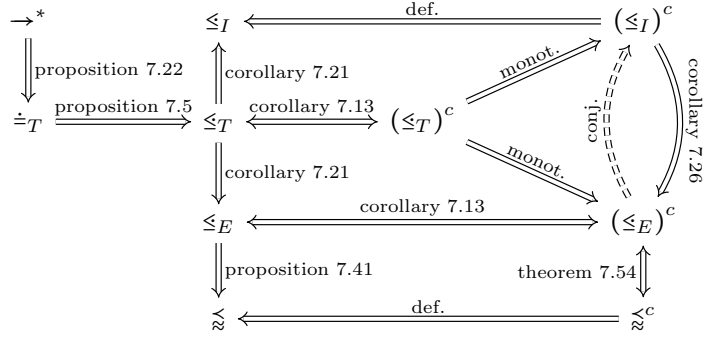

We conjecture in \cref{fig:sill-obs-equiv/extern-observ:1} that external communication precongruence implies internal communication precongruence.
We give an intuitive justification for this conjecture.
Assume that \(\jtrelc{\eocsprec}{\Gamma}{\mc{C}}{\mc{D}}{\Delta}\).
This means that for all \(\jcfgt{\Lambda}{\ctxh{\mc{F}}{\Gamma}{\Delta}}{\Phi}\) and \(\jcfgt{\Psi}{\ctxh{\mc{E}}{\Lambda}{\Phi}}{\Xi}\), we observe similar communications on the external channels \(\Psi\Xi\) of \(\ctxh[{\ctxh[\mc{C}]{\mc{F}}{\Gamma}{\Delta}}]{\mc{E}}{\Lambda}{\Phi}\) and \(\ctxh[{\ctxh[\mc{C}]{\mc{F}}{\Gamma}{\Delta}}]{\mc{E}}{\Lambda}{\Phi}\):
\begin{equation}
  \label{eq:sill-obs-equiv-summary-relations:1}
  \obsbr{\jcfgt{\Psi}{\ctxh[{\ctxh[\mc{C}]{\mc{F}}{\Gamma}{\Delta}}]{\mc{E}}{\Lambda}{\Phi}}{\Xi}}_{\Psi\Xi} \commsim[\unirel]
  \obsbr{\jcfgt{\Psi}{\ctxh[{\ctxh[\mc{D}]{\mc{F}}{\Gamma}{\Delta}}]{\mc{E}}{\Lambda}{\Phi}}{\Xi}}_{\Psi\Xi}.
\end{equation}
To prove the conjecture, we must show that we observe similar communications on the internal channels~\(\Lambda\Phi\):
\[
  \obsbr{\jcfgt{\Psi}{\ctxh[{\ctxh[\mc{C}]{\mc{F}}{\Gamma}{\Delta}}]{\mc{E}}{\Lambda}{\Phi}}{\Xi}}_{\Lambda\Phi} \commsim[\unirel]
  \obsbr{\jcfgt{\Psi}{\ctxh[{\ctxh[\mc{D}]{\mc{F}}{\Gamma}{\Delta}}]{\mc{E}}{\Lambda}{\Phi}}{\Xi}}_{\Lambda\Phi}.
\]
Intuitively, this must be the case: if we observed dissimilar communications on \(\Lambda\Phi\), then we should be able to construct an experiment \(\ctxh{\mc{E}}{\Lambda}{\Phi}\) that uses these dissimilar communications to produce dissimilar communications on its external channels, contradicting \cref{eq:sill-obs-equiv-summary-relations:1}.
Put differently, we suspect that \cref{theorem:sill-obs-equiv/extern-observ:1} could be adapted to establish the conjecture.
The difficulty in doing so stems from the fact that the families of experiment contexts used in \cref{theorem:sill-obs-equiv/extern-observ:1} only had to account for unidirectional communications (we were only concerned with communication on external channels), while in this case, we must construct families of experiment contexts that account for bidirectional communications (we are now concerned with communication on internal channels).

\subsection{Precongruences for Processes}
\label{sec:sill-obs-equiv:proc-equiv}

In this section, we relate relations on configurations to relations on processes.
Recall that to show that two configurations are total or external precongruent, it is sufficient by \cref{prop:sill-obs-equiv/extern-observ:6} to consider only simply branched contexts.
In \cref{sec:sill-obs-equiv:proc-equiv:relat-simply-branch}, we show that simply branched configuration contexts closely mirror ``observation contexts'' for processes.
In \cref{sec:sill-obs-equiv:proc-equiv:lift-relat-conf}, we show how to lift relations from configurations to open processes.
We show that there is a tight correspondence between various notions of precongruence on configurations and on open processes.

\subsubsection{Relating Simply Branched Contexts and Observation Contexts}
\label{sec:sill-obs-equiv:proc-equiv:relat-simply-branch}

Observation contexts capture the idea of processes experimenting on processes through communication.
Intuitively, an observation context is a process context that only interacts with its hole by communicating:

\begin{definition}
  \label{def:sill-obs-equiv/relat-equiv:4}
  An \defin[context!observation|defin]{observation context} is a typed context derived using exactly one instance of the axiom \getrn{Ct-p-hole}, plus zero or more instances of the derived rules \getrn{Ct-hole-cut-l} and \getrn{Ct-hole-cut-r},
  \begin{gather*}
    \getrule{Ct-hole-cut-l}
    \\
    \getrule{Ct-hole-cut-r}
  \end{gather*}
  such that the context satisfies the grammar:
  \[
    \ctxh{O}{\Delta}{a : A} \Coloneqq \ctxh{}{\Delta}{a:A} \mid \tCut{b}{\ctxh{O}{\Delta}{a:A}}{P} \mid \tCut{b}{P}{\ctxh{O}{\Delta}{a:A}}.\qedhere
  \]
\end{definition}

Recall from \cref{def:sill-obs-equiv/bisimulations:2} the definition of a contextual relation on processes.
``Observational contextuality'' weakens this notion from arbitrary contexts to observation contexts:

\begin{definition}
  \label{def:sill-obs-equiv/proc-equiv:2}
  A typed-indexed relation $\relfnt{R}$ on processes is \defin{observationally contextual} if ${\jtrelp{\relfnt{R}}{\cdot}{\Delta}{P}{Q}{a}{A}}$ implies $\jtrelp{\relfnt{R}}{\cdot}{\Lambda}{\ctxh[P]{O}{\Delta}{a : A}}{\ctxh[Q]{O}{\Delta}{a : A}}{b}{B}$ for all observation contexts ${\jtypem{\cdot}{\Lambda}{\ctxh{O}{\Delta}{a:A}}{b}{B}}$.
\end{definition}

\begin{definition}
  \label{def:sill-obs-equiv/proc-equiv:4}
  The \defin{observationally contextual interior} of a typed relation $\relfnt{R}$ on processes is the greatest observationally contextual typed relation $\mkocg{\relfnt{R}}$\glsadd{mkocg} contained in $\relfnt{R}$.
\end{definition}

By analogy with \cref{prop:sill-obs-equiv/bisimulations:2}, we can think of \(\mkocg{\relfnt{R}} \subseteq \relfnt{R}\) as a limited precongruence on processes when \(\relfnt{R}\) is a preorder.

There is an obvious translation from observation contexts to configuration contexts, where we inductively map \getrn{Ct-p-hole} to \getrn{conf-h}, and \getrn{Ct-hole-cut-l} and \getrn{Ct-hole-cut-r} to \getrn{conf-c}:

\begin{proposition}
  \label{prop:sill-obs-equiv/relat-equiv:1}
  Let $\jtypem{\cdot}{\Lambda}{\ctxh{O}{\Delta}{a:A}}{c}{C}$ be an observation context.
  There exists a configuration context $\jcfgt{\Lambda}{\ctxh{\mc{O}}{\Delta}{a : A}}{c : C}$ such that $\jproc{c}{\ctxh[Q]{O}{\Delta}{a : A}} \mssteps \ctxh[\jproc{a}{Q}]{\mc{O}}{\Delta}{a : A}$ for all processes $\jtypem{\cdot}{\Delta}{Q}{a}{A}$.
\end{proposition}

\begin{proof}
  By induction on the derivation of the observation context.
  \begin{proofcases}
  \item[\getrn{Ct-p-hole}] Let $\mc{O}$ be given by \getrn{conf-h}.
    The step is given by reflexivity.
  \item[\getrn{Ct-hole-cut-l}] The observation context is $\getrc{Ct-hole-cut-l}$, and it is formed by:
    \[
      \getrule{Ct-hole-cut-l}
    \]
    By the induction hypothesis, there exists a configuration context $\jcfgti{\Delta_1}{\Iota'}{\ctxh{\mc{O}'}{\Delta}{a : A}}{b : B}$ such that $\jproc{b}{\ctxh[Q]{O}{\Delta}{a : A}} \mssteps \ctxh[\jproc{a}{Q}]{\mc{O}'}{\Delta}{a : A}$.
    This implies that
    \[
      \jproc{c}{\tCut{b}{\ctxh[Q]{O}{\Delta}{a : A}}{P}} \msstep \jproc{b}{\ctxh[Q]{O}{\Delta}{a : A}}, \jproc{c}{P} \mssteps \ctxh[\jproc{a}{Q}]{\mc{O}'}{\Delta}{a : A}, \jproc{c}{P}.
    \]
    Let $\mc{O}$ be given by:
    \[
      \infer[\getrn{conf-c}]{
        \jcfgti{\Delta_1 \Delta_2}{\Iota', b : B}{\ctxh{\mc{O}'}{\Delta}{a : A}, \jproc{c}{P}}{c : C}
      }{
        \jcfgti{\Delta_1}{\Iota'}{\ctxh{\mc{O}'}{\Delta}{a : A}}{b : B}
        &
        \infer[\getrn{conf-p}]{
          \jcfgti{\Delta_2, b : B}{\cdot}{\jproc{c}{P}}{c : C}
        }{
          \getrh{Ct-hole-cut-l}{2}
        }
      }
    \]
    Plugging $\jproc{a}{Q}$ into $\mc{O}$ gives the configuration
    \[
      \jcfgti{\Delta_1, \Delta_2}{\Iota', b : B}{\ctxh[\jproc{a}{Q}]{\mc{O}'}{\Delta}{a : A}, \jproc{c}{P}}{c : C}
    \]
    that we recognize as the result of the above sequence of rewrite steps.
    We conclude that
    \[
      \jproc{c}{\tCut{b}{\ctxh[Q]{O}{\Delta}{a : A}}{P}} \mssteps \ctxh[\jproc{a}{Q}]{\mc{O}}{\Delta}{a : A}.
    \]
  \item[\getrn{Ct-hole-cut-r}]
    This case is symmetric to the previous case.\qedhere
  \end{proofcases}
\end{proof}

The translation in the opposite direction is more subtle.
To translate configuration contexts to observation contexts, one must be able to translate configurations to processes.
Naïvely, one would hope that:

\begin{falsehood}
  If $\jcfgt{\Gamma}{\mc{C}}{c : C}$, then there exists a process $\jtypem{\cdot}{\Gamma}{P}{c}{C}$ such that $\jproc{c}{P} \mssteps \mc{C}$.
\end{falsehood}

This is often impossible when configurations contain message facts.
Consider, for example, the configuration $\jcfgti{d : A}{\cdot}{\jmsg{c}{\mSendCP{c}{k}{d}}}{c : \Tplus \{ k : A \}}$.
The only plausible solutions are variations on the theme $\jtypem{\cdot}{d : A}{\tSendL{c}{k}{\tFwdP{d}{c}}}{d}{c : \Tplus \{ k : A \}}$.
However,
\[
  \jproc{c}{\tSendL{c}{k}{\tFwdP{d}{c}}} \msstep \jproc{d'}{\tFwdP{d}{d'}}, \jmsg{c}{\mSendCP{c}{k}{d'}}
\]
and there is no way to get rid of the forwarding process fact.
Instead, we settle for \cref{prop:sill-obs-equiv/observ-comm-equiv:5}.
It states that every simply branched configuration \(\mc{P}\) can be encoded as a process \(P\), and that \(\mc{P}\) is total observationally equivalent to the process fact \(\jproc{c}{P}\).
Because total observational equivalence is the finest observational equivalence (\cref{prop:sill-obs-equiv/total-observ:1}), this implies that no observational simulation can distinguish \(\mc{P}\) and \(\jproc{c}{P}\).

\begin{proposition}
  \label{prop:sill-obs-equiv/observ-comm-equiv:5}
  If $\jcfgt{\Delta}{\mc{P}}{c : A}$, then $\jtrelc{\strobsc}{\Delta}{\mc{P}}{\jproc{a}{P}}{c : A}$ for some process $\jtypem{\cdot}{\Delta}{P}{c}{A}$
\end{proposition}

\begin{proof}
  By \cref{prop:sill-obs-equiv/observ-comm:3}, $\jcfgt{\Delta}{\mc{P}}{c : A}$ has a simply-branched derivation.
  We proceed by induction on this derivation.
  We give only the illustrative cases.
  \begin{proofcases}
  \item[\getrn{conf-m}] We proceed by case analysis on the particular message fact in
    \[
      \getrule{conf-m}
    \]
    \begin{proofcases}
    \item[$m = \mClose c$] Take $P = \tClose{c}$ and apply \cref{eq:sill:msr-tu-r} and \cref{prop:sill-obs-equiv/observ-comm:8} (closure of \(\strobsc\) under stepping).
    \item[$m = \mSendLP{c}{k}{d}$] Take $P = m$ and apply \cref{eq:sill:msr-tplus-r,prop:sill-obs-equiv/observ-comm:8} to get:
      \[
        \jtrelc{\strobsc}{d : A_k}{\jproc{c}{\mSendLP{c}{k}{d}}}{\jproc{e}{\tFwdP{d}{e}}, \jmsg{c}{\mSendLP{c}{k}{e}}}{c : \Tplus \{ l : A_l \}_{l \in L}}.
      \]
      By \cref{prop:sill-obs-equiv/observ-comm-equiv:2},
      \[
        \jtrelc{\strobsc}{d : A_k}{\jmsg{c}{\mSendLP{c}{k}{d}}}{\jproc{e}{\tFwdP{d}{e}}, \jmsg{c}{\mSendLP{c}{k}{e}}}{c : \Tplus \{ l : A_l \}_{l \in L}}.
      \]
      We conclude the result by transitivity and symmetry.
    \end{proofcases}
  \item[\getrn{conf-p}] Immediate by reflexivity.
  \item[\getrn{conf-c}] By assumption, both branches of the rule
    \[
      \getrule{conf-c}
    \]
    are simply branched, and $\Pi = b : B$ contains a single channel.
    By the induction hypothesis, there exist processes $C$ and $D$ such that
    \( \jtrelc{\strobsc}{\Gamma}{\mc{C}}{\jproc{b}{C}}{b : B} \) and \( \jtrelc{\strobsc}{b : B, \Lambda}{\mc{D}}{\jproc{c}{D}}{c : A} \).
    Take the process $\jtypem{\cdot}{\Gamma\Lambda}{\tCut{b}{C}{D}}{c}{A}$.
    Then
    \[
      \jproc{a}{\tCut{b}{D}{C}} \msstep \jproc{b}{C}, \jproc{c}{D},
    \]
    so $\jtrelc{\strobsc}{\Gamma\Lambda}{\jproc{a}{\tCut{b}{C}{D}}}{\jproc{b}{C}, \jproc{c}{D}}{c : A}$ by \cref{prop:sill-obs-equiv/observ-comm:8}.
    Because $\strobsc$ is a congruence,
    \[
      \jtrelc{\strobsc}{\Gamma\Lambda}{\jproc{a}{\tCut{b}{C}{D}}}{\mc{C}, \mc{D}}{c : A}
    \]
    as desired.\qedhere
  \end{proofcases}
\end{proof}

\Cref{prop:sill-obs-equiv/observ-comm-equiv:6} extends \cref{prop:sill-obs-equiv/relat-equiv:1} to give the correspondence between observation contexts for processes and simply branched configuration contexts.

\begin{proposition}
  \label{prop:sill-obs-equiv/observ-comm-equiv:6}
  \leavevmode
  \begin{enumerate}
  \item For all configuration contexts $\jcfgt{\Lambda}{\ctxh{\mc{O}}{\Delta}{a : A}}{c : C}$, there exists an observation context $\jtypem{\cdot}{\Lambda}{\ctxh{O}{\Delta}{a : A}}{c}{C}$ such that for all $\jtypem{\cdot}{\Delta}{Q}{a}{A}$,
    \[
      \jtrelc{\strobsc}{\Lambda}{\ctxh[\jproc{a}{Q}]{\mc{O}}{\Delta}{a:A}}{\jproc{c}{\ctxh[Q]{O}{\Delta}{a : A}}}{c : C}.
    \]
  \item For all observation contexts $\jtypem{\cdot}{\Lambda}{\ctxh{O}{\Delta}{a:A}}{c}{C}$, there exists a configuration context $\jcfgt{\Lambda}{\ctxh{\mc{O}}{\Delta}{a : A}}{c : C}$ such that for all $\jtypem{\cdot}{\Delta}{Q}{a}{A}$,
    \[
      \jtrelc{\strobsc}{\Lambda}{\ctxh[\jproc{a}{Q}]{\mc{O}}{\Delta}{a : A}}{\jproc{c}{\ctxh[Q]{O}{\Delta}{a : A}}}{c : C}.
    \]
  \end{enumerate}
\end{proposition}

\begin{proof}
  We show the first part of the proposition.
  Let $\jcfgt{\Lambda}{\ctxh{\mc{O}}{\Delta}{a : A}}{c : C}$ be arbitrary.
  By \cref{prop:sill-obs-equiv/observ-comm:3}, it has a simply-branched derivation.
  We proceed by induction on this derivation to construct $\jtypem{\cdot}{\Lambda}{\ctxh{O}{\Delta}{a : A}}{c}{C}$.
  The possible cases are:
  \begin{proofcases}
  \item[\getrn{conf-h}] If $\mc{O}$ is a hole, then let $O = \ctxh{}{\Delta}{a : A}$.
  \item[\getrn{conf-c}] Then the context is formed by an instance of
    \[
      \getrule{conf-c}
    \]
    By simple-branching, $\Pi = b : B$ contains a single channel.
    If the hole is in the left branch, \ie, if $\mc{C} = \ctxh{\mc{O}'}{\Delta}{a : A}$, then by the induction hypothesis, there exists an observation context $\jtypem{\cdot}{\Lambda}{\ctxh{O'}{\Delta}{a : A}}{b}{B}$ such that for all $\jtypem{\cdot}{\Delta}{Q}{a}{A}$,
    \[
      \jtrelc{\strobsc}{\Lambda}{\ctxh[\jproc{a}{Q}]{\mc{O}'}{\Delta}{a:A}}{\jproc{c}{\ctxh[Q]{O'}{\Delta}{a : A}}}{b : B}.
    \]
    Let $D$ be given for $\mc{D}$ by \cref{prop:sill-obs-equiv/observ-comm-equiv:5} such that $\jtrelc{\strobsc}{b : B, \Lambda}{\mc{D}}{\jproc{c}{D}}{c : A}$.
    Take $\jtypem{\cdot}{\Lambda}{\ctxh{O}{\Delta}{a : A}}{c}{C}$ to be given by $\tCut{b}{\ctxh{O'}{\Delta}{a : A}}{D}$.
    Then for all $\jtypem{\cdot}{\Delta}{Q}{a}{A}$,
    \[
      \jproc{c}{\ctxh[Q]{O}{\Delta}{a : A}} \msstep \jproc{b}{\ctxh[Q]{O'}{\Delta}{a : A}}, \jproc{c}{D},
    \]
    so by \cref{prop:sill-obs-equiv/observ-comm:8}, \( \jtrelc{\strobsc}{\Lambda}{\jproc{c}{\ctxh[Q]{O}{\Delta}{a : A}}}{\jproc{b}{\ctxh[Q]{O'}{\Delta}{a : A}}, \jproc{c}{D}}{c : C} \).
    But $\strobsc$ is a congruence, so
    \[
      \jtrelc{\strobsc}{\Lambda}{\jproc{c}{\ctxh[Q]{O}{\Delta}{a : A}}}{\ctxh[\jproc{a}{Q}]{\mc{O}'}{\Delta}{a:A}, \mc{D}}{c : C},
    \]
    \ie, $\jtrelc{\strobsc}{\Lambda}{\jproc{c}{\ctxh[Q]{O}{\Delta}{a : A}}}{\ctxh[\jproc{a}{Q}]{\mc{O}}{\Delta}{a:A}}{c : C}$.
    The result follows by symmetry.
    The case for when the hole is in the right branch is analogous.
  \end{proofcases}

  The second part of the proposition is immediate by \cref{prop:sill-obs-equiv/observ-comm:8,prop:sill-obs-equiv/relat-equiv:1}.
\end{proof}

\subsubsection{Relating Precongruences on Configurations and Processes}
\label{sec:sill-obs-equiv:proc-equiv:lift-relat-conf}

The precongruence relations we studied in \cref{sec:total-obs-equiv:total-observ,sec:sill-obs-equiv:intern-observ,sec:sill-obs-equiv:extern-observ,sec:sill-obs-equiv:extern-observ:relations} were all defined on configurations.
In particular, they are \emph{configuration} precongruences.
By \cref{prop:sill-obs-equiv/bisimulations:2}, this means that they are closed under composition with configuration contexts.
In \cref{sec:sill-obs-equiv:proc-equiv:relat-simply-branch}, we showed how to translate between configuration contexts and observation contexts on processes.
In this subsection, we show that contextual relations on configurations lift to contextual relations on open processes.
By analogy with \cref{prop:sill-obs-equiv/bisimulations:2}, we can think of these results as stating that restricted forms of precongruence on configurations lift to restricted forms of precongruences on processes.

In this section, we frequently assume that if \(\preccurlyeq\) is a relation on configurations, then ${\strobsc} \subseteq {\preccurlyeq}$.
Recall from \cref{fig:sill-obs-equiv/extern-observ:1} that this assumption is satisfied by all observational preorders $\preccurlyeq$ that we have considered thus far.

\Cref{def:sill-obs-equiv/observ-comm:1} lifts type-indexed relations \(\relfnt{R}\) on configurations to type-indexed relations \(\relcpl{\relfnt{R}}\) on open processes using an approach reminiscent of Howe's ``open extensions''~\cite[Definition~2.2]{howe_1996:_provin_congr_bisim}, which used closing substitutions to extend relations on closed terms to relations on open terms.
Recall that a substitution ${\jcmf{\sigma}{\Phi}{x_1 : \tau_1, \dotsc, x_n : \tau_n}}$ is a list $\sigma$ of terms $N_1, \dotsc, N_n$ satisfying $\jtypef{\Phi}{N_i}{\tau_i}$ for all $1 \leq i \leq n$, and that we write $\apprs{\sigma}{P}$ for the simultaneous substitution $\subst{\vec N}{\vec x}{P}$.

\begin{definition}
  \label{def:sill-obs-equiv/proc-equiv:1}
  A \defin{closing substitution} is a substitution $\jcmf{\sigma}{\cdot}{\Psi}$ such that $\fnval{\sigma(x)}$ for all $x : \tau \in \Psi$.
\end{definition}

\begin{definition}
  \label{def:sill-obs-equiv/observ-comm:1}
  Let $\relfnt{R}$ be a type-indexed relation on configurations.
  Write $\jtrelp{\relcpl{\relfnt{R}}}{\Psi}{\Delta}{P}{Q}{c}{C}$ if $\jtrelc{\relfnt{R}}{\Delta}{\jproc{c}{\apprs{\sigma}{P}}}{\jproc{c}{\apprs{\sigma}{Q}}}{c : C}$ for all closing substitutions $\jcmf{\sigma}{\cdot}{\Psi}$.
\end{definition}

When \(\relfnt{R}\) is an observational \(\mc{S}\)-equivalence, $\jtrelp{\relcpl{\relfnt{R}}}{\Psi}{\Delta}{P}{Q}{c}{C}$ means that whenever we close \(P\) and \(Q\) using some closing substitution \(\sigma\), the experiments in \(\mc{S}\) cannot distinguish the processes \(\apprs{\sigma}{P}\) and \(\apprs{\sigma}{Q}\).
Closing substitutions are defined in terms of closed \emph{values} to capture the call-by-value semantics of Polarized SILL's functional layer: variables in processes stand for values and, operationally, we only ever substitute closed values for variables.

Recall the definition of simply branched contextual interior $\mkbcg{\relfnt{R}}$ of a relation on configurations from \cref{def:sill-background-relat-equiv:1}, and the observationally contextual interior \(\mkocg{\relfnt{R}}\) of a relation on processes from \cref{def:sill-obs-equiv/proc-equiv:4}.
The following \namecref{cor:sill-obs-equiv/proc-equiv:4} uses the translation between configuration contexts and observation contexts (\cref{prop:sill-obs-equiv/observ-comm-equiv:6}) to translate between a restricted form of precongruence \(\mkbcg{\preccurlyeq}\) on configurations and a restricted form of precongruence \(\mkocg{\relcpl*{\preccurlyeq}}\) on processes.

\begin{proposition}
  \label{cor:sill-obs-equiv/proc-equiv:4}
  Let $\preccurlyeq$ be a transitive type-indexed relation on configurations such that ${\strobsc} \subseteq {\preccurlyeq}$.
  The following are equivalent:
  \begin{enumerate}
  \item $\jtrelp{\relcpl*{\mkbcg{\preccurlyeq}}}{\Psi}{\Delta}{P}{Q}{c}{C}$;\label{item:sill-obs-equiv/observ-comm-equiv:1}
  \item $\jtrelp{\mkocg{\relcpl*{\preccurlyeq}}}{\Psi}{\Delta}{P}{Q}{c}{C}$.\label{item:sill-obs-equiv/observ-comm-equiv:3}
  \end{enumerate}
\end{proposition}

\begin{proof}
  Remark that, because $\strobsc$ is a congruence, ${\strobsc} \subseteq {\preccurlyeq}$ implies ${\strobsc} \subseteq {\mkbcg{\preccurlyeq}}$.
  Observe that $\jtrelp{\mkocg{\relcpl*{\preccurlyeq}}}{\Psi}{\Delta}{P}{Q}{c}{C}$ if and only if both
  \begin{enumerate}[label=(\roman*)]
  \item $\jtrelp{\relcpl*{\preccurlyeq}}{\Psi}{\Delta}{P}{Q}{c}{C}$; and
  \item $\jtrelp{\relcpl*{\preccurlyeq}}{\cdot}{\Gamma}{\ctxh[\apprs{\sigma}{P}]{O}{\Delta}{c : C}}{\ctxh[\apprs{\sigma}{Q}]{O}{\Delta}{c : C}}{b}{B}$ for all closing substitutions $\jcmf{\sigma}{\cdot}{\Psi}$ and all observation contexts $\jtypem{\cdot}{\Gamma}{\ctxh{O}{\Delta}{c:C}}{b}{B}$.
  \end{enumerate}

  To see that \cref{item:sill-obs-equiv/observ-comm-equiv:1} implies \cref{item:sill-obs-equiv/observ-comm-equiv:3}, assume that $\jtrelp{\relcpl*{\mkbcg{\preccurlyeq}}}{\Psi}{\Delta}{P}{Q}{c}{C}$.
  This implies that $\jtrelp{\relcpl*{\mkbcg{\preccurlyeq}}}{\cdot}{\Delta}{\apprs{\sigma}{P}}{\apprs{\sigma}{Q}}{c}{C}$ for all $\jcmf{\sigma}{\cdot}{\Psi}$.
  Let $\jtypem{\cdot}{\Gamma}{\ctxh{O}{\Delta}{c:C}}{b}{B}$ be an arbitrary observation context, and let $\jcmf{\sigma}{\cdot}{\Psi}$ be an arbitrary closing substitution.
  We must show that
  \begin{equation}
    \label{eq:sill-obs-equiv/proc-equiv:3}
    \jtrelp{\relcpl*{\preccurlyeq}}{\cdot}{\Gamma}{\ctxh[\apprs{\sigma}{P}]{O}{\Delta}{c : C}}{\ctxh[\apprs{\sigma}{Q}]{O}{\Delta}{c : C}}{b}{B}.
  \end{equation}
  By \cref{prop:sill-obs-equiv/observ-comm-equiv:6}, there exists a simply branched context $\jcfgt{\Gamma}{\ctxh{\mc{O}}{\Delta}{c : C}}{b : B}$ such that
  \begin{gather*}
    \jtrelc{\strobsc}{\Gamma}{\ctxh[\jproc{c}{\apprs{\sigma}{P}}]{\mc{O}}{\Delta}{c : C}}{\jproc{b}{\ctxh[\apprs{\sigma}{P}]{O}{\Delta}{c : C}}}{b : B},\\
    \jtrelc{\strobsc}{\Gamma}{\ctxh[\jproc{c}{\apprs{\sigma}{Q}}]{\mc{O}}{\Delta}{c : C}}{\jproc{b}{\ctxh[\apprs{\sigma}{Q}]{O}{\Delta}{c : C}}}{b : B}.
  \end{gather*}
  Because $\jtrelp{\relcpl*{\mkbcg{\preccurlyeq}}}{\Psi}{\Delta}{P}{Q}{c}{C}$ and $\mkbcg{\preccurlyeq}$ is simply branched contextual,
  \[
    \jtrelc{\mkbcg{\preccurlyeq}}{\Gamma}{\ctxh[\jproc{c}{\apprs{\sigma}{P}}]{\mc{O}}{\Delta}{c : C}}{\ctxh[\jproc{c}{\apprs{\sigma}{Q}}]{\mc{O}}{\Delta}{c : C}}{b : B}.
  \]
  By assumption, the symmetric relation $\strobsc$ is contained in $\mkbcg{\preccurlyeq}$.
  By transitivity of $\mkbcg{\preccurlyeq}$,
  \[
    \jtrelc{\mkbcg{\preccurlyeq}}{\Gamma}{\jproc{b}{\ctxh[\apprs{\sigma}{P}]{O}{\Delta}{c : C}}}{\jproc{b}{\ctxh[\apprs{\sigma}{Q}]{O}{\Delta}{c : C}}}{b : B}.
  \]
  But ${\mkbcg{\preccurlyeq}} \subseteq {\preccurlyeq}$, so we conclude \cref{eq:sill-obs-equiv/proc-equiv:3}.

  To see that \cref{item:sill-obs-equiv/observ-comm-equiv:3} implies \cref{item:sill-obs-equiv/observ-comm-equiv:1}, assume that $\jtrelp{\mkocg{\relcpl*{\preccurlyeq}}}{\Psi}{\Delta}{P}{Q}{c}{C}$, and let $\jcfgt{\Gamma}{\ctxh{\mc{O}}{\Delta}{c : C}}{b : B}$ and $\jcmf{\sigma}{\cdot}{\Psi}$ be arbitrary.
  We must show that $\jtrelc{\preccurlyeq}{\Gamma}{\ctxh[\jproc{c}{\apprs{\sigma}{P}}]{\mc{O}}{\Delta}{c : C}}{\ctxh[\jproc{c}{\apprs{\sigma}{Q}}]{\mc{O}}{\Delta}{c : C}}{b : B}$.
  By  \cref{prop:sill-obs-equiv/observ-comm-equiv:6}, there exists an observation context $\jcfgt{\Gamma}{\ctxh{O}{\Delta}{c : C}}{b : B}$ such that
  \begin{gather*}
    \jtrelc{\strobsc}{\Gamma}{\ctxh[\jproc{c}{\apprs{\sigma}{P}}]{\mc{O}}{\Delta}{c : C}}{\jproc{b}{\ctxh[\apprs{\sigma}{P}]{O}{\Delta}{c : C}}}{b : B},\\
    \jtrelc{\strobsc}{\Gamma}{\ctxh[\jproc{c}{\apprs{\sigma}{Q}}]{\mc{O}}{\Delta}{c : C}}{\jproc{b}{\ctxh[\apprs{\sigma}{Q}]{O}{\Delta}{c : C}}}{b : B}.
  \end{gather*}
  The symmetric relation $\strobsc$ is contained in $\preccurlyeq$, so
  \begin{gather*}
    \jtrelc{\preccurlyeq}{\Gamma}{\ctxh[\jproc{c}{\apprs{\sigma}{P}}]{\mc{O}}{\Delta}{c : C}}{\jproc{b}{\ctxh[\apprs{\sigma}{P}]{O}{\Delta}{c : C}}}{b : B},\\
    \jtrelc{\preccurlyeq}{\Gamma}{\jproc{b}{\ctxh[\apprs{\sigma}{Q}]{O}{\Delta}{c : C}}}{\ctxh[\jproc{c}{\apprs{\sigma}{Q}}]{\mc{O}}{\Delta}{c : C}}{b : B}.
  \end{gather*}
  By assumption, $\jtrelc{\preccurlyeq}{\Gamma}{\jproc{b}{\ctxh[\apprs{\sigma}{P}]{O}{\Delta}{c:C}}}{\jproc{b}{\ctxh[\apprs{\sigma}{Q}]{O}{\Delta}{c:C}}}{b : B}$.
  We are done by transitivity of $\preccurlyeq$.
\end{proof}

We would like to strengthen the above equivalence to give a full precongruence on processes.
Recall that we write \(\mkcg{\relfnt{R}}\) for the contextual interior of \(\relfnt{R}\), \ie, the greatest precongruence contained in \(\relfnt{R}\).

\begin{conjecture}
  \label{conj:sill-obs-equiv/proc-equiv:1}
  Let $\preccurlyeq$ be a transitive type-indexed relation on configurations such that ${\strobsc} \subseteq {\preccurlyeq}$.
  Then $\jtrelp{\mkocg{\relcpl*{\preccurlyeq}}}{\Psi}{\Delta}{P}{Q}{a}{A}$ if and only if $\jtrelp{\mkcg{\relcpl*{\preccurlyeq}}}{\Psi}{\Delta}{P}{Q}{a}{A}$.
\end{conjecture}

A proof of \cref{conj:sill-obs-equiv/proc-equiv:1} is elusive because of the subtle interplay between the process and functional layers.
Naïvely, we proceed by induction on the process context (see the proof of \cref{theorem:sill-obs-equiv/proc-equiv:1} below).
In the case where the context's hole is within a quoted process, we need a substitution property to account for the fact that the quoted process can be copied by terms in the functional layer.\footnote{Consider, for example, the process context
  \begin{equation}
    \label{eq:sill-obs-equiv-proc-equiv:1}
    \jtypem{\cdot}{a : \Tplus \{ l : \Tu, r : \Tu \}}{\tProc{a}{\left( \lambda x : \Tproc{b : \Tu}{a : \Tu} . \tCase{a}{\Set{ l \Rightarrow \tProc{a}{x}{b} \mid r \Rightarrow \tProc{a}{x}{b} }}\right)\left(\tProc{a}{\ctxh{}{a : \Tu}{b : \Tu}}{b}\right)}{b}}{b}{\Tu}.
  \end{equation}
  At run-time, the process in the hole will get copied to both branches of the case statement.
}
Unfortunately, this substitution property is extremely similar to the result we are trying to show and we get stuck.

This difficulty is similar to the one that arises when trying to naïvely show that Abramsky's~\cite{abramsky_1990:_lazy_lambd_calcul} applicative bisimilarity for the lazy \(\lambda\)-calculus is congruence (\cf~\cite[\S~5.4]{pitts_2012:_howes_method_higher_order_languag}).
Howe's method~\cite{howe_1996:_provin_congr_bisim} was designed to address this difficulty, and it is a standard technique for showing congruence properties of program equivalences.
We have made preliminary attempts to generalize Howe's method to prove this result, but we do not present these attempts here.
This generalization is non-trivial because a single relation on processes is insufficient: we also need a relation on terms.
We must show that these two relations agree with each other, \ie, that quoting and unquoting preserves relatedness.
We must also show that all constructions in Howe's method preserve this agreement.

Abramsky used a denotational semantics to show that applicative bisimilarity is a congruence.
We conjecture that the same approach could show restricted forms of \cref{conj:sill-obs-equiv/proc-equiv:1}.
In \cite{kavanagh_2021:_commun_based_seman}, we give Polarized SILL a denotational semantics and show that, subject to a few simplify assumptions, it is sound relative to external observational congruence.
We conjecture that these results could be extended to show \cref{conj:sill-obs-equiv/proc-equiv:1} where \(\preccurlyeq\) is external observational congruence.

We suspect that a step-indexed logical relations approach to \cref{conj:sill-obs-equiv/proc-equiv:1} would be challenging.
First, we do not want to conflate non-terminating processes at a given type, but step-indexed models typically do so~\cite[138]{toninho_2015:_logic_found_session_concur_comput}.
Second, Polarized SILL's two layers have very different semantics: the process layer is specified by a multiset rewriting semantics, while the functional layer is given by an evaluation semantics.
These distinct semantics increase the technical complexity of any logical relations argument.

Despite these difficulties, we can still significantly generalize \cref{cor:sill-obs-equiv/proc-equiv:4} to handle contexts whose hole does not cross the boundary between processes and functional programs.
This restriction eliminates the difficulty posed by contexts that copy their holes.
We call these contexts ``pure process contexts''.

\begin{definition}
  \label{def:sill-obs-equiv/proc-equiv:7}
  A \defin{pure process context} $\jtypem{\Psi}{\Delta}{\ctxh{C_p}{\Gamma;\Lambda}{b:B}}{a}{A}$ is a process context with exactly one hole such that its instance of \getrn{Ct-p-hole} does not appear in a subderivation of \getrn{I-proc}.
\end{definition}

\begin{remark}
  \Cref{def:sill-obs-equiv/proc-equiv:7} imposes no other restrictions on a context's use of the functional layer.
  Pure process contexts can send and receive functional values (including quoted processes), and they can quote and unquote processes.
  They cannot, however, send functional values involving the hole, and they cannot quote or unquote processes involving the hole.
  We place no restrictions on which processes can be placed in a pure process context's hole.
\end{remark}

\begin{definition}
  \label{def:sill-obs-equiv/proc-equiv:3}
  A typed-indexed relation $\relfnt{R}$ on processes is \defin{purely process contextual} if ${\jtrelp{\relfnt{R}}{\Psi}{\Delta}{P}{Q}{a}{A}}$ implies $\jtrelp{\relfnt{R}}{\Phi}{\Lambda}{\ctxh[P]{C_p}{\Psi; \Delta}{a : A}}{\ctxh[Q]{C_p}{\Delta}{a : A}}{b}{B}$ for all pure process contexts ${\jtypem{\Phi}{\Lambda}{\ctxh{C_p}{\Psi; \Delta}{a:A}}{b}{B}}$.
  The \defin{purely process contextual interior} of a typed relation $\relfnt{R}$ on processes is the greatest purely process contextual typed relation $\mkpcg{\relfnt{R}}$ contained in $\relfnt{R}$.
\end{definition}

\begin{theorem}
  \label{theorem:sill-obs-equiv/proc-equiv:1}
  Let $\preccurlyeq$ be a transitive type-indexed relation on configurations such that ${\strobsc} \subseteq {\preccurlyeq}$.
  Then $\jtrelp{\mkocg{\relcpl*{\preccurlyeq}}}{\Psi}{\Delta}{P}{Q}{a}{A}$ if and only if $\jtrelp{\mkpcg{\relcpl*{\preccurlyeq}}}{\Psi}{\Delta}{P}{Q}{a}{A}$.
\end{theorem}

\begin{proof}
  Necessity is immediate, so we show sufficiency.
  Assume that $\jtrelp{\mkocg{\relcpl{\preccurlyeq}}}{\Psi}{\Delta}{P}{Q}{c}{C}$.
  By \cref{cor:sill-obs-equiv/proc-equiv:4}, this implies for all simply branched contexts $\jcfgt{\Lambda}{\ctxh{\mc{B}}{\Delta}{c : C}}{b : B}$ and all closing substitutions $\jcmf{\sigma}{\cdot}{\Phi}$ that
  \[
    \jtrelc{\preccurlyeq}{\Lambda}{\ctxh[\jproc{c}{\apprs{\sigma}{P}}]{\mc{B}}{\Delta}{c : C}}{\ctxh[\jproc{c}{\apprs{\sigma}{Q}}]{\mc{B}}{\Delta}{c : C}}{b : B}.
  \]
  We show the stronger property that $\jtrelp{\mkpcg{\relcpl*{\mkbcg{\preccurlyeq}}}}{\Psi}{\Delta}{P}{Q}{a}{A}$.
  This means that we must show for all pure process contexts ${\jtypem{\Gamma}{\Phi}{\ctxh{C_p}{\Psi; \Delta}{a:A}}{b}{B}}$ and all $\jcmf{\sigma}{\cdot}{\Gamma}$ that
  \[
    \jtrelc{\mkbcg{\preccurlyeq}}{\Phi}{%
      \jproc{b}{\apprs{\sigma}{\left(\ctxh[P]{C_p}{\Psi; \Delta}{a : A}\right)}}
    }{%
      \jproc{b}{\apprs{\sigma}{\left(\ctxh[Q]{C_p}{\Psi; \Delta}{a : A}\right)}}
    }{b : B}.
  \]
  This in turn requires that we show for all simply branched contexts $\jcfgt{\Lambda}{\ctxh{\mc{B}}{\Phi}{b : B}}{d : D}$ that:
  \[
    \jtrelc{\preccurlyeq}{\Lambda}{%
      \ctxh[\jproc{b}{\apprs{\sigma}{\left(\ctxh[P]{C_p}{\Psi; \Delta}{a : A}\right)}}]{\mc{B}}{\Phi}{b : B}
    }{%
      \ctxh[\jproc{b}{\apprs{\sigma}{\left(\ctxh[Q]{C_p}{\Psi; \Delta}{a : A}\right)}}]{\mc{B}}{\Phi}{b : B}
    }{d : D}
  \]
  Let $\jcfgt{\Lambda}{\ctxh{\mc{B}}{\Phi}{b : B}}{d : D}$ be an arbitrary simply branched context.
  We proceed by induction on the context ${\jtypem{\cdot}{\Phi}{\ctxh{C_p}{\Psi; \Delta}{a:A}}{b}{B}}$ and give only the illustrative cases.
  \begin{proofcases}
  \item[\getrn{Ct-p-hole}] The result is immediate by assumption.
  \item[\getrn{fwdp}] This case is impossible because there is no hole, while pure process contexts have exactly one hole.
  \item[\getrn{cut}] Then $C_p$ is either $\tCut{e}{\ctxh{C'_p}{\Psi; \Delta}{a:A}}{R}$ or $\tCut{e}{L}{\ctxh{C'_p}{\Psi; \Delta}{a:A}}$ for some \( C'_p \) and $L$ or $R$.
    Assume that we fall in the first case.
    Let $e'$ be globally fresh.
    Then by \cref{eq:sill:msr-cut},
    \[
      \ctxh[
      \jproc{b}{\apprs{\sigma}{\ctxh[P]{C_p}{\Psi; \Delta}{a : A}}}
      ]{\mc{B}}{\Phi}{b : B}
      \msstep
      \ctxh[
      \jproc{e'}{\subst{e'}{e}{\left(\apprs{\sigma}{\left(\ctxh[P]{C'_p}{\Psi; \Delta}{a : A}\right)}\right)}},
      \jproc{b}{\subst{e'}{e}{\left(\apprs{\sigma}{R}\right)}}
      ]{\mc{B}}{\Phi}{b : B}.
    \]
    By \cref{prop:sill-obs-equiv/observ-comm:8},
    \[
      \jtrelc{\strobsc}{\Lambda}{%
        \ctxh[
        \jproc{b}{\apprs{\sigma}{\ctxh[P]{C_p}{\Psi; \Delta}{a : A}}}
        ]{\mc{B}}{\Phi}{b : B}
      }{%
        \ctxh[
        \jproc{e'}{\subst{e'}{e}{\left(\apprs{\sigma}{\left(\ctxh[P]{C'_p}{\Psi; \Delta}{a : A}\right)}\right)}},
        \jproc{b}{\subst{e'}{e}{\left(\apprs{\sigma}{R}\right)}}
        ]{\mc{B}}{\Phi}{b : B}
      }{d : D}.
    \]
    We recognize the right side as $\ctxh[\jproc{e'}{\subst{e'}{e}{(\apprs{\sigma}{(\ctxh[P]{C'_p}{\Psi; \Delta}{a : A})})}}]{\mc{B}'}{}{}$ where $\mc{B}'$ is the simply branched context $\ctxh[ \ctxh{}{\Delta}{a : A}, \jproc{b}{\subst{e'}{e}{R}} ]{\mc{B}}{\Phi}{b : B}$.
    Analogously,
    \begin{gather*}
      \jtrelc{\strobsc}{\Lambda}{%
        \ctxh[
        \jproc{b}{\apprs{\sigma}{\ctxh[Q]{C_p}{\Psi; \Delta}{a : A}}}
        ]{\mc{B}}{\Phi}{b : B}
      }{%
        \ctxh[
        \jproc{e'}{\subst{e'}{e}{\left(\apprs{\sigma}{\left(\ctxh[Q]{C'_p}{\Psi; \Delta}{a : A}\right)}\right)}}
        ]{\mc{B}'}{}{}
      }{d : D}.
    \end{gather*}
    By the induction hypothesis and the fact that $\preccurlyeq$ is type-indexed, so closed under renamings of channels,
    \[
      \jtrelc{\preccurlyeq}{\Lambda}{%
        \ctxh[
        \jproc{e'}{\subst{e'}{e}{\left(\apprs{\sigma}{\left(\ctxh[P]{C'_p}{\Psi; \Delta}{a : A}\right)}\right)}}
        ]{\mc{B}'}{}{}
      }{%
        \ctxh[
        \jproc{e'}{\subst{e'}{e}{\left(\apprs{\sigma}{\left(\ctxh[Q]{C'_p}{\Psi; \Delta}{a : A}\right)}\right)}}
        ]{\mc{B}'}{}{}
      }{d : D}.
    \]
    We are done by transitivity, the assumption that ${\strobsc} \subseteq {\preccurlyeq}$, and symmetry of $\strobsc$.
    The case of $\tCut{e}{L}{\ctxh{C'_p}{\Psi; \Delta}{a:A}}$ for some \( C'_p \) and $L$ is analogous.
  \item[\getrn{L-tplus}]
    Then $C_p$ is of the form $\tCase{e}{\{ l \Rightarrow P_l \}_{l \in L}}$ where $P_k = \ctxh{C'_p}{\Psi; \Delta}{a : A}$ for some unique $k \in L$.
    We observe that
    \[
      \ctxh[\jproc{b}{\apprs{\sigma}{\left(\ctxh[P]{C_p}{\Psi; \Delta}{a : A}\right)}}]{\mc{B}}{\Phi}{b : B}
      \msstep
      \ctxh[\jmsg{e}{\mSendLP{e}{l}{e'}}, \jproc{b}{\apprs{\sigma}{\left(\ctxh[P]{C_p}{\Psi; \Delta}{a : A}\right)}}]{\mc{B}'}{\Phi'}{b : B}
    \]
    if and only if
    \[
      \ctxh[\jproc{b}{\apprs{\sigma}{\left(\ctxh[Q]{C_p}{\Psi; \Delta}{a : A}\right)}}]{\mc{B}}{\Phi}{b : B}
      \msstep
      \ctxh[\jmsg{e}{\mSendLP{e}{l}{e'}}, \jproc{b}{\apprs{\sigma}{\left(\ctxh[Q]{C_p}{\Psi; \Delta}{a : A}\right)}}]{\mc{B}'}{\Phi'}{b : B}.
    \]
    If this is the case and $l \neq k$, then both
    \begin{gather*}
      \ctxh[\jproc{b}{\apprs{\sigma}{\left(\ctxh[P]{C_p}{\Psi; \Delta}{a : A}\right)}}]{\mc{B}}{\Phi}{b : B}
      \msstep
      \ctxh[\jproc{b}{\subst{e'}{e}{\left(\apprs{\sigma}{P_l}\right)}}]{\mc{B}'}{\Phi'}{b : B},\\
      \ctxh[\jproc{b}{\apprs{\sigma}{\left(\ctxh[Q]{C_p}{\Psi; \Delta}{a : A}\right)}}]{\mc{B}}{\Phi}{b : B}
      \msstep
      \ctxh[\jproc{b}{\subst{e'}{e}{\left(\apprs{\sigma}{P_l}\right)}}]{\mc{B}'}{\Phi'}{b : B}.
    \end{gather*}
    We are done by \cref{prop:sill-obs-equiv/observ-comm:8}, transitivity, and the inclusion ${\strobsc} \subseteq {\preccurlyeq}$.
    If $l = k$, then
    \begin{gather*}
      \ctxh[\jproc{b}{\apprs{\sigma}{\left(\ctxh[P]{C_p}{\Psi; \Delta}{a : A}\right)}}]{\mc{B}}{\Phi}{b : B}
      \msstep
      \ctxh[\jproc{b}{\subst{e'}{e}{\left(\apprs{\sigma}{\left(\ctxh[P]{C'_p}{\Psi; \Delta}{a : A}\right)}\right)}}]{\mc{B}'}{\Phi'}{b : B},\\
      \ctxh[\jproc{b}{\apprs{\sigma}{\left(\ctxh[Q]{C_p}{\Psi; \Delta}{a : A}\right)}}]{\mc{B}}{\Phi}{b : B}
      \msstep
      \ctxh[\jproc{b}{\subst{e'}{e}{\left(\apprs{\sigma}{\left(\ctxh[Q]{C'_p}{\Psi; \Delta}{a : A}\right)}\right)}}]{\mc{B}'}{\Phi'}{b : B}.
    \end{gather*}
    By the induction hypothesis and the fact that $\preccurlyeq$ is type-indexed, so closed under renamings of channels,
    \[
      \jtrelc{\preccurlyeq}{\Lambda}{%
        \ctxh[\jproc{b}{\subst{e'}{e}{\left(\apprs{\sigma}{\left(\ctxh[P]{C'_p}{\Psi; \Delta}{a : A}\right)}\right)}}]{\mc{B}'}{\Phi'}{b : B}
      }{%
        \ctxh[\jproc{b}{\subst{e'}{e}{\left(\apprs{\sigma}{\left(\ctxh[Q]{C'_p}{\Psi; \Delta}{a : A}\right)}\right)}}]{\mc{B}'}{\Phi'}{b : B}
      }{d : D}.
    \]
    We are done by \cref{prop:sill-obs-equiv/observ-comm:8}, transitivity, the assumption that ${\strobsc} \subseteq {\preccurlyeq}$, and symmetry of $\strobsc$.

    Finally, assume that in no fair trace do we get a message fact $\jmsg{e}{\mSendLP{e}{l}{e'}}$.
    Then by case analysis on the rules, no rule ever applies to $\jproc{b}{\apprs{\sigma}{\left(\ctxh[P]{C_p}{\Psi; \Delta}{a : A}\right)}}$ or $\jproc{b}{\apprs{\sigma}{\left(\ctxh[P]{C_p}{\Psi; \Delta}{a : A}\right)}}$.
    It follows that $\ctxh[\jproc{b}{\apprs{\sigma}{\left(\ctxh[P]{C_p}{\Psi; \Delta}{a : A}\right)}}]{\mc{B}}{\Phi}{b : B}$ and $\ctxh[\jproc{b}{\apprs{\sigma}{\left(\ctxh[Q]{C_p}{\Psi; \Delta}{a : A}\right)}}]{\mc{B}}{\Phi}{b : B}$ have the same traces (modulo the unused process fact), so the same observable messages and observed communications.
    This completes the case.
  \item[\getrn{L-tand}] Then $C_p$ is of the form $\tRecvV{x}{e}{\ctxh{C'_p}{\Psi; \Delta}{a : A}}$.
    We observe that
    \[
      \ctxh[\jproc{b}{\apprs{\sigma}{\left(\ctxh[P]{C_p}{\Psi; \Delta}{a : A}\right)}}]{\mc{B}}{\Phi}{b : B}
      \msstep
      \ctxh[\jmsg{e}{\mSendVP{e}{v}{e'}}, \jproc{b}{\apprs{\sigma}{\left(\ctxh[P]{C_p}{\Psi; \Delta}{a : A}\right)}}]{\mc{B}'}{\Phi'}{b : B}
    \]
    if and only if
    \[
      \ctxh[\jproc{b}{\apprs{\sigma}{\left(\ctxh[P]{C_p}{\Psi; \Delta}{a : A}\right)}}]{\mc{B}}{\Phi}{b : B}
      \msstep
      \ctxh[\jmsg{e}{\mSendVP{e}{v}{e'}}, \jproc{b}{\apprs{\sigma}{\left(\ctxh[Q]{C_p}{\Psi; \Delta}{a : A}\right)}}]{\mc{B}'}{\Phi'}{b : B}.
    \]
    If this is the case, then both
    \begin{align}
      \ctxh[\jproc{b}{\apprs{\sigma}{\left(\ctxh[P]{C_p}{\Psi; \Delta}{a : A}\right)}}]{\mc{B}}{\Phi}{b : B}
      &\msstep
      \ctxh[\jproc{b}{\subst{e',v}{e,x}{\left(\apprs{\sigma}{\left(\ctxh[P]{C'_p}{\Psi; \Delta}{a : A}\right)}\right)}}]{\mc{B}'}{\Phi'}{b : B},\label{eq:sill-obs-equiv/proc-equiv:1}\\
      \ctxh[\jproc{b}{\apprs{\sigma}{\left(\ctxh[Q]{C_p}{\Psi; \Delta}{a : A}\right)}}]{\mc{B}}{\Phi}{b : B}
      &\msstep
      \ctxh[\jproc{b}{\subst{e',v}{e,x}{\left(\apprs{\sigma}{\left(\ctxh[Q]{C'_p}{\Psi; \Delta}{a : A}\right)}\right)}}]{\mc{B}'}{\Phi'}{b : B}.\label{eq:sill-obs-equiv/proc-equiv:2}
    \end{align}
    We remark that the composition $\subst{v}{x}{} \circ \sigma$ determines a closing substitution $\jcmf{\sigma'}{\cdot}{\Gamma,x : \tau}$ for $\ctxh{C'_p}{\Psi; \Delta}{a : A}$.
    So the right sides of \labelcref{eq:sill-obs-equiv/proc-equiv:1,eq:sill-obs-equiv/proc-equiv:2} are respectively equal to:
    \begin{align*}
      &\ctxh[\jproc{b}{\subst{e'}{e}{\left(\apprs{\sigma'}{\left(\ctxh[P]{C'_p}{\Psi; \Delta}{a : A}\right)}\right)}}]{\mc{B}'}{\Phi'}{b : B},\\
      &\ctxh[\jproc{b}{\subst{e'}{e}{\left(\apprs{\sigma'}{\left(\ctxh[Q]{C'_p}{\Psi; \Delta}{a : A}\right)}\right)}}]{\mc{B}'}{\Phi'}{b : B}.
    \end{align*}
    By the induction hypothesis and the fact that $\preccurlyeq$ is type-indexed, so closed under renamings of channels,
    \[
      \jtrelc{\preccurlyeq}{\Lambda}{%
        \ctxh[\jproc{b}{\subst{e'}{e}{\left(\apprs{\sigma'}{\left(\ctxh[P]{C'_p}{\Psi; \Delta}{a : A}\right)}\right)}}]{\mc{B}'}{\Phi'}{b : B}
      }{%
        \ctxh[\jproc{b}{\subst{e'}{e}{\left(\apprs{\sigma'}{\left(\ctxh[Q]{C'_p}{\Psi; \Delta}{a : A}\right)}\right)}}]{\mc{B}'}{\Phi'}{b : B}
      }{d : D}.
    \]
    We are done by \cref{prop:sill-obs-equiv/observ-comm:8}, transitivity, the assumption that ${\strobsc} \subseteq {\preccurlyeq}$, and symmetry of $\strobsc$.
    Finally, assume that in no fair trace do we get a message fact $\jmsg{e}{\mSendVP{e}{v}{d}}$.
    Then the analysis is the same as in the previous case.
  \end{proofcases}
  All other cases are analogous to the above.
  Explicitly, \getrn{fwdn} and \getrn{R-tu} are analogous to \getrn{fwdp}.
  The definition of pure process context implies that there is no hole in the case of \getrn{E-proc}, so it too is analogous to \getrn{fwdp}.
  All of the cases in which the hole sends a message are analogous to \getrn{cut}.
  In particular, the definition of pure process context guarantees that the hole does not appear in the sent term in the cases of \getrn{R-tand} and \getrn{L-timp}.
  All of the cases in which the hole receives a message are analogous to \getrn{L-tplus} or \getrn{L-tand}, depending on whether or not the message carries a functional value.
\end{proof}

\Cref{cor:sill-obs-equiv/observ-comm-equiv:1} summarizes our results.
This \namecref{cor:sill-obs-equiv/observ-comm-equiv:1} lets us use the precongruences of \cref{fig:sill-obs-equiv/extern-observ:1} to reason about processes.
Intuitively, \cref{cor:sill-obs-equiv/observ-comm-equiv:1} states that, though these precongruences were defined on configurations, they lift to relations on open processes, and that these liftings are contextual.
Put differently, it states that these liftings are certain restricted forms of (pre)congruences.
As a result, we can use \cref{cor:sill-obs-equiv/observ-comm-equiv:1} to reason modularly about processes: it lets us ``replace equals by equals'' in a large class of contexts.

\begin{corollary}
  \label{cor:sill-obs-equiv/observ-comm-equiv:1}
  Let $\preccurlyeq$ be a transitive, type-indexed precongruence on configurations such that ${\strobsc} \subseteq {\preccurlyeq}$.
  The following are equivalent:
  \begin{enumerate}
  \item $\jtrelp{\relcpl{\preccurlyeq}}{\Psi}{\Delta}{P}{Q}{a}{A}$;
  \item $\jtrelp{\mkocg{\relcpl{\preccurlyeq}}}{\Psi}{\Delta}{P}{Q}{a}{A}$, where \(\mkocg{\relcpl{\preccurlyeq}}\) is the observationally contextual interior of \(\relcpl{\preccurlyeq}\);
  \item $\jtrelp{\mkpcg{\relcpl{\preccurlyeq}}}{\Psi}{\Delta}{P}{Q}{a}{A}$, where \(\mkpcg{\relcpl{\preccurlyeq}}\) is the purely process contextual interior of \(\relcpl{\preccurlyeq}\).
  \end{enumerate}
\end{corollary}

\begin{proof}
  Observe that ${\mkbcg{\left(\mkcg{\left(\preccurlyeq\right)}\right)}} = {\mkcg{\preccurlyeq}}$.
  The result follows from \cref{cor:sill-obs-equiv/proc-equiv:4,theorem:sill-obs-equiv/proc-equiv:1}.
\end{proof}

We consider a few examples to illustrate how we can (and cannot) use \cref{cor:sill-obs-equiv/observ-comm-equiv:1} to reason modularly about processes.

\begin{example}
  \label{ex:sill-obs-equiv-proc-equiv:1}
  We revisit \cref{ex:sill-obs-equiv-main-equiv:3,ex:sill-obs-equiv-main-equiv:4} and use \cref{cor:sill-obs-equiv/observ-comm-equiv:1} to show that:
  \begin{equation}
    \label{eq:sill-obs-equiv-proc-equiv:2}
    \jtrelp{\relcpl{\strobsc}}{\cdot}{a : \Tamp \{ l : \Tus{\Tu} \}}{\tSendL{a}{l}{\tSendL{b}{r}{\tSendS{a}{\tFwdP{a}{b}}}}}{\tSendL{b}{r}{\tSendL{a}{l}{\tSendS{a}{\tWait{a}{\tClose b}}}}}{b}{\Tplus \{r : \Tu\}}.
  \end{equation}
  By an argument analogous to the one in \cref{ex:sill-obs-equiv-main-equiv:4},
  \[
    \jtrelc{\strobsc}{a : \Tamp \{ l : \Tus{\Tu} \}}{\jproc{b}{\tSendL{a}{l}{\tSendL{b}{r}{\tSendS{a}{\tFwdP{a}{b}}}}}}{\jproc{b}{\tSendL{b}{r}{\tSendL{a}{l}{\tSendS{a}{\tFwdP{a}{b}}}}}}{b : \Tplus \{r : \Tu\}}.
  \]
  It follows that \( \jtrelp{\relcpl{\strobsc}}{\cdot}{a : \Tamp \{ l : \Tus{\Tu} \}}{\tSendL{a}{l}{\tSendL{b}{r}{\tSendS{a}{\tFwdP{a}{b}}}}}{\tSendL{b}{r}{\tSendL{a}{l}{\tSendS{a}{\tFwdP{a}{b}}}}}{b}{\Tplus \{r : \Tu\}}\).
  In \cref{ex:sill-obs-equiv-main-equiv:3}, we showed that \( \jtrelc{\strobsc}{a : \Tu}{\jproc{b}{\tFwdP{a}{b}}}{\jproc{b}{\tWait{a}{\tClose{b}}}}{b : \Tu} \).
  This implies by definition of \(\relcpl{\strobsc}\) that \(\jtrelp{\relcpl{\strobsc}}{\cdot}{a : \Tu}{\tFwdP{a}{b}}{\tWait{a}{\tClose{b}}}{b}{\Tu}\).
  The context \(\tSendL{b}{r}{\tSendL{a}{l}{\tSendS{a}{\ctxh{}{a:\Tu}{b:\Tu}}}}\) is a pure process context, so by \cref{cor:sill-obs-equiv/observ-comm-equiv:1} and definition of \(\mkpcg{\relcpl{\strobsc}}\):
  \[
    \jtrelp{\relcpl{\strobsc}}{\cdot}{a : \Tamp \{ l : \Tus{\Tu} \}}{\tSendL{b}{r}{\tSendL{a}{l}{\tSendS{a}{\tFwdP{a}{b}}}}}{\tSendL{b}{r}{\tSendL{a}{l}{\tSendS{a}{\tWait{a}{\tClose b}}}}}{b}{\Tplus \{r : \Tu\}}.
  \]
  \Cref{eq:sill-obs-equiv-proc-equiv:2} now follows by transitivity of \(\relcpl{\strobsc}\).
\end{example}

\begin{example}
  \label{ex:sill-obs-equiv-proc-equiv:2}
  In \cref{ex:sill-obs-equiv-proc-equiv:1}, we used contextuality and \cref{cor:sill-obs-equiv/observ-comm-equiv:1} to show that two processes were equivalent.
  In this example, we emphasize that we cannot rely on full contextuality, but only on observational contextuality or on pure process contextuality.
  Let \(\jtypem{\cdot}{a : \Tplus \{ l : \Tu, r : \Tu \}}{\ctxh{C}{a : \Tu}{b : \Tu}}{b}{\Tu}\) be given by \cref{eq:sill-obs-equiv-proc-equiv:1} on \cpageref{eq:sill-obs-equiv-proc-equiv:1}.
  Though \( \jtrelp{\relcpl{\strobsc}}{\cdot}{a : \Tu}{\tFwdP{a}{b}}{\tWait{a}{\tClose{b}}}{b}{\Tu} \), \cref{cor:sill-obs-equiv/observ-comm-equiv:1} is not strong enough to show directly that \( \jtrelp{\relcpl{\strobsc}}{\cdot}{a : \Tplus \{ l : \Tu, r : \Tu \}}{\ctxh[\tFwdP{a}{b}]{C}{a : \Tu}{b : \Tu}}{\ctxh[\tWait{a}{\tClose{b}}]{C}{a : \Tu}{b : \Tu}}{b}{\Tu}\).
  This is because \(\ctxh{C}{a : \Tu}{b : \Tu}\) is not a pure process context (its hole occurs in a subderivation of \getrn{I-proc}), and \cref{cor:sill-obs-equiv/observ-comm-equiv:1} only applies to observation contexts or pure process contexts.
\end{example}

\section{Related Work}
\label{sec:ssos-fairness:related-work}

\subsection{Multiset Rewriting Systems}

Multiset rewriting systems with existential quantification were first introduced by \textcite{cervesato_1999:_meta_notat_protoc_analy}.
They were used to study security protocols and were identified as the first-order Horn fragment of linear logic.
Since, MRSs have modelled other security protocols, and strand spaces~\cite{cervesato_2000:_inter_stran_linear_logic,cervesato_2005:_compar_between_stran}.
\Textcite{cervesato_scedrov_2009:_relat_state_based} studied the relationship between MRSs and linear logic.
These works do not explore~fairness.

Substructural operational semantics~\cite{simmons_2012:_subst_logic_specif} based on multiset rewriting are widely used to specify the operational behaviour of session-typed languages arising from proofs-as-processes interpretations of linear logic and adjoint logic.
Examples include functional languages with session-typed concurrency~\cite{toninho_2013:_higher_order_proces_funct_session}, languages with run-time monitoring~\cite{gommerstadt_2018:_session_typed_concur_contr}, message-passing interpretations of adjoint logic~\cite{pruiksma_pfenning_2019:_messag_passin_inter_adjoin_logic}, and session-typed languages with sharing~\cite{balzer_pfenning_2017:_manif_sharin_with_session_types}.

\subsection{Fairness}

Fairness finds its roots in work of \textcite{lamport_1977:_provin_correc_multip_progr,park_1980:_seman_fair_paral}.
\Textcite{lamport_1977:_provin_correc_multip_progr} studied the correctness of multiprocessor programs.
He described fairness constraints on schedulers using clocks, where process clocks were assumed to advance a certain amount in every period of real time.
Instead of using clocks, \textcite{park_1980:_seman_fair_paral} defined fairness using a ``fair merge'' operator on traces.
Early work~\cite{lehmann_1981:_impar_justic_fairn,park_1980:_seman_fair_paral} on fairness was concerned with fair termination.
\Textcite{francez_1986:_fairn} gave a comprehensive analysis of fair termination.
Weak and strong fairness were introduced by \textcite{apt_olderog_1982:_proof_rules_dealin_with_fairn, park_1982:_predic_trans_weak_fair_iterat} in the context of do-od languages.
Fairness was subsequently adapted to process calculi, \eg, by \textcite{grumberg_1984:_fair_termin_commun_proces} for CSP-like languages and by \textcite{costa_stirling_1987:_weak_stron_fairn_ccs} for Milner's CCS.
\Textcite{hennessy_1987:_algeb_theor_fair} studied fairness in the setting of asynchronous communicating processes.
\Textcite{leu_1988:_inter_among_various} introduced fairness for Petri nets.
\Textcite{kwiatkowska_1989:_survey_fairn_notion} surveys these notions and others, and gives a taxonomy of varieties of fairness.
Bounded fairness~\cite{dershowitz_2003:_bound_fairn} places bounds on how long we must wait before an event occurs.

Fair scheduling algorithms are an active research area in the programming languages and systems communities.
For example, \textcite{sistla_1983:_theor_issues_desig} studied the complexity of fair scheduling algorithms.
\Textcite{henry_1984:_unix_system} gave a fair scheduler for processes on UNIX systems.
\Textcite{muller_2019:_fairn_respon_paral, muller_2018:_compet_paral} studied fair scheduling for interactive computation and in the presence of priorities.
\Textcite{lahav_2020:_makin_weak_memor_model_fair} gave an account of process fairness under weak memory models.

Several applications of fairness to session-typed processes have been studied.
\Textcite{glabbeek_2021:_assum_just_enoug} develop a session-typed system that is complete for lock-freedom when assuming justness, a minimal fairness assumption.
Their calculus supports internal and external choice and guarded reduction.
Its semantics is given by a labelled transition system.
\Textcite{ciccone_padovani_2021:_infer_system_with} used ``generalized inference systems'' to specify fair compliance and fair subtyping for session-typed processes.
Fair compliance is a liveness property that ensures that clients interacting with servers will always reach a configuration where the client is satisfied and the server has not failed.
Fair subtyping is a refinement notion for session types that preserves fair compliance.

\subsection{Session-Typed Languages}

\Textcite{honda_1993:_types_dyadic_inter,takeuchi_1994:_inter_based_languag} introduced session types to describe sessions of interaction.
\Textcite{caires_pfenning_2010:_session_types_intuit_linear_propos} observed a proofs-as-programs correspondence between the session-typed $\pi$-calculus and intuitionistic linear logic, where the \rn{Cut} rule captures process communication.
\Textcite{toninho_2013:_higher_order_proces_funct_session} built on this correspondence and introduced SILL's monadic integration between functional and synchronous message-passing programming.
They specified SILL's operational behaviour using a substructural operational semantics.
\Textcite{gay_vasconcelos_2009:_linear_type_theor} introduced asynchronous communication for session-typed languages.
They used an operational semantics and buffers to model asynchronicity.
\Textcite{pfenning_griffith_2015:_polar_subst_session_types} observed that the polarity of a type determines the direction of communication along a channel.
They observed that synchronous communication can be encoded in an asynchronous setting using explicit shift operators.
They gave a computational interpretation to polarized adjoint logic.
In this interpretation, linear propositions, affine propositions, and unrestricted propositions correspond to different modes in which resources can be used.

There are several process calculi and session-typed programming languages that are closely related to Polarized SILL, and to which we conjecture our techniques could be extended.
\Textcite{wadler_2014:_propos_as_session} introduced ``Classical Processes'' (CP), a proofs-as-programs interpretation of classical linear logic that builds on the ideas of \textcite{caires_pfenning_2010:_session_types_intuit_linear_propos}.
CP supports replication but not recursion.
Though CP does not natively support functional programming, Wadler gives a translation for GV, a linear functional language with pairs but no recursion, into CP.
In contrast, Polarized SILL uniformly integrates functional programming and message-passing concurrency.
CP has a synchronous communication semantics and does not have an explicit treatment of polarities.
Polarized SILL has an asynchronous communication semantics, and synchronous communication is encoded using polarity shifts, even though we do not detail this construction here.
\Textcite{montesi_2017:_class_higher_order_proces} extended CP to support process mobility as in the higher-order \(\pi\)-calculus.

\Textcite{kokke_2019:_better_late_than_never} introduced ``hypersequent classical processes'' (HCP).
HCP is a revised proofs-as-processes interpretation between classical linear logic and the $\pi$-calculus.
Building on Atkey's~\cite{atkey_2017:_obser_commun_seman_class_proces} semantics for CP, they gave HCP a denotational semantics using Brzozowski derivatives~\cite{brzozowski_1964:_deriv_regul_expres}.
HCP does not include recursion, shifts, or functional value transmission.
\Textcite{fowler_2021:_separ_session_smoot} introduce ``Hypersequent GV'' (HGV), a calculus for functional programming with session types that enjoys deadlock-freedom, confluence, and strong normalization.
They give translations between HGV and HCP that preserve and reflect reduction.

\Textcite{montesi_peressotti_2021:_linear_logic_their_metat} gave a design recipe for developing session-typed process calculi based on linear logic.
This recipe unifies prior proofs-as-processes interpretations, gives a logical reconstruction of session fidelity, and elucidates connections between the metatheory of derivations and behavioural theory.

\Textcite{gommerstadt_2018:_session_typed_concur_contr} introduced run-time monitors for a dependent version of Polarized SILL.
Our type system for configurations is inspired by theirs~\cite[786]{gommerstadt_2018:_session_typed_concur_contr}.

\Textcite{pruiksma_pfenning_2021:_messag_passin_inter_adjoin_logic} gave a message passing interpretation to adjoint logic.
It supports richer communication topologies than Polarized SILL.
For example, it supports \emph{multicast}, \emph{replicable services}, and \emph{cancellation}.
Its operational semantics is specified by a multiset rewriting system.
It enjoys session fidelity and deadlock-freedom.

\Textcite{toninho_yoshida_2018:_polym_session_funct} gave a fully abstract and mutually inverse encoding between the polymorphic session-typed \(\pi\)-calculus and a linear formulation of System F.
Using this encoding, they use properties of their polymorphic \(\lambda\)-calculus to deduce properties of their polymorphic session-typed \(\pi\)-calculus, including the existence of inductive and coinductive sessions.

\subsection{Equivalence for Session-Typed Processes}

Various notions of program equivalence have been introduced for session-typed processes.
\Textcite{castellan_yoshida_2019:_two_sides_same_coin} gave a game semantics interpretation of the session $\pi$-calculus with recursion.
They showed that it is fully abstract relative to barbed congruence.
\Textcite{perez_2012:_linear_logic_relat,perez_2014:_linear_logic_relat} introduced linear logical relations for session-typed processes.
The denotational semantics \textcite{kokke_2019:_better_late_than_never} gave for hypersequent classical processes is fully abstract relative to barbed congruence.
It is unclear how to extend these approaches to support functional value transmission, general recursion, or polarity shifts.

Our observed communication semantics was partly motivated by efforts to relate a denotational semantics for Polarized SILL \cite{kavanagh_2021:_commun_based_seman} to Polarized SILL's substructural operational semantics.
There, processes denote continuous functions between domains of session-typed communications, and process composition is given by a trace operator.
Subject to certain simplifying assumptions, this denotational semantics is sound relative to barbed congruence and external observational congruence.

Barbed equivalences were introduced by \textcite{milner_sangiorgi_1992:_barbed_bisim} for CCS.
\Textcite{sangiorgi_1992:_expres_mobil_proces_algeb} introduced barbed simulations for the $\pi$-calculus~\cite{milner_1992:_calcul_mobil_proces_i,milner_1992:_calcul_mobil_proces_ii}, HO$\pi$, and CHOCS~\cite{thomsen_1993:_plain_chocs_secon}.
Barbed equivalences have since been generalized to a variety of calculi.

Defining process equivalence through experimentation dates at least as far back as the early 1980s~\cite{milner_1980:_calcul_commun_system,darondeau_1982:_enlar_defin_compl,hennessy_1983:_synch_async_exper_proces}.
\Textcite{milner_1980:_calcul_commun_system,denicola_hennessy_1984:_testin_equiv_proces} experimented on programs using programs.
\Textcite{denicola_hennessy_1984:_testin_equiv_proces,hennessy_1983:_synch_async_exper_proces} used a state-based approach, and experiments were successful if they reached a ``success'' state.
In contrast, our experiments on processes are ``successful'' if their observed communications are suitably related by a preorder or equivalence relation.
Our typed setting also simplifies our experimentation.
In particular, we do not need to worry about ``useless'' observers that cannot interact with the process being tested (\cf~\cite[\S~2.3.2]{denicola_1985:_testin_equiv_fully}): session types ensure that observers and observed processes can communicate.
Because fairness ensures a unique observation for each process, we do not need to worry about nondeterministic observations.
In particular, this means that the ``may'' and ``must'' preorders of \textcite{hennessy_1983:_synch_async_exper_proces, denicola_1985:_testin_equiv_fully} coincide in our setting.

\section{Summary and Future Work}
\label{sec:conclusion}

We gave the first analysis of fairness for multiset rewriting systems, and we developed communication-based techniques for reasoning about Polarized SILL and its programs.
We summarize these contributions and we discuss their potential applications to future research.
We also discuss open problems that are directly related to our contributions.

In \cref{sec:sill-obs-equiv:observ-comm}, we developed an observed communication semantics for Polarized SILL.
We defined the meaning of a session type to be the set of communications it allows, and we showed that this set could be endowed with a notion of approximation.
Then, we showed how to observe the communications sent by processes and configurations in the course of an execution.
Importantly, we showed that all fair executions of configurations resulted in the same observed communications.
This fact reflects the confluence property satisfied by Polarized SILL.
We also showed that the type of an observed communication agrees with the type of the channel on which it was observed, giving a semantic account of session fidelity.

We introduced a framework for extensional, observational notions of equivalence for Polarized SILL in \cref{cha:sill-obs-equiv}.
It was inspired by the ``testing equivalences'' framework of De Nicola and Hennessy~\cite{denicola_hennessy_1984:_testin_equiv_proces,hennessy_1983:_synch_async_exper_proces,denicola_1985:_testin_equiv_fully}.
Both frameworks are similar in that they deem processes to be equivalent whenever they are indistinguishable through experimentation.
The frameworks differ, however, in the notion of experimentation.
Subjecting processes to classical experiments could potentially result in a ``success'' state, and two processes were equivalent if they succeeded the same experiments.
Instead of defining experimental indistinguishability using observed states, we defined it in terms of observed communications.
In particular, our experiments communicated with processes (strictly speaking, with configurations of processes), and we deemed processes to be equivalent if we could not observe any differences in their communications.
We had a certain latitude in choosing which channels to observe, and this latitude resulted in different notions of process equivalence.
One of these, ``external observational equivalence'', coincided with barbed congruence.
We showed how to lift observational congruences on configurations to restricted forms of congruences on~processes.

For our observed communication semantics to be well-defined and for it to capture our semantic intuitions, we had to first develop fairness for multiset rewriting systems in \cref{sec:three-vari-fairn,sec:ssos-fairness:prop-fair-trac}.
We discovered three independent varieties of fairness---rule fairness, fact fairness, and instantiation fairness---and saw how each subdivided along the axis of weak and strong fairness.
These notions were all subsumed by a particularly strong form of fairness called \emph{über fairness}.
We studied properties of fair traces, constructed a scheduler, and gave sufficient conditions for multiset rewriting systems to have fair traces.
We observed that under certain conditions, all varieties of fairness coincided.
We introduced ``union equivalence'' for traces and studied the effects of permutations on fairness.
In particular, we showed that subject to certain conditions, fairness was preserved by permutation and that all fair executions were permutations of each other.

There are many open questions related to the above contributions.
We highlight some of the most important:
\begin{enumerate}
\item \emph{How do we lift observational congruences on configurations to (full) congruences on processes?}
  We showed in \cref{sec:sill-obs-equiv:proc-equiv} that observational congruences on configurations induced certain restricted classes of congruences on processes.
  However, the subtle interplay between the process and functional layers prevented us from showing that they induced full congruences.
  We conjecture that we could adapt Howe's method~\cite{howe_1996:_provin_congr_bisim} to show this result.

  This question has important implications for practical applications of our observational congruences.
  Indeed, the reason congruence relations are so sought after is that they allow us to replace equals by equals.
  If we could do so, then we could use them to reason about, \eg, program optimizations.
\item \emph{Can we use our observed communication semantics and communication-based testing equivalences to relate processes written in different languages?}
  Though our observed communication semantics is not denotational (it is not defined by induction on programs), it has a denotational flavour: it abstracts away a program's operational behaviour to define its meaning in terms of language-independent mathematical objects.
  We conjecture that these language-independent objects---observed communications---could be used to relate processes across languages.
  What properties must each language's observed communication semantics satisfy for these cross-language relations to be semantically meaningful?
\item \emph{What is the relationship between internal and external observational precongruence?}
  \Cref{cor:sill-obs-equiv/intern-observ:1} states that internal observational precongruence implies external observational precongruence, and we conjectured the converse in \cref{sec:sill-obs-equiv:proc-equiv}.
  This converse would simplify reasoning about processes, for it would let us reason directly using internal communications around the experiment's hole, instead of having to consider how they percolate through the experiment to its external channels.
  In light of \cref{theorem:sill-obs-equiv/extern-observ:2}, this converse would also imply that internal observational precongruence and barbed precongruence coincide.
\item \emph{Can we meaningfully redefine external, internal, and total observational simulation using a finer relation on functional values?}
  These three simulations were defined using the universal relation \(\unirel\) because we deemed values in Polarized SILL to be unobservable (they are either functions or quoted processes).
  This meant that values \(\lambda : \tau.x\) and \(\lambda x : \tau.\tFix{y}{y}\) were indistinguishable when sent over channels.
  We used this fact to show that internal and external observational simulation are distinct (\cref{prop:sill-obs-equiv-main-equiv:1}), and that internal observational simulation is not a precongruence (\cref{prop:sill-obs-equiv/intern-observ:2}).
  We conjecture that by using a finer relation on values, that internal and external observational simulation would coincide, and that internal observational simulation would be a precongruence.
\end{enumerate}

The unifying theme of our contributions is that we defined the meaning of processes in terms of their communications.
In doing so, we have stayed faithful to the process abstraction, \ie, to the premise that communication is the only phenomenon of processes.
A key benefit is that by defining meaning in terms of communication, we can abstract away concrete implementation or language details.
As a result, we believe our contributions can scale to handle more interesting protocols or communication patterns.
For example, we conjecture that our work could be extended to handle the dependent protocols captured by dependent session types~\cite{toninho_yoshida_2018:_depen_session_typed_proces, thiemann_vasconcelos_2019:_label_depen_session_types, toninho_2011:_depen_session_types, das_pfenning_2020:_session_types_arith}.
We also believe that our techniques can scale to handle more interesting communication patterns.
For example, we conjecture that they could be extended to handle features like \emph{multicast} (sending one message to multiple clients) and \emph{cancellation} (discarding channels without communicating on them) found in computational interpretations of adjoint logic~\cite{pruiksma_pfenning_2021:_messag_passin_inter_adjoin_logic}.

\section*{Acknowledgements}
\label{sec:ssos-fairness:concl-ackn}

The author thanks Stephen Brookes, Iliano Cervesato, Frank Pfenning, and the anonymous reviewers for their helpful comments.

This research was sponsored by Microsoft Corporation award 5005283 and by a Carnegie Mellon University School of Computer Science Presidential Fellowship.
The views and conclusions contained in this document are those of the author and should not be interpreted as representing the official policies, either expressed or implied, of any sponsoring institution or any other entity.

\setcounter{biburlnumpenalty}{100}  %
\setcounter{biburlucpenalty}{100}   %
\setcounter{biburllcpenalty}{100}   %
\emergencystretch=0.2em
\printbibliography
\emergencystretch=0em

\printunsrtglossary[type={symbols},style={topic}]

\appendix

\makeatletter
\gdef\thesection{\@Alph\c@section}%
\makeatother

\section{Complete Definition of Polarized SILL}

\subsection{Complete Listing of Typing Rules for Polarized SILL}
\label{sec:sill-background:compl-list-typing}

For ease of reference, we collect all of the rules for Polarized SILL in this appendix.

\subsubsection{Rules for Term Formation}
\label{sec:sill-den-sem:rules-term-formation}

\begin{gather*}
  \getrule{I-proc}
  \quad
  \getrule{F-var}
  \quad
  \getrule{F-fix}
  \\
  \getrule{F-fun}
  \quad
  \getrule{F-app}
\end{gather*}

\subsubsection{Rules for Process Formation}
\label{sec:sill-den-sem:rules-proc-form}

\begingroup
\allowdisplaybreaks

\begin{gather*}
  \getrule{fwdp}
  \quad
  \getrule{fwdn}
  \\
  \getrule{cut}
  \\
  \getrule{E-proc}
  \\
  \getrule{R-tu}
  \quad
  \getrule{L-tu}
  \\
  \getrule{R-tds}
  \quad
  \getrule{L-tds}
  \\
  \getrule{R-tus}
  \quad
  \getrule{L-tus}
  \\
  \begin{adjustbox}{max width=\textwidth}
    \getrule{R-tplus}
    \quad
    \getrule{L-tplus}
  \end{adjustbox}
  \\
  \begin{adjustbox}{max width=\textwidth}
    \getrule{R-tamp}
    \quad
    \getrule{L-tamp}
  \end{adjustbox}
  \\
  \getrule{R-tot}
  \quad
  \getrule{L-tot}
  \\
  \getrule{R-tlolly}
  \quad
  \getrule{L-tlolly}
  \\
  \getrule{R-tand}
  \quad
  \getrule{L-tand}
  \\
  \begin{adjustbox}{max width=\textwidth}
    \getrule{R-timp}
    \quad
    \getrule{L-timp}
  \end{adjustbox}
  \\
  \begin{adjustbox}{max width=\textwidth}
    \getrule{R-rhop}
    \,
    \getrule{L-rhop}
  \end{adjustbox}
  \\
  \begin{adjustbox}{max width=\textwidth}
    \getrule{R-rhon}
    \,
    \getrule{L-rhon}
  \end{adjustbox}
\end{gather*}

\endgroup

\subsubsection{Rules for Type Formation}
\label{sec:sill-den-sem:rules-type-formation}

\begingroup
\allowdisplaybreaks

\begin{gather*}
  \getrule{T-tu}
  \quad
  \getrule{T-var}
  \\
  \getrule{T-rhop}
  \quad
  \getrule{T-rhon}
  \\
  \getrule{T-tds}
  \quad
  \getrule{T-tus}
  \\
  \getrule{T-tplus}
  \quad
  \getrule{T-tamp}
  \\
  \getrule{T-tot}
  \quad
  \getrule{T-tlolly}
  \\
  \getrule{T-tand}
  \quad
  \getrule{T-timp}
  \\
  \getrule{T-proc}
  \quad
  \getrule{T-to}
\end{gather*}

\endgroup

\subsection{Complete Listing of Multiset-Rewriting Rules for Polarized SILL}
\label{sec:sill-background:compl-list-mult}

\newcommand{\numberthis}{\stepcounter{equation}\tag{\theequation}}

\begingroup
\allowdisplaybreaks

\begin{gather*}
  \tag{\ref{eq:sill:msr-fwdp}}
  \jmsg{a}{m^+},
\jproc{b}{\tFwdP{a}{b}}
\to
\jmsg{b}{\subst{b}{a}{m^+}}
\\
  \tag{\ref{eq:sill:msr-fwdn}}
  \jproc{b}{\tFwdN{a}{b}},
\jmsg{c}{m^-_{b,c}}
\to
\jmsg{c}{\subst{a}{b}{m^-_{b,c}}}
\\
  \tag{\ref{eq:sill:msr-cut}}
  \jproc{c}{\tCut{a}{P}{Q}} \to \exists b. \jproc{b}{\subst{b}{a}{P}}, \jproc{c}{\subst{b}{a}{Q}}
\\
  \tag{\ref{eq:sill:msr-e-qt}}
  \jeval{M}{\tProc{a}{P}{\overline{a_i}}}, \jproc{a}{\tProc{a}{M}{\overline{a_i}}} \msstep \jproc{a}{P}
\\
  \tag{\ref{eq:sill:msr-tu-r}}
  \jproc{a}{\tClose{a}}
\to
\jmsg{a}{\mClose{a}}
\\
  \tag{\ref{eq:sill:msr-tu-l}}
  \jmsg{a}{\mClose{a}},
\jproc{c}{\tWait{a}{P}}
\to
\jproc{c}{P}
\\
  \tag{\ref{eq:sill:msr-tds-r}}
  \jproc{a}{\tSendS{a}{P}}
\to
\exists d.
\jproc{d}{\subst{d}{a}{P}},
\jmsg{a}{\mSendSP{a}{d}}
\\
  \tag{\ref{eq:sill:msr-tds-l}}
  \jmsg{a}{\mSendSP{a}{d}},
\jproc{c}{\tRecvS{a}{P}}
\to
\jproc{c}{\subst{d}{a}{P}}
\\
  \label[msr]{eq:sill:msr-tus-r}\numberthis
  \jproc{a}{\tRecvS{a}{P}},
\jmsg{d}{\mSendSN{a}{d}}
\to
\jproc{d}{\subst{d}{a}{P}}
\\
  \label[msr]{eq:sill:msr-tus-l}\numberthis
  \jproc{c}{\tSendS{a}{P}}
\to
\exists d.
\jmsg{d}{\mSendSN{a}{d}},
\jproc{c}{\subst{d}{a}{P}}
\\
  \tag{\ref{eq:sill:msr-tplus-r}}
  \jproc{a}{\tSendL{a}{k}{P}}
\to {}
\exists d.
\jproc{d}{\subst{d}{a}{P}},
\jmsg{a}{\mSendLP{a}{k}{d}}
\\
  \tag{\ref{eq:sill:msr-tplus-l}}
  \jmsg{a}{\mSendLP{a}{k}{d}},
\jproc{c}{\tCase{a}{\left\{l \Rightarrow P_l\right\}_{l \in L}}}
\to
\jproc{c}{\subst{d}{a}{P_k}}
\\
  \label[msr]{eq:sill:msr-tamp-r}\numberthis
  \jproc{a}{\tCase{a}{\left\{l \Rightarrow P_l\right\}_{l \in L}}},
\jmsg{d}{\mSendLN{a}{k}{d}}
\to
\jproc{d}{\subst{d}{a}{P_k}}
\\
  \label[msr]{eq:sill:msr-tamp-l}\numberthis
  \jproc{c}{\tSendL{a}{k}{P}}
\to {}
\exists d.
\jmsg{d}{\mSendLN{a}{k}{d}},
\jproc{c}{\subst{d}{a}{P}}
\\
  \tag{\ref{eq:sill:msr-tot-r}}
  \jproc{a}{\tSendC{a}{b}{P}}
\to
\exists d.
\jproc{d}{\subst{d}{a}{P}},
\jmsg{a}{\mSendCP{a}{b}{d}}
\\
  \tag{\ref{eq:sill:msr-tot-l}}
  \jmsg{a}{\mSendCP{a}{e}{d}},
\jproc{c}{\tRecvC{b}{a}{P}}
\to
\jproc{c}{\subst{e,d}{b,a}{P}}
\\
  \tag{\ref{eq:sill:msr-tlolly-r}}
  \jproc{a}{\tRecvC{b}{a}{P}},
\jmsg{d}{\mSendCN{a}{e}{d}}
\to
\jproc{d}{\subst{e,d}{b,a}{P}}
\\
  \tag{\ref{eq:sill:msr-tlolly-l}}
  \jproc{c}{\tSendC{a}{b}{P}}
\to
\exists d.
\jmsg{d}{\mSendCN{a}{b}{d}},
\jproc{c}{\subst{d}{a}P}
\\
  \tag{\ref{eq:sill:msr-tand-r}}
  \jeval{M}{v}, \jproc{a}{\tSendV{a}{M}{P}} \to \exists d . \jproc{d}{\subst{d}{a}{P}}, \jmsg{a}{\mSendVP{a}{v}{d}}
\\
  \tag{\ref{eq:sill:msr-tand-l}}
  \jmsg{a}{\mSendVP{a}{v}{d}},\jproc{c}{\tRecvV{x}{a}{P}} \to \jproc{c}{\subst{d,v}{a,x}{P}}
\\
  \label[msr]{eq:sill:msr-timp-r}\numberthis
  \jproc{a}{\tRecvV{x}{a}{P}}, \jmsg{d}{\mSendVN{a}{v}{d}} \to \jproc{d}{\subst{d,v}{a,x}{P}}
\\
  \label[msr]{eq:sill:msr-timp-l}\numberthis
  \jeval{M}{v}, \jproc{c}{\tSendV{a}{M}{P}} \to \exists d . \jmsg{d}{\mSendVN{a}{v}{d}}, \jproc{c}{\subst{d}{a}{P}}
\\
  \tag{\ref{eq:sill:msr-rhop-r}}
  \jproc{a}{\tSendU{a}{P}}
\to
\exists d.
\jproc{d}{\subst{d}{a}{P}},
\jmsg{a}{\mSendUP{a}{d}}
\\
  \tag{\ref{eq:sill:msr-rhop-l}}
  \jmsg{a}{\mSendUP{a}{d}},
\jproc{c}{\tRecvU{a}{P}}
\to
\jproc{c}{\subst{d}{a}{P}}
\\
  \label[msr]{eq:sill:msr-rhon-r}\numberthis
  \jproc{a}{\tRecvU{a}{P}},
\jmsg{d}{\mSendUN{a}{d}}
\to
\jproc{d}{\subst{d}{a}{P}}
\\
  \label[msr]{eq:sill:msr-rhon-l}\numberthis
  \jproc{c}{\tSendU{a}{P}}
\to
\exists d.
\jmsg{d}{\mSendUN{a}{d}},
\jproc{c}{\subst{d}{a}{P}}

\end{gather*}

\endgroup

\end{document}